\documentclass[
12pt, % The default document font size, options: 10pt, 11pt, 12pt
oneside, % Two side (alternating margins) for binding by default, uncomment to switch to one side
english, % ngerman for German
onehalfspacing, % Single line spacing, alternatives: onehalfspacing or doublespacing
headsepline, % Uncomment to get a line under the header
]{MastersDoctoralThesis} % The class file specifying the document structure

\usepackage[utf8]{inputenc} % Required for inputting international characters
\usepackage[T1]{fontenc} % Output font encoding for international characters
\usepackage{amsthm}
\usepackage{amssymb}
\usepackage{graphicx} % only needed if you include graphics files other than TpX
\usepackage{hyperref} % produces nice hyperlinks in your document
\usepackage{tikz}
\usepackage{braket}
\usepackage{dsfont}
\usepackage{mathtools}
\usepackage{mathrsfs}
\newtheorem{remark}{Remark}
\newtheorem{exmp}{Example}
\newtheorem{Lemma}{Lemma}
\newtheorem{defn}{Definition}
\newtheorem{thm}{Theorem}
\newtheorem{Corollary}{Corollary}
\theoremstyle{plain}
\newtheorem{Claim}{Claim}

\usepackage[backend=bibtex,style=numeric-comp,natbib=true, sorting=none]{biblatex} % Use the bibtex backend with the authoryear citation style (which resembles APA)

\thesistitle{Identification and Estimation of Quantum Linear Input-Output Systems} % Your thesis title, this is used in the title and abstract, print it elsewhere with \ttitle
\supervisor{M\u{a}d\u{a}lin Gu\c{t}\u{a}} % Your supervisor's name, this is used in the title page, print it elsewhere with \supname
\supervisorw{Theodore Kypraios}
\examiner{} % Your examiner's name, this is not currently used anywhere in the template, print it elsewhere with \examname
\degree{Doctor of Philosophy} % Your degree name, this is used in the title page and abstract, print it elsewhere with \degreename
\author{Matthew Levitt} % Your name, this is used in the title page and abstract, print it elsewhere with \authorname
\addresses{} % Your address, this is not currently used anywhere in the template, print it elsewhere with \addressname

\subject{Mathematical Sciences} % Your subject area, this is not currently used anywhere in the template, print it elsewhere with \subjectname
 % Keywords for your thesis, this is not currently used anywhere in the template, print it elsewhere with \keywordnames
\university{\href{http://www.nottingham.ac.uk}{University of Nottingham}} % Your university's name and URL, this is used in the title page and abstract, print it elsewhere with \univname
\department{\href{http://www.nottingham.ac.uk/Mathematics/}{School of Mathematical Sciences}} % Your department's name and URL, this is used in the title page and abstract, print it elsewhere with \deptname
\group{\href{http://www.nottingham.ac.uk/mathematics/research/mathematical-physics.aspx}{Mathematical Physics}} % Your research group's name and URL, this is used in the title page, print it elsewhere with \groupname

\AtBeginDocument{
\hypersetup{pdftitle=\ttitle} % Set the PDF's title to your title
\hypersetup{pdfauthor=\authorname} % Set the PDF's author to your name
\hypersetup{pdfkeywords=\keywordnames} % Set the PDF's keywords to your keywords
}

\begin{document}

\frontmatter % Use roman page numbering style (i, ii, iii, iv...) for the pre-content pages

\pagestyle{plain} % Default to the plain heading style until the thesis style is called for the body content

%----------------------------------------------------------------------------------------
%	TITLE PAGE
%----------------------------------------------------------------------------------------

\begin{titlepage}
\begin{center}

\vspace*{.06\textheight}
{\scshape\LARGE \univname\par}\vspace{1.0cm} % University name
\textsc{\Large Doctoral Thesis}\\[0.5cm] % Thesis type

\HRule \\[0.4cm] % Horizontal line
{\LARGE \bfseries \ttitle\par}\vspace{0.4cm} % Thesis title
\HRule \\[1.1cm] % Horizontal line
 
\begin{minipage}[t]{0.4\textwidth}
\begin{flushleft} \large
\emph{Author:}\\
\href{https://arxiv.org/find/quant-ph/1/au:+Levitt_M/0/1/0/all/0/1}{\authorname} % Author name - remove the \href bracket to remove the link
\end{flushleft}
\end{minipage}
\begin{minipage}[t]{0.4\textwidth}
\begin{flushright} \large
\emph{Supervisors:} \\
\href{https://www.maths.nottingham.ac.uk/personal/pmzmig/}{\supname} \\
\href{https://www.maths.nottingham.ac.uk/personal/tk/}{\supnamew}% Supervisor name - remove the \href bracket to remove the link  
\end{flushright}
\end{minipage}\\[1.4cm]

\large \textit{A thesis submitted to the University of Nottingham\\ for the degree of\\ \textbf{\degreename}}\\[0.3cm] % University requirement text

\groupname\\\deptname\\[2cm] % Research group name and department name

\end{center}
\end{titlepage}

%----------------------------------------------------------------------------------------
%	DECLARATION PAGE
%----------------------------------------------------------------------------------------

\begin{declaration}
\addchaptertocentry{\authorshipname} % Add the declaration to the table of contents
\noindent I, \authorname, declare that this thesis titled, \enquote{\ttitle} and the work presented in it are my own. I confirm that:

\begin{itemize} 
\item This work was done wholly or mainly while in candidature for a research degree at this University.
\item Where I have consulted the published work of others, this is always clearly attributed.
\item Where I have quoted from the work of others, the source is always given. With the exception of such quotations, this thesis is entirely my own work.
\item I have acknowledged all main sources of help.
\item Where the thesis is based on work done by myself jointly with others, I have made clear exactly what was done by others and what I have contributed myself.\\
\end{itemize}
 
\noindent Signed:\\
\rule[0.5em]{25em}{0.5pt} % This prints a line for the signature
 
\noindent Date:\\
\rule[0.5em]{25em}{0.5pt} % This prints a line to write the date
\end{declaration}

\cleardoublepage

%----------------------------------------------------------------------------------------
%	QUOTATION PAGE
%----------------------------------------------------------------------------------------

%----------------------------------------------------------------------------------------
%	ABSTRACT PAGE
%----------------------------------------------------------------------------------------

\begin{abstract}

The system identification  problem is to estimate dynamical parameters from the output data, obtained by performing measurements on the output fields. We investigate system identification for quantum linear systems. Our main objectives are to address the following general problems: 
\begin{enumerate}
\item Which parameters can be identified by measuring the output? 
\item How can we construct a system realisation from sufficient input-output data? 
\item How well can we estimate the parameters governing the dynamics? 
\end{enumerate}
 We investigate these problems in two contrasting approaches; using time-dependent inputs (Sec. \ref{ghu} or time-stationary (quantum noise) inputs (Sec. \ref{stato}).

In the time-dependent approach the output fields are characterised by the \textit{transfer function}.
We show that indistinguishable minimal systems in the transfer function are related by symplectic transformations acting on the space of system modes (Ch. \ref{powers34}).
 We also present techniques enabling one  to find a physical realisation of the system from the input-output data. We present  realistic schemes for estimating \textit{passive}  quantum linear systems at the Heisenberg limit (Ch. \ref{QEEP}) under energy resource constraint.
 `Realistic' is our primary concern here, in the sense that there exists both experimentally feasible states and practical measurement choices that enable this heightened performance for all passive quantum linear systems. We consider both single parameter and multiple parameter estimation.
 
 In the stationary approach the characteristic quantity is the \textit{power spectrum}. We define the notion of global minimality for a given power spectrum, and characterise globally minimal systems as those with fully mixed stationary state (Sec. \ref{G.Ms}).  The power spectrum depends on the system parameters via the transfer function. Our main result here is that   under global minimality the power spectrum uniquely determines  the transfer function, so the system can be identified up to a symplectic transformation (see Secs. \ref{256}, \ref{dogs} \ref{JUKKA1}). We also give methods for constructing a globally minimal subsystem directly from the power spectrum (see Sec. \ref{jute}). These results hold for pure inputs, we discuss extensions to mixed inputs and the use of additional input channels; using an appropriately chosen input in the latter case ensures that the system is always globally minimal (hence identifiable).
 
Finally, we discuss a particular feedback control estimation problem in Chs. \ref{FEDERER} and \ref{DUAL}. In general, information about a parameter within a quantum linear system  may be obtained at a linear rate with respect to time (in both approaches above);  the so called \textit{standard scaling}. However, we see that when the system destabilises, so that its system matrix has eigenvalues very close to the imaginary axis, the quantum Fisher information is enhanced, to quadratic (\textit{Heisenberg}) level. We give feedback methods enabling one to destabilise the system and give adaptive procedures for realising the Heisenberg bounds.

 %We using squeezed states.  Our We employ quantum estimation theory and show for a given energy constraint that it is theoretically possible to achieve Heisenberg scaling. Furthermore, it is shown that this scaling can be achieved using realistic measurements, in particular a simple homodyne measurement. %

 %We using squeezed states.  Our We employ quantum estimation theory and show for a given energy constraint that it is theoretically possible to achieve Heisenberg scaling. Furthermore, it is shown that this scaling can be achieved using realistic measurements, in particular a simple homodyne measurement. %

\end{abstract}

%----------------------------------------------------------------------------------------
%	ACKNOWLEDGEMENTS
%----------------------------------------------------------------------------------------

\begin{constants}{lr@{${}{}$}l} % The list of physical constants is a three column table

% The \SI{}{} command is provided by the siunitx package, see its documentation for instructions on how to use it

Some of the ideas in this thesis (Chs. \ref{T.F} and \ref{powers34})  have previously appeared in the &\\following   publications: \cite{Levitt1, Levitt2}. We have two articles in preparation based on the &\\ ideas in  Ch. \ref{QEEP} (joint work with M\u{a}d\u{a}lin Gu\c{t}\u{a} and Naoki Yamamoto) and &\\ Ch. \ref{DUAL} (joint work with  M\u{a}d\u{a}lin Gu\c{t}\u{a}). In all of the above I am  the main author. \\
%Constant Name & $Symbol$ & $Constant Value$ with units\\

\end{constants}

%----------------------------------------------------------------------------------------
%	QUOTATION PAGE
%----------------------------------------------------------------------------------------

\begin{acknowledgements}
\addchaptertocentry{\acknowledgementname} % Add the acknowledgements to the table of contents

First and foremost  I would like to thank my project supervisors M\u{a}d\u{a}lin Gu\c{t}\u{a} and Theodore Kypraios for the incredible support, patience and insight they have given during my time at the University of Nottingham. I am grateful to you for  encouraging my research and for allowing me to grow as a researcher and it has been an honour to be your PhD student.

I would also like to thank my collaborators, Hendra Nurdin and Naoki Yamamoto, for providing many fruitful ideas and discussions.  Also the Quantum Correlations group  and the numerous visitors that we  have had here at Nottingham, I would like to thank for the invaluable source of inspiration, advice and ideas.

Finally, this page would not be complete without mentioning AD for keeping me sane (and putting up with me) through the many ups and downs during the production of this work. Words cannot express my gratitude to my parents, brother and also grandparents for their love and sacrifices they've made over the years; this would not have been possible without you.

\end{acknowledgements}

%----------------------------------------------------------------------------------------
%	LIST OF CONTENTS/FIGURES/TABLES PAGES
%----------------------------------------------------------------------------------------

\tableofcontents % Prints the main table of contents

%----------------------------------------------------------------------------------------
%	PHYSICAL CONSTANTS/OTHER DEFINITIONS
%----------------------------------------------------------------------------------------

%----------------------------------------------------------------------------------------
%	SYMBOLS
%----------------------------------------------------------------------------------------

\begin{symbols}{ll} \label{troop}% Include a list of Symbols (a three column table)
$M=\left(M_{ij}\right)$& Matrix\\
$\mathbf{M}=\left(\mathbf{M}_{ij}\right)$& Matrix of operators\\
$1_n$ & Identity matrix of size $n$\\
$\mathds{1}$& Identity operator\\
$x^*$ or $\overline{x}$& Complex conjugate of $x\in\mathbb{C}$\\
$|x|$& Modulus of $x\in\mathbb{C}$\\
$\mathrm{Re}(\alpha)$& Real part of  $x\in\mathbb{C}$\\
$\mathrm{Im}(\alpha)$& Imaginary part of  $x\in\mathbb{C}$\\
$M^{\#}=\left(M^*_{ij}\right)      $&  Complex conjugate matrix\\
$M^{T}=\left(M_{ji}\right)      $& Transpose matrix\\
$M^{\dag}=\left(M^*_{ji}\right)      $& Adjoint matrix\\
$\left[A, B\right]=AB-BA$ & Commutator\\
$\breve{X}=\left[\begin{smallmatrix} X\\X^{\#}\end{smallmatrix}\right]$ &Doubled-up matrix\\
 $\left(\begin{smallmatrix} A&B\\B^{\#}&A^{\#}\end{smallmatrix}\right)$& $\Delta\left(A, B\right)$ \\
$ \left(\begin{smallmatrix} 1_n&0\\0&-1_n\end{smallmatrix}\right)$&$J_n$\\
$Z^{\flat}=J_mZ^{\dag}J_n$ for $Z\in\mathbb{C}^{2n\times 2n}$&Symplectic adjoint\\
$\mathrm{Tr}(\cdot)$& Trace\\
$\mathrm{Det}(\cdot)$ & Determinant\\
$\mathrm{Spec}(\cdot)$& List of eigenvalues of a matrix\\
$\mathrm{Diag}(x_1,...,  x_m)$& An $m\times m$  diagonal matrix with (diagonal)\\ $\quad$& entries $x_1,...x_m$\\
$\delta_{ij}$ & Kronecker Delta\\
$\delta(t)$& Dirac Delta\\
$\bra{\cdot}$ and $\ket{\cdot}$& `Bra' and `Ket' notation for a quantum state\\
$\rho$& Density matrix\\
$\left<\mathbf{X}\right>$ &Quantum expectation of observable $\mathbf{X}$\\
$ \mathbb{E}\left[X\right]$ &Classical expectation of random variable $X$\\
$\mathrm{Var}\left(\cdot\right)$ & Variance of an observable or random variable\\
$\mathrm{Cov}\left(\cdot\right)$ & Covariance of an observable or random variable\\
$\int^{\infty}_{-\infty}e^{-st}x(t)$ for $s\in\mathbb{C}$& Laplace transformation. \\
Semisimple matrix& A diagonalizable matrix\\
Monic Rational Function & A rational function $R(x):=\frac{P(x)}{Q(x)}$ where the\\& leading coefficient of both polynomials $P(x)$ and \\ &$Q(x)$ is unity.

%DEFINITION OF SEMISIMPLE: &&Could even put in nomenclature \\
%Call this chapter Nomenclature&&\\
%^define&Laplace &transform also for $RE(s)\geq0$\\
%explain&ordering in cascade &ie $\triangleleft$\\
%need to define quadratures&&\\
%need to define expectation i.e. E and bra ket &&\\
%i have yellow sheet with these on &&\\
%denote the measurements outcomes by unbounded matrices...check have done this throughout.&&\\
%Symbol & Name & Unit \\

%\addlinespace % Gap to separate the Roman symbols from the Greek

%$\omega$ & angular frequency & \si{\radian} \\

\end{symbols}

%----------------------------------------------------------------------------------------
%	DEDICATION
%----------------------------------------------------------------------------------------

%----------------------------------------------------------------------------------------
%	THESIS CONTENT - CHAPTERS
%----------------------------------------------------------------------------------------

\mainmatter % Begin numeric (1,2,3...) page numbering

\pagestyle{thesis} % Return the page headers back to the "thesis" style

% Include the chapters of the thesis as separate files from the Chapters folder
% Uncomment the lines as you write the chapters

\chapter{Introduction}
We stand on the brink  of a quantum technological revolution, poised to deliver groundbreaking applications in computation, communication and metrology \cite{Nielsen1, Dowling1, Guta2}.
In order to surpass the fundamental limits set by ``classical'' measurement, information processing and control theory, these applications  must exploit  powerful quantum mechanical phenomena, such as entanglement and coherence, which have no classical analogue.
%We stand on the brink  of a quantum technological revolution, poised to deliver groundbreaking applications in computation, communication and metrology \cite{Nielsen1, Dowling1}. The exploitation of quantum mechanical phenomena, such as entanglement and coherence, to exceed the fundamental limits set by ``classical'' measurement, information processing and control, 
%common feature of these applications is the exploitation of quantum mechanical phenomena, such as entanglement and coherence, to exceed the fundamental limits set by ``classical'' measurement, information processing  and control.
However, the main difficulty is that the enhancements attributed to quantum effects  are notoriously sensitive to noise \cite{Demo1, Dong1}. Therefore, the key challenge (both theoretically and experimentally) is to devise and develop quantum control methods for systems interacting with noisy environments.
%However, as the enhancements attributed to quantum effects are notoriously sensitive to noise and disturbances, one of the key challenges  (both theoretically and experimentally) is to devise and develop quantum control methods for systems interacting with noisy environments. 
This has motivated the development of quantum filtering theory  \cite{Wiseman1, Bouten2, Bouten1, Gough1}, quantum feedback control \cite{Soma1, Peter1, James1, Gough1, Guta2, Nurdin3, Griv1, Gough2, Soma2, Doherty1, Yanagisawa1, Gough5} 
and network theory \cite{Gough3, Gough4, Zhang3, Nurdin2}, which build on the existing classical stochastic control theory.

In particular, there has been a rapid growth in the study of quantum linear systems (QLSs). 
QLSs are a class of models used in quantum optics, opto-mechanical systems, electrodynamical  systems,  cavity QED systems  and elsewhere  \cite{Yamamoto2, Wall1, Tian1, Gardner1, Stockton1, Doherty1, Matyas1}.  
They have many applications, such as quantum memories, entanglement generation, quantum information processing and quantum control \cite{Yamamoto1, Nurdin1, James2, Nurdin4, Wiseman1, Bouten2, Dong1, Yamamoto4, Zhang2}.
Broadly speaking, a QLS consists of a continuous variable quantum system (e.g an electromagnetic field in an optical cavity) weakly coupled with a Bosonic environment (e.g. external laser fields).
They are analogous to classical   electrical networks;  classical circuits are built from elementary elements such as resistors, capacitors, inducers, etc, whereas QLSs are quantum circuits built from beam-splitters, optical, cavities and squeezers.
QLSs are input-output models, where one prepares an input in the field and then infers system parameters indirectly  through quantum measurements in the field.

System identification theory \cite{Ljung1, Ljung2, Pintelon1, Guta3, Guta4} lies at the boundary of  control theory and statistical inference, and deals with the estimation of unknown dynamical parameters, such as the system hamiltonian or coupling constants, from the input-output data. 
%The integration of control and identification techniques plays an important role in adaptive control \cite{Astrom1}.  
The identification of linear systems is by now a well developed and mature topic in classical systems theory \cite{Glover1, Kalman1, Ljung1, Ljung2, Ho1, Anders1, Youla1, Zhou1, Pintelon1, Davies1}, but has not been fully explored in the quantum domain \cite{Guta2}. 
 Here we will strive towards this.

We distinguish two contrasting approaches to the identification of QLSs, 
which we illustrate in Fig. \ref{trew}. 
In the first approach, one probes the system with a known \emph{time-dependent} input signal (e.g., coherent state), then uses the output measurement data to compute an estimator of the unknown dynamical parameter. The characteristic quantity here is the \textit{transfer function} and as such QLSs with the same transfer function are called \textit{transfer function equivalent} (TFE). We investigate system identification for the time-dependent approach in Ch. \ref{T.F}. In the second approach 
 the input fields are prepared in a stationary (in time) pure Gaussian state with independent increments (squeezed vacuum noise) and the characteristic quantity is the power spectrum. We explore this approach in Ch. \ref{powers34}.
 
 Following the results in Chs. \ref{T.F} and \ref{powers34} one will understand what parameters can be identified in a QLS and how to identify them. The next natural question following this is how well such parameters can be identified? This question is the main focus of Chs. \ref{QEEP} and \ref{FEDERER}. The setup is similar to the standard metrology setup (which is also reviewed in Ch. \ref{plm}), however the added difficulty in  the QLS setup arises from working in continuous time and therefore requires a different analysis of behaviour \cite{Jacobs1}.
 In  Sec. \ref{Toes1} we provide a realistic scheme to identify a single unknown parameter of a \textit{passive} QLS at the Heisenberg limit with respect to an energy constraint. `Realistic' is our primary concern here, in the sense that there exists both experimentally feasible states and practical measurement choices. 
Our method uses a generalisation of the seminal interferometric method proposed by Caves in 1981 \cite{Caves1}, which places squeezed and coherent states in the interferometer arms. Our method is genuinely quantum system identification using entanglement, which is a new topic in systems and control theory. Note that this strategy is currently being implemented in the most advanced interferometers designed to detect \textit{gravitational waves} \cite{LIGO1, LIGO2}.   In Ch. \ref{FEDERER} we switch focus, so that time is our main resource constraint, and  investigate in detail the phenomenon that when the system destabilises, so that the system is quasi-decoherence free, metrological information from the system is enhanced and the \textit{quantum fisher information} becomes quadratic (rather than linear) in the observation time. This is particularly surprising considering that  information about the system is extracted indirectly via the field and such destabilisation will make the  system-field coupling smaller.
 
 Finally, in Ch. \ref{DUAL} we consider a particular \textit{reservoir engineering} problem. Reservoir engineering is a hot topic at the moment and is concerned with designing the dissipative dynamics of a  system with the environment to drive the system into a desired pure stationary state. It has various applications in laser cooling and optical pumping \cite{Plenio1, Stan1, Yamamoto2, Yamamoto3, Metcalf1}. Given a QLS is it possible to design a second QLS so that the combined system has a pure stationary state and (by results  in Ch. \ref{powers34}) the output of the first system is negated by the second. This provides a natural purification for the stationary state and we see later its possible applications.
 
In Part \ref{partb} we review the necessary background material required for this thesis
In particular, in  Chs. \ref{CL1} and \ref{WEEP}  we review in detail the important aspects of classical and quantum linear systems, respectively, and  Ch. \ref{plm}  surveys the necessary aspects of quantum estimation theory. Part \ref{partr} comprises the main results outlined above.

\part{Background}\label{partb}

\chapter{Classical Linear Systems}\label{CL1}

Classical linear systems (CLSs)  are dynamical models describing a range of real world phenomena. They  have many   applications, from \textit{automatic control systems} and \textit{communication}, to \textit{aeronautics} and \textit{engineering} \cite{Zhou1, Yang1, Anna1, Enq}. 
Most recently the applications of linear systems have  extended to  quantum mechanics, which subsequently led to the birth of  \textit{quantum linear systems theory} \cite{Peter1};. Many CLS theory results  transfer over to quantum linear systems theory directly, therefore it is worthwhile for us to review CLSs in detail first.

%Most real systems have non-linear input-output characteristics. However, their linear approximation works  reasonably well in practice for many applications \cite{Enq}. 
% By introducing stochastic processes into these models  one can describe the sort of unpredictable patterns arising from  noise or observation errors, or even complexity within the model eluding  our ability to describe the system deterministically. 
 %The formulation of the model may be straightforward because the governing equations are known.

CLSs are examples of input-output or `black-box' models. Typically, one can access the system indirectly by preparing a  time-dependent input signal, which acts as a probe to the system. After the coupling, the parameters of the system  are imprinted on the output signal.  From the  observations, the task is to estimate parameters within the system. A huge body of theory has been developed to treat various aspects of this problem, including \textit{system identification, realization theory} and \textit{statistical estimation theory} \cite{Zhou1, Hayden1, Kalman1, Youla1}.

\section{Description of CLSs}

CLSs are described by the following pair of differential equations \cite{Zhou1} 
\begin{align}
\label{CSL1}dx(t)&=Ax(t)dt+Bu(t)dt\\
\label{CSL2}dy(t)&=Cx(t)dt+Du(t)dt
\end{align}
where ${x}(t)\in\mathbb{C}^n$ is called the \textit{system state}, ${u}(t)\in\mathbb{C}^m$ is an  \textit{input signal}  and ${y}(t)\in\mathbb{C}^p$ is an \textit{output signal}. The matrices $A, B, C, D$ are appropriately dimensioned complex matrices\footnote{Usually everything here is real valued, however from a quantum point of view later it is more natural to work with complex matrices and signals.}. 
%This is a `black-box' model where the observer can control the input signal and observe the output, but does not have access to the internal state of the system. 
A CLS with one input channel ($m=1$) and one output channel ($p=1$) is called SISO (single input and single output), otherwise it is called MIMO (multiple input and multiple output). The input to the system can be deterministic or stochastic, resulting in a set of differential  or stochastic differential equations respectively. 

Now, Eqs. \eqref{CSL1} \eqref{CSL2} can be solved directly \cite{Zhou1, Ljung1}; we give the solution for $x(t)$ here:
\begin{equation}\label{rubbish1}
 x(t)=e^{A(t-t_0)}x(t_0)+\int^t_{t_0}e^{A(t-\tau)}Bu(\tau)d\tau
 \end{equation}
 where $t_0$ is the initial time. 
From  Eq. \eqref{rubbish1}, if any of the eigenvalues of $A$ are in right half plane (Re$(\lambda(A))\geq0$) then $x(t)$ will be dominated by the first term and will grow without bound in the long time limit. Therefore, from now we will only be interested in \textit{stable} CLSs (see Def. \ref{stable}). 
\begin{defn}\label{stable}
A CLS $(A, B, C, D)$ is said to be (Hurwitz) stable if all eigenvalues of $A$ lie in the open left half plane.
% (Re$(\lambda(A))\geq0$).
\end{defn}

\section{Controllability, Observability and Minimality}\label{systheo}

Apart from the input (and the initial state of the system) the dynamics is completely determined by the quadruple $(A, B, C, D)$ \cite{Guta2}. The following two concepts, defined in terms of these matrices,  are very important in linear systems theory \cite{Zhou1}:

\begin{defn}
A CLS $(A, B, C, D)$ or the pair $(A,B)$ is \textit{controllable} if, for any initial initial state $x_0$, final state $x_1$ and times $t_0< t_1$ there exists a (piecewise continuous) input $u(\cdot)$ such that $x(t_0)=x_0$ and $x(t_1)=x_1$. Otherwise the pair $(A, B)$ is \textit{uncontrollable}.
\end{defn}

\begin{defn}
A CLS $(A, B, C, D)$ or the pair $(C,A)$ is \textit{observable} if, for any  times $t_0< t_1$ the initial state $x(t_0)=x_0$ can be determined from the past history of the input $u(t)$  and the output $y(t)$ in the time interval $[t_0, t_1]$. Otherwise the pair $(C, A)$ is \textit{unobservable}.
\end{defn}

We will see precisely why these definitions are so important shortly. But first observe that we can take the Laplace transformation (see  Nomenclature) of \eqref{CSL1}  and \eqref{CSL2} to obtain the following input-output map\footnote{under the assumption of stability}:
\begin{equation}
Y(s)=G(s)U(s),
\end{equation}
where $U(s)$ and $Y(s)$ are the Laplace transforms of $u(s)$ and $y(s)$ and the \textit{transfer function} $G(s)$ is given by: $$G(s)=C\left(sI-A\right)^{-1}B+D.$$

The most an experimenter can  hope to obtain from measurements of the output is the transfer function. However, not only is it possible that two CLSs can have the same transfer function, such systems may also have differing dimensions (i.e different number of modes). In seeking the simplest model of the input-output behaviour we make the following  definition.
\begin{defn}
A state space realisation $(A, B, C, D)$ of the transfer function $G(s)$ is minimal if there is no other state space with smaller dimension.
\end{defn}
Now, the importance of the controllability and observability can be seen in the following fact \cite{Zhou1}: if a system $(A,B,C,D)$  is \textbf{not} observable or controllable then there exists a lower dimensional system with the same transfer function as the original one. Furthermore, a minimal system may be obtained from the former by using a technique called the \textit{Kalman decomposition} \cite{Zhou1}. Therefore, in this sense describing the transfer system by a non-minimal system would be superfluous.

The following system theoretic results linking the concepts will be very useful in the following; they are stated without proof (see for example \cite{Zhou1}).

\begin{thm}\label{cl5}
The following are equivalent:
\begin{enumerate}
\item \label{cl1}$(A,B)$ is controllable
\item \label{cl2}The \textit{controllability matrix} 
\begin{equation}
\mathcal{C}=\left[\begin{smallmatrix}B&AB&...&A^{n-1}B\end{smallmatrix}\right]
\end{equation}
has full column rank.
\item \label{cl3} For any left-eigenvector, $x$, of $A$ with corresponding  eigenvalue $\lambda$, i.e $x^{\dag}A=x^{\dag}\lambda$, then $x^{\dag}B\neq0$.
\item \label{cl4} $(B^{\dag}, A^{\dag})$ is observable.
\end{enumerate}
\end{thm}

\begin{thm}\label{obs5}
The following are equivalent:
\begin{enumerate}
\item \label{obs1} $(C,A)$ is observable
\item \label{obs2} The \textit{observability matrix} 
\begin{equation}
\mathcal{O}=\left[\begin{smallmatrix}C\\CA\\ \vdots\\CA^{n-1}\end{smallmatrix}\right]
\end{equation}
has full row rank.
\item \label{obs3} For any right-eigenvector, $y$, of $A$ with corresponding  eigenvalue $\lambda$, i.e $Ay=\lambda y$, then $By\neq0$.
\item \label{obs4} $( A^{\dag}, C^{\dag})$ is controllable.
\end{enumerate}
\end{thm}

\section{Time-Dependent Versus Stationary Inputs}\label{feb16}

We now distinguish two parallel setups within the CLS theory, which are determined by the choice of input. That is, we discuss the use of both time-dependent and stationary inputs in this subsection.

\subsection{Time-Dependent Inputs}
In the time-dependent approach, one probes the system with a known time-dependent pulse $u(t)$. In the Laplace-domain the input-output map is therefore is entirely captured by the transfer function, $G(s)$, and hence the most that one can identify in this approach is the transfer function.

\subsection{Stationary inputs}
In the stationary approach we assume that the system is driven by  stationary white noise $u(t)$ characterised by  covariance $\mathbb{E}\left[u(t)u^{\dag}(\tau)\right]=1_m\delta(t-\tau)$. From \eqref{CSL2} (and working in the Laplace-domain) the most information about the system in the output  is given by the \textit{power spectral density} (or power spectrum) \cite{Hayden1}:
$$\Psi(s):=\mathbb{E}\left[  Y(s)Y(-s^*)^{\dag}   \right]=G(s)G(-s^*)^{\dag}.$$

\section{Identifiability of CLSs}\label{class.id}

As hinted earlier, the following question is  very important question in linear systems  theory \cite{ Zhou1, Ljung1}: which dynamical parameters 
of a CLS can be identified by observing the output?

\subsection{Time-Dependent Inputs}

Since the transfer function may be recovered from the input-output map, the identifiability question in the time-dependent approach reduces to finding the  equivalence classes of (minimal) systems with the same transfer function.  This \textit{system identification problem} has been addressed in the literature \cite{Kalman1, Ho1, Anders1} and we state it here for convenience.

\begin{thm}\label{class.id2}
Let $(A, B, C, D)$ and $(A', B', C', D')$ be two minimal CLSs. Then they have the same transfer function if and only if there exists a similarity transformation $T$ such that 
\begin{equation}\label{eq.similarity}
A'=TAT^{-1}, \quad B'=TB, \quad C'=CT^{-1}, \quad D=D'.
\end{equation}
\end{thm}

\subsection{Stationary Inputs}\label{ICE12}

 In the stationary input  approach the power spectrum can  be computed from the output correlations.
Understanding which parameters can be identified from the CLS  reduces to finding all CLSs with the same power spectrum.  Observe that the power spectrum depends on the system parameters via the transfer function, therefore it cannot be possible to identify  more than the transfer function. The problem of finding the transfer function from the power spectrum is of the type: `for a square rational matrix $R(s)$, where $s\in\mathbb{C}$, find rational matrix $W(s)$ such that 
$$R(s)=W(s)W(-{s}^*)^{\dag}$$
for all $s\in \mathbb{C}$'. This type of problem is called the 
 \textit{spectral factorisation problem} \cite{Youla1, Anders1}. There are known algorithms to do this \cite{Youla1, Davies1}. From the latter,  one then finds a system realisation (i.e. matrices governing the system dynamics) for the given transfer function \cite{Ljung1}. The problem is that the map from power spectrum to transfer functions is  non-unique, and each factorisation could lead to (minimal) system realisations of differing dimension. For this reason, the concept of \textit{global minimality} was introduced in \cite{Kalman1} to select the transfer function with smallest system dimension. 
 
 \begin{defn}\label{CLSGM}
 A CLS $(A, B, C, D)$ is globally minimal if there exists no lower dimensional system with the same power spectrum, $\Psi(s)$. 
 \end{defn}
 
 This raises the following question: is global minimality sufficient to \textbf{uniquely} identify the transfer function from the power spectrum? %In other words, does global minimality imply that the \emph{only} non-identifiable parameter directions in the power spectrum are the ones relating equivalent systems with the transfer function ? 
The answer is in general negative\footnote{However, under the assumption that the transfer function be \textit{outer} the construction of the transfer function from the power spectrum is unique (see \cite{Hayden1}).}, as discussed in \cite{Anders1, Glover1} (see also Lemma 2 and Corollary 1 in \cite{Hayden1}).

\section{Operations on CLSs}\label{network1}

We conclude this section by discussing two methods of connecting CLSs. These are analogous to the notions of components connected in parallel or in series in electrical circuits. 

\begin{figure}
\centering
\includegraphics[scale=0.15]{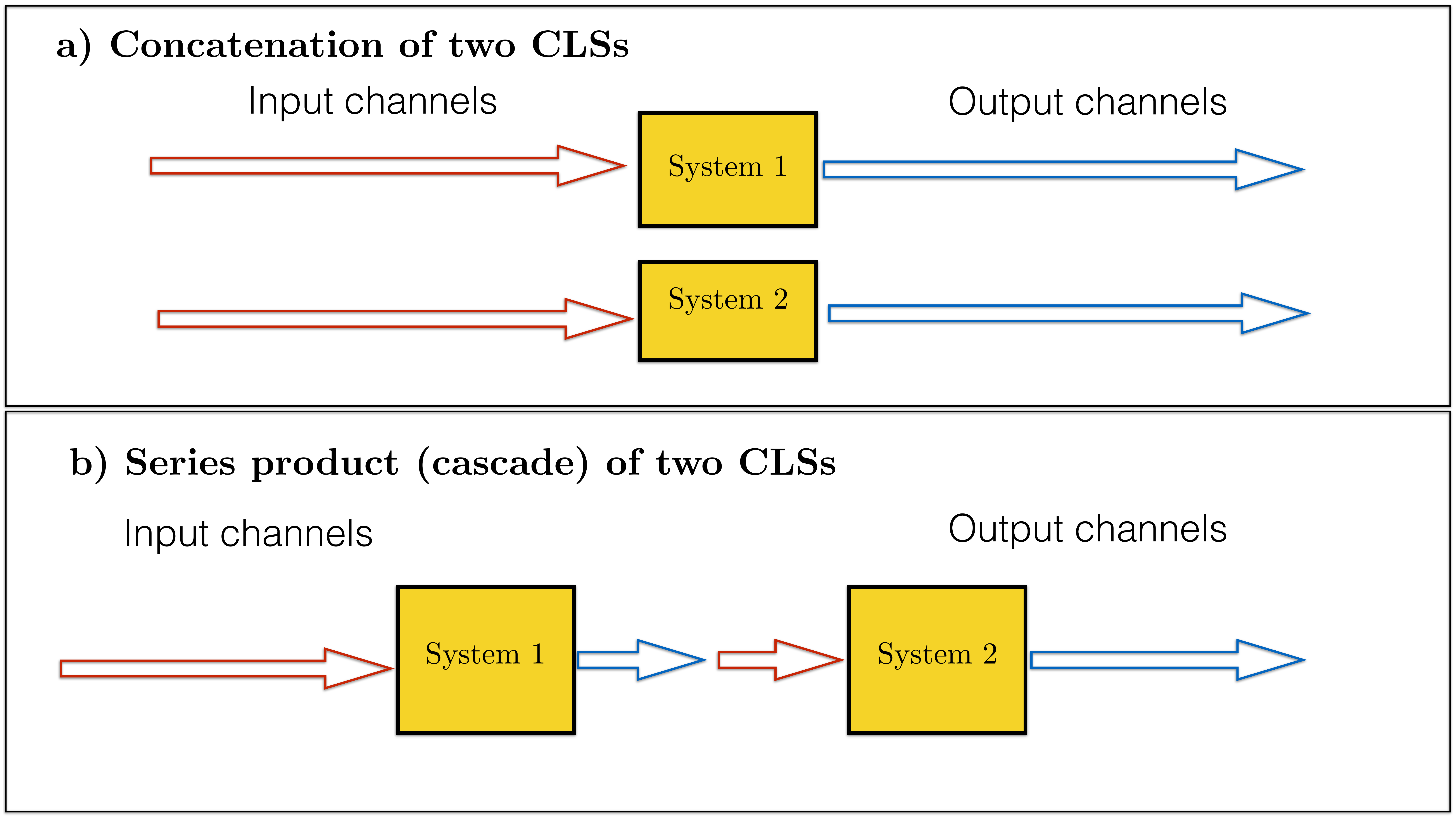}
\caption{Illustration on a) concatenation connection  and b) series product connection.
\label{sercon1}}
\end{figure}

The \textit{concatenation} connection \cite{Zhou1}, which consists of lying two systems  $\left(A_1, B_1, C_1, D_1\right)$ and    $\left(A_2, B_2, C_2, D_2\right)$ side by side (cf. Fig. \ref{sercon1}), has system matrices 
\begin{equation}\label{concat}
\left(\left(\begin{smallmatrix}A_1&0\\0&A_2\end{smallmatrix}\right), \left(\begin{smallmatrix}B_1&0\\0&B_2\end{smallmatrix}\right), \left(\begin{smallmatrix}C_1&0\\0&C_2\end{smallmatrix}\right), \left(\begin{smallmatrix}D_1&0\\0&D_2\end{smallmatrix}\right)\right).
\end{equation}

The \textit{series product} connection \cite{Zhou1}, which consists of connecting two systems in series with the output of the $\left(A_1, B_1, C_1, D_1\right)$  as the input of the $\left(A_2, B_2, C_2, D_2\right)$ (cf. Fig. \ref{sercon1}), has system matrices 
\begin{equation}\label{concat1}
\left(\left(\begin{smallmatrix}A_1&0\\B_2C_1&A_2\end{smallmatrix}\right), \left(\begin{smallmatrix}B_1\\B_2D_1\end{smallmatrix}\right), \left(\begin{smallmatrix}D_2C_1&C_2\end{smallmatrix}\right), D_2D_1\right).
\end{equation}

\chapter{Quantum Linear Systems}\label{WEEP}

Quantum linear systems (QLSs) 
 are composed of a continuous variables \textit{system}  coupled to a series of quantum stochastic input \textit{fields} consisting of non-commuting noise processes \cite{Peter1, Guta2}. Following the input-system interaction, the input is  transformed into an output, which can be measured to produce a classical stochastic process (see Ch. \ref{CL1}) \cite{Guta2}. %The components of QLSs include optical cavities, squeezers, and beam splitters. 
Their uses extend to  quantum optics, opto-mechanical systems, electrodynamical  systems,  cavity QED systems  and elsewhere  and have a huge array of applications \cite{Wiseman1, Doherty1, Yamamoto2, Tian1, Gardner1, Stockton1, Yamamoto4, Zhang2, Matyas1}.  
%They have many applications, such as quantum memories, entanglement generation, quantum information processing and quantum control \cite{Yamamoto1, Nurdin1, James2, Nurdin4, Wiseman1, Bouten2, Dong1}.
 The  framework required to describe these is the celebrated \textit{quantum stochastic calculus} \cite{Parth1}. 

 %Indeed, the model consists of an open quantum system with harmonic modes and quadratic hamiltonian coupled linearly to a series of Bosonic input channels.  

 In this section we  review the QLS theory; firstly discussing in detail the \textit{system} and \textit{field} constituents. We then highlight some results of relevance for this thesis. 
We refer the reader to \cite{Gardner1} and \cite{Sarovar1} for a more detailed discussion on the input-output 
formalism, and  to the review papers \cite{Peter1, Parth1, Parth2, Nurdin2} 
for the theory of QLSs.

\section{Quantum Harmonic Oscillator (QHO) and the Symplectic Group}\label{QHO}

Consider a collection of $n$ QHOs, which    are described by a column vector of annihilation operators, 
$\mathbf{a}:=[\mathbf{a}_1,\mathbf{a}_2, \dots, \mathbf{a}_n]^T$. Together with their respective creation operators 
$\mathbf{a}^{\#}:=[\mathbf{a}^{\#}_1,\mathbf{a}^{\#}_2, \dots, \mathbf{a}^{\#}_n]^T$ 
they satisfy the canonical commutation relations (CCR)
$
\left[\mathbf{a}_i, \mathbf{a}^{*}_j\right]=\delta_{i j}\mathds{1}.
$
Alternatively, we may write these in doubled-up notation as
$$\left[\breve{\mathbf{a}}_i,\breve{\mathbf{a}}_k^{\#}\right]=J_{ik}.$$
We denote by $\mathcal{H}:= L^2(\mathbb{R}^n)$ the Hilbert space of this system carrying the standard representation of the $n$ modes.

\begin{defn}\label{def.symplectic}
A matrix $S\in\mathbb{C}^{2n\times 2n}$ is said to be \textit{$\flat$-unitary} if it is invertible and satisfies
\[S^\flat S=SS^\flat =1_{2n}.\]
If additionally, $S$ is of the form $S=\Delta(S_-, S_+)$ for some $S_-, S_+\in\mathbb{C}^{m\times m}$ then we call it \textit{symplectic}. Such matrices form a group called the \textit{symplectic group} \cite{Gough1}.
\end{defn}

The \textit{symplectic transformation} 
\begin{equation}\label{canonical}
\breve{\mathbf{a}}\mapsto S\breve{\mathbf{a}}
\end{equation}
for symplectic matrix $S$
 preserves the CCR and is a mapping from one QHO system to another. In addition such a transformation may be unitarily implemented \cite{Gough1}; that is,  by Shale's Theorem \cite{Shale1} there exists a unitary operator $\mathbf{U}$ such that 
$$S\breve{\mathbf{a}}=\mathbf{U}^{\dag}\breve{\mathbf{a}}\mathbf{U}.$$ From the familiar quantum mechanical point of view, there is a Hamiltonian $\mathbf{H}$ generating \eqref{canonical}. That is, $\mathbf{U}=e^{-i\mathbf{H}}$ where 
\begin{equation}\label{Hamiltonian}
\mathbf{H}=\mathbf{a}^{\dag}\Omega_-\mathbf{a}+\frac{1}{2}\mathbf{a}^T\Omega_+^{\dag}\mathbf{a}+\frac{1}{2}\mathbf{a}^{\dag}\Omega_+\mathbf{a}^{\#}
\end{equation}
for $n\times n$ matrices $\Omega_-=\Omega_-^{\dag}$, $\Omega_+=\Omega_+^T$ \cite{Weedbrook1}. 

A state on $\mathcal{H}$ is said to be \textit{Gaussian} if  
\begin{equation}\label{cough}
\left<e^{i\breve{u}^{\dag}\breve{\mathbf{a}}}\right>=e^{-\frac{1}{2}\breve{u}^{\dag}V\breve{u}+i\breve{u}^{\dag}\breve{\alpha}},
\end{equation}
where $V\geq0$. The Gaussian state is characterised by the mean $\breve{\alpha}=\left<\breve{\mathbf{a}}\right>$ and the covariance matrix $$V:=\left<\breve{\mathbf{a}}\breve{\mathbf{a}}^{\dag}\right>=\left(\begin{smallmatrix}1+N^T&M\\M^{\dag}&N\end{smallmatrix}\right)$$
where $N=N^{\dag}$ and $M=M^T$.
% which ensures that the state does not violate the uncertainty principle. 
Now, we can also interpret \eqref{canonical} as a change of basis, since in the co-ordinates $\breve{\mathbf{a}}':=S\breve{\mathbf{a}}$  the state will be Gaussian with mean and covariance $S\breve{\alpha}$ and $SVS^{\dag}$ respectively. 
A \textit{coherent state} with amplitude $\alpha$ corresponds to the case $N=M=0$ (note that the vacuum is a coherent state with amplitude $\alpha=0$). Assuming that $\alpha=0$ for simplicity, then the state's purity can be characterised in terms of the symplectic eigenvalues of $V$. That is, there exists a symplectic matrix, $S$, such that the modes  $\breve{\mathbf{a}}'=S\breve{\mathbf{a}}$ are independent of each other and each of them is in a vacuum or thermal state i.e.  $V_i=\left<\breve{\mathbf{a}_i}'   \breve{\mathbf{a}'_i}^{\dag}\right>=             \left(\begin{smallmatrix} n_i+1&0\\0&n_i\end{smallmatrix}\right)$, where $n_i$ is the mean photon number. We call $\breve{\mathbf{a}}'$ a canonical basis and the elements of the ordered sequence $n_1\leq...\leq n_k$ the \textit{symplectic eigenvalues} of $V$. The latter give information about the state's purity: if all $n_i=0$ the state is pure, if all $n_i>0$ the state is fully mixed. This result is known as \textit{Williamson's Theorem} \cite{Weedbrook1, Adesso1, Williamson1}. In particular, a pure stationary state may be viewed in a different basis as vacuum by performing the symplectic transformation  \cite{Weedbrook1}
\begin{equation}\label{vtrick}
S=\Delta\left((N^T+1)^{1/2},M\left(N^{\dag}+1\right)^{-1/2}\right).
%\left(\begin{smallmatrix} (N^T+1)^{1/2}& M\left((N^T+1)^{-1/2}\right)^{\#}\\ M^{\#}(N^T+1)^{-1/2}&\left((N^T+1)^{1/2}\right)^{\#}\end{smallmatrix}\right) 
\end{equation}
 More generally, we can separate the pure and mixed modes and write ${\mathbf{a}}'=\left(\mathbf{a}_p^T,\mathbf{a}^T_m\right)^T$ (cf symplectic decomposition \cite{Wolf1}). For more details on Gaussian states see \cite{Weedbrook1, Adesso1}.
A further useful  result for this thesis is the following,  which relates to bipartite entanglement of Gaussian states  \cite[Theorem 1]{Botero1}.

\begin{thm}\label{WOLFF}
Consider a bipartition of a (zero mean) pure Gaussian state with $n$ modes, so that $A$ has $n_1$ modes (and $B$ has $n-n_1$ modes). Then there exists local symplectic transformations at each site $A$ and $B$ so that 
\begin{equation}
\left|\psi\right>_{A,B}=\left|\psi\right>_{\tilde{A}_1,\tilde{B}_1}\otimes...\otimes\left|\psi\right>_{\tilde{A}_s,\tilde{B}_s}\otimes\left|0\right>_{\tilde{A}_F}\otimes\left|0\right>_{\tilde{B}_F},
\end{equation}
for some $s\leq\mathrm{min}(n_1, n-n_1)$, where $\{\tilde{A}_1, ... \tilde{A}_{n_1}\}$ and $\{\tilde{B}_1, ... \tilde{B}_{n-n_1}\}$ are the new transformed sets of modes. 
Here 
$\left|\psi\right>_{\tilde{A}_i,\tilde{B}_i}$ are two-mode squeezed states \cite{Lvovsky1} characterised by covariance matrix $V(N_i,M_i)$ with
$$N_i=\left(\begin{smallmatrix}n_i&0\\0&n_i\end{smallmatrix}\right) \quad\mathrm{and}\quad 
M_i=\left(\begin{smallmatrix}0&m_i\\m_i&0\end{smallmatrix}\right)$$
with $n_i, m_i\in\mathbb{R}$ and $n_i(n_i+1)=m_i^2$ (purity condition). The states
$\left|0\right>_{\tilde{A}_F}$ and $\left|0\right>_{\tilde{B}_F}$ are vacuum states on the remaining modes in $\tilde{A}_i$ and $\tilde{B}_i$ respectively.
\end{thm}

This result says that any bipartite division of modes in a pure Gaussian state can always be expressed as a product state involving either two-mode squeezed states or single mode squeezed states at each site.
The proof is a consequence of the \textit{Schmidt decomposition} \cite{Nielsen1}.
Essentially, 
since the reduced state (with respect to the bipartition) of the original pure Gaussian state is also Gaussian, then one may apply 
Williamson's Theorem on each site. Careful comparison of this with the  result of the Schmidt decomposition gives the result. 

\begin{remark}
Observe that the symplectic eigenvalues of the reduced states for observers $\tilde{A}$ and $\tilde{B}$ are identical. Furthermore, the reduced states on modes $\tilde{A}_i$ and $\tilde{B}_i$ are the same and are given by $V_{\tilde{A}_i}=V_{\tilde{B}_i}=\left(\begin{smallmatrix}n_i+1&0\\0&n_1\end{smallmatrix}\right)$. 
Moreover the character of entanglement can be understood mode-wise because each mode on one side is entangled with only one on the other side. 
%A consequence of this is that the total bipartite entanglement is the sum of mode wise entanglement. 
\end{remark}

%[[PUrity condition in terms of N and M has been left out....could include if need later]]]

\section{Bosonic Fields}\label{Bosonic}
A bosonic environment with $m$ channels is described by  
 fundamental variables (fields) 
$\mathbf{B}(t):=\left[\mathbf{B}_{1}(t), \mathbf{B}_{2}(t), \ldots, 
\mathbf{B}_{m}(t)\right]^T$, where $t\in \mathbb{R}$ represents time. 
The fields satisfy the CCR 
\begin{eqnarray}
\left[\mathbf{B}_{i}(t), \mathbf{B}^{*}_{j}(s)\right]=\mathrm{min}\{t,s\}\delta_{ij}\mathds{1}.
\end{eqnarray}
Equivalently, this can be written as 
$\left[\breve{\mathbf{b}}_{i}(t),\breve{\mathbf{b}}_{k}^{*}(s)\right]
=\delta(t-s)J_{ik}\mathds{1}$, where 
$\mathbf{b}_{i}(t)$ are the infinitesimal (white noise) annihilation operators formally defined as 
$\mathbf{b}_{i}(t):=d\mathbf{B}_{i}(t)/dt$ \cite{Peter1}. 
The operators can be defined in a standard fashion on the Fock space 
$\mathcal{F}= \mathcal{F}(L^2(\mathbb{R})\otimes \mathbb{C}^m)$ \cite{Bouten1}.  At each instance of time the modes $b_i(t)$ may be interpreted as an $m$-mode  QHO.
For more details on Bosonic fields and quantum stochastic processes see \cite{Parth1, Bouten1}.

\section{The Model: Time-Domain Representation}\label{MTDR}
A quantum linear system (QLS) is defined as a continuous 
variable (cv) system (Sec. \ref{QHO}), called the \textit{system}, coupled to a Bosonic environment (Sec. \ref{Bosonic}), called the \textit{field(s)}. 
QLSs are  input-output models (see Fig. \ref{inout}), in which one  prepares an input in the field and then infers system parameters indirectly from measurements of the output.

\begin{figure}
\centering
\includegraphics[scale=0.18]{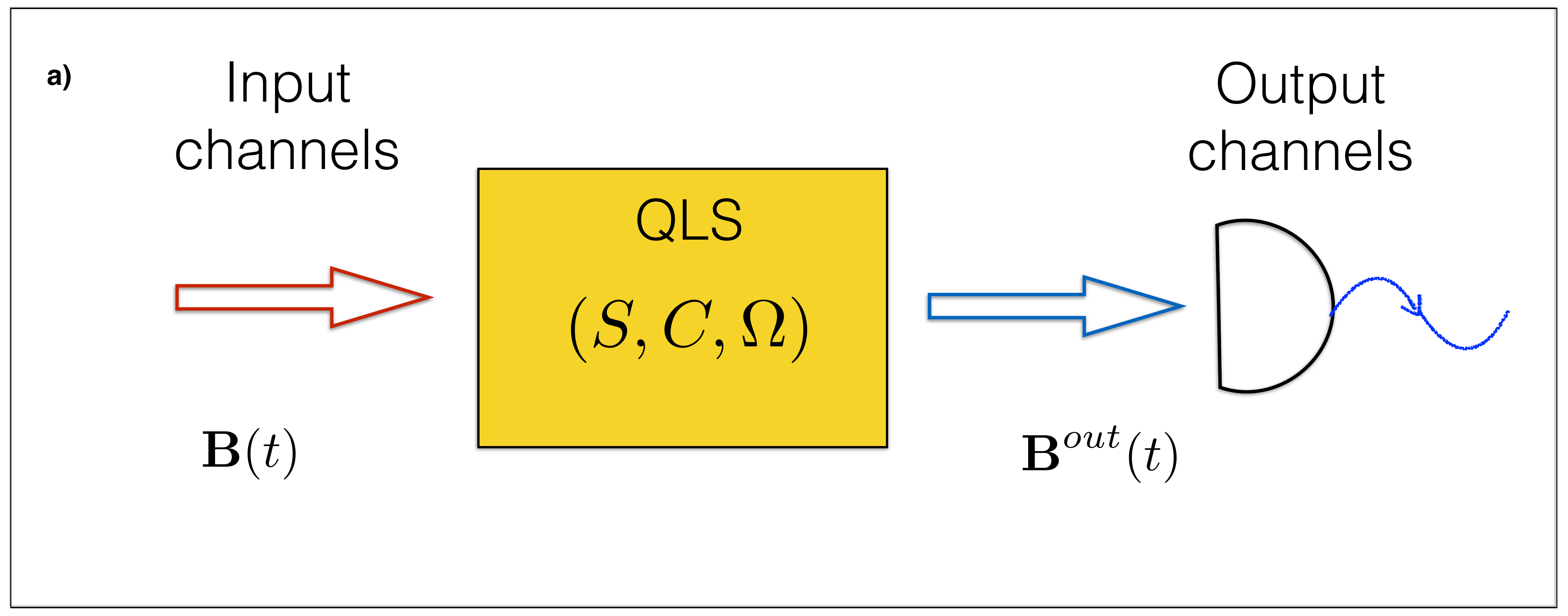}
\caption{ System identification problem: find parameters $(S, C, \Omega)$ of a QLS by measuring output.
\label{inout}}
\end{figure}

The dynamics of a QLS is determined by the system's Hamiltonian \eqref{Hamiltonian} and the coupling operator 
\begin{equation}\label{couplef}
\mathbf{L}=C_-\mathbf{a}+C_{+}\mathbf{a}^{\#},
\end{equation}
 where $C_-$ and $C_+$ are $m\times n$ matrices. In the Markov approximation, the joint unitary evolution of ``system $\otimes$ field'' is described (in the interaction picture) by the unitary 
$${\bf U}(t)=\overleftarrow{\mathrm{exp}}\left[-i\int_0^t\left(\mathbf{H}+\mathbf{H}_{int}(s)\right)ds\right]$$
 on the joint space 
$\mathcal{H}\otimes \mathcal{F}$, where $\mathbf{H}_{int}:=i\left[\mathbf{b}^{\dag}(t)\mathbf{L}-\mathbf{L}^{\dag}\mathbf{b}(t)\right]$ is the \textit{interaction Hamiltonian}. Notice that the total Hamiltonian $\mathbf{H}+\mathbf{H}_{int}(t)$ is quadratic in the canonical variables. 
The unitary ${\bf U}(t)$ is the (unique) solution of the quantum stochastic differential equation (QSDE) \cite{Bouten1, Bouten2, Dong1, Gardner1, Parth1, Parth2}
\begin{eqnarray}
\label{eq.unitary.cocycle}
& & \hspace*{-3.9em}
      d\mathbf{U}(t) 
          :=\mathbf{U}(t+dt)-\mathbf{U}(t)
\nonumber \\ & & \hspace*{-1em}
      = \left(-i \mathbf{H}dt+({\bf S}-\mathds{1}_m)
               d\mathbf{\Lambda}(t)
              + d\mathbf{A}^{\dag}(t)\mathbf{L}
               - \mathbf{L}^{\dag}d\mathbf{A}(t) 
               -\frac{1}{2}\mathbf{L}^{\dag}\mathbf{L}dt\right)\mathbf{U}(t),
\end{eqnarray}
with initial condition ${\bf U}(0)= \mathbf{I}$. 
Here ${\bf S}=S\mathbf{1}_m$ describes the 
 scattering between 
the fields; 
$d{\bf A}_i(t), d{\bf A}_i^*(t), 
d\mathbf{\Lambda}_{ij}(t)$ 
are increments of fundamental quantum stochastic processes describing the 
creation, annihilation and scattering between the input channels. In particular, $\mathbf{\Lambda}(t)$ is an $m\times m$ matrix whose elements 
$\mathbf{\Lambda}_{ij}(t)$ are the second quantisation operators 
$d\Gamma(P_{[0,t]}\otimes |j\rangle \langle i| )$, where 
$P_{[0,t]}$ is the projection onto wave functions in $L^2(\mathbb{R})$ with 
support in $(0,t)$. 
These operators represent scattering between the input channels and are formally 
given by $\mathbf{\Lambda}_{ij}(t) 
:=\int^t_0\mathbf{b}_i^{*}(s)\mathbf{b}_j(s)ds$.

Now let $\mathbf{a}(t)$ and 
$\mathbf{B}^{out}(t)$ be the Heisenberg evolved system and output variables 
\begin{eqnarray}
       \mathbf{a}(t):=\mathbf{U}(t)^{\dag}\mathbf{a}\mathbf{U}(t),~~~ 
       \mathbf{B}^{out}(t):=\mathbf{U}(t)^{\dag}\mathbf{B}(t)\mathbf{U}(t).
\end{eqnarray}
A consequence of the quadratic Hamiltonian is that these variables satisfy linear (doubled-up) QSDEs
 (see for example \cite{Gough1, Parth1, Bouten1}): 
\begin{eqnarray}\label{fungi}
       d\breve{\bf a}(t) &=& 
            A \breve{\bf a}(t)dt-C^{\flat}\Delta(S,0)d \breve{\bf B}(t),\\
       d \breve{\bf B}^{out}(t) &=& 
            C \breve{\bf a}(t)dt+ \Delta(S,0)d \breve{\bf B}(t),
\end{eqnarray}
where   $C:=\Delta\left(C_{-}, C_{+}\right)$ and $A:=\Delta\left(A_{-}, A_{+}\right)=-\frac{1}{2}C^\flat C-iJ_n\Omega$ with $\Omega= \Delta\left(\Omega_{-}, \Omega_{+}\right)$ and
\[
A_{\mp}:=-\frac{1}{2}\left(C_{-}^{\dag}C_{\mp}-C_{+}^{T}C_{\pm}^{\#}\right)-i\Omega_{\mp}.
\] These are formally equivalent to the \textit{Langevin} equations
\begin{align}
 \label{langevin}      \breve{\bf \dot{\mathbf{a}}}(t) =& 
            A \breve{\bf a}(t)-C^{\flat}\Delta(S,0)\breve{\bf b}(t),\\
  \label{langevin3}     \breve{\bf b}^{out}(t) =& 
            C \breve{\bf a}(t)+ \Delta(S,0)\breve{\bf b}(t),
\end{align}
with $\breve{\bf b}^{out}(t):=d\breve{\bf B}(t)/dt$.
These equations may be solved; we give the solution to Eq. \eqref{langevin}:
\begin{equation}\label{lan2}
\breve{\mathbf{a}}(t)=e^{At}\breve{\mathbf{a}}(0)+e^{At}\left(\int_0^t e^{-As}(-C^b)d\breve{\mathbf{B}}(s)\right).
\end{equation}

More generally one may allow for static squeezing operations in the field in addition to the scattering processes. This is achieved by extending the  unitary scattering matrices $S=\Delta(S,0)$ to the symplectic group. We therefore denote by $S$ the squeezing and/or scattering in the field. 
%DIfficulty with modelling in QSDE

To be explicit, an arbitrary QLS is completely 
characterised by the parameters 
$\left(S, C, \Omega\right)$. Synonymously we also use the notation $\left({S}, C, A\right)$ here to characterise  the QLS. 
However one should be aware  that  not all choices of  $A$ are physically realisable as QLSs \cite{James1}. 
In fact, a general $(S, C, A)$ satisfying the equations \eqref{langevin} and \eqref{langevin3} realises a QLS if and only if 
 \cite{James2, Maalouf1}
\begin{equation}\label{PR}
A+A^{\flat}+C^{b}C=0.
\end{equation}
In some instances we shall restrict our attention to the case of no scattering or squeezing; where this is clear from the context we shall  use  the notation $(A, C)$ (or $(\Omega, C)$) to denote that $S=1$ in our model

Throughout this thesis we shall assume that the QLSs considered are ergodic. This means that if the state of the input is the vacuum, then in the long time limit the system converges to the vacuum state as the unique stationary state. 

\section{The Model: Frequency-Domain Representation}\label{freqrep}

We can also switch from the time domain dynamics 
described above to the frequency domain picture by taking 
 Laplace transforms. That is,  
\cite{Yanagisawa1}:
\begin{equation}\label{iol}
\breve{\bf b}^{out}(s) = 
%\mathcal{L}[\mathbf{b}^{out}](s)=
\Xi(s)
 \breve{\bf b}(s),
%\mathcal{L}[\mathbf{b}](s),
\end{equation}
where $ \Xi(s)$ is \emph{transfer function matrix} of the system%
\begin{equation}
\label{eq.transfer.function.general}
       \Xi(s)=\Big\{{1}_m-C(s{1}_n-A)^{-1}C^{\flat}S\Big\}=
       \left(\begin{smallmatrix}
       \Xi_{-}(s)&\Xi_{+}(s)
       \\
       \Xi_{+}(\overline{s})^{\#}& \Xi_{-}(\overline{s})^{\#}
       \end{smallmatrix}\right).
\end{equation}
and   $ \breve{\bf b}(s)$ and $\breve{\bf b}^{out}(s)$ are the Laplace transforms of $ \breve{\bf b}(t)$         and $\breve{\bf b}^{out}(t)$.
Although the same notation has been used for the input-output operators in the time- and Laplace-domains, it should be understood going forward that when we use $t$ ($s$) we are referring  to the time domain (Laplace-domain).

In particular, the frequency domain input-output 
relation is obtained by choosing $s=i\omega$ for $\omega\in\mathbb{R}$.
%
%$\breve{\bf b}^{out}(-i\omega) = 
%\mathcal{L}[\mathbf{b}^{out}](s)=
%\Xi(-i\omega)
% \breve{\bf b}(-i\omega).$
 %
The corresponding commutation relations are
$\left[\mathbf{b}(-i\omega),\mathbf{b}(-i\omega')^{\#}\right]
=i\delta(\omega-\omega')\mathds{1}$, and similarly for the output modes\footnote{Note that the position of the conjugation sign is important here because in general  $\mathbf{b}(-i\omega')^{\#}$ and $\mathbf{b}^{\#}(-i\omega')$ are not the same (see the definition of the Laplace transform).}. 
As a consequence, the transfer matrix $\Xi(-i\omega)$ is symplectic for all frequencies $\omega$ \cite{Gough1}. 
%Note that we have ignored the initial value contribution of the system modes, which is possible because the system is stable (and hence all poles are in the left hand complex plane). 
%Due to the simplicity of the input-output relation, our analysis will often be focused on 
%the frequency domain. 
In fact not only is this condition necessary for a QLS, but it was recently shown  to be sufficient (see \cite{James2}). That is, if $\Xi(s)$ is symplectic on the imaginary axis, then there is guaranteed to exist a system $(A, C)$ realising it. For this reason, this condition is termed the \textit{frequency domain physical relizability} (FPR) condition.

%One can also ask the question of whether a given transfer function is physically realisable. That is, given a transfer function, is there a genuine quantum system realising it. The answer to this was given in  \cite{PETERSEN?}, where such a necessary and sufficient  \textit{frequency domain relizability condition}was shown to be that $\Xi(s)$ be symplectic for all frequencies.

Finally, we note that while the transfer function is uniquely determined by the triple $(S, C, \Omega)$, the converse statement is not true, which is the subject of Ch. \ref{T.F}.
%as discussed in  further detail in Ch. \ref{T.F}.

\section{PQLS}\label{honk}

A special case of linear systems is that  of \textit{passive} quantum linear systems (PQLSs) for which $C_{+}=0$, 
$\Omega_{+}=0$ and $S_+=0$ \cite{Guta2}.   
They reason they are referred to as passive systems is because neither the Hamiltonian nor the coupling contain terms require an external source of quanta  \cite{Gough1, Nurdin2}, i.e., the evolution is purely dissipative. This class of QLS is still sufficiently rich to arise in many applications \cite{Wiseman1, Guta2, Petersen2, Nurdin6, Gough5}, and  include  optical cavities and beam splitters.

 When dealing exclusively with PQLSs,  the doubled-up notation is no longer necessary and we shall drop it; PQLSs are characterised by the triple 
$\left(S_-, C_-, \Omega_{-}\right)$ (or $\left(S_-, C_-, A_{-}\right)$), where the scattering matrix $S_-$ is unitary.

 The 
 input-output relation becomes \cite{Yanagisawa1,Guta2}
\begin{equation}
\label{iolp}
\mathbf{b}^{out}(s) = 
%\mathcal{L}[\mathbf{b}^{out}](s)=
\Xi(s)
 \mathbf{ b} (s),
%\mathcal{L}[\mathbf{b}](s),
\end{equation}
where the transfer function is given by
\begin{equation}
\label{eq.transfer.function}
       \Xi(s)=\Big\{{1}_m-C_{-}(s{1}_n-A_{-})^{-1}C_{-}^{\dag}\Big\}S,
\end{equation}
which is \textbf{unitary} for all $s= -i\omega\in i \mathbb{R}$.

The PR and FPR conditions in this case are: $A_-+A_-^{\dag}+C_-^{\dag}C=0$, and $\Xi(s)$ is unitary on the imaginary axis of the complex plane, respectively.

\section{Examples of Quantum Linear Systems}\label{pity}

\begin{figure}
\centering
\includegraphics[scale=0.18]{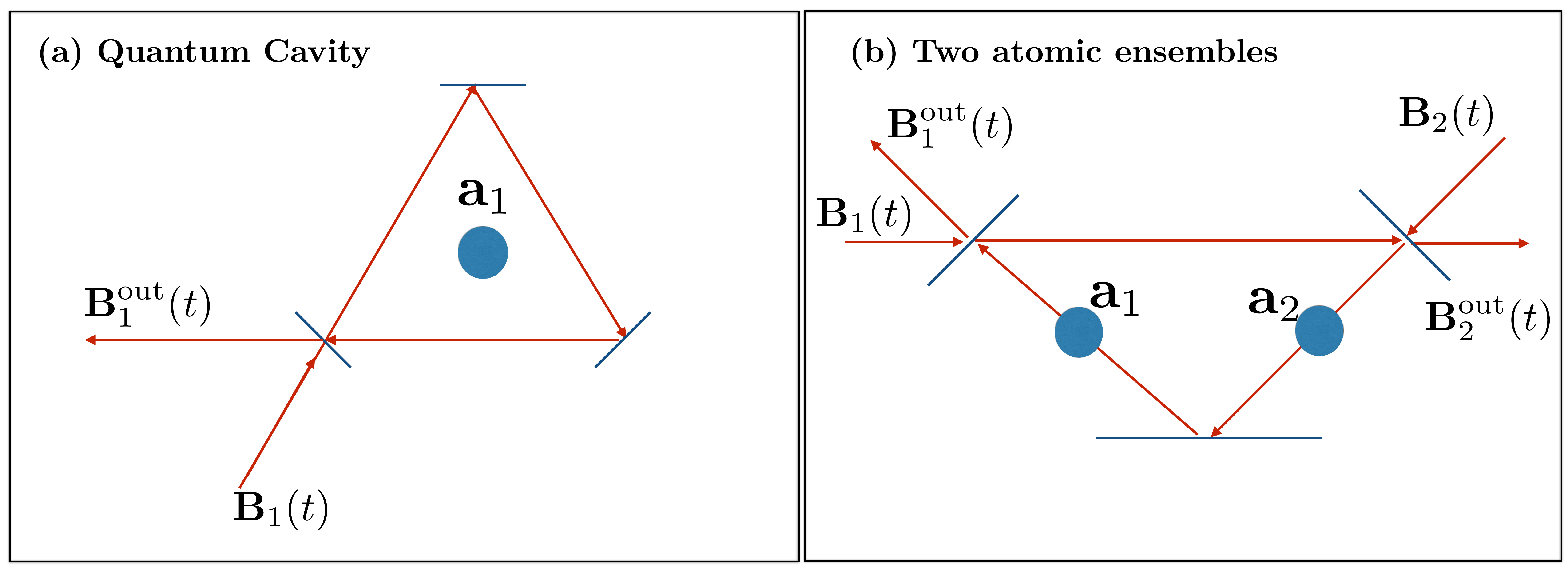}
\caption{ Circuit diagram illustration for (a) a quantum cavity and (b) two atomic ensembles.
\label{ensem}}
\end{figure}

\begin{exmp}\cite{Guta2, Gough3}\label{opticalcavity}
Our first example is an optical cavity, illustrated in Fig. \ref{ensem} (a), which is a passive QLS with one internal mode $\mathbf{a}$ and one field $\mathbf{B}(t)$. The system dynamics are given by
\begin{align*}
&d\mathbf{a}=(-i\omega_0-\frac{\kappa}{2})d\mathbf{a}-\sqrt{\kappa}d\mathbf{B}\\
&d\mathbf{B}^{out}=\sqrt{\kappa}d\mathbf{a}+d\mathbf{B},
\end{align*}
where $\kappa$ is the transmitivity of the coupling mirrors and $\omega_0$ is the detuning, which represents the frequency difference between the inner and outer optical fields \cite{Guta2}.
Note that in this case $\left(S, C, \Omega\right)=(1,\sqrt{\kappa}, \Omega_0)$.
\end{exmp}

\begin{exmp} \cite{Gardner1, Gough1} \label{DPA1} 
A degenerate parametric amplifier (DPA) can be modelled as  a single mode, $\mathbf{a}$, coupled to a single field, $\mathbf{B}(t)$. Here $S=1$, $\Omega_-=0$, $\Omega_+=\frac{i}{2}\epsilon$ ($\epsilon>0$), $C_-=\sqrt{\kappa}$ and $C_+=0$.

Now, using \eqref{eq.transfer.function.general} the transfer function of this system is given by
$$\Xi(s)=\frac{1}{s^2+\kappa s+\frac{\kappa^2+\epsilon^2}{4}}
\left(\begin{smallmatrix}s^2-\frac{\kappa^2+\epsilon^2}{4}&-\frac{1}{2}\epsilon\kappa\\-\frac{1}{2}\epsilon\kappa&s^2-\frac{\kappa^2+\epsilon^2}{4}\end{smallmatrix}\right).$$
Interestingly, by rescaling $\kappa=k\kappa_0$, $\epsilon=k\epsilon_0$, $s\mapsto \frac{s}{k}$ and taking the limit $k\mapsto\infty$ , results in the following symplectic transformation as the input-output map: 
$$\mathbf{b}^{out}(s)=\mathrm{cosh}(r_0)\mathbf{b}(s)+\mathrm{sinh}(r_0)\mathbf{b}^{\dag}(s),$$
where $r_0=\mathrm{ln}\left[(\kappa_0-\epsilon_0)/(\kappa_0+\epsilon_0)\right]$    \cite{Gardner1}.
 Therefore, the output is squeezed white noise from vacuum, which is constant across $s\in\mathbb{C}$ (see Sec. \ref{stato}). Essentially  this DPA implements a static squeezing device in the field for which the internal degrees of freedom of the DPA system are  eliminated \cite{Gough1}. This example was taken from \cite{Gough1}.
\end{exmp}

\begin{exmp}\cite{Muschik1, Guta2}\label{ensemble}
Finally,   consider the setup of  two coupled atomic ensembles investigated in \cite{Muschik1} (see Fig. \ref{ensem} (b)). In an ideal large ensemble limit, this system has the form of linear passive 
system with two probe inputs $\mathbf{b}=[\mathbf{b}_1, \mathbf{b}_2]^T$ 
and two outputs 
$\mathbf{b}^{out}=[\mathbf{b}_1^{out}, \mathbf{b}_2^{out}]^T$: 
\[
     \dot{\mathbf{a}} = -\frac{\kappa}{2} Y\mathbf{a} 
            - \sqrt{\kappa Y} \mathbf{b},~~~
     \mathbf{b}^{out} = \sqrt{\kappa Y} \mathbf{a} + \mathbf{b}, 
\]
where $\mathbf{a}=[\mathbf{a}_1, \mathbf{a}_2]^T$ denotes the pair of 
collective lowering operators 
of the atomic ensembles and $\kappa$ is the decay rate. 
%The inputs $\mathbf{b}_1$ and $\mathbf{b}_2$ are continuous modes with 
%different frequencies. 
The system-probe coupling is governed by the following matrix: 
\[
   Y =\left[ \begin{array}{cc}
       \cosh(2\theta) & -\sinh(2\theta) \\
       -\sinh(2\theta) & \cosh(2\theta) \\
     \end{array} \right], 
\]
which stems from the coupling operators 
$\mathbf{L}_1=\mathbf{a}_1 \cosh(\theta) + \mathbf{a}_2^{\dag} \sinh(\theta)$ 
and $\mathbf{L}_2=\mathbf{a}_2 \cosh(\theta) + \mathbf{a}_1^{\dag} \sinh(\theta)$. 
The parameter $\theta$ represents the coupling strength between the two 
ensembles.
In fact in the long time limit $t\rightarrow \infty$ the system's steady state 
becomes a two-mode squeezed state with squeezing level $\theta$; 
that is, this system deterministically generates an entangled state between 
two large-scale atomic ensembles by dissipation (if $\theta=0$ 
the ensembles are not entangled). 
Hence, the precise estimation of $\theta$ is important for testing whether such a macroscopic system exhibits entanglement (see Sec. \ref{Macroscopic} where we treat this parameter as unknown and try to identify it). 
 The more general case of \textit{pure Gaussian cluster states} may also be generated via a PQLS composed of atomic ensembles \cite{Guta2, Li1}. This example was worked from \cite{Guta2}.
 \end{exmp}

%\section{Filtering and weak measurements}
%[[[reference john gough filtering paper for coherent or single photon input filtering equations...i could discuss this weak measurement stuff later...seperate section]]]]...describes wider context.

\section{Time-Dependent vs Stationary Inputs}
As mentioned above, we are free to choose the input to the field channels. We distinguish two contrasting approaches; using time-dependent inputs or time stationary Gaussian (quantum noise) inputs (note the parallels with Sec. \ref{feb16}).  We discuss  both approaches  in turn in this subsection.

 \begin{figure}[h]
\centering
\includegraphics[scale=0.24]{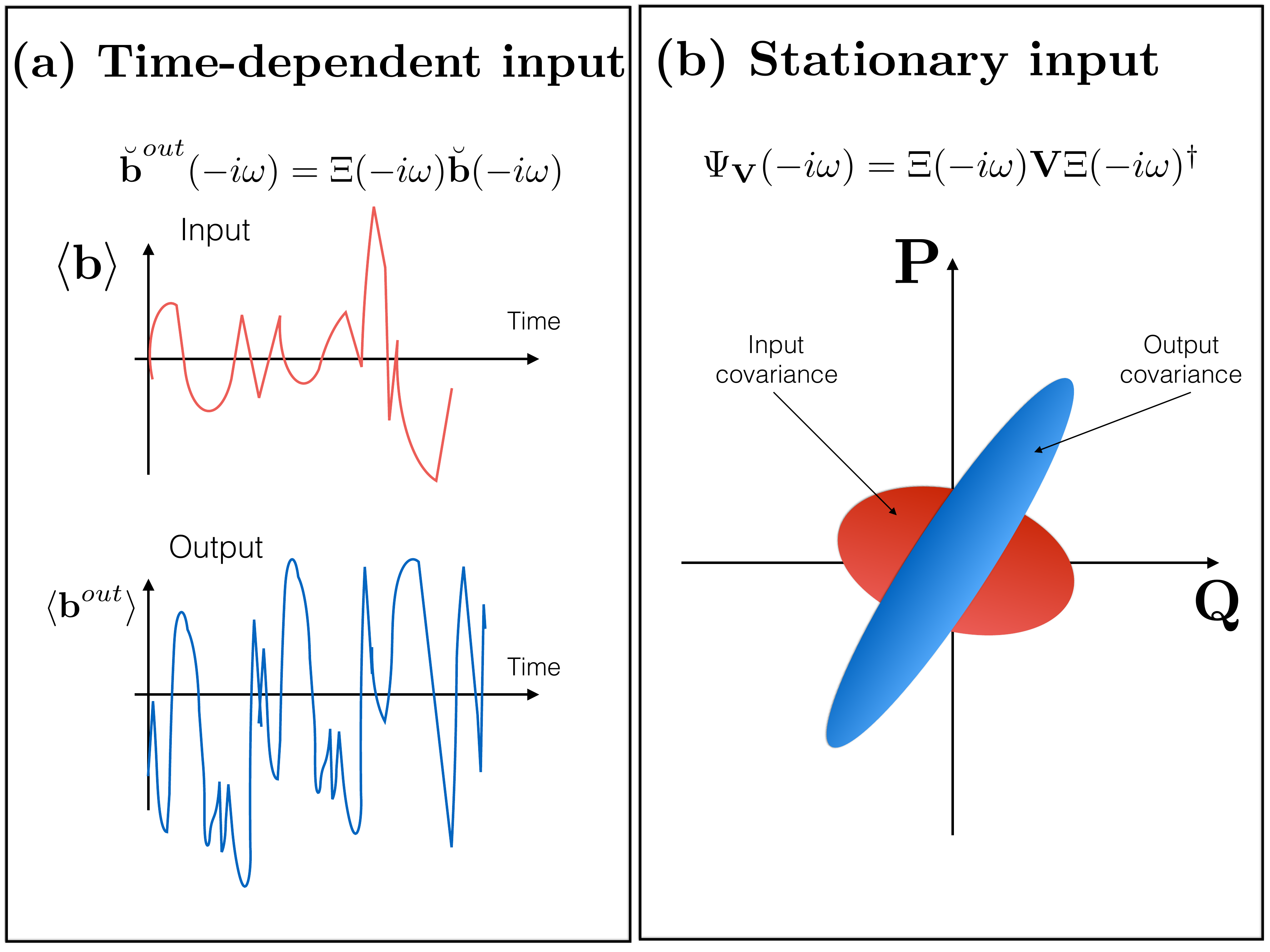}
\caption{(a) Time-dependent scenario: in frequency domain, input and output are related by the {\it transfer function} $\Xi(-i\omega)$ which depends on $(S,C,\Omega)$.  (b) Stationary scenario: \textit{power spectrum} describes output covariance which is quadratic with respect to 
$\Xi(-i\omega)$ .\label{trew}}
\end{figure}
 
\subsection{Time-Dependent Inputs}\label{ghu}

In the time-dependent approach, one probes the system with a known time-dependent input signal (e.g. a coherent state (example \ref{iffr})). One then performs measurements on the  time-dependent output to infer unknown dynamical parameters within the QLS (see Fig. \ref{trew}). The input-output relation \eqref{iol} shows that the experimenter can at most identify the transfer function $\Xi(s)$ of the system. Therefore, in this setting the transfer function entirely encapsulates the QLSs behaviour.
%Systems which have the same transfer function are called \textit{equivalent} and belong to the same equivalence class. 

\begin{exmp}\label{iffr}
As an example, suppose that we have a PQLS and probe with a coherent state of time dependent amplitude $\alpha(t)$, which is defined as 
\begin{equation}\label{cohg}
\left|\alpha(t)\right>=\mathrm{exp}\left[i\int^{\infty}_{-\infty}\breve{\alpha}(t)\breve{\mathbf{b}}^{\dag}(t)dt\right]\left|0\right>,
\end{equation}
where $\left|0\right>$ is vacuum (ground) state. Note that the quantity $\int^{\infty}_{-\infty}|\alpha(t)|^2dt$ gives the energy of the coherent pulse; so $|\alpha(t)|^2$ may be interpreted as the mean number of photons at time $t$. Notice the similarities with coherent states of QHOs (see Eq. \eqref{cough}). 

In the frequency domain, the input is  a product of coherent states over frequency modes, each with amplitude $\alpha(\omega)$ (the Fourier transform of $\alpha(t)$).
It follows from \eqref{iolp} that the  output state frequency  profile will be $\Xi(-i\omega)\alpha(\omega)$. 
%Note that in practice one works in the time-domain and obtains the input and output frequency profile by taking the Fourier transform.
Finally, as the transfer function is unitary, the result of the interaction with the system is a frequency dependent rotation of the coherent state amplitude. 
\end{exmp}

\subsection{Stationary Inputs}\label{stato}

In the  second approach we consider probing  the system with \textit{time-stationary} pure gaussian states (see Fig. \ref{trew}).
That is, the input state is `squeezed quantum noise', i.e. a zero-mean, pure Gaussian state with time-independent increments. It is completely characterised by its covariance matrix ${ V}$, which is given by 
\begin{align}\label{Ito}
\nonumber\left<\begin{smallmatrix}d\mathbf{B}(t)d\mathbf{B}(t)^{\dag}&d\mathbf{B}(t)d\mathbf{B}(t)^T\\d\mathbf{B}^{\#}(t)d\mathbf{B}(t)^{\dag}&d\mathbf{B}^{\#}(t)d\mathbf{B}(t)^T\end{smallmatrix}\right>&=\left(\begin{smallmatrix} N^T+{1}&M\\M^{\dag}&N\end{smallmatrix}\right)dt
\\&:={V}(N,M)dt.
\end{align}
 In the frequency domain the input state turns out to be a continuous tensor product over frequency QHO modes of squeezed states  with covariance ${V}(N,M)$ (recall Sec. \ref{QHO}). Therefore, the  conditions on $V(N, M)$ transfer directly from Sec. \ref{QHO}. By the requirement that the input be \textit{pure}, we mean that the covariance matrix  $V(N,M)$ is that of pure state, where the conditions are defined in Sec. \ref{QHO}.

Under ergodicity, the dynamics exhibits an initial transience period after which it reaches stationarity. Therefore, since we deal with a linear systems the input-output map consists of applying a (frequency dependent) unitary Bogolubov transformation whose linear symplectic action on the frequency modes is given by the transfer function 
$$
\breve{\bf b}^{out} (-i\omega)  =  
%S(-i\omega)^\dagger \breve{\bf b} (-i\omega) S(-i\omega) = 
\Xi (-i\omega) \breve{\bf b} (-i\omega).
$$
Consequently, the output state is a gaussian state consisting of independent frequency modes with covariance matrix
\begin{eqnarray*} 
\left< \breve{\bf b}^{out} (-i\omega) \breve{\bf b}^{out} (-i\omega^\prime)^\dagger \right>=   
 \Psi_{ V}(-i\omega) \delta(\omega-\omega^\prime)  
 %=: \Xi(-i\omega)V\Xi(-i\omega)^{\dag} \delta(\omega-\omega^\prime).
 \end{eqnarray*}
 where $ \Psi_{ V}(-i\omega)$ is the restriction to the imaginary axis of the \emph{power spectral density} (PSD) (or power spectrum) defined in the Laplace domain by
\begin{equation}\label{powers}
\Psi_{ V}(s)= \Xi(s){ V}\Xi(-s^*)^{\dag}.
\end{equation}

We should comment on the physicality of (quantum) white noise. As we have seen above, the PSD of white noise is constant and as a result, by integrating over all frequencies, the average power will be infinite. As no real process can have infinite signal power, this type of input is therefore   purely theoretical. However, many real and important stochastic processes have a PSD that is approximately constant constant over a wide range of frequencies. In practice one works with band-limited white noise, that is, where the frequency is constant over a finite range and zero outside. Band-limited white-noise necessarily has finite signal power. Band-limited white noise is generated using a low-pass filter.

\section{Controllability, Observability, Minimality, Stability}\label{hastings}
We now reconnect with the system theoretic concepts from Sec. \ref{systheo}.

By taking the expectation with respect to the initial joint system state of Eqs. (\ref{langevin}) we obtain the following classical linear system
\begin{eqnarray}\label{classicallangevin}
       d\left<\breve{\bf a}(t) \right>= 
            A\left< \breve{\bf a}(t) \right>dt-C^{\flat}d\left<\breve{\bf B}(t)\right>,\\
       d\left<\breve{\bf B}^{out}(t)\right> = 
            C\left<\breve{\bf a}(t) \right>dt+d\left<\breve{\bf B}(t)\right>.
\end{eqnarray}

\begin{defn}
The quantum linear system (\ref{langevin}) is said to be Hurwitz stable (respectively controllable, observable, minimal) if the corresponding classical system (\ref{classicallangevin}) is Hurwitz stable (respectively controllable, observable, minimal).
\end{defn}

For a general QLS  observability and controllability are equivalent \cite{Gough2} (see also \cite{Guta2} for PQLS result). Therefore, to verify minimality we need only check one of these properties. 
In the case of passive systems Hurwitz stability is further equivalent to minimality of the system \cite{Guta1,Guta2}. However for active systems, only the statement [Hurwitz $\implies$ minimal]  is true, whereas the converse statement ([minimal $\implies$ Hurwitz]) is not necessarily so. We see the former statement  in the following Lemma  and then discuss  a counterexample to the latter in example \ref{wrong}.

\begin{Lemma}
If a QLS $(S, C, A)$ is Hurwitz, then it is minimal.
\end{Lemma}
\begin{proof}
We will prove this statement by contradiction. To this end, suppose that $(S, C, A)$ is non-minimal, which is equivalent to it being non-observable. Therefore, by  Theorem \ref{obs5} part \ref{obs3}, there exists   an eigenvector $y$ and eigenvalue $\lambda$ of $A$, such that $Cy=0$. This entails that $-iJ\Omega y=\lambda y$. Hence using the self-adjointness of $\Omega$, we have
\begin{equation}\label{footy}
\left(y^{\dag}J\right)\left(iJ\Omega\right)=\overline{\lambda}\left(y^{\dag}J\right)
\end{equation}
Furthermore, \begin{equation}\label{footy2}
(y^{\dag}J)(C^{\flat}C)=0\end{equation} from $Cy=0$.
Combining \eqref{footy} and \eqref{footy2} implies that $\left(y^{\dag}J\right)A=-\left(y^{\dag}J\right)\overline{\lambda}$. 
So in summary $\lambda$ and $-\overline{\lambda}$ are eigenvalues of $A$, which cannot be under stability.

Note that an alternative proof of this  Lemma is given  in  \cite{Yamamoto2}.
\end{proof}

\begin{exmp}\label{wrong}
Consider a general one-mode SISO QLS, which is parameterised by $\Omega=\Delta(\omega_{-},  \omega_{+})$ and $C=\Delta\left(c_-, c_+\right)$.
The system is Hurwitz stable (i.e. the eigenvalues of $A$ have strictly negative real part) if and only if
\begin{enumerate}
\item $|c_-|>|c_+|$ and $|\omega_-|\geq|\omega_+|$, or
\item $|\omega_+|>|\omega_-|$  and $\sqrt{|\omega_+|^2-|\omega_-|^2}<\frac{1}{2}\left(|c_-|^2-|c_+|^2\right)$.
\end{enumerate}
A system is non-minimal if and only if the following matrix has rank less than two: 
\[
\left[
\begin{smallmatrix}
C\\CJ_n\Omega
\end{smallmatrix}
\right]=
\left[\begin{array}{cc}
c_-&c_+
\\
{\overline{c}_+}&{\overline{c}_-}
\\
c_-\omega_- -c_+{\overline{\omega}_+} &c_-\omega_+ -c_+\omega_-
\\ 
{\overline{c}_+}\omega_- -{\overline{c}_-}{\overline{\omega}_+}   &  {\overline{c}_+}\omega_+ -{\overline{c}_-}\omega_-
\end{array}\right]. 
\]
For a  counterexample to the statement:  [minimal $\implies$ Hurwitz] consider for example   $|c_+|>|c_-|$ with $\omega_+=\omega_-$.
\end{exmp}

In light of the previous example, we make the physical assumption that all systems considered throughout this paper are Hurwitz (hence minimal). Note that the \textit{ergodic} assumption in Sec. \ref{MTDR} is equivalent to the system being Hurwitz. 
Notice that if $t|\mathrm{Re}(\lambda_s(A))|\ll 1$, where $\lambda_s(A)$ is the eigenvalue of $A$ lying closest to the imaginary axis, then the first term in Eq. \eqref{lan2} contributes negligibly. Therefore $\lambda_s(A)$ determines the time taken for the system  to stabilise. We call $|\mathrm{Re}(\lambda_s(A))|$ the \textit{spectral gap}  and will be useful in Ch. \ref{FEDERER}. Moreover, in the case of stationary inputs the relaxation of the system to stationary state is related to the spectral gap.  That is, $||\rho(t)-\rho_{ss}||\sim e^{t     |\mathrm{Re}(\lambda_s(A))|  } $ \cite{Kasia1}, where $\rho(t)$ and $\rho_{ss}$ are the state of the system at time $t$ and the stationary state of the system, respectively.

Finally we remark that minimality is a reasonable assumption to make in the quantum regime too.  For if a system is non-minimal, then there exists  a minimal system that has the  same transfer function as the original system and crucially it satisfies the PR conditions \cite{Zhang4}.
 Such a system can be obtained by using the \textit{ quantum Kalman decomposition}  \cite{Zhang4}. This result was essentially obtained by combining the  \textit{classical} Kalman decomposition \cite{Zhou1}  with the FPR condition (see Sec. \ref{freqrep}). 
As the transfer function contains maximal information about the system, therefore
 minimality ensures the simplest  description of the input-output behaviour.

%Put in statement about eigenvalues of A.

%Link to classical stuff...results there valid here...we have more because of quantum nature. 

\section{Operations on QLSs}\label{greed}

The network rules discussed in Sec \ref{network1} for CLSs, i.e the concatenation and series connections, also hold for QLSs.  Notice the simplicity of the frequency domain in dealing with cascaded systems - their transfer functions  multiply. Note that if $(S_1, C_1, \Omega_1)$ with transfer function $\Xi_1(s)$ is fed into system $(S_2, C_2, \Omega_2)$ with transfer function $\Xi_2(s)$ the overall transfer function is  given by $\Xi(s)=\Xi_2(s)\Xi_1(s)$, rather than $\Xi_1(s)\Xi_2(s)$.
\begin{remark}
A remark that these rules are valid provided that the delays in the circuits  between the components are negligible compared with the interaction times within the QLS. Treating the case of finite time delays leads to a non-Markovian model and as such  is often not tractable. However, we would like to make the reader aware of  \cite{Grimsmo1}, where  a novel application of tensor networks is used for non-negligible time-delays.
\end{remark}
%Series product and concatenation...explain meaning. do quantum rule and classical so that can refer later.

%Use classical result to get quantium one . instantaneous feedforward...refernce if not instantaneous in quantum case but say don't refer to. (Grimsmo) 

%[[[NEED to look at identifiability paper to see what need to put in from there. 

%could do passive also 

%REFERENCE JOHN GOUGH PAPER>

%Comment that have reorder a's and a daggers .

%[[[LOOK at end of yr reports to see whether included everything]]]

%%%%%%%%%%%%%%%%%%%%%%%%%%%%%%%%%%%%%%
%%%%%%%%%%%%%%%%%%%%%%%%%%%%%%%%%%%%%%
\chapter{Quantum Fisher-Information and the Heisenberg Limit}
\label{plm}
%%%%%%%%%%%%%%%%%%%%%%%%%%%%%%%%%%%%%%
%%%%%%%%%%%%%%%%%%%%%%%%%%%%%%%%%%%%%%

\section{Estimation Theory}\label{plm1}

In this section we review the relevant aspects of quantum estimation theory 
that are used in this thesis. 
%Important concepts here are the quantum Fisher information and the 
%quantum Cram\'{e}r-Rao lower bound, which give a measure of the 
%performance of a given experimental design \cite{qcrlbw}. 
%This will be followed by a discussion on the fundamental bound of 
%quantum-enhanced precision for PQLSs. This section concludes with a 
%discussion on the current status of the field which provides a basis for 
%our new results.
Let  $\rho_{{\theta}}$ be a quantum state which depends on a $d$-dimensional unknown parameter 
${\theta}=[\theta_1, \theta_2,\dots , \theta_d]^T$. In quantum estimation theory the goal is to estimate $\theta$ by performing a measurement 
$\mathbf{M}$ and constructing an estimator ${\hat{{\theta}}}(X)$ based on the measurement outcome $X$. 
Suppose for simplicity that $X$ takes values in a discrete set $\{1,\dots, p\}$. Then the  probability distribution of the outcomes is of the form $\mathbb{P}_\theta(X=i) = {\rm Tr} (\rho_\theta \mathbf{M}_i)$ where $\{\mathbf{M}_1,\dots,  \mathbf{M}_p\}$ is the positive operator valued measure associated to $\mathbf{M}$. The (classical) Cram\'{e}r-Rao 
bound gives a measure of the 
performance of a given measurement \cite{Braun1}; if $\hat\theta$ is an arbitrary unbiased estimator 
($\mathbb{E}({\hat{{\theta}}})=\theta$) then its error covariance matrix is 
lower bounded as  
\begin{equation}
\label{eq.classical.cr}
   {\rm Cov}(\hat{\theta}):= 
     \mathbb{E}_\theta
         \Big[ \big({\hat{\theta}} - \theta\big) \big({\hat{\theta}} -\theta\big)^T
                    \Big] 
            \geq I^{(M)}({\theta})^{-1}.
\end{equation}
The right side is the inverse of the (classical) Fisher information matrix $I^{(M)}({\theta})$ (CFI), associated to the outcome $X$ of $\mathbf{M}$, whose matrix elements are given by 
\begin{eqnarray}\label{cod}
I^{(M)}({\theta})_{i,j} = 
\mathbb{E}_\theta\left[\frac{ \partial \mathbb{P}_\theta (X)}{\partial \theta_i}\frac{ \partial \mathbb{P}_\theta(X)}{\partial \theta_j}\right]. %= \Sum_{k=1}^p .
\end{eqnarray}

%Recall from statistics, that the Cramer-Rao lower bound expresses a lower bound on the error covariance of an estimator of certain deterministic parameters to be estimated \cite{cramer, rao}. 
If one considers arbitrary measurements,  an additional level of optimization 
is required. 
The \emph{quantum Cram\'{e}r-Rao bound} (QCRB) establishes that any 
CFI matrix is upper bounded as
\begin{equation}\label{eq.quantum.cr}
 I^{(M)}(\theta)\leq F(\theta).
\end{equation}
On the right side, $F({\theta})$ is the quantum Fisher-information (QFI) matrix 
for the family $\rho_{{\theta}}$, whose elements are defined by \cite{Holevo1} 
\[
     F({\theta})_{lm}
       =\frac{1}{2}\mathrm{Tr}\left[\rho_{{\theta}}\{\mathbf{L}_l, \mathbf{L}_m\}\right],
\]
where $\mathbf{L}_m$ is the symmetric logarithmic derivative (SLD) defined implicitly 
by $$\partial\rho_{{\theta}}/\partial\theta_m
=\left(\rho_{{\theta}}\mathbf{L}_m+\mathbf{L}_m\rho_{{\theta}}\right)/2.$$

%One of the properties of the QFI is that it is convex and so mixing quantum states can only reduce the sensitivity. Therefore, the state with the largest QFI must be pure. 
For families of pure states $\rho_{{\theta}}=\Ket{\psi_{{\theta}}}\Bra{\psi_{{\theta}}}$ the QFI has an explicit expression
\begin{equation}\label{qfiformula}
F({\theta})_{lm}=4\mathrm{Re}\left(\Braket{\partial_l\psi_{{\theta}}|\partial_m\psi_{{\theta}}}-\Braket{\partial_l\psi_{{\theta}}|\psi_{{\theta}}}\Braket{\psi_{{\theta}}|\partial_m\psi_{{\theta}}}\right),
\end{equation}
where $\ket{\partial_m\psi_{{\theta}}}:=\partial\ket{\psi_{{\theta}}}/\partial\theta_m$. In particular, if 
$\ket{\psi_{\theta}}=\exp(-i\mathbf{G}\theta)\ket{\psi_0}$ is a one dimensional unitary 
rotation model with generator $G$, the QFI is 
\begin{equation}\label{unit} 
F(\theta)=4\mathrm{Var}\left(\mathbf{G}\right) = 4\left[\langle\psi_0 |\mathbf{G}^2|\psi_0\rangle - \langle\psi_0 | \mathbf{G}|\psi_0\rangle^2\right].
\end{equation}
The QCRB provides a lower bound for the risk of unbiased estimators, for instance 
the mean square error. 
By taking trace on both sides of Eq. \eqref{eq.classical.cr} and using 
Eq. \eqref{eq.quantum.cr} we have
\begin{equation}
\label{rtvz} 
     \mathbb{E} \| \hat{\theta}- \theta\|^2 
           \geq\mathrm{Tr}\left(F({\theta})^{-1}\right).
\end{equation}
The above bounds are generally not tight, for reasons which are both classical 
and quantum. 
The classical Cram\'{e}r-Rao bound (\ref{eq.classical.cr}) is achievable only 
for special ``exponential models"; it is however \emph{asymptotically} achievable 
for any model, in the limit of large numbers of \emph{independent} samples 
from the distribution $\mathbb{P}_\theta$, which could be obtained by 
performing repeated measurements on identically prepared systems. 
The QCRB \eqref{eq.quantum.cr} is in general not 
achievable, even asymptotically, i.e. when the experimenter is allowed to 
perform arbitrary measurements on the ensemble $\rho_\theta^{\otimes n}$ 
of $n$ identically prepared systems. 
Indeed, it can be shown \cite{Guta5} that in order to optimally estimate the component 
$\theta_i$ one needs to measure the collective observable 
$\mathbf{L}^n_i:= \sum_{j=1}^n \mathbf{L}_i^{(j)}$ where $\mathbf{L}_i^{(j)}$ is the symmetric logarithmic 
derivative of the $j$th system with respect to $\theta_i$. 
These observables are generally incompatible unless the following commutativity 
condition holds: ${\rm Tr}( \rho_\theta [\mathbf{L}_i , \mathbf{L}_j])= 0$, for all $i=1,\dots , d$. 
In the special case of one-dimensional parameter, the QCRB is asymptotically 
achievable for large sample numbers, by separately measuring the SLD 
$\mathbf{L}^{(j)}$ of each system.

%%%%%%%%%%%%%%%%%%%%%%%%%%%%%%%%%%%%%%
\section{Quantum Metrology and the Heisenberg Limit}
\label{class}
%%%%%%%%%%%%%%%%%%%%%%%%%%%%%%%%%%%%%%

Quantum metrology is the ingenuity of exploiting powerful properties of quantum theory, such as entanglement and squeezing, in order to develop high-resolution  measurements of physical parameters \cite{Braun1}. Such properties have no classical analogue, and as a result one is able to develop measurement techniques that give better precision (i.e. the so-called Heisenberg limit) than is possible in the classical framework. Already, various aspects of quantum metrology for system identification have been considered in the recent literature \cite{New1, New2, New4, New5, New6}. 

Above, we considered the case of a single, or several independent, identically 
prepared systems. 
In quantum metrology we often encounter a scenario in which a single unknown 
parameter $\theta$ is encoded by means of a unitary transformation 
$\mathbf{U}_{\theta}=e^{-i\theta \mathbf{G}}$ acting in parallel on an ensemble of $N$ quantum 
systems (probes). 
We distinguish two set-ups, as illustrated in Fig. \ref{classquant}; either the 
probes are prepared in the product state $\ket{\psi}^{\otimes N}$ (or more 
generally a separable state), or they are prepared in a general entangled state.

%\begin{itemize}
%\item[1)] The probes are prepared in the separable state $\ket{\psi}^{\otimes N}$. 
%Local measurements are performed on each independent channel.
%\item[2)] The probes are prepared in an entangled state, with evolution and measurement steps the same as in Strategy 1. 
%\end{itemize}

The first case corresponds to a classical procedure in which $N$ identical 
experiments are repeated. 
As shown in Eq. (\ref{unit}), in this case the QFI is $4N\mathrm{Var}({\mathbf{G}})$. 
This can be maximised by choosing a probe that is an equal superposition of 
the largest and smallest eigenstates of the generator $\mathbf{G}$, namely 
$\left(\ket{\lambda_{min}}+\ket{\lambda_{max}}\right)/\sqrt{2}$ \cite{Giov1}, 
so that $F_{max}(\theta)=N(\lambda_{max}-\lambda_{min})^2$. 
The linear scaling in $N$ is called the \textit{standard quantum limit} (SQL), 
and is the best one can achieve with separable states. 
%Moreover, for independent samples the  will converge to a Gaussian distribution 
%with variance $\mathcal{O}(N^{-1})$. This limit in quantum metrology is referred 
%to as the \textit{standard quantum limit} (SQL).

In the second case, the unitary transformation is applied to the joint state 
$|\psi^N\rangle$ such that 
$|\psi^N_\theta\rangle: = \mathbf{U}_\theta^{\otimes N} |\psi^N\rangle$. 
In this case, the best probe state is an equal superposition of the maximum 
and minimum eigenvectors of the global operator $\mathbf{G}^N:= \sum_i \mathbf{G}^{(i)}$. 
That is, choosing $\ket{\psi^{\otimes N}}=\left(\ket{\lambda_{min}}^{\otimes N}+
\ket{\lambda_{max}}^{\otimes N}\right)/\sqrt{2}$ leads to the QFI  
$F(\theta)=N^2\left(\lambda_{max}-\lambda_{min}\right)^2$. 
Therefore, by exploiting quantum resources such as entanglement, we obtain an 
$\mathcal{O}(N)$ improvement over purely classical schemes. 
This $\mathcal{O}(N^2)$ scaling is referred to as the \textit{Heisenberg limit} 
and for a fixed number of systems has been claimed as fundamental as it is 
brought about by Heisenberg-like uncertainty relations \cite{Bollinger1}. 
 It is worth  mentioning here that in general Heisenberg scaling holds for unitary channels, while for noisy channels one typically recovers the standard scaling  \cite{Demo1}.

\begin{figure}
\centering
\includegraphics[scale=0.16]{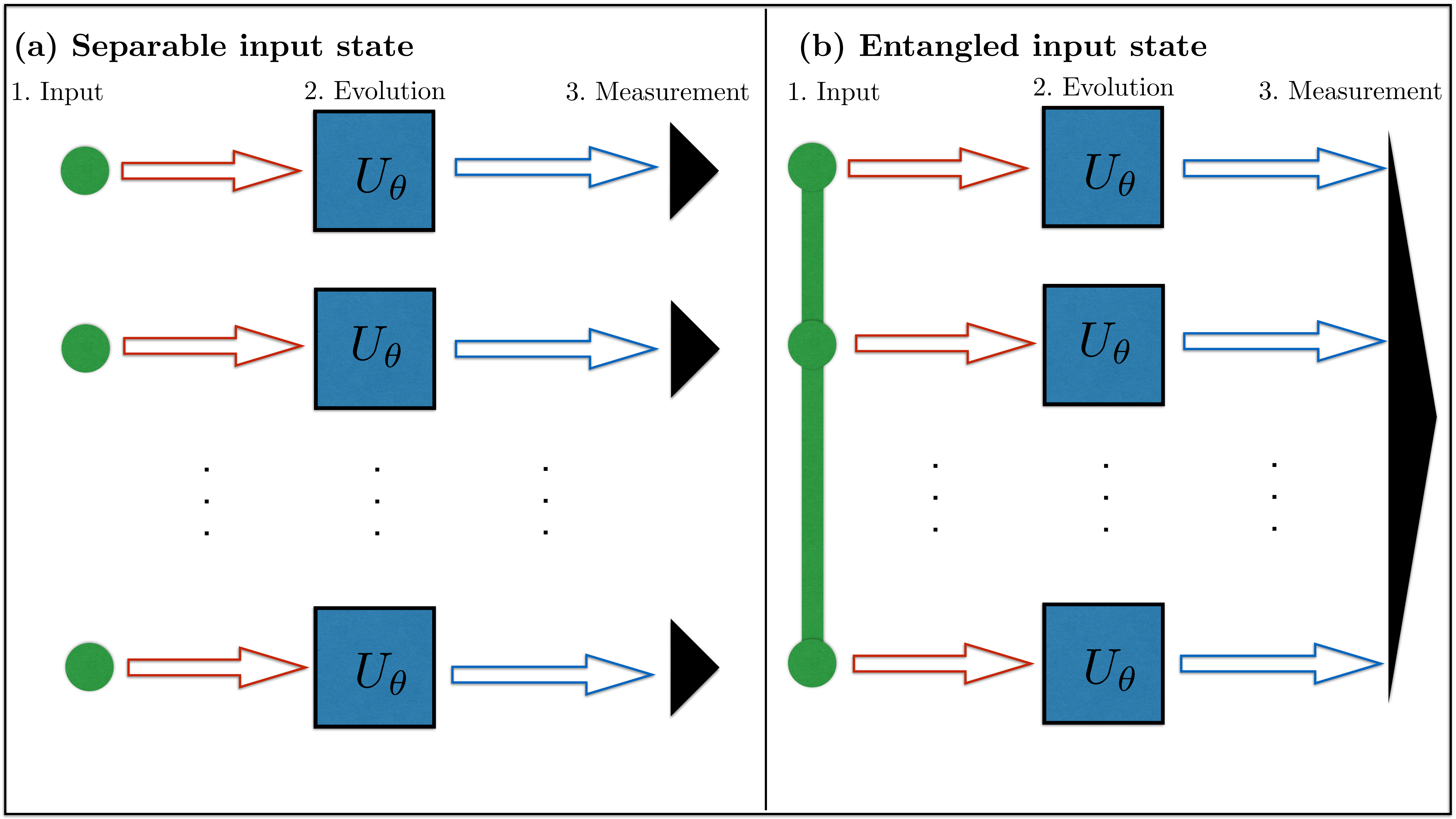}
\caption{This figure compares the classical and quantum estimation strategies 
discussed in Sec. \ref{class}. 
In panel (a) one inputs a product state. 
The upshot is that each channel corresponds to an identical and independent 
experiment. 
In panel (b) the input is entangled over $N$ channels. 
In each setup, there are three stages for the estimation procedure, which are: 
1) preparing an input, 2) evolution of the state, and 3) measurement of the output. 
\label{classquant}}
\end{figure}

%In general, the above bound is not tight for several reasons. due to incompatible measurements corresponding to different parameter components.  For pure states, the condition that Equation \ref{rtvz} be tight is that
%$\Braket{\partial_l\psi_{{\theta}}|\partial_m\psi_{{\theta}}}\in\mathbb{R}$ for all $l$ and $m$ \cite{podds}. However, even if this bound is tight then the optimal measurement may either depend on the unknown parameter or involve projections onto complicated states, thus rendering it experimentally infeasible.

%For a one-dimensional parameter space $\theta\in\mathbb{C}$ the QCRLB inequality simplifies. A useful class of parametric models for us are \textit{unitary rotational families}. That is, $\theta$ is encoded on a given input state by a unitary operator, $\mathbf{U}(\theta)$, according to $\ket{\psi_{\theta}}=\mathbf{U}(\theta)\ket{\psi_0}$. Using Stone's theorem \cite{Stone}, we may write $\mathbf{U}=e^{-i\mathbf{G(\theta)}}$ for some Hermitian operator $\mathbf{G}$.
%The QFI for unitary rotational families with pure input can be shown to be given by 
%\begin{equation}\label{unit} 
%I(\theta)=4\mathrm{Var}\left(\frac{d\mathbf{G}(\theta)}{d\theta}\right),
%\end{equation}
%where the variance of this operator is understood to be taken over the state $\ket{\psi_{\theta}}$. Moreover $I(\theta)=I(0)$ for all $\theta$.
%Finally, in this case the above tight bound condition is satisfied.

\part{Results}\label{partr}

\chapter{Transfer Function Identification }\label{T.F}

The purpose of this chapter is to investigate system identification for QLSs using  \textbf{time-dependent inputs} and by observing the output fields. In particular, we address the following questions: 
\begin{enumerate}
\item[(1)]Which parameters can be identified by measuring the output?
\item[(2)]  How can we construct a system realisation from sufficient input-output data? 
\end{enumerate}

We saw in Sec. \ref{ghu} that the input-output dynamics  is completely characterised by the transfer function for the case of time-dependent inputs. The first question above reduces to finding all equivalent systems with the same transfer function. We call two  QLSs with the  same transfer function,  \textit{transfer function equivalent} (TFE). We  answer this in Sec.  \ref{idenf} by finding the quantum analogue of the classical result in Theorem \ref{class.id2}. Note that this problem has been addressed for the special class of passive QLSs in \cite{Guta2}. The tricky part in answering the second question is in ensuring that the realisation of the transfer function corresponds to a physical QLS in the sense that it satisfies the physical realisability conditions \eqref{PR}. 
We give two two such realisation methods in Secs. \ref{pond} and \ref{indirect}.

Finally, we consider these identifiability problems in the context of noisy QLSs in Sec. \ref{noisek}. Noise may be modelled with the use of additional channels that cannot be accessed (i.e., measured or observed) \cite{Peter1}. As a result, when restricting the output to  only a subset of the channels the equivalence classes from Sec. \ref{idenf} grow larger.  
%As such the identifiability problems become more involved.     
Finally, in Sec. \ref{Noisek4} we investigate the notion of \textit{noise unobservable subspaces} where part of the system is shielded from the noise. Noise unobservable subspaces could potentially offer applications to metrology or control.

%\section{Minimality}

%We now formulate the notion of \textit{minimality}, which will be a central concept throughout this thesis. 

%\begin{defn}
%A quantum linear system is said to be minimal iff it is both observable and controllable. 
%\end{defn}

%In fact, for a quantum linear system observability and controllability are equivalent \cite{Indep}. Therefore, in order to verifying minimality is equivalent to checking whether the following observability matrix has rank $2n$:
%\[\mathcal{O}=[C^T, (CJ_n\Omega)^T,...., \left(C(J_n\Omega)^{2n-1}\right)^T]^T.\]

\section{Identifiability}\label{idenf}

In this section we address the first  identifiability question above for  general QLSs.
%The identifiability problem is illustrated in Fig. \ref{one}.
%Note that systems having the same transfer function cannot be 
%distinguished by any output measurements \cite{Guta}. 
Before giving our main result let us review the case of passive QLSs considered in 
\cite{Guta2}. Firstly, the transfer function  can be identified for example by following the method in example \ref{iffr}. For PQLS it is known that
two minimal systems with parameters $(\Omega,C,S)$ and $(\Omega', C',S')$ are 
equivalent if and only if their parameters are related by a unitary transformation, 
i.e. $C'=CT$ and $\Omega'=T\Omega T^{\dag}$ for some $n\times n$ unitary 
matrix $T$, and $S= S'$. 
The first part of this result was shown in \cite{Guta1}; the fact that the scattering 
matrices must be equal follows by choosing $s=-i\omega$ and taking the limit 
$\omega\to \infty$ in Eq. (\ref{eq.transfer.function}). 
Physically, this means that at frequencies far from the internal frequencies 
of the system, the input-output is dominated by the scattering/squeezing between the 
input fields. 
%Equivalently, the system modes are related by a unitary transformation. 
% Thus, the parameter space may be partitioned into equivalence classes. 

Our  main identifiability result for this section  is given in the following theorem.

\begin{thm}\label{symplecticequivalence}
Let $\left({S}, C, \Omega\right)$ and $\left({S'}, C', \Omega'\right)$ be two minimal and stable QLSs. 
Then they have the same transfer function if and only if there exists a symplectic matrix $T$  (see Definition \ref{def.symplectic}) such that 
\begin{equation}\label{eq.equivalene.classes}
J_n\Omega'=TJ_n\Omega T^{\flat}, \,\,\, C'=CT^{\flat}\,\,\,S=S'.
\end{equation}
\end{thm}
\begin{proof}
Firstly, using the same  argument as above, the scattering/squeezing matrices $S$ and $S'$ must be equal.

It is known \cite{Ljung1} that two minimal classical linear systems
\begin{equation}\label{CE1}
d\mathbf{x}(t)=A\mathbf{x}(t)dt+B\mathbf{u}(t)dt,\,\,\,d\mathbf{y}(t)=C\mathbf{x}(t)dt+D\mathbf{u}(t)dt
\end{equation}
and 
\begin{equation}\label{CE2}
d\mathbf{x}(t)=A'\mathbf{x}(t)dt+B'\mathbf{u}(t)dt,\,\,\,d\mathbf{y}(t)=C'\mathbf{x}(t)dt+D'\mathbf{u}(t)dt
\end{equation}
for input $\mathbf{u}(t)$, output $\mathbf{y}(t)$ and system state $\mathbf{x}(t)$ have the same transfer function if and only if
\[A'=TAT^{-1}\,\,\, B'=TB\,\,\,C'=CT^{-1}\,\,\,D'=D\]
for some invertible matrix $T$.
Therefore $C\left(s{1}-A\right)^{-1}C^{\flat}=C'\left(s{1}-A'\right)^{-1}C'^{\flat}$
if and only if there exists an invertible matrix $T$ such that 
  \begin{equation}\label{chelsea}
  A'=TAT^{-1}\,\,\, C'^{\flat}=TC^{\flat}\,\,\,C'=CT^{-1}.
  \end{equation}
Note that at this stage $T$ is not assumed to be symplectic; we show that by imposing the restriction that the two systems \eqref{CE1} and \eqref{CE2} be quantum implies that $T$ must be symplectic.
Now, the second and third conditions in \eqref{chelsea} implies $C=C\left(T^\flat T\right)$, which further implies that $\lbrack T^\flat T,C^{\flat}C\rbrack=0$. Also, by earlier definitions $A=-\frac{1}{2}C^\flat C-iJ_n\Omega$, so that the second and third conditions applied to the first condition (in \eqref{chelsea}) implies that $J_n\Omega'=TJ_n\Omega T^{-1}$. Next, using this and the  observation $\left(J_n\Omega\right)^\flat =J_n\Omega$ it follows that 
$\lbrack T^\flat T, J_n\Omega\rbrack=0.$

Now, $C\left(J_n\Omega\right)^k=C\left(T^\flat T\right)\left(J_n\Omega\right)^k=C\left(J_n\Omega\right)^k\left(T^\flat T\right)$ which means that the observability matrix $\mathcal{O}$ satisfies $\mathcal{O}=\mathcal{O}T^\flat T$. Because the system is minimal $\mathcal{O}$ must be full rank, hence $T^\flat T={1}$.

Finally, it remains to show that the matrix $T$ generating the equivalence class is of the form 
\[T=\left(\begin{smallmatrix} T_1 &T_2\\T_2^{\#}&T_1^{\#}\end{smallmatrix}\right).\]
To see this, observe that $CA^k$, $C'A'^k$ must be of the of this doubled up form for $k\in\{0,1,2,...\}$. Writing $CA^k$, $C'A'^k$ and $T$ as $\left(\begin{smallmatrix} P_{(k)} &Q_{(k)}\\Q_{(k)}^{\#}&P_{(k)}^{\#}\end{smallmatrix}\right)$, $\left(\begin{smallmatrix} P'_{(k)} &Q'_{(k)}\\Q_{(k)}'^{\#}&P_{(k)}'^{\#}\end{smallmatrix}\right)$ and $T=  \left(\begin{smallmatrix} T_1 &T_2\\T_3&T_4\end{smallmatrix}\right)$, and using the above result, $C'A'^k= CA^k T^\flat $, 
it follows that 
\[P_{(k)}(T_1^{\dag}-T_4^T)+Q_{(k)}(T_3^T-T_2^{\dag})=0\]
and
\[Q_{(k)}^{\#}(T_1^{\dag}-T_4^T)+P_{(k)}^{\#}(T_3^T-T_2^{\dag})=0.\]
Hence 
\[\mathcal{O}\left[\begin{smallmatrix}    T^{\dag}_1-T_4^T\\  T_3^T-T_2^{\dag}\end{smallmatrix}\right]=0\] and so using the fact that $\mathcal{O}$ is full rank gives the required result. 
\end{proof}
Therefore, without any additional information, we can at most identify the equivalence 
class of systems related by a symplectic transformation (on the system). 
Note that the above transformation of the system matrices is equivalent  to a change of co-ordinates $\breve{\mathbf{a}}\mapsto T^\flat \breve{\mathbf{a}}$ in Eq. (\ref{langevin}).

We can obtain the result for PQLSs from \cite{Guta2} as a Corollary to Theorem \ref{symplecticequivalence}. 

\begin{Corollary}\label{colo}
Let $\left({S}, C_-, \Omega_-\right)$ and $\left({S'}, C'_-, \Omega_-'\right)$ be two minimal PQLSs. Then they have the same transfer function if and only if there exists a unitary transformation $U$ such that 
\begin{equation}\label{oki34}
\Omega_-'=U\Omega_- U^{\dag}\quad C_-'=C_-U^{\dag} \quad S=S'.
\end{equation}
\end{Corollary}
\begin{proof}
Firstly the scattering matrix $S$ and $S'$ must be equal using the argument above Theorem \ref{symplecticequivalence}.

Write the two PQLSs in doubled-up form so that $A=\left(\begin{smallmatrix}A_-&0\\0&A_-^{\#}\end{smallmatrix}\right)$, $C=\left(\begin{smallmatrix}C_-&0\\0&{C_-}^{\#}\end{smallmatrix}\right)$, $A_-=-i\Omega_--\frac{1}{2}C_-^{\dag}C_-$ and similarly for the primed matrices. 

Since PQLSs are special cases of QLSs, we can apply Theorem  \ref{symplecticequivalence} so that the two systems are related via:
\begin{equation}\label{cheese}
A'=UA U^{\flat}\quad C'=CU^{\flat},
\end{equation}
where $U$ is symplectic. 
The relations \eqref{cheese} imply that $C'A'^k=CA^kU^{\dag}$ for all $k\in\mathbb{N}_0$. Therefore, writing $U=\Delta(U_1, U_2)$ and using the block diagonality of $C, C', A, A'$ implies that $C_-A_-^kU_2=0$ for all $k$. Hence $U_2=0$ by minimality. Therefore $U$ is unitary and symplectic, which is equivalent to the statement in the Theorem on the non-doubled-up space.
\end{proof}

 It is worth noting that while in the classical set-up equivalent linear systems are related by similarity transformations (see  Theorem \ref{class.id2}), in both quantum scenarios described above the transformations are more restrictive due to the unitary nature of the dynamics.
 
 In light of these results,  we make the following definition \cite{Guta2}. 
\begin{defn}\label{identy}
Let $\left(S(\theta), C(\theta), \Omega(\theta)\right)$ be a minimal system and $\theta\in\Theta$ an unknown parameter. Then $\theta$ is \textit{identifiable} if and only if
$$S(\theta')=S(\theta)\quad C(\theta')=C(\theta)T^{\flat} \quad J_n\Omega(\theta')=TJ_n\Omega(\theta) T^{\flat},$$ where $T$ is a symplectic matrix,
implies $\theta=\theta'$.
\end{defn}

 %%%%%%%%%%%%%%%%%%%%%%%%%
\section{Cascade Realisation of QLS}\label{seriesp}
%%%%%%%%%%%%%%%%%%%%%%%%%

Recently, a synthesis result has been established showing that the transfer function of a `generic'  QLS has a pure cascade realisation \cite{Nurdin3}. Translated to our setting, this means that given an $n$-mode  QLS $(C,\Omega)$, one can construct an equivalent system (i.e. with the same transfer function) which is a series product of single mode systems. The result holds for a large class of systems, including all PQLSs \cite{Nurdin6}. At the heart of the construction of the cascade realisation for a generic QLS is a sort of symplectic Schur decomposition  \cite{Nurdin3} (using symplectic similarity transformations, rather than unitary ones). Just like the ordinary Schur decomposition, one is free to reorder the elements on the main diagonal by applying a suitable symplectic transformation (in the ordinary Schur decomposition, this is would be a unitary transformation). By Theorem \ref{symplecticequivalence} this leaves the transfer function invariant. It turns out that this degeneracy exactly corresponds to the ordering of the systems.  It should be stressed though that the one-mode systems in each Schur decomposition may be different, which is a signature of non-commuting systems in the cascade realisation (see example \ref{beat} below). We remark that the cascade realisation  is an example of one type of realisation for QLSs. Other realisations  are the  \textit{independent oscillator representation} and the \textit{chain-mode realisation} for example \cite{Gough2}. The underlying feature of any good realisation is that they provide a simple model for the system's behaviour.

%characterised by the fact that the matrix $A$ admits a certain symplectic Schur decomposition, which holds for a dense, open subset of the relevant set of matrices, including all PQLSs \cite{Nurdin6}. 

\begin{exmp}\label{beat}
Consider a two mode cascaded PQLS, where the system $\left(C_1, \Omega_1\right)=\left(\left(\begin{smallmatrix}2\\0\end{smallmatrix}\right), 0\right)$ is fed into the system $\left(C_2, \Omega_2\right)=\left(\left(\begin{smallmatrix}6\\3\end{smallmatrix}\right), 4\right)$ and denote their transfer functions by $\Xi^{(1)}(s)$ and $\Xi^{(2)}(s)$ respectively. The combined transfer function is given by $\Xi(s)=\Xi_2(s)\Xi_1(s)$. 

Now the pair of cascaded systems $\left(\tilde{C}_1, \tilde{\Omega}_1\right)=\left(\left(\begin{smallmatrix}(2-6x)/(1+|x|^2)\\-3x/ (1+|x|^2)    \end{smallmatrix}\right), 4\right)$ with transfer function $\tilde{\Xi}_1(s)$ and $\left(\tilde{C}_2, \tilde{\Omega}_2\right)=\left(\left(\begin{smallmatrix}(2x+6)/(1+|x|^2)\\3/ (1+|x|^2)    \end{smallmatrix}\right), 0\right)$ with transfer function $\tilde{\Xi}_2(s)$, where $x=41/24-1/3i$, also has the same combined  transfer function. Notice that $\Xi_1(s)\neq \tilde{\Xi}_1(s)$ and $\Xi_2(s)\neq \tilde{\Xi}_2(s)$.
\end{exmp}

In this remainder of this subsection we shall assume that the QLS is SISO. If such a cascade realisation  is possible, then the  transfer function can be written as a product of  $n$  transfer functions of single mode systems, which are given by 
$$
\Xi_i (s)=\left(\begin{smallmatrix}\Xi_{i-}(s)&\Xi_{i+}(s)\\ &\\{\Xi_{i+}(\overline{s})}^{\#}& {\Xi_{i-}(\overline{s})}^{\#}\end{smallmatrix}\right).
$$
Further, we can stipulate that the coupling to the field is of the form $C=\Delta(C_-, 0)$, with each element of $C_-$ being real and positive. Indeed, since the system is assumed to be stable,  there exists a local symplectic transformation on each mode so that coupling is purely passive.  The point of this requirement is that it fixes all the parameters, so that under these restrictions each equivalence class from Sec. \ref{symplecticequivalence} contains exactly one element. Note that the Hamiltonian may still have both active and passive parts.  
Therefore, each one mode system in the series product is characterised by three parameters, $c_i, \Omega_{i-}\in\mathbb{R}$ with 
$c_i\neq 0$, and  $\Omega_{i+}\in\mathbb{C}$. If $\Omega_{i+}=0$ then the mode is passive. Actually, it is more convenient for us here to reparameterize the coefficients so that 
\begin{equation}\label{base1}
  \Xi_{i-}(s)= \frac{s^2-x^2_i-y^2_i +  2{ix_i\theta_i}   }{\left(s+x_i+y_i \right) \left(s+x_i-y_i\right)    },  \end{equation}
  \begin{equation}\label{base2}
\Xi_{i+}(s)=\frac{-2ix_ie^{i\phi_i}\sqrt{y_i^2+\theta_i^2}}{\left(s+x_i+y_i \right) \left(s+x_i-y_i\right)    },
\end{equation}
where $x_i=\frac{1}{2}c_i^2, y_i= \sqrt{|\Omega_{i+}|^2-\Omega_{i-}^2}, \theta_i=\Omega_{i-}$ and $\phi_i=\mathrm{arg}(\Omega_{i+})$.
Therefore, from the properties of the individual $\Xi_{i\pm}(s)$,  $\Xi_{-}(s)$ and $\Xi_{+}(s)$ can be written as
\begin{eqnarray}
\label{form1}
\Xi_{-}(s)&=&\prod\limits_{i=1}^n\frac{\left(s-\lambda_i\right)\left(s+\lambda_i\right)}{\left(s+x_i+y_i\right)\left(s+x_i-y_i\right)}
\\
\label{form2}
\Xi_{+}(s)&=&\gamma\frac{  \prod\limits_{i=1}^{j}     \left(s-\gamma_i\right)\left(s+\gamma_i\right)}{ \prod\limits_{i=1}^n \left(s+x_i+y_i\right)\left(s+x_i-y_i\right)},
\end{eqnarray}
with $\gamma, \gamma_i, \lambda_i\in\mathbb{C}$, $x_i\in\mathbb{R}$, and $y_i$ either real or imaginary, while $j$ is some number between $1$ and $n-1$. In particular, the poles are either in real pairs  or  complex conjugate pairs.

Furthermore, there is a possibility that some of the poles and zeros may cancel in \eqref{form1} and \eqref{form2}, and as a result some of these poles and zeros could be fictitious (see proof of Theorem \ref{mainresult} later where this becomes important).

For PQLSs such a cascade realisation is always possible \cite{Gough2, Petersen2} and  each single mode system is passive. We see an example of this.

\begin{exmp}\label{sisoexample}
Consider a  SISO PQLS $(C, \Omega)$ and let  $z_1, z_2, \dots ,z_m$ be the eigenvalues of $A=-i\Omega-\frac{1}{2}C^{\dag}C$. Then the transfer function is given by
\begin{align*}\Xi(s)&=\frac{\mathrm{Det}(s-{A}^{\#})}{\mathrm{Det}(s-{A})}\\&=\frac{s-\overline{z}_1   }{s-z_1}\times \frac{s-\overline{z}_2   }{s-z_2}\times ...\times \frac{s-\overline{z}_m   }{s-z_m}.\end{align*}
Now, comparing each term in the product with the transfer function of a SISO system of one mode, i.e., 
\[\Xi(s)=\frac{s+i\Omega-\frac{1}{2}|c|^2}{s+i\Omega+\frac{1}{2}|c|^2},\]
it is clear that each represents the transfer function of a bona-fide PQLS with Hamiltonian and coupling parameters given  by $\Omega_i=-\mathrm{Im}(z_i)$ and $1/2|c_i|^2=-\mathrm{Re}(z_i)$. This realisation of the transfer function is   a cascade of optical cavities.
Furthermore, we note that the order of the elements in the series product is irrelevant; in fact a differing order can be achieved by a change of basis on the system space. 
\end{exmp}

In actual fact this example enables us to  find a system realisation directly from the transfer function for SISO PQLSs, thus offering a parallel strategy to the realisation method in \cite{Guta2}. We will see in the next subsection that a similar brute-force approach for finding a cascade realisation of a geneic SISO active QLS is also possible.  
%However, the active case is more involved than the passive case, as the transfer function is characterised by two quantities, $ \Xi_-(s)$ and $\Xi_+(s)$, rather than just one.

 \section{Identification Method 1: Direct Method (SISO Systems)}\label{pond}
 Suppose that we have constructed the transfer function from the input-output data, using for instance one of the techniques of \cite{Ljung1}\footnote{Typically this can be done by probing the system with a
known input (e.g a coherent state with a time-dependent amplitude) and performing a measurement (e.g homodyne or heterodyne measurement) on the output field and post-processing the data (e.g using maximum likelihood or some other classical method \cite{Ljung1}.}. 
Here  we   outline one  system identification  method, which gives a cascaded system realisation of the   transfer function, and is possible  for all generic\footnote{The QLS is generic in the sense that the cascade exists (see Sec. \ref{seriesp})} minimal SISO  QLSs.

The work outstanding here is  to identify $x_i, \Omega_{i-}, y_{i}, \theta_i, \phi_i$ in \eqref{base1} \eqref{base2} from the expressions for $\Xi_-(s)$ and $\Xi_+(s)$ [\eqref{form1} and \eqref{form2}]. We use a brute-force algorithm to do this, which goes as follows:
\begin{enumerate}
\item Consider a bipartition of the system as $1|n-1$ and identify the parameters of the first system (see Fig. \ref{splitting}). \label{iti1}
\item Since the transfer function is $\flat$-unitary, divide through to obtain  an $n-1$-mode transfer function and repeat the procedure with the remaining $n-1$ modes until one mode is left. \label{iti2}
\item Identify the one remaining mode.\label{iti3}
\end{enumerate}
Let us now explain in more detail step \eqref{iti1} of the algorithm (steps \eqref{iti2} and \eqref{iti3} are self-explanatory). We shall assume for simplicity that none of the poles of the transfer function lie on the real axis, i.e. they are  in complex conjugate pairs. 

Firstly, since the transfer function is known, therefore the locations of all poles and zeros of $\Xi_-(s)$ and $\Xi_+(s) $ are too.
Note that the transfer function has up to $2n$ poles. However,  without loss of generality we may assume that there are exactly $2n$. Moreover, we may assume that we also have the full list of zeros in \eqref{form1} and \eqref{form2}. Indeed, since all poles are in the left complex plane and the zeros appear in pairs, $\pm z$, then if there is a zero at $z$ in the right-plane but no zero at $-z$ then there must have been a pole-zero cancellation at $-z$. Therefore we can identify any ``missing'' poles or zeros.

%so that if a pole and zero were to cancel at, say $p$,  in either quantity, then because there must also be a zero at $-p$ (which is necessarily in the right-plane) we can identify this ``missing'' pole.  That is, if there is a zero at $z$ in the right-plane but no zero at $-z$ then there must have been a pole-zero cancellation at $-z$. 
 
 \begin{figure}
\centering
\centering\includegraphics[scale=0.14]{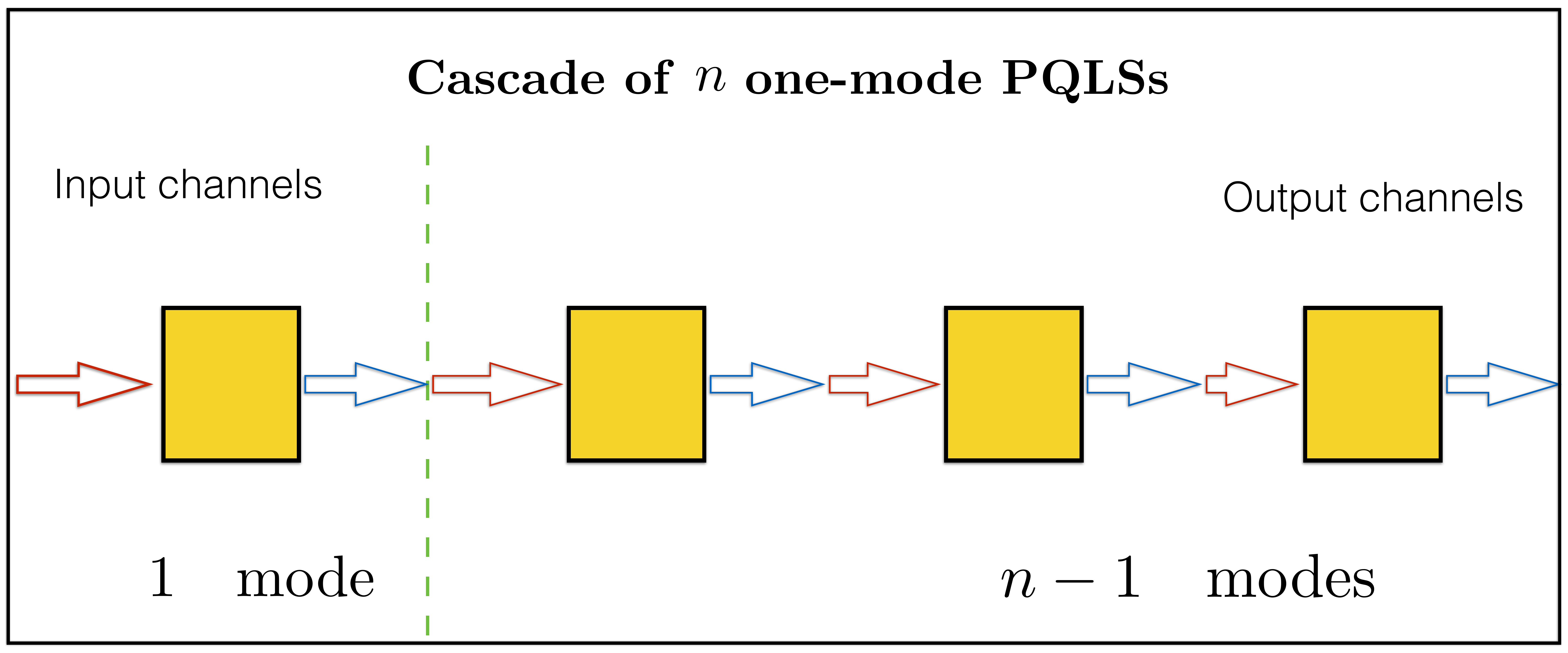}
\caption{
\label{splitting}}
\end{figure}

Now write $\Xi_{-}(s)$ and $\Xi_{+}(s)$ as 
\begin{equation}\label{q1}
\Xi_{-}(s)=\Xi^{(n-1)}_-(s)\Xi_{1-}(s)+\Xi^{(n-1)}_{+}(s)\overline{\Xi_{1+}(s^*)},
\end{equation}
\begin{equation}\label{q2}
\Xi_{+}(s)=\Xi^{(n-1)}_-(s)\Xi_{1+}(s)+\Xi^{(n-1)}_{+}(s)\overline{\Xi_{1-}(s^*)},
\end{equation}
with 
\[\Xi^{(n-1)}_-(s)=\frac{s^{2(n-1)}+a_{n-2}s^{2(n-2)}+...+ a_1s^{2(1)}+a_0           }{  \prod\limits_{i=2}^{n}  \left(s+x_i+y_i\right)\left(s+x_i-y_i\right)},\]    
\[\Xi^{(n-1)}_{+}(s)=\frac{   b_{n-2} s^{2(n-2)}+b_{n-3}s^{2(n-3)}+...+ b_1s^{2(1)}+b_0            }{\prod\limits_{i=2}^{n}  \left(s+x_i+y_i\right)\left(s+x_i-y_i\right)},
\]
\[\Xi_{1-}(s)=\frac{s^2+e}{\left(s+x_1+y_1\right)\left(s+x_1-y_1\right)},\]
\[\Xi_{1+}(s)=\frac{f}{\left(s+x_n+y_n\right)\left(s+x_n-y_n\right)},\]
where $e=-x_1^2-y_1^2+2ix_1\theta_1$ and $f=-2ix_1\Omega_{1+}$ (see Sec. \ref{seriesp}).

For  step \eqref{iti1} of the algorithm choose  a pair of complex conjugate poles, which must belong to one of the systems in the cascade. We can take this system to be the  one system in the partition, i.e., $\Xi_1(s)$. As alluded to above, let us identify this system in the cascade.

Firstly, we can identify the quantities $x_1, y_1$ from the poles. 
 It is clear at this stage that it remains to identify $\theta_1, \Omega_{1+}$. Note that we already know the value of $|\Omega_{1+}|^2-\theta_1^2$, which is $y_1^2$.

Now from the numerators of \eqref{q1} and \eqref{q2} we can identify the following:
From (\ref{q1}):
\begin{align*}
&s^0: a_0e+b_0\overline{f}=:\alpha_0\\
&s^2: a_1e+b_1\overline{f}+a_0=:\alpha_1\\
&\vdots\\
&s^{2(n-2)}: a_{n-2}e+b_{n-2}\overline{f}+a_{n-3}=:\alpha_{n-2}\\
&s^{2(n-1)}: 1\cdot e+a_{n-2}=:\alpha_{n-1}
\end{align*}
From (\ref{q2}):
\begin{align*}
&s^0: a_0f+b_0\overline{e}=:\beta_0\\
&s^2: a_1f+b_1\overline{e}+b_0=:\beta_1\\
&\vdots\\
&s^{2(n-2)}: a_{n-2}f+b_{n-2}\overline{e}+b_{n-3}=:\beta_{n-2}\\
&s^{2(n-1)}: 1\cdot f+b_{n-2}=:\beta_{n-1}.
\end{align*}
We now work recursively in decreasing powers  of $s$ to obtain expressions for $a_i$ and $b_i$ (for $i=n-2,...,0$). We see that at each stage of the recursion the expression we obtain is written  in terms of known quantities and is linear in the two unknowns $\theta_1, \Omega_{1+}$. Clearly,  this is true for $a_{n-2}$ and $b_{n-2}$. That is, take the  $s^{2(n-1)}$ equations:
\[a_{n-2}=\alpha_{n-1}-e\]
\[b_{n-2}=\beta_{n-1}-f\]
and observe that $e, f$ are linear in $\theta_1, \phi_1$. To see that this is true for all $a_i$ $b_i$, we proceed by induction. Suppose that after the $k-1$ step we have expressions for $a_{n-k}$ and $b_{n-k}$ that are linear in $\theta_1, \phi_1$ and are of the form: 
\[a_{n-k}=s_1+s_2\overline{\Omega_{1+}}-s_3\theta_1\]
\[b_{n-k}=s_4-s_2\theta_1+s_3\Omega_{1+},\]
where the $s_i$ are known quantities. Note that it shouldn't be too hard to see that $a_{n-2}, b_{n-2}$ were also of this form. 
We will now show that $a_{n-(k+1)}$ is of this form (the proof for $b_{n-(k+1)}$ is similar).
\begin{align*}
a_{n-(k+1)}&=\alpha_{n-k}-a_{n-k}e-b_{n-k}\overline{f}\\
&=\alpha_{n-k}-\left(s_1+s_2\overline{\Omega_{1+}}-s_3\theta_1\right)\left(-x_1^2-y_1^2+2ix_1\theta_1\right)-\left(s_4-s_2\theta_1+s_3\Omega_{1+}\right)\left(2ix_1\overline{\Omega_{1+}}\right)\\
&=\tilde{s}_1+\tilde{s}_2\overline{\Omega_{1+}}-\tilde{s}_3\theta_1,
\end{align*}
where 
\[\tilde{s_1}:=\alpha_{n-k}+s_1(x_1^2+y_1^2)+2is_3x_1(\theta^2+|\Omega_{1+}|^2)\]
\[\tilde{s_2}:=s_2(x_1^2+y_1^2)-2is_4x_1\]
\[\tilde{s_3}:=-s_3(x_1^2+y_1^2)-2ix_1s_1.\]
Therefore $\tilde{s_i}$ are known, so we are done.

Finally  take the expressions for $a_0$ and $b_0$  that we have just obtained and substitute them into the two remaining equations coming from the $s^0$ powers. The result of this is again two linear equations in $\theta_1, \Omega_{1+}$. A consequence of the fact we obtain  linear expressions at each stage is that the final solution will be  unique. This uniqueness should not be surprising because of the constraints that  we placed on the realisation in Sec. \ref{seriesp}, i.e.,  that each element of the equivalence class contains exactly one element.

\begin{exmp}\label{ordert} 
We now give an example of the above procedure. Suppose that we have the following two-mode SISO QLS:
\[C_-=\left(8,12\right),\,\,\,C_+=\left(0,-1\right),\,\,\,\Omega_-=\left(\begin{smallmatrix}6& -1\\-1&2\end{smallmatrix}\right)\,\,\,\mathrm{and}\,\,\,\Omega_+=\left(\begin{smallmatrix}0& i\\i&0\end{smallmatrix}\right).\]
Its transfer function is given by
\begin{align}\label{ex.p}
&\Xi_-(s)=\frac{s^4+(-10672.25+482)s^2+(338313-12284i)}{  s^4+207s^3+10752.25s^2+6940s+338505    }\\\label{exp.p2}
&\Xi_+(s)=\frac{(-192+32i)s^2+(-4276+1632i)}{  s^4+207s^3+10752.25s^2+6940s+338505   }.
\end{align}

We will now find a  cascade realisation directly from the transfer function using the above procedure. 
Firstly, the poles of the transfer function are given by $-103.48 - 2.12i, -103.48 + 2.12i$ and $-0.0187 + 5.62i, -0.0187 - 5.62i$. Let us identify the poles $-103.48 - 2.12i, -103.48 + 2.12i$ as belonging to the first system. Using $x_1=\frac{1}{2}c^2_1$ and the recursive procedure above,  we obtain $c_1=14.39, \Omega_{1-}=2.33$ and $\Omega_{1+}=-0.15 - 0.93i$. Finally, we may identify the second system by  $\Xi_2(s)=\Xi(s)\Xi_1(s)^{\flat}$, so that $c_2=0.2, \Omega_{2-}=7.38$ and $\Omega_{2+}=-1.48 - 4.55i$. 

On the other hand we may identify the poles $-0.0187 + 5.62i, -0.0187 - 5.62i$ as the first system and $-103.48 - 2.12i, -103.48 + 2.12i$ as the second system. As a result we obtain $c_1=0.2, \Omega_{1-}=7.30, \Omega_{1+}=1.61-4.37i$ and $c_2=14.39 \Omega_{2-}=2.33, \Omega_{2+}=-0.15 - 0.93i$.This degeneracy corresponds corresponds to the \textit{(symplectic) Schur decomposition} degeneracy mentioned above. Notice that the subsystems in both constructions are different, which is a signature of the fact that the subsystems in each don't commute. 
\end{exmp}

%This realisation method was already quite involved for SISO systems.  For MIMO it is expected that it would get  far more complicated and  may not even be tractable. Therefore we need an alternative, which is the subject of the next subsection. 

The brute-force approach in this section gets too cumbersome for  MIMO systems so we need an alternative method, which is the subject of the next subsection.

\section{Identification Method 2: Indirect Method}\label{indirect}
Here, the realisation is obtained indirectly by first finding a non-physical realisation and then constructing a physical one from this by applying a criterion developed in \cite{Gough2}.  
The construction follows similar lines to the method described in \cite{Guta1,Guta2} for PQLSs. 

%Not all of the details will be given here and is left for future works. 

Let  $(A_0, B_0, C_0)$ be a triple of doubled-up matrices which constitute a  minimal realisation of $\Xi(s)$, i.e.  
\begin{equation}
\label{TF1}
\Xi(s)={1}+C_0(sI-A_0)^{-1}B_0.
\end{equation}
The identification method here is applicable to generic (not necessarily SISO) QLSs, provided such a classical realisation can be found beforehand. 
For example, in  Appendix \ref{APP7} such a classical realisation is found for an $n$-mode minimal QLS, with matrices
$(A,C)$, possessing $2n$ distinct poles each with non-zero imaginary
part.
%There are of course many such possible realisations and all will work in the following.

Any other minimal realisation of the transfer function can be generated via a \textbf{similarity} transformation:
\begin{equation}\label{trivialtrans}
A=TA_0T^{-1}\,\,\, B=TB_0 \,\,\, C=C_0T^{-1}.
\end{equation}
The problem here  is that in general these matrices may not describe a genuine quantum system in the sense that from a given $A, B, C$ one cannot reconstruct the pair $(\Omega, C)$. Our goal is to find a special transformation $T$ mapping $(A_0, B_0, C_0)$ to a triple $(A, B, C)$ that does represent a genuine quantum system. To this end, a triple $(A, B, C)$ corresponds to  a quantum linear system if and only if it satisfies the  \textit{physical realisability conditions} \cite{Gough2}:
\begin{equation}\label{PRs}
A+A^{\flat}+C^{\flat}C=0 \,\,\, \mathrm{and}\,\,\, B=-C^\flat.
\end{equation}
Therefore, substituting \eqref{trivialtrans} into the left equation of \eqref{PRs} one finds 
\begin{equation}\label{fog}
\left(T^{\dag}JT\right)A_0+A_0^{\dag}\left(T^{\dag}JT\right)+C^{\dag}_0JC_0=0 ,
\end{equation}
%and  $\left(T^{\dag}JT\right)B_0=-C^{\dag}_0J$, 
where the matrices $J$ here are of appropriate dimensions. The goal here is to find a $T$ satisfying \eqref{fog}, for then the triple $(A, B, C)$ obtained via \eqref{trivialtrans} will be a minimal QLS realisation of the transfer function.

Next, because  $A$ and $A_0$ are similar and the system is assumed to be stable, $A_0$ will have eigenvalues in the left-half plane. Hence  it follows from \cite[Lemma 3.18]{Zhou1} that \eqref{fog} has a unique solution, given by  
\begin{equation}\label{solute}
T^\flat T= J \left(T^{\dag}JT\right)=   \int^{\infty}_0 J\left(C_0e^{A_0 t}\right)^{\dag}J\left(C_0e^{A_0t}\right)dt.
%&= \int^{\infty}_0 J\left(C_0e^{A_0 t}\right)^{\dag}J\left(C_0e^{A_0 t}\right).
\end{equation}
%Hence $T^\flat T$ is of the form $JK^{\dag}JK$ for doubled-up $K^{2n\times 2n}$.

We now need to use a result from \cite{Griv1}, which is a sort of singular value decomposition for symplectic matrices. We state the result in a slightly different way here. %See \cite{cascade2} for the proof.
\begin{Lemma}\label{peterf}
Let $N^{2n\times 2n}$ be a complex, invertible, doubled-up matrix and let $\mathcal{N}=N^\flat N$.
\begin{enumerate}
\item Assume that all eigenvalues of $\mathcal{N}$ are semisimple. Then there exists a symplectic matrix $W$ such that $\mathcal{N}=W\hat{N}W^\flat$ where $\hat{N}=\left(\begin{smallmatrix} \hat{N}_1&\hat{N}_2\\\hat{N}_2^{\#}&\hat{N}_1^{\#}\end{smallmatrix}\right)$ with
\[\hat{N}_1=\mathrm{Diag}\left(\lambda^{+}_1, ..., \lambda^{+}_{r_1}, \lambda^{-}_1, ..., \lambda^{-}_{r_2}, \mu_1 {1}_2, ..., \mu_{r_3} {1}_2\right)\]
\[\hat{N}_2=\mathrm{Diag}\left(0, ..., 0, 0, ..., 0, -\nu_1 \sigma , ..., -\nu_{r_3} \sigma\right).\]
Here $\lambda_{i}^{+}>0$, $\lambda_{i}^{-} <0$ and $\lambda^{c}_i := \mu_i+i\nu_i$ (with $\mu_i, \nu_i\in\mathbb{R}$ $\nu_i >0$) are the eigenvalues of $\mathcal{N}$. The matrix $\sigma=\left(\begin{smallmatrix} 0&-i\\i&0\end{smallmatrix}\right)$ is one of the Pauli matrices. 

\item
There exists another symplectic matrix $V$ such that $N=V\bar{N}W^\flat$ where $\bar{N}$ is the factorization of $\hat{N}$ $\left(\hat{N}=\bar{N}^\flat\bar{N}\right)$ given by
$\bar{N}=\left(\begin{smallmatrix} \bar{N}_1&\bar{N}_2\\\bar{N}_2^{\#}&\bar{N}_1^{\#}\end{smallmatrix}\right)$ with
\[\bar{N}_1=\mathrm{Diag}\left(\sqrt{\lambda^{+}_1}, ..., \sqrt{\lambda^{+}_{r_1}}, 0, ..., 0, \alpha_1 {1}_2, ..., \alpha_{r_3} {1}_2\right)\]
\[\hat{N}_2=\mathrm{Diag}\left(0, ..., 0, \sqrt{|\lambda^{-}_1|}, ..., \sqrt{|\lambda^{-}_{r_2}|},-\beta_1 \sigma , ...,-\beta_{r_3} \sigma\right).\]
The coefficients $\alpha_i$ and $\beta_i$ are determined from $\mu_i$ and $\nu_i$ via
\begin{itemize}
\item If $\mu_i\geq0$, then $\alpha_i=\sqrt{\mu_i}\mathrm{cosh}x_i$, $\beta_i=\sqrt{\mu_i}\mathrm{sinh}x_i$, with $x_i=\frac{1}{2}\mathrm{sinh}^{-1}\frac{\nu}{\mu}$.
\item If $\mu_i\leq0$, then $\alpha_i=\sqrt{|\mu_i|}\mathrm{sinh}x_i$, $\beta_i=\sqrt{|\mu_i|}\mathrm{cosh}x_i$, with $x_i=\frac{1}{2}\mathrm{sinh}^{-1}\frac{\nu}{|\mu|}$.
\item If $\mu_i=0$, then $\alpha_i=\beta_i=\sqrt{\frac{\nu_i}{2}}$.
\end{itemize}
\end{enumerate}
\end{Lemma} 
The Lemma can be extended beyond the semisimple assumption, but since the latter holds for generic matrices \cite{Griv1}, it suffices for our purposes.

We can therefore use  Lemma \ref{peterf} together with equation \eqref{solute} in order to write the `physical' 
$T$ as $T=V\bar{T}W^\flat$, where $W$ and $\bar{T}$ can be computed as in the Lemma above, and $V$ is a symplectic matrix. However, since the QLS equivalence classes are characterised by symplectic transformation, this means that 
$T_0=\bar{T}W^\flat$ transforms 
$(A_0, B_0, C_0)$ to the matrices of a quantum systems satisfying the realisability conditions. Finally, we can solve to find the set of physical parameters $(\Omega, C)$, which are given in terms of $(A_0, B_0, C_0)$, as
\begin{align*}
C&= C_0W\bar{T}^{-1}  \\
\Omega&=i\left(\bar{T}W^\flat A_0W\bar{T}^{-1} +\frac{1}{2}\left(\bar{T}^{\flat}\right)^{-1}W^\flat C_0^\flat C_0W\bar{T}^{-1}\right) .
\end{align*}

\begin{remark}
Note that by assumption $\Xi(s)$ is the  transfer function of a QLS. 
Since the original triple $(A_0, B_0, C_0)$ is minimal, this implies that there exists a non-singular $T$ satisfying \eqref{solute}, so the right side of \eqref{solute} is non-singular, which eventually leads to a non-singular transformation 
$T$ computed using Lemma  \ref{peterf}.
%If the integral on the right side of \eqref{solute} had any zero eigenvalues then  the original triple $(A_0, B_0, C_0)$ would be non-minimal, and would therefore lead to a non unique solution. To find the minimal solution one could restrict the transformation $T$ to the non-zero eigenvalue eigenspace. 
\end{remark}

% Note that the constructive procedure enabling one to compute $\bar{N}$ and $W$ is given in \cite{cascade2}.

%\begin{remark}
%The proof also holds for MIMO systems provided that  one can find a minimal doubled-up (non-physical) realisation beforehand. 
%\end{remark}

\section{Noisy Transfer Function Identification}\label{noisek}

We will now discuss the situation where there is noise in the system. Recall that the system in our QLSs model is an example of an open quantum system. In open quantum systems theory  information  is dissipated to the environment \cite{Rivas1}, which is typically modelled as (bosonic) quantum bath (see Sec.  \ref{Bosonic}). Therefore noisy QLSs may be modelled  with additional channels that we cannot access \cite{Peter1}(see Fig. \ref{decorg}). 

 \begin{figure}
\centering
\centering\includegraphics[scale=0.20]{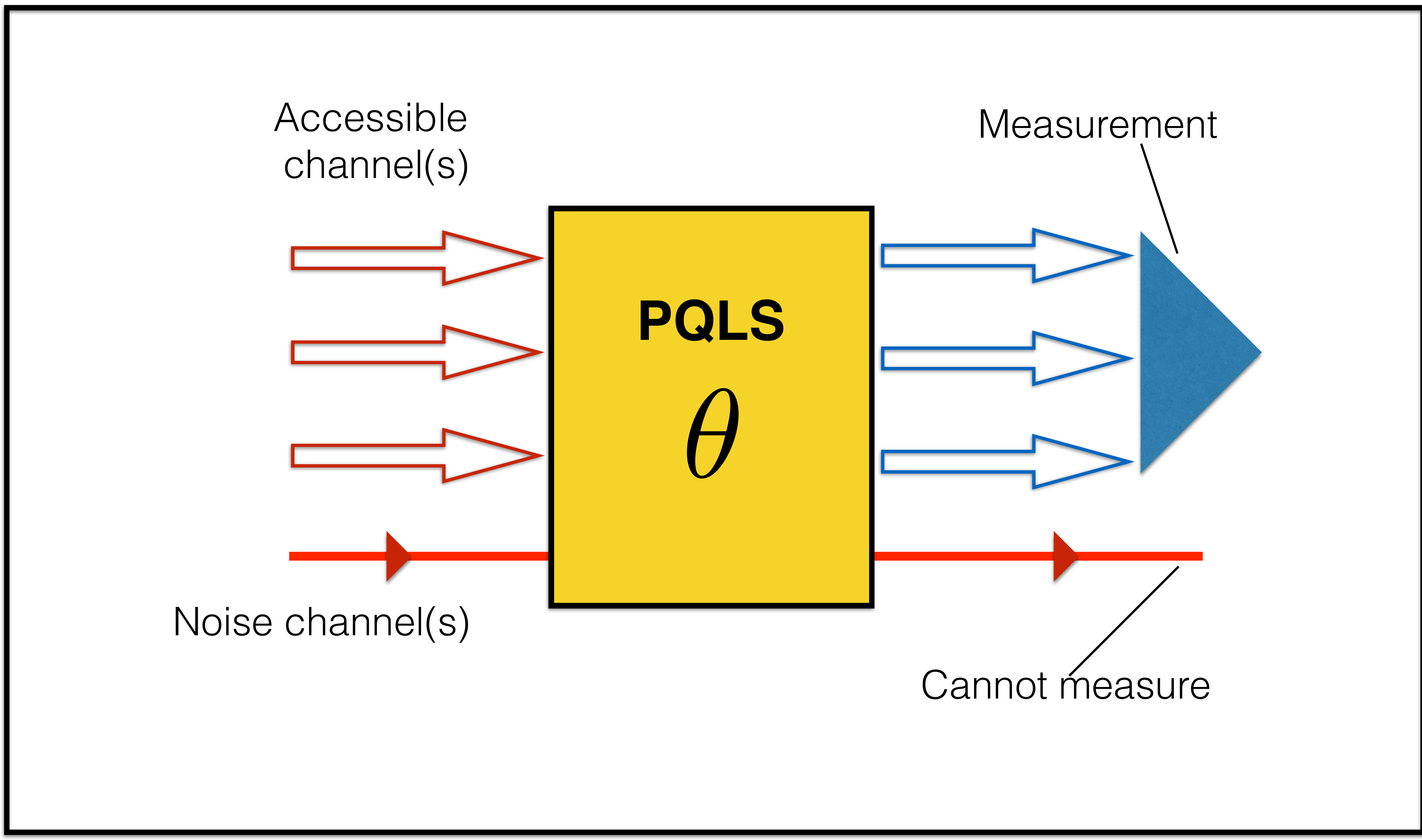}
\caption{Illustration of how to model noise in QLSs. \label{decorg}}
\end{figure}
Throughout this section we shall work with \textbf{passive} QLSs.
\subsection{Identifiability}\label{noisek1}

Suppose that  we have an unknown PQLS with $n$ internal modes and $m$ channels. We assume that we can only access $m_1$ of them. We call the $m_1$ channels that we can access the \textit{accessible channels} and the remaining channels the \textit{noise channels}. We would like to ask the following question: what is identifiable from measurements on the accessible channels? That is, we seek to identify the system  $\mathcal{G}:=(C=(c_1,c_2), \Omega)$ where $c_1$ is supported on the accessible channels and $c_2$ on noise  channels. Write the transfer function for $\mathcal{G}$ as $$\Xi(s)=\left(\begin{smallmatrix}\Xi_{11}(s)&\Xi_{12}(s)\\\Xi_{21}(s)&\Xi_{22}(s)\end{smallmatrix}\right),$$
where, for example, $\Xi_{12}(s)$ represents the transfer function from the accessible channels  to the noise channels. We assume  throughout Sec.  \ref{noisek1} that the input on the noise channels is vacuum. Therefore,  we have access to only the quantity $\Xi_{11}(s)=1-c_1\left(s-A\right)^{-1}c_1^{\dag}$    from the full transfer function $\Xi(s)$.

Under the assumption of minimality and because $\Xi_{11}(s)$ is a classical transfer function, then the worst case scenario  is that we may identify $(c_1, \Omega)$ up to a similarity  transformation. On the other hand, the best case  to hope for is that we can   identify $(C,\Omega)$ up to a unitary transformation as in Corollary \ref{colo}. We give a simple example illustrating that the actual level of identifiability lies somewhere between these two extremes.

\begin{exmp}\label{probes}
Suppose that we have a one-mode PQLS with two channels, of which we only have access to one. The full transfer function is given by 
$$\Xi(s)=\left(\begin{smallmatrix}  1-\frac{|c_1|^2}{s+i\Omega+\frac{1}{2}(|c_1|^2+|c_2|^2)}&\frac{c_1\overline{c_2}}{s+i\Omega+\frac{1}{2}(|c_1|^2+|c_2|^2)}\\\frac{c_2\overline{c_1}}{s+i\Omega+\frac{1}{2}(|c_1|^2+|c_2|^2)}&1-\frac{|c_1|^2}{s+i\Omega+\frac{1}{2}(|c_1|^2+|c_2|^2)}\end{smallmatrix}\right).$$
Clearly from the full transfer function the most that we can hope to identify are the quantities $|c_1|, |c_2| \Omega$ and $c_1\overline{c_2}$. This corresponds exactly to the equivalence classes in Corollary \ref{colo}. 

However, from just the top-left element of $\Xi(s)$ we an only identify $|c_1|, |c_2|,$ and $\Omega$. That is, we have no information about the relative  phase of the coupling between   the two channels; this should  not be  a surprise. 
\end{exmp}

Also, restricting  to accessible  channels observability and controllability are not  equivalent in general, unlike the noiseless case. We show give a  counterexample for this. 
\begin{exmp}
Consider a two mode system SISO system with an additional noisy channel. Let $\Omega=\mathrm{Diag}(4,1)$ and $c_1=(2,y), c_2=(-1,3)$, where $y$ is a variable to be specified later. Therefore,
$$A=\left(\begin{smallmatrix}-4i-4&3-2y\\3-2\overline{y}&-i-|y|^2-9\end{smallmatrix}\right).$$
Using  condition \eqref{obs2} in Theorem \ref{obs5} the system $(c_1, A)$ is observable iff the matrix $$\left(\begin{smallmatrix}c_1\\c_1A\end{smallmatrix}\right)$$ has full column rank. After some algebra it turn out that this condition is equivalent to 
\begin{equation}\label{ce1}
y^2+y\left(-2i+\frac{16}{3}\right)-12=0.
\end{equation}

On the other hand, by using \eqref{cl2} from Theorem \ref{cl5}, $(c_1^{\dag}, A)$ is controllable  iff
\begin{equation}\label{ce2}
y^2+y\left(+2i+\frac{16}{3}\right)-12=0.
\end{equation}
The equations \eqref{ce1} and \eqref{ce2} cannot be satisfied simultaneously, hence we have counterexamples to the statements [observability$\implies$ controllability] and [controllability$\implies$ observability].
\end{exmp}

The previous two examples indicate that the identifiability problem for noisy systems is more complicated than the noiseless case. To recap, the following statements no longer hold:
\begin{itemize}
\item Observability $\iff$ controllability.
\item Minimal $\implies$  TFE PQLSs are  related by a unitary transformation. 
\end{itemize}

We will now understand what is identifiable in the noisy case. 
%That is, we will address the following:   find all systems $(C,\Omega)$ satisfying the  transfer function $\Xi_{11}(s)$. 
%The common theme is that of attempting to find a physical realisation from the transfer function. In doing this we obtain the class of equivalent systems for free.

%\underline{Method}

First construct a minimal classical realisation for the transfer function, $\Xi_{11}(s)$, given by $(A_0, B_0, C_0)$, by using (for instance) the Gilbert realisation method \cite{Zhou1}. Here $A_0$ is of size $n\times n$, $B_0$ is of size $n\times m_1$ and $C_0$ is of size $m_1\times n$. That is,
$$\Xi_{11}(s)=1+C_0\left(s-A_0\right)^{-1}B_0.$$
Under assumption of minimality this realisation will be unique up to a similarity transformation (see Theorem \ref{class.id2}).

Now we must  find a pair of matrices $B_1\in\mathbb{C}^{n\times (m-m_1)}$ and $C_1\in\mathbb{C}^{(m-m_1)\times n}$ such that the extended system $$\left(A_0, \left(\begin{smallmatrix}B_0& B_1\end{smallmatrix}\right), \left(\begin{smallmatrix}C_0\\C_1\end{smallmatrix}\right), 1\right)$$
represents a physical quantum  system.  Note that there may be many such matrices $B_1$ and $C_1$; this would correspond to  PQLSs with the same transfer function on the accessible channels but different transfer function on the extended space. Let us put this non-uniqueness issue to one side for the moment and  proceed to find a realisation. Also note that  the extended system (once found) is guaranteed to  be  minimal because  $(A_0, B_0, C_0)$ is.
Now the requirement that the extended system be physical means that there exits a  similarity transformation $T$ such that the matrices 
\begin{equation}\label{older}A:=TA_0T^{-1}, \quad B:=T\left(\begin{smallmatrix}B_0& B_1\end{smallmatrix}\right)\quad C:=\left(\begin{smallmatrix}C_0\\C_1\end{smallmatrix}\right)T^{-1}
\end{equation}
correspond to a genuine PQLS, i.e., (see Sec. \ref{MTDR} or \cite{Guta2})
\begin{equation}\label{collier}
A+A^{\dag}+C^{\dag}C=0\quad \mathrm{and} \quad B=-C^{\dag}.
\end{equation}
Crucially such a similarity transformation preserves the transfer function $\Xi_{11}(s)$.
It follows that 
\begin{align}
&\label{bold}\left(T^{\dag}T\right)A_0+A_0^{\dag}\left(T^{\dag}T\right)+C_0^{\dag}C_0+{C}_1^{\dag}{C}_1=0\\
&\label{bold1}\left(T^{\dag}T\right)(B_0, {B}_1)=-(C_0^{\dag}, {C}_1^{\dag}).
\end{align}
As $A_0$ is Hurwitz, Eq. \eqref{bold} is equivalent to 
 \begin{equation}\label{bold2}T^{\dag}T=\int^{\infty}_0e^{A^{\dag}_0t}\left(C_0^{\dag}C_0+{C}_1^{\dag}{C}_1\right)e^{A_0t}.\end{equation}
 Note that $T^{\dag}T$ must be invertible for any choice of $C_1$   by minimality \cite{Zhou1}. 
% of $(A_0, B_0, C_0, 1)$. 
Finally, combining this expression with \eqref{bold1}, one obtains
\begin{equation}\label{pits}B_0=-\left(\int^{\infty}_0e^{A^{\dag}_0t}\left(C_0^{\dag}C_0+C_1^{\dag}{C}_1\right)e^{A_0t}\right)^{-1}C_0^{\dag}\end{equation}
\begin{equation}\label{pits1}{B}_1=-\left(\int^{\infty}_0e^{A^{\dag}_0t}\left(C_0^{\dag}C_0+C_1^{\dag}{C}_1\right)e^{A_0t}\right)^{-1}{C}_1^{\dag}.\end{equation}

In summary in order to find a physical realisation from $(A, B, C)$ one begins by finding  a $C_1$ (perhaps numerically) satisfying \eqref{pits}. The matrix $B_1$  is then fixed by the physicality condition \eqref{pits1}. From these one can obtain a unique (given a particular solution $B_1, C_1$) $T^{\dag}T$ from \eqref{bold2}. As $T$ is a normal matrix, we can diagonalise $T^{\dag}T$, i.e. $T^{\dag}T=U_0\Lambda_0U_0^{\dag}$, so that $T=U^{\dag}_0\sqrt{\Lambda_0}$. Finally, the physical system is given by $(A, C)$ satisfying \eqref{older}.

As mentioned earlier  $B_1$ and $C_1$ are not unique here (although $B_1$ is fixed by $C_1$). This freedom appears in Eq. \eqref{pits}. However,      given a particular solution  $B_1$ and $C_1$ the matrix $T$ is unique up to a unitary matrix in Corrolary  \ref{colo}, as the problem simply reduces to the noiseless case. Therefore, the non-uniqueness in the choice of  $C_1$ represents the relaxation of the unitary equivalence classes from the noiseless to the noisy case. Moreover, the conditions \eqref{older} reduce to the classical similarity transformations on $(A_0, B_0, C_0)$, whereas the condition that the full system be physical (i.e. \eqref{collier}) imposes restrictions on the admissible similarity transformations. Therefore, identifiability does indeed lie somewhere between the similarity and unitary transformation extremes, as mentioned above. 

Furthermore, there is no guarantee that the transfer function $\tilde{\Xi}(s)$ is equal to the original transfer function $\Xi(s)$. Moreover, as minimality on the extended space of channels $m$  does not imply minimality on the channels $m_1$, the transfer function $\tilde{\Xi}(s)$ may be simpler than $\Xi(s)$  in the sense that less modes are required to describe it. 
%This being the case is an example of a \textit{noise-unidentifiable subspace}, which we discuss in more detail in the following subsection.

We now give an example of this method in action.

\begin{exmp}
Consider a one mode SISO system with an additional noisy channel. Suppose that the coupling is given by $C_1=6i$, $C_2=3+2i$ and the Hamiltonian is $\Omega=3$.
The question that we intend to answer is the following:
can we recover these parameters from the part of the transfer function that we have access to, namely
$\Xi_{11}(s)=1-\frac{36}{s+3i+39/2}$?

Firstly, suppose that we obtain the classical realisation $(A_0=-\frac{39}{2},B_0=-1, C_0=36)$ for this  transfer function, 
Now solving the integral in \eqref{pits} leads to the condition
$1=\frac{(39)(36)}{|\tilde{C}_0|^2+36^2}.$
Hence $|\tilde{C}_0|=6\sqrt{3}$. From \eqref{pits1} we also obtain
$\tilde{B}_0=-\frac{\overline{\tilde{C}_0}}{36}$. As there are no further constraints on $\tilde{B}_0$ and $\tilde{C}_0$, we can take them to be real. Hence $\tilde{C}_0=6\sqrt{3}$ and $\tilde{B}_0=-\frac{\sqrt{3}}{6}$. Therefore, from \eqref{bold2} one obtains $\left(T^{\dag}T\right)=\left(\begin{smallmatrix}36&0\\0&36\end{smallmatrix}\right)$. Hence $T=\left(\begin{smallmatrix}6&0\\0&6\end{smallmatrix}\right)$.
Finally, from these we can obtain the  physical system  $C_1=6, C_2=\sqrt{3}, \Omega=3$. One can verify that this realisation has noisy transfer function $\Xi_{11}(s)$. However, the transfer function of the extended system characterised by matrices  $C_1=6, C_2=\sqrt{3}, \Omega=3$ is not the same as the full transfer function of the original system $C_1=6i, C_2=3+2i, \Omega=3$, but nevertheless remains consistent with Example \ref{probes}.
\end{exmp} 

Throughout this subsection we have implicitly assumed that we know the number of noise channels. In fact we may remove this assumption using the following observation: for a system with $n$ modes,  it is always possible to find $\tilde{B_0}$ and $\tilde{C}_0$ of size $n\times n$ satisfying the requirements above. This means that we can model our MIMO $n$ mode system as having (at most) $n$ noisy channels and a solution is guaranteed to exist; any additional channels  will result in superfluous parameters. Therefore, we don't need to know $m$ in our analysis above. To see why this statement is true, suppose  as above that we have an $n$ mode system with $m_1$ accessible  channels and $m-m_1>n$ noise channels, which can be written as a cascade of one-mode systems \cite{Nurdin6}. Suppose that the cascaded system is given by 
$$\left(\left(\begin{smallmatrix}c_{1n}\\c_{2n}\end{smallmatrix}\right),\Omega_n\right)\triangleleft...\triangleleft\left(\left(\begin{smallmatrix}c_{11}\\c_{21}\end{smallmatrix}\right),\Omega_1\right).$$
Now we can change basis on the noise fields $1,....m-m_1$ (hence the accessible field remains unchanged)  so that the noisy coupling to the first mode in the cascade is supported on only the first noisy channel.  Next perform a change of basis on the noise  fields $2,....m-m_1$ so that the noisy coupling to the second mode in the cascade is supported on only the first two noisy channels. Repeat this procedure a further $n-2$ times leads to an accessible channel TFE system which couples with only $n$ noise channels (there will be $m-m_1-n$ redundant noise channels) (see Fig. \ref{basischange}). 
%Note that this procedure  worked because the input on the noisy channels was assumed to be vacuum. 

 \begin{figure}
\centering
\centering\includegraphics[scale=0.20]{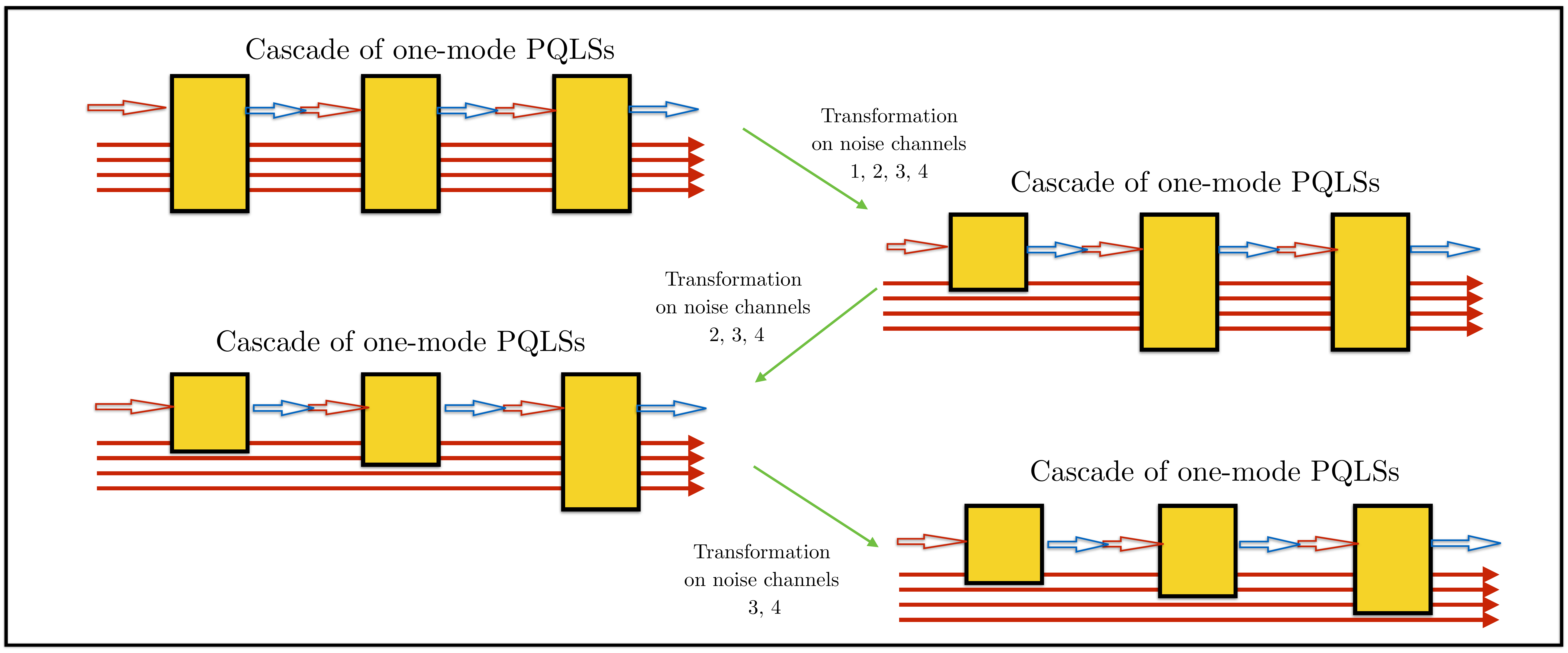}
\caption{
\label{basischange}}
\end{figure}

\subsection{Noise Unobservable Subspaces}\label{Noisek4}

Here we consider a  modified problem to the one in the last subsection, where we suppose that we also have access to the quantity $\Xi_{m_1, m-m_1}(s)$. This could be achieved for example under the assumption that the input  on the noise channel is known (and is different from vacuum). The question of  identifiability 
comes down to understanding what is identifiable from $\Xi_{m_1, m_1}(s)$ and  $\Xi_{m_1, m-m_1}(s)$, and can be answered in a similar way to the problem in Sec. \ref{noisek1}. However, this problem is not our primary concern here. Instead we discuss an interesting phenomenon when part of the system is  not visible in the noise output. 

%is then to understand what can be identified about the system from quantities $\Xi_{m_1, m_1}(s)$ and  $\Xi_{m_1, m-m_1}(s)$.  Note that this  identifiability problem could be  understood by using a similar technique to the one in the Sec. \ref{noisek1}. 

Suppose that we have the PQLS $\mathcal{G}:=(C=(c_1,c_2), \Omega)$ where $c_1$ is supported on the first $m_1$ channels and $c_2$ on the remaining $m-m_1$ channels.  The transfer function that we can access is given by
$$\Xi_{acc}(s):=\left(1,0\right)-c_1\left(s-A\right)^{-1}\left(c_1^{\dag}, c_2^{\dag}\right),$$
whereas the transfer function on the noisy output is given by
$$\Xi_{noi}(s):=\left(0,1\right)-c_2\left(s-A\right)^{-1}\left(c_1^{\dag}, c_2^{\dag}\right).$$
Note that the  ``$(A, B, C, D)$ forms''  for the  accessible and noisy  channels are  given by $\left(A, -\left(c_1^{\dag}, c_2^{\dag}\right),c_1, 1\right)$ and $\left(A, -\left(c_1^{\dag}, c_2^{\dag}\right),c_2, 1\right)$, respectively.

\begin{defn}
We say that a PQLS is observable or controllable on the accessible (noisy) channel if the system  is observable or controllable, respectively, when restricted to the accessible (noisy) channel.
\end{defn}

%We say that the PQLS is observable on the accessible (noisy) channel if the system $\left(A, c_1\right)$ $\left(\left(A, c_2^{\dag}\right),c_1, 1\right)\right)$ is observable.

%We say that the PQLS is controllable on the accessible (noisy) channel if the system $\left(A, c_1\right)$ $\left(\left(A, c_2^{\dag}\right),c_1, 1\right)\right)$ is observable.

\begin{thm}
Suppose that a PQLS is observable on the accessible (noisy) channel, then it is  controllable on the accessible (noisy) channel. 
\end{thm}
\begin{proof}
Suppose that the system is observable on the accessible channel. 
%
%Suppose that $\left(A, c_1\right)$ is observable. 
By Theorem \ref{obs5} this is equivalent to the statement: if $y$ is a (right) eigenvector of $A$ with eigenvalue $\lambda$ then $c_1y\neq0$. We need to show that $\left(A,  -\left(c_1^{\dag}, c_2^{\dag}\right)\right)$ is controllable, or equivalently $\left(A^{\dag},  (c_1, c_2)\right)$ is observable. Suppose for a contradiction that $\left(A^{\dag},  (c_1, c_2)\right)$ is not observable. Therefore there exists a (right)-eigenvector, $z$, of $A^{\dag}$ with eigenvalue $\mu$ such that $c_1z=0$ and $c_2z=0$. Using these
conditions we have 
\begin{align*}
\mu  z&=
A^{\dag}z\\&=-\left(A+c_1^{\dag}c_1+c_2^{\dag}c_2\right)\\&=-Az
\end{align*}
However, controllability would imply that $c_1z\neq0$, which is a contradiction. The proof of the statement on the noisy channel can be obtained in a similar way. 
\end{proof}

Note that the reverse direction in this theorem is not necessarily true and is a fact we look to exploit in the following. 

\begin{defn}
We say that a subspace of a PQLS is noise unobservable if it is unobservable  from the noise channel. That is, there exists a series of (left)-eigenvectors of $A$, denoted by $y_1, ...y_k$, such that $c_2y_i=0$ for all $i$. The \textit{noise unobservable subspace} (NUS) is given by Span$\{y_1, ...y_n\}$.
%we call such a subspace a ... 
%there exist a \textit{protected subspace} (PS)
%We say that there exist a \textit{protected subspace} (PS) if there exists left a eigenvector of $A$, i.e $Ay=\lambda y$ for some eigenvector $y$, such that $c_2y=0$. 
%The protected subspace is given by Span$\{y\}$. 
%Suppose that there exist $k$ such eigenvectors, $y_1, ...y_k$ then the protected subspace is given by Span$\{y_1, ...y_n\}$
\end{defn}

The motivation behind the definition of a NUS is that there will be modes that  are not visible from the noise output channel.   However, these modes may still be controllable from the noise input.
  %In other words by restricting only to the accessible channel,  then the effect of these modes will be  shielded from the noise.
 Note that the there are similarities between our definition of an NUS and that of a \textit{protected subspace} in \cite{Maj1}, in the sense that information flows from this subspace to the accessible channel(s) without loss to the noise channel(s). However, our emphasis here is on identification, rather than  information flow,  and their  setup is  also different to ours.
 %We also envisage that the term \textit{protected subspace} would be here 

 We have the following result, which says that the NUS  can be identified as in the noiseless case (see Theorem \ref{symplecticequivalence}).  

\begin{thm}\label{protected}
Suppose that we have a PQLS   $\left(A, -\left(c_1^{\dag}, c_2^{\dag}\right),\left(\begin{smallmatrix} c_1\\c_2\end{smallmatrix}\right), 1\right)$, where:
\begin{itemize}
\item 
The accessible system, $\left(A, -\left(c_1^{\dag}, c_2^{\dag}\right),c_1, 1\right)$, is observable (hence minimal) 
\item The PQLS has a    NUS. 
\end{itemize}
%with respect to the noise channels. 
Then on the NUS  the transfer function equivalent systems (TFEs) are related by a unitary transformation,  as in  Theorem \ref{symplecticequivalence}.
\end{thm}
\begin{proof}
Suppose that $\left(A, -\left(c_1^{\dag}, c_2^{\dag}\right),c_1, 1\right)$ and $\left(A', -\left(c_1^{'\dag}, c_2^{'\dag}\right),c'_1, 1\right)$ are minimal and have the same transfer function $\Xi_{acc}(s)$. Therefore, they must be related via 
\begin{equation}\label{useful}
A'=UAU^{-1}\quad \left(c_1^{'\dag}, c_2^{'\dag}\right)=\left(Uc_1^{'\dag}, Uc_2^{'\dag}\right)\quad c'_1=c_1U^{-1}
\end{equation}
for some invertible matrix $U$ (see Theorem \ref{eq.similarity}).
Now, using a similar technique as in the proof of  Theorem \ref{symplecticequivalence} it is easy to show from these conditions that $c_1=c_1U^{\dag}U$. 
We also have 
\begin{align}
\nonumber A'&=-A'^{\dag}-c_1^{'\dag}c'_1-c_2^{'\dag}c'_2\\
\nonumber&=-\left(U^{\dag}\right)^{-1}A^{\dag}U^{\dag}-Uc_1^{\dag}c_1U^{\dag}-c_2^{'\dag}c'_2\\
\nonumber&=\left(U^{\dag}\right)^{-1}AU^{\dag}+\left(U^{\dag}\right)^{-1}c_1^{\dag}c_1U^{\dag}+\left(U^{\dag}\right)^{-1}c_2^{\dag}c_2U^{\dag}-Uc_1^{\dag}c_1U^{\dag}-c_2^{'\dag}c'_2\\
\label{ponce}&=\left(U^{\dag}\right)^{-1}AU^{\dag}+\left(U^{\dag}\right)^{-1}c_2^{\dag}c'_2-c_2^{'\dag}c'_2,
\end{align}
where we have used the conditions \eqref{useful} as well as the physical realisation conditions (Eq.  \eqref{PR}). Under the assumption that  $\left(A', -\left(c_1^{'\dag}, c_2^{'\dag}\right),c'_1, 1\right)$ has a NUS, then \eqref{ponce} implies that  
\begin{equation}\label{ponce1}
A'y=\left(U^{\dag}\right)^{-1}AU^{\dag}y
\end{equation} for all vectors $y$ in the NUS. Combining \eqref{ponce1}  with $c_1=c_1U^{\dag}U$, it follows that
\begin{equation}
c_1AU^{-1}y=c_1AU^{\dag}y.
\end{equation}
It is not too difficult to show that this condition may be extended to 
\begin{equation}
c_1A^kU^{-1}y=c_1A^kU^{\dag}y
\end{equation}
for all $k=0,1,2,...$.
Therefore, by observability we have that $U^{\dag}y=U^{-1}y$. Hence on the NUS the equivalence class of systems \eqref{useful} are related by a unitary transformation, as in Theorem \ref{symplecticequivalence}.
\end{proof}

We have seen in the previous subsection how noise decreases identifiability. In quantum mechanics, the inclusion of noise in models also has  a detrimental effect on many other problems, for example metrology \cite{Demo1} or control \cite{Dong1}. In some cases this can be catastrophic, for example the use of N00N states in quantum metrology has the effect of destroying the enhanced level of Heisenberg scaling. Therefore, the fact that part of the system can be identified as if there is no noise present is very interesting. It can also be shown that an identical result holds for general QLSs.

A consequence of this theorem, which was apparent from the proof but not written in the theorem, is that if one system, $\left(A, -\left(c_1^{\dag}, c_2^{\dag}\right),c_1, 1\right)$, has a NUS and $\left(A', -\left(c_1^{'\dag}, c_2^{'\dag}\right),c'_1, 1\right)$ has the same transfer function then this second system must also have a NUS.

% this reduces to the trivial scenario of  concatenated system (see Sec. \ref{greed}). 

If a  NUS is also uncontrollable with respect to  the noise channel, then it becomes  a \textit{decoherence-free subsystem} (DFS) \cite{Maj1, Yamamoto5}.
A DFS is a subsystem of the QLS that is completely isolated from the noise; that is whose variables are not affected by the input and do not appear in the output. 
 To see this observe that the transfer function of the combined accessible and noise outputs is in this case given by:
\begin{align}\label{bg}
\Xi(s)&=1-\left(\begin{smallmatrix}c_1\\ c_2\end{smallmatrix}\right)(s-A)^{-1}\left(c_1^{\dag}, c_2^{\dag}\right)\\&= 1-\left(\begin{smallmatrix}c_1\\ c_2\end{smallmatrix}\right)\sum_{i}\frac{   R_iL_i    }{(s-\lambda_i)}\left(c_1^{\dag}, c_2^{\dag}\right),\end{align}
where $L_i, R_i, \lambda_i$ are the left eigenvectors, right eigenvectors and eigenvalues of $A$. Therefore if there is an unobservable-uncontrollable noise mode then $c_2R_i=0$ and $L_ic_2^{\dag}=0$ for some $i$ and so only the $(1,1)$ block of the transfer function has a contribution on that mode. This subsystem is clearly completely identifiable in this case (in the sense of  Corollary \ref{colo}) from the accessible output.
 DFSs are trivially NUSs, but interestingly  Theorem \ref{protected} says that all NUSs (beyond the trivial case of DFSs) are completely identifiable. We revisit DFSs later in Ch. \ref{FEDERER}.

%The other extreme to protected subsystems are \textit{noise unidentifiable subspaces}, where a subsystem is not visible from the accessible channel and hence cannot be identified from the output.
%\begin{defn}
%We call a subsystem a \textit{noise unidentifiable subspaces} (NUS) 
%if it is unobservable  from the accessible  channel. That is, there exists a series of (left)-eigenvectors of $A$, denoted by $y_1, ...y_k$, such that $c_1y_i=0$ for all $i$. The NUS is given by Span$\{y_1, ...y_n\}$.
%\end{defn}
%Note that  saw one example of this in Sec. \ref{noisek1}

\section{Summary and Outlook}

In summary we have addressed the system identification problem and  characterised all QLSs with the same  transfer function. Such equivalent systems are related by a symplectic transformation on the space of modes. Therefore we have extended the result of \cite{Guta2} beyond the class of passive systems. 
 We then outlined two methods to construct a (minimal and physical) realization of the system from the transfer function. 
 
We also  considered these same problems in the context of  noisy QLSs, which are modelled with the use of additional inaccessible channels. We  have investigated the notion of \textit{noise unobservable subspaces}, where part of the system is shielded from the noise. Interestingly, we found that such a subsystem can be identified as in the noiseless case. 
 An interesting topic of research is to understand more about noise unobservable subsystems and how far reaching their applications could be.
 
Finally, given that we now understand what is identifiable, the next  step is   
 to understand how well parameters can be estimated. This will be the subject of Chs. \ref{QEEP} and \ref{FEDERER}.

%%%%%%%%%%%%%%%%%%%%%%%%%
\chapter{Power Spectrum Identification}\label{powers34}
%%%%%%%%%%%%%%%%%%%%%%%%%

In the previous chapter we addressed system identification problem from a \emph{time-dependent} input perspective. 
We are now going to change viewpoint and consider a setting where the input fields are \emph{time-stationary} pure\footnote{The purity of the input state should be understood as that defined in Sec. \ref{stato}} Gaussian states as in Sec. \ref{stato}. The output is uniquely defined by its power spectrum \eqref{powers}.
The power spectrum identifiability is a natural and relevant setting in the quantum context, as it is in the classical one, where it was treated in the references \cite{Anders1, Youla1}. This setup is relevant when it may not be possible for the experimenter to use time-dependent inputs, e.g. when identification is performed in conjunction with control.

Our aim is to answer the same identifiability questions as those outlined in the time-dependent setting for this setting. That is:  
\begin{enumerate}
\item[(1)] Which parameters can be identified by measuring the output for a given input covariance matrix  $V(N, M)=:V$ (see Eq.  \eqref{Ito})? 
\item[(2)] How can we construct a system realisation from sufficient input-output data? 
\end{enumerate}
Notice that unlike the time-dependent case there is an explicit dependence on the input.  
The characterisation of the equivalence classes in the first question  boils down to finding which systems have the same \textit{power spectrum}, a problem which is well understood in the classical setting \cite{Anders1} but has not been addressed in the quantum domain. If two QLSs have the same power spectrum, then we call them \textit{power spectrum  equivalent} (PSE).
Moreover, since the power spectrum depends on the system parameters via the transfer function, it is clear that one can identify `at most as much as' in the time-dependent setting discussed in Sec. \ref{T.F}. In other words  
the corresponding equivalence classes are at least as large as those described by symplectic transformations \eqref{eq.equivalene.classes}. 
However, recall that in the analogous classical problem it was generally not possible to reconstruct the transfer function from the power spectrum uniquely (see Sec. \ref{ICE12}) even under global minimality \cite{Youla1, Anders1, Davies1}, which requires one to reconstruct the transfer function with the smallest dimension.

Consider the system's stationary  state and note that it can be uniquely written as a tensor product between a pure and a mixed Gaussian state (see Sec. \ref{QHO}). We see in Sec. \ref{G.Ms} that by restricting the system to the mixed component leaves the power spectrum unchanged. Furthermore, the pure component is passive, which ties in with previous results of \cite{Yamamoto1}. Conversely, if the stationary state is fully mixed, there exists no smaller dimensional system with the same power spectrum. Such systems will be called globally minimal, and can be seen as the analogue of minimal systems for the stationary setting.

The main result of this chapter is to show that under global minimality the power spectrum determines the transfer function, and therefore the equivalence classes are the same as those in the transfer function (i.e Theorem \ref{symplecticequivalence}). We  give three proofs for this. The first is given in Sec. \ref{jute}  for a generic class of SISO QLSs and is obtained by using a brute force argument to identify the poles and zeros of the transfer function from those of the power spectrum. The second proof, given in Sec. \ref{dogs}, holds for general QLSs and the key there is in reducing the power spectrum identifiability problem to an equivalent transfer function identifiability problem. Both of these methods use tools from classical systems theory. Our final method in Sec. \ref{JUKKA1} is a purely quantum one. In particular,  we use the observation that unidentifiable directions will have  zero QFI rate   (infinitesimal) directions in the parameter space with the same power spectrum.  We also  give an identification method in Sec. \ref{256} to reconstruct a system realisation of the power spectrum and discuss an example in Sec. \ref{god1}. 

In Sec. \ref{pass5} we restrict our attention to   PQLSs (with non-vacuum inputs). In particular, the identifiability problems turn out to have a much simpler solution for SISO PQLSs (Sec. \ref{SISO1}). We investigate global minimality  in more detail  and  understand which systems are globally minimal for both SISO and MIMO PQLSs. 

We also see in Sec. \ref{EI} we show that by using additional ancillary channels with an appropriately chosen entangled input ensures that one can identify the transfer function from the power spectrum for all \textbf{minimal} systems. The key point is that we are changing the input, which matters for the power spectrum, in order  to create a globally minimal system-input pair that is identifiable.

All of these identifiability problems have been discussed for \textbf{pure} inputs only. Finally, in Sec. \ref{thermy} we extend the identifiability results to thermal inputs, which are an  interesting class of mixed inputs.

\section{Global Minimality}\label{G.Ms}

We now formally introduce the definition of global minimality, which is analogous to the classical definition (see Definition \ref{CLSGM}).

\begin{defn}
A system $(S, C, \Omega)$ is  globally minimal for (pure) input covariance, $V$, if there exists no lower dimensional system with the same power spectrum, $\Psi_{V}$.
\end{defn}

To see why this definition important, consider for example a passive system with vacuum input.  In this case the  power spectrum will be vacuum, which is  the same as that of a zero-dimensional system. 

In fact, we can assume without loss of generality that the input is vacuum, i.e., $V=\left(\begin{smallmatrix}1&0\\0&0\end{smallmatrix}\right)=:V_{\mathrm{vac}}$. Essentially, as the 
 input is known (i.e the choice of the experimenter) and pure,  we can change the basis of the field so that the input is  vacuum and the system and output covariance (power spectrum) are modified as 
$\left(\tilde{S}, \tilde{C}, \tilde{\Omega}\right):=\left(S_{\mathrm{in}}^{\flat}SS_{\mathrm{in}}, S_{\mathrm{in}}^{\flat}C, \Omega\right)$ and
 $\tilde{\Psi}(s)=S_{\mathrm{in}}^{\flat}\Psi(s)\left(S_{\mathrm{in}}^{\dag}\right)^{\flat}$, respectively (recall the definition of $\flat$-adjoint from the Nomenclature). Here 
 $S_{\mathrm{in}}$ is given by Eq. \eqref{vtrick} in Sec. \eqref{QHO}.

The stationary state of the system is characterised by it's covariance matrix, $P$, which from   \eqref{langevin} is the unique  solution of the \textit{Lyapunov equation} \cite{Levitt1}
 \begin{equation}\label{eq.Lyapunov}
 {A}P+P{A}^{\dag}+{C}^{\flat}V_{\mathrm{vac}}({C^\flat})^{\dag}=0.
 \end{equation}
 
The following theorem links global minimality with the purity of the stationary state of the system.

 \begin{thm}\label{equivalence} 
 Let $\mathcal{G}:= \left(S, C, \Omega\right)$ be a minimal QLS with pure input $V_{\mathrm{vac}}$.
 
1. The system  is globally minimal  if and only if the (Gaussian) stationary state with covariance $P$ satisfying the Lyapunov equation \eqref{eq.Lyapunov}  is fully mixed.

2. A non-globally minimal system is transfer function equivalent (TFE) (see Sec. \ref{T.F}) to a QLS which is a series product of two  systems;  the first system has a  pure stationary state, whereas the second has a fully mixed stationary state (see Fig. \ref{spane}). We call these systems the \textit{pure component} and \textit{mixed component}, respectively. 
%Moreover, its pure part is TFE to a system of the form: 
%\begin{equation}\label{split}
% \left({1}, S_{in}\Delta(C_{p},0), \Delta(\Omega_{p},0)\right).
%\end{equation}

3. The reduction to the mixed component is globally minimal and has the same power spectrum as the original system. 
\end{thm}
 %Finally, the proof of Part (3) follows from Theorem 1 in \cite{Naoki}. We do not give the details here as an  alternative proof is given in Theorem \ref{mainresult}.???????
  \begin{figure}
\centering
\includegraphics[scale=0.22]{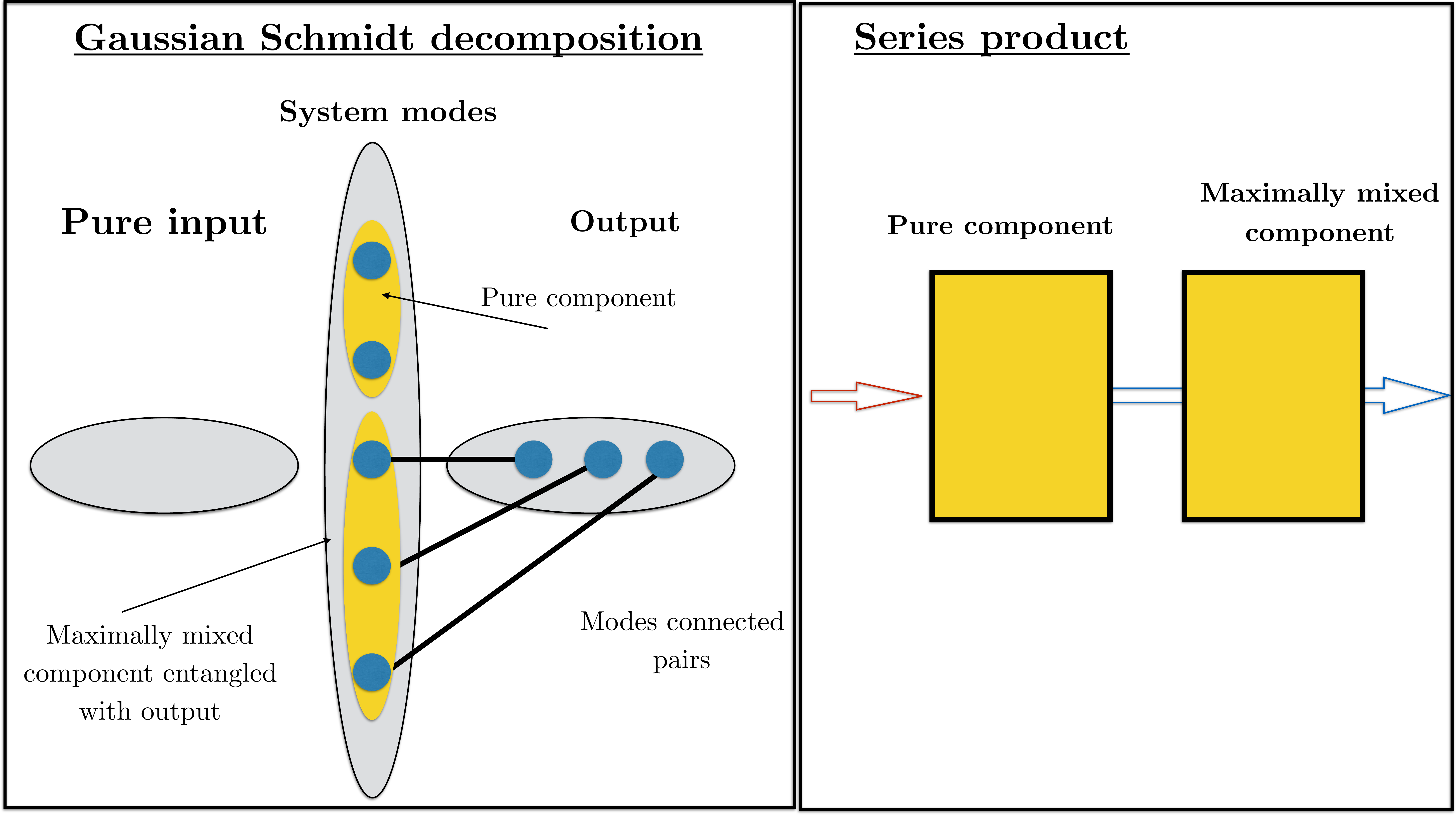}
\caption{  Pure and maximally mixed components connected in series.  \label{spane}}
\end{figure}
\begin{proof}
%Firstly, note that any squeezing/scattering matrix $S$ does not affect the question of global minimality, so without loss of generality we can take it to be identity. 
Let us prove the result first in the case $S={1}$.

Firstly, perform a change of system coordinates\footnote{Note that such a symplectic transformation on the system  is of form prescribed by Theorem \ref{symplecticequivalence}, but the interpretation here is that we are dealing with the same system seen in a different basis, rather than a different system with the same transfer function.}
as described in Sec. \ref{QHO}, so that the input is in the vacuum state, while the system modes decompose into its `pure' and `mixed' parts ${\bf a}^{\prime T} =({\bf a}_p^T , {\bf a}_m^T)$. Note that this transformation will alter the  coupling and Hamiltonian matrices accordingly, but  we still denote them $\Omega$ and $C$ to simplify notations.

Therefore,  in this basis the stationary state of the system is given by the covariance
$$
P=\left(\begin{smallmatrix} R+{1} & 0\\0 & R\end{smallmatrix}\right), \qquad R=\left(\begin{smallmatrix} 0&0\\0& R_{m}\end{smallmatrix}\right)
$$
and satisfies the Lyapunov equation \eqref{eq.Lyapunov}.

($\implies$)  We show that if the system has a pure component, then it is globally reducible. Let us write  $A_{\pm}$ and $C_{\pm}$ as block matrices according to the pure-mixed splitting 
%\[
%A_{-}=\left(\begin{matrix} A_{pp}&A_{pm}\\A_{mp}&A_{mm}\end{matrix}\right), \,\,\, A_{+}=\left(\begin{matrix} B_{pp}&B_{pm}\\B_{mp}&B_{mm}\end{matrix}\right),\]
%\[C_{-}=\left(C_p, C_m\right), \,\,\,\mathrm{and}\,\,\, C_{+}=\left(D_p, D_m\right),\]

\[
A_{\pm}=\left(\begin{matrix} A_{\pm}^{pp}&A_{\pm}^{pm}\\A_{\pm}^{mp}&A_{\pm}^{mm}\end{matrix}\right), \qquad
%\,\,\, 
%A_{+}=\left(\begin{matrix} B_{pp}&B_{pm}\\B_{mp}&B_{mm}\end{matrix}\right),\]
%\[
C_{\pm}=\left(C_{\pm}^p, C_{\pm}^m\right),
 %\,\,\,\mathrm{and}\,\,\, C_{+}=\left(D_p, D_m\right),
 \]
so that the Lyapunov equation \eqref{eq.Lyapunov} can be seen as a system of 16 block matrix equations.  
Taking the (1,1) and (1,3) blocks, which correspond to the $\left<\mathbf{a}_p\mathbf{a}_p^{\dag}\right>$ and $\left<\mathbf{a}_p\mathbf{a}_p\right>$ components of the stationary state, one obtains
\begin{align}
&A_{-}^{pp}+A_{-}^{pp \dag}+C_{-}^{p\dag}C_{-}^p=0\label{block1}\\
&A_{+}^{pp T}-C_{-}^{p \dag}C_+^p=0\label{block3}.
\end{align}
Since $A_{-}^{pp}=-i\Omega_{-}^{pp}  - 1/2(C_{-}^{p \dag} C_{-}^p - C_{+}^{p T} C_{+}^{p\#})$, Eq. (\ref{block1}) implies that $C_{+}^{pT} C_{+}^{p \#}=0$, hence $C_{+}^{p}=0$. Therefore, using this fact in Eq. (\ref{block3}) gives $A_+^{pp}=0$, hence $\Omega_{+}^{pp}=0$. 
These two tell us that the pure part contains only passive terms.

Consider now the $(1, 2)$ and $(2, 3)$ blocks, which correspond to the $\left<\mathbf{a}_p\mathbf{a}_m^{\dag}\right>$ and $\left<\mathbf{a}_m\mathbf{a}_p\right>$ components of the stationary state. From this,  we get 
\begin{align}
&A_{-}^{pm}(R_m+{1})+
%B_{pm}Q_m^{\dag}+
A_{-}^{pm \dag}+C_{-}^{p \dag}C_{-}^m=0\label{block4}\\
&(R_m+1)A_+^{pm T}
%+Q_mA_{pm}^T
=0\label{block5}.
\end{align}
Since $A_{-}^{pm}+A_{-}^{pm \dag}+C_{-}^{p \dag}C_{-}^m=0$, and $R_m$ is invertible, 
equation (\ref{block4}) implies $A_{-}^{pm}=0$. Similarly, Eq. (\ref{block5}) implies that $A_{+}^{pm}=0$.
%\begin{equation}
%\left(\begin{smallmatrix} R_m+1&Q_m\\ Q_m^{\dag} &R_m^T\end{smallmatrix}\right)\left(\begin{smallmatrix}B_{pm}^T\\A_{pm}^T\end{smallmatrix}\right)=0.
%\end{equation} 

%Now we use a result from \cite{WOLF}. Essentially, by treating the pure multi-mode gaussian state over the system and output field as a bipartite division, there exists a canonical basis so that the entanglement is localised to one mode at each side (between the mixed part of the system and the field). That is, the modes are entangled in pairs (see Figure \ref{coupling}). The meaning of this is that the mixed part of the system is fully mixed with the output and so the determinant of $\left(\begin{smallmatrix} R_m+1&Q_m\\ Q^{\dag}_m &R_m^T\end{smallmatrix}\right)$  is non-zero \cite{squeezing}. Hence
%$A_{pm}=0$ and $B_{pm}=0$.
%This result implies (see \cite{cascade, Gough3}) that the mixed and pure parts of the system are connected via the series product, as required. 
Let $\mathcal{G}^p:= ({1}, \Omega^{pp}, C^p)$ be the system consisting of the pure modes, with 
$\Omega^{pp}= \Delta(\Omega^{pp}_{-}, 0)$ and 
$C^p = \Delta (C_-^p, 0)$. Let $\mathcal{G}^m:= ({1}, \Omega^{mm}, C^m) $ be the system consisting of the mixed modes with 
$\Omega^{mm}= \Delta(\Omega^{mm}_{-}, \Omega^{mm}_{+})$ and $C^m = \Delta (C_-^m, C_+^m)$. We can now show that the original system is the series product (concatenation) of the pure and mixed restrictions
$$
\mathcal{G} = \mathcal{G}^m \triangleleft \mathcal{G}^p.
$$
Indeed, using the fact that $C_{+}^p= \Omega_{+}^{pp}= A_{-}^{pm} = A_{+}^{pm}=0$, one can check that the series product has required matrices  \cite{Gough3}
$$
C_{series} = \tilde{C}^p + \tilde{C}^m = C
$$
and 
$$
\Omega_{series} = \tilde{\Omega}^{pp}+ \tilde{\Omega}^{mm} + \mathrm{Im}_\flat (\tilde{C}_m^\flat \tilde{C}_p)
$$
where the  `tilde' notation stands for block matrices where only one block is non-zero, e.g. $\tilde{C}^p= (C^p,  0)$, and 
$ \mathrm{Im}_\flat X:= (X- X^\flat)/2i $.

Now, let $\Xi^{p,m}(s)$ denote the transfer functions of  $\mathcal{G}^{p,m}$; since the transfer function of a series product is the product of the transfer functions, we have $\Xi(s) = \Xi^{m}(s)\cdot \Xi^{p}(s)$.  Furthermore, since  
$\mathcal{G}^{p}$ is passive and the input is vacuum, we have $\Psi^p_{ V}(s)   =  \Xi^p(s) {V} \Xi^p(-\overline{s})^\dagger = {V} $ so that 
$$
\Psi_{ V}(s) = \Xi(s) {V} \Xi(-\overline{s})^\dagger =  \Xi_m(s) {V} \Xi_m(-\overline{s})^\dagger 
$$
which means that the original system was globally reducible (not minimal).

($\impliedby$)
We now show that if the system's stationary state is fully mixed, then it is globally minimal. The key idea is that a sufficiently long block of output has a finite symplectic rank (number of modes in a mixed state in the canonical decomposition) equal to twice the dimension of the system. Therefore the dimension of a globally minimal system is ``encoded'' in the output. This is the linear dynamics analogue of the fact that stationary outputs of finite dimensional systems (or translation invariant finitely correlated states) have rank equal to the square of the system dimension (or bond dimension) \cite{Guta4}. To understand this property consider the system (S) together with the output at a long time 
$2T$, and split the output into two blocks: A corresponding to an initial time interval $[0,T]$ and  $B$ corresponding to 
$[T,2T]$. If the system starts in a pure Gaussian state, then the $S+A+B$ state is also pure. By ergodicity, at time $T$ the system's state is close to the stationary state with symplectic rank $d_m$. At this point the system and output block A are in a pure state so by appealing to the `Gaussian Schmidt decomposition' \cite{Wolf1} we find that the state of the block $A$ has the same symplectic eigenvalues (and rank $d_m$) as that of the system (see Fig. \ref{spane}). In the interval $[T,2T]$ the output $A$ is only shifted without changing its state, but the correlations between A and S decay. Therefore the joint $S+A$ state is close to a product state and has symplectic rank $2 d_m$. On the other hand we can apply the Schmidt decomposition argument to the pure bipartite system consisting of $S+A$ and $B$ to find that the symplectic rank of $B$ is $2d_m$. By ergodicity, $B$ is close to the stationary state in the limit of large times, which proves the assertion.

To extend the result to $S\neq{1}$, 
instead perform the change of field co-ordinates $V\mapsto S_{\mathrm{in}}S^bV\left(S_{\mathrm{in}}S^b\right)^{\dag}$ in \eqref{vtrick}. The proof then follows as above because in this basis $S={1}$.
\end{proof}

This result has no classical analogue and is particularly interesting because it relates a classical concept, i.e, global minimality, with the quantum concept of purity.
This theorem enables one to check global minimality by computing the symplectic eigenvalues of the stationary state (see Sec. \ref{QHO}). If all eigenvalues are non-zero, then the state is fully mixed and the system is globally minimal. We emphasise that the argument relies on the fact that the input is a pure state. For mixed input states, the stationary state may be fully mixed while the system is non globally minimal (see Sec. \ref{thermy}   later).

The following Lemma will be of use later.

\begin{Lemma}\label{LEM1}
Suppose that we have a QLS $\left(S, C, \Omega\right)$ with input $V_{\mathrm{vac}}$, then the following are equivalent:
\begin{enumerate}
\item The system is  globally minimal
 \item $\left(A, C^{\flat}S{ V}_{\mathrm{vac}}\right)$ is controllable.
 \item  $\left({ V}_{\mathrm{vac}}S^{\flat}C, A^{\flat}\right)$ is observable.
 \end{enumerate}
 \end{Lemma}
 
\begin{proof}
For the equivalence between (1) and (2): Using Theorem \ref{equivalence}, global minimality is equivalent to a fully mixed stationary state, which is in turn equivalent to  $P>0$ in \eqref{eq.Lyapunov}.  Furthermore, by Theorem 3.1 in \cite{Zhou1}  $P>0$ in Eq.  \eqref{eq.Lyapunov}  is equivalent to  $\left(A, C^{\flat}SV_{\mathrm{vac}}\right)$ being controllable. 

It remains to show equivalence between (2) and (3). Firstly, by the duality condition \eqref{cl4}  in Theorem \ref{cl5}  $\left(A, C^{\flat}S{ V}_{\mathrm{vac}}\right)$  controllable is equivalent to   $\left({V}_{\mathrm{vac}}S^{\dag}\left(C^{\flat}\right)^{\dag}, A^{\dag}\right)$ observable. It therefore remains to show equivalence between the observability of $\left({V}_{\mathrm{vac}}S^{\dag}\left(C^{\flat}\right)^{\dag}, A^{\dag}\right)$ and $\left({V}_{\mathrm{vac}}S^{\flat}C, A^{\flat}\right)$. 

Suppose that $\left({ V}_{\mathrm{vac}}S^{\dag}\left(C^{\flat}\right)^{\dag}, A^{\dag}\right)$ is observable. To show observability of $\left({V}_{\mathrm{vac}}S^{\flat}C, A^{\flat}\right)$ we need to show that for all eigenvectors and eigenvalues of $A^{\flat}$, i.e. $A^{\flat}y=\lambda y$, then ${ V}_{\mathrm{vac}}S^{\flat}Cy\neq0$ Theorem \ref{obs5}. 
To this end suppose that $A^{\flat}y=\lambda y$, then $A^{\dag}\left(Jy\right)=\lambda\left(Jy\right)$, which by the observability of $\left({ V}_{\mathrm{vac}}S^{\dag}\left(C^{\flat}\right)^{\dag}, A^{\dag}\right)$ implies that ${ V}_{\mathrm{vac}}S^{\dag}\left(C^{\flat}\right)^{\dag}\left(Jy\right)\neq0$. Therefore, ${ V}_{\mathrm{vac}}S^{\flat}Cy\neq0$ and we are done.
The reverse implication follows similarly.
\end{proof}
For simplicity we shall  now assume  (until Sec. \ref{sq+sc}) that there is no squeezing or scattering in the field, i.e. $S=1$.  We discuss  the case $S\neq1$ in detail in Sec. \ref{sq+sc}.

\section{Description of the Power Spectrum as  Cascaded CLSs}\label{altos}

In this subsection we show that the power spectrum of  our QLS can be treated as a transfer function of a cascade of two classical systems (with the combined system having twice as many modes). Furthermore, the resultant cascaded system will be minimal iff the original system is globally minimal. This result will be particularly important in  Sec. \ref{dogs} because the power spectrum identification problem reduces to a transfer function identification problem, which is which is much simpler to solve. 

Using Eq. \eqref{powers} for the power spectrum,  write $\Psi(s)J$ as a transfer function of the following two cascaded systems:
\begin{itemize}
\item The first system  is $\left(-A^{\flat}, -C^{\flat}, -V_{\mathrm{vac}}C, V_{\mathrm{vac}}\right)$
\item The second system is  $\left(A, -C^{\flat}V_{\mathrm{vac}}, C, V_{\mathrm{vac}}\right)$. 
\end{itemize} 
It should be understood that the first system is fed into the second (see Fig. \ref{casc}). Note that the first system is  unstable, whereas the second is stable. Using Eq. \eqref{concat1}, a representation for the resultant system with transfer function $\Psi(s)J$ is 
\begin{equation}\label{cask}
\left(\tilde{A}, \tilde{B}, \tilde{C}, \tilde{D}\right):=
\left(\left(\begin{smallmatrix} -A^{\flat}&0\\C^{\flat}V_{\mathrm{vac}}C&A\end{smallmatrix}\right), \left(\begin{smallmatrix} -C^{\flat}\\-C^{\flat} V_{\mathrm{vac}}\end{smallmatrix}\right), \left(\begin{smallmatrix} -V_{\mathrm{vac}}C& C\end{smallmatrix}\right), V_{\mathrm{vac}}       \right).
\end{equation}

%%%%%%%
 \begin{figure}[h]
\centering
\includegraphics[scale=0.4]{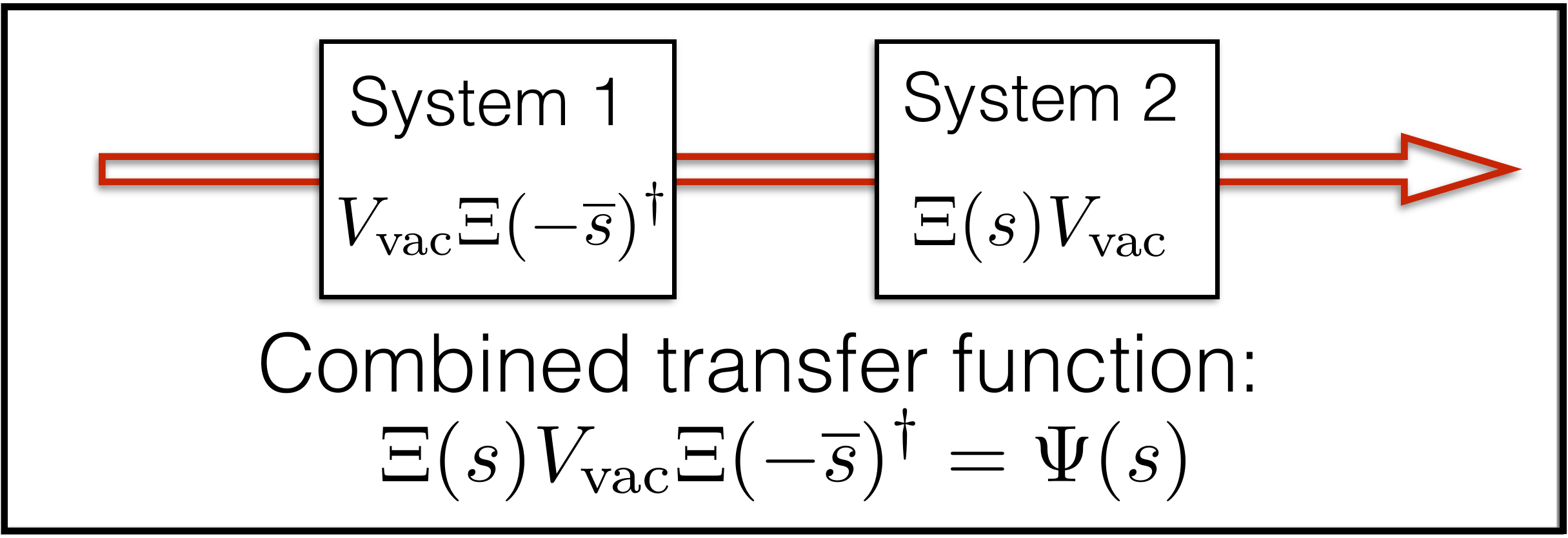}
\caption{The setup in Sec. \ref{altos} where the power spectrum is treated as two systems connected in series. \label{casc}}
\end{figure}
%%%%%%%%%%%%%%%%%%%%%%%

Now, $\tilde{A}$ has $4n$ eigenvalues. It is also 
 lower block triangular (LBT) with the following properties:
\begin{itemize}
\item[(1)] \label{pil1} It has $2n$ right-(generalised\footnote{A matrix is diagonalisable iff it has a full basis of eigenvectors. Generalised eigenvectors are a next best thing to eigenvectors enabling one to `almost diagonalise' a matrix. More specifically, a vector  $x$ is a generalised eigenvector of rank $m$ with corresponding eigenvalue $\lambda$ if 
$$\left(A-\lambda {1}\right)^mx=0$$
(but
$\left(A-\lambda {1}\right)^{m-1}x\neq0$).
For every matrix $A$ there exists an invertible matrix $M$, whose columns consist of the generalised eigenvectors, such that $J=M^{-1}AM$ where $J$ is a matrix called the \textit{Jordan normal matrix} and is given by 
$$J=\mathrm{Diag}(J_1, J_2,..., J_r) \quad \mathrm{where} \quad
J_i=\left(\begin{smallmatrix} \lambda_i &1&&\\
&\lambda_i &1&\\
&&\ddots&1\\
&&&\lambda_i\end{smallmatrix}\right).$$
 })-eigenvectors of the form $\left(\begin{smallmatrix}0\\y_2^{(i)}\end{smallmatrix}\right)$ with (possibly non-distinct) eigenvalues $\lambda^{(i)}$, which satisfy  $\mathrm{Re}(\lambda^{(i)})<0$. Note that $y_2^{(i)} $ and $\lambda^{(i)}$ are right-(generalised) eigenvectors and eigenvalues of $A$.
\item[(2)] \label{pil2} 
It has $2n$ left-(generalised)-eigenvectors of the form  $\left(\begin{smallmatrix}x_1^{(i)},&0\end{smallmatrix}\right)$ with  (possibly non-distinct) eigenvalues $\mu^{(i)}$, which satisfy $\mathrm{Re}(\mu^{(i)})>0$. Note that $x_1^{(i)}$ and $\mu^{(i)}$ are left-eigenvectors and eigenvalues of $-A^{\flat}$.
\end{itemize}

\begin{defn}
A matrix $A$ is called \textit{proper ordered lower block triangular (proper LBT)} if it is it LBT and satisfies (1) and (2).
\end{defn}

%The following Lemma will be useful later. 

\begin{Lemma}\label{hud}
If two proper LBT matrices, $\tilde{A}$ and $\tilde{A}'$,  are related via $\tilde{A}'=T\tilde{A}T^{-1}$, where $T$ is invertible, then $T$ is LBT. 
\end{Lemma}

The proof is in Appendix \ref{goh}.
The final result of this subsection, which is another equivalent formulation of global minimality, 
%perhaps
 will be key to our main identifiability result in this chapter. 

\begin{thm}\label{games}
The quantum system $(C, \Omega)$ is globally minimal if and only if the system \eqref{cask} is minimal. 
\end{thm}

\begin{proof}
The reverse implication here is trivial. For the forward implication we need to prove controllability and observability. 

Firstly, the observability of $\left(\tilde{C}, \tilde{A}\right)$. Suppose that  
\begin{equation}\label{contro} 
\left(\begin{smallmatrix} -A^{\flat}&0\\C^{\flat}V_{\mathrm{vac}}C&A\end{smallmatrix}\right)\left(\begin{smallmatrix}y_1\\y_2\end{smallmatrix}\right)=\left(\begin{smallmatrix}\lambda y_1\\\lambda y_2\end{smallmatrix}\right),
\end{equation}
then in order to show observability we require that $\left(\begin{smallmatrix} -V_{\mathrm{vac}}C& C\end{smallmatrix}\right)\left(\begin{smallmatrix}y_1\\y_2\end{smallmatrix}\right)\neq0$. There are two cases; either $y_1=0$ or $y_1\neq0$. 
\begin{itemize}
\item If $y_1=0$ then \eqref{contro} reduces to   $Ay_2=\lambda y_2$ and so the observability of $A$  tells us that $Cy_2\neq0$. Hence 
$\left(\begin{smallmatrix} -V_{\mathrm{vac}}C& C\end{smallmatrix}\right)\left(\begin{smallmatrix}0\\y_2\end{smallmatrix}\right)\neq0$. 
\item For $y_1\neq0$, the proof is a little trickier.  Suppose to the contrary that the system is not observable. That is, there exists a vector $\left(\begin{smallmatrix}y_1\\y_2\end{smallmatrix}\right)$ satisfying \eqref{contro} such that 
\begin{equation}\label{ps}
V_{\mathrm{vac}}Cy_1=Cy_2
\end{equation} 
Firstly, from 
 \eqref{contro} it is clear that  $-A^{\flat}y_1=\lambda y_1$, hence  $V_{\mathrm{vac}}Cy_1\neq0$ by global minimality (Lemma \ref{LEM1}). 
 We also have $C^{\flat}V_{\mathrm{vac}}Cy_1+Ay_2=\lambda y_2$  from \eqref{contro}, hence $-A^{\flat}y_2=\lambda y_2$ using \eqref{ps}. 
 On the other hand, letting  $y_2=\left(\begin{smallmatrix}u_1\\u_2\end{smallmatrix}\right)$, where $u_1, u_2$ are $n$ dimensional complex vectors, then by the doubled-up properties of $A^{\flat}$ it follows that $\left(\begin{smallmatrix}\overline{u}_2\\\overline{u}_1\end{smallmatrix}\right)$ is  also an eigenvector of $-A^{\flat}$ (with eigenvalue $\overline{\lambda}$).
 Therefore, $V_{\mathrm{vac}}C\left(\begin{smallmatrix}\overline{u}_2\\\overline{u}_1\end{smallmatrix}\right)\neq0$ by global minimality (Lemma \ref{LEM1}). Finally, this condition implies that $\overline{C_-}u_2+\overline{C_+}u_1\neq0$, which is a contradiction to \eqref{ps}.
% Finally, writing $C$ as $\left(\begin{smallmatrix}C_{\mathrm{upper}}\\C_{\mathrm{lower}}\end{smallmatrix}\right)$, it follows from $V_{\mathrm{vac}}C\left(\begin{smallmatrix}\overline{u}_2\\\overline{u}_1\end{smallmatrix}\right)\neq0$  that $C_{\mathrm{lower}}y_2\neq0$, which is a contradiction to \eqref{ps}.
 %(recall that $X=\left(\begin{smallmatrix}1&0\\0&0\end{smallmatrix}\right)$)
  Hence the system is observable.

% any eigenvalue of $-A^{\flat}$writing $y_2$ as $\left(\begin{smallmatrix}u_1\\u_2\end{smallmatrix}\right)$, where $u_1, u_2$ are $n$ dimensional complex vectors, then $\left(\begin{smallmatrix}\overline{u}_2\\\overline{u}_1\end{smallmatrix}\right)$ is also an eigenvector (with eigenvalue $\overline{\mu}$).

% for any eigenvector and eigenvalue of $-A^{\flat}$, written as 
%$$-\left(\begin{smallmatrix}B_1&B_2\\B_2^{\#}&B_1^{\#}\end{smallmatrix}\right)\left(\begin{smallmatrix}u_1\\u_2\end{smallmatrix}\right)=\mu\left(\begin{smallmatrix}u_1\\u_2\end{smallmatrix}\right),$$
%then $\left(\begin{smallmatrix}\overline{u}_2\\\overline{u}_1\end{smallmatrix}\right)$ is also an eigenvector (with eigenvalue $\overline{\mu}$). Hence, writing $y_2$ as $\left(\begin{smallmatrix}u_1\\u_2\end{smallmatrix}\right)$, then $XC\left(\begin{smallmatrix}\overline{u}_2\\\overline{u}_1\end{smallmatrix}\right)\neq0$ by global minimality. Finally, writing $C$ as $\left(\begin{smallmatrix}C_{\mathrm{upper}}\\C_{\mathrm{lower}}\end{smallmatrix}\right)$, it follows from $XC\left(\begin{smallmatrix}\overline{u}_2\\\overline{u}_1\end{smallmatrix}\right)\neq0$  that $C_{\mathrm{lower}}y_2\neq0$, which contradicts \eqref{p1} (recall that $X=\left(\begin{smallmatrix}1&0\\0&0\end{smallmatrix}\right)$).

%SHOW MUST HAVE BEEN OBSERVABLE

%And so $y_1=y_2$ using the genericness assumption. Finally, it follows that for equation \eqref{p1} to be satisfied we require  $\left(iXC+C\right)y_1=0$, which is not possible under the previous assumption $XCy_1\neq0$.
\end{itemize}

Showing controllability of $\left(\tilde{A}, \tilde{B}\right)$ can be achieved by similar means. Alternatively, we can use the dual properties of observability and controllability to show this. To this end, in order to show that $\left(\tilde{A}, \tilde{B}\right)$ is controllable it is enough to show that $\left(\tilde{B}^{\dag}, \tilde{A}^{\dag}\right)$ is observable (see Theorem \ref{cl5}).
In light of this, suppose that $\tilde{A}^{\dag}\left(\begin{smallmatrix}z_1\\z_2\end{smallmatrix}\right)=\lambda\left(\begin{smallmatrix}z_1\\z_2\end{smallmatrix}\right)$, which, by using the definition of $\tilde{A}$, is equivalent to 
\begin{align*}
-JAJz_1+C^{\dag}V_{\mathrm{vac}}CJz_2&=\lambda z_1 \quad\mathrm{and}\quad
A^{\dag}z_2=\lambda z_2.
\end{align*}
These equations can be written in matrix form as
%\begin{align*}
%A\left(Jz_1\right)-C^{\flat}XC\left(Jz_2\right)&=-\lambda \left(Jz_1\right)\\
%-A^{\flat}\left(Jz_2\right)&=-\lambda \left(Jz_2\right),
%\end{align*}
%or, by writing in matrix form we have 
%
$$
\tilde{A}\left(\begin{smallmatrix}Jz_2\\-Jz_1\end{smallmatrix}\right)=-\lambda\left(\begin{smallmatrix}Jz_2\\-Jz_1\end{smallmatrix}\right).
$$
Now, because $\left(\tilde{C}, \tilde{A}\right)$ is observable, it follows that
$$-C\left(Jz_1\right)-V_{\mathrm{vac}}C\left(Jz_2\right)\neq0.$$
This condition is equivalent to $\tilde{B}^{\dag}\left(\begin{smallmatrix}z_1\\z_2\end{smallmatrix}\right)\neq0$. 
%So in summary, we have $\tilde{A}^{\dag}\left(\begin{smallmatrix}z_1\\z_2\end{smallmatrix}\right)=\lambda\left(\begin{smallmatrix}z_1\\z_2\end{smallmatrix}\right)$ implies $\tilde{B}^{\dag}\left(\begin{smallmatrix}z_1\\z_2\end{smallmatrix}\right)\neq0$, which is to say that  $\left(\tilde{B}^{\dag}, \tilde{A}^{\dag}\right)$ is observable. 
%Note that we could also run this argument in reverse to show that controllability of the system implies observability. 
\end{proof}

\section{Power Spectrum Identification of SISO QLSs}\label{jute}

The following theorem shows that two generic\footnote{Under the conditions discussed in Sec. \ref{seriesp} allowing the transfer function to be realised as a cascade of one mode systems.} globally minimal SISO QLSs have the same power spectrum if and only if they have the same transfer function.  In particular are related by a symplectic transformation, as described in Theorem \ref{symplecticequivalence}.

\begin{thm}\label{mainresult}
Let $\left(C_1, \Omega_1\right)$ and $\left(C_2, \Omega_2\right)$ be two globally minimal SISO systems for fixed pure 
input with covariance ${V}_{\mathrm{vac}}$, which are assumed to be generic in the sense of \cite{Nurdin3}. Then
$$
\Psi_{1}(s)=\Psi_{2}(s) \,\, for~all~s \quad \Leftrightarrow \quad \Xi_1(s)=\Xi_2(s)\,\, for~all~s 
$$
\end{thm}
 \begin{proof}
Recall that the power spectrum of a system $\left(C, \Omega\right)$ is given by $\Xi(s)V_{\mathrm{vac}}\Xi(-\overline{s})^{\dag}$. 
Therefore, if $\Xi_1(s)=\Xi_2(s)$ then $\Psi_{1}(s)=\Psi_{2}(s) $. We will now prove the converse. 

The power spectrum in the SISO case  is given by 
\begin{equation}\label{PS1}
\left(\begin{array}{cc} 
\Xi_{-}(s)  {\Xi_{-}(-\overline{s})}^{\#}      &\Xi_{-}(s)\Xi_{+}(-s)\\[2mm]
%&\\
{\Xi_{+}(\overline{s})^{\#}\Xi_{-}(-\overline{s})^{\#}}&{\Xi_{+}(\overline{s})}^{\#} \Xi_+(-s)    \end{array}\right).
\end{equation}
%We have dropped the tildes from this modified transfer function   for notational convenience.

The transfer function is completely characterised by the elements in the top row of its matrix, i.e., $\Xi_{-}(s) $ and     $\Xi_{+}(s)$. Also, $\Xi_{-}(s) $ and     $\Xi_{+}(s)$ must be of the the form \eqref{form1} and \eqref{form2}.  Our first observation is that  $\Xi_{-}(s) $ and  $\Xi_{+}(s)$ in \eqref{form1} and \eqref{form2} cannot contain poles and zeros in the following arrangement: $\Xi_{-}(s) $ has a factor like
\begin{equation}\label{topple1}
\frac{(s-{\overline{\lambda}}_i)(s+\overline{\lambda}_i)}{(s-\overline{\lambda}_i)(s-\lambda_i)}=\frac{(s+\overline{\lambda}_i)}{(s-\lambda_i)}
\end{equation}
 \textbf{and} $\Xi_{+}(s)$ contains a factor like 
\begin{equation}\label{tipple1}
\frac{(s-\lambda_i)(s+\lambda_i)}{(s-\overline{\lambda}_i)(s-\lambda_i)}=\frac{(s+\lambda_i)}{(s-\overline{\lambda}_i)}.
\end{equation}
For if this were the case and assuming that this could be done $k$ times, then our original system could be decomposed as a cascade (series product) of two systems. 
\begin{itemize}
\item
The first system is an $k$-mode passive system with transfer function 
\begin{equation}\label{huck1}
\Xi^{(1)}(s)=\left(\begin{smallmatrix} \Xi_-^{(1)}(s)&0\\0&\Xi_-^{(1)}(\overline{s})^{\#}\end{smallmatrix}\right), 
\end{equation}
where 
$$
\Xi_-^{(1)}(s)=\prod^{k}_{i=1}\frac{(s+\overline{\lambda}_i)}{(s-\lambda_i)}, \qquad
\Xi_-^{(1)}(\overline{s})^{\#}=\prod^{k}_{i=1}\frac{(s+\lambda_i)}{(s-\overline{\lambda}_i)}.
$$
Note that by Example \ref{sisoexample} it is physical. 
\item The second system has transfer function 
\begin{equation}\label{huck2}
\Xi^{(2)}(s)=   \left(\begin{smallmatrix} \Xi_-^{(2)}(s)&\Xi_+^{(2)}(s)\\\Xi_+^{(2)}(\overline{s})^{\#}&\Xi_-^{(2)}(\overline{s})^{\#}\end{smallmatrix}\right),\end{equation}
where 
$$ \Xi_-^{(2)}(s)=\Xi_-(s)   \prod^{k}_{i=1}\frac{\left(s+\overline{\lambda}_i\right)}{\left(s-\lambda_i\right)},$$
$$ \Xi_+^{(2)}(s)=    \Xi_+(s)   \prod^{k}_{i=1}\frac{\left(s+\lambda_i\right)}{\left(s-\overline{\lambda}_i\right)}.$$
It can be shown that there exists an $n-k$ mode minimal physical quantum system with this transfer function (see Appendix \ref{APP2}).
\end{itemize}
Since $\Xi^{(1)}(s)$ is passive, 
$$\Xi^{(1)}(s)V_{\mathrm{vac}}\Xi^{(1)}(-\overline{s})^{\dag}=V_{\mathrm{vac}}$$ and hence this $k$-mode system is not visible from the power spectrum, while the power spectrum is the same as that of the lower dimensional system $\Xi^{(2)}(s)$. 
Therefore we have a contradiction to global minimality.

We will now construct $\Xi_{-}(s) $ and     $\Xi_{+}(s)$ directly from the power spectrum. 
This is equivalent to identifying their poles and zeros \footnote{Note that some of the poles and zeros in \eqref{form1} and \eqref{form2} may be  ``fictitious'' and so will not be required to be identified.}. To do this we must identify all poles and zeros of $\Xi_-(s)$ and $\Xi_+(s)$ from the three quantities:
\begin{align}
& \label{blue}  \Xi_{-}(s)  {\Xi_{-}(-\overline{s})}^{\#}   \\
& \label{red} \Xi_{-}(s)\Xi_{+}(-s)\\
& \label{yellow} {\Xi_{+}(\overline{s})}^{\#} \Xi_+(-s).
\end{align}

Firstly, all poles of $\Xi_-(s)$ and $\Xi_+(s)$ may be identified from the power spectrum. Indeed, due to stability, each 
pole in \eqref{blue}, \eqref{red}, \eqref{yellow} can be assigned unambiguously to either $\Xi_{-}(s) $ or $ \Xi_{+}(-s)$. However, cancelations between zeros and poles of the two terms in the product may lead to some transfer function poles not being identifiable, so we need to show that this is not possible.  Suppose that a pole $\lambda$ of  $\Xi_-(s)$ is not visible from the power spectrum. This implies  
\begin{itemize}
\item From \eqref{blue}, $\lambda$ is a zero of $ {\Xi_{-}(-\overline{s})}^{\#}$  (equivalently $-\overline{\lambda}$ is a zero of  $ \Xi_{-}(s) $), and 
\item From \eqref{red}, $\lambda$ is a zero of $\Xi_{+}(-s)$ (equivalently $-\lambda$ is a zero of $\Xi_{+}(s)$).
\end{itemize}
We consider two separate cases: $\lambda$ non-real or real.
\begin{itemize}
\item  If $\lambda$ is non-real then from 
 the symmetries of the poles and zeros in \eqref{form1} and \eqref{form2}, $\Xi_{-}(s)$ will contain a term like
\eqref{topple1} and
$\Xi_{+}(s) $ will contain a term like \eqref{tipple1}.
By the argument above,  the system is non-globally minimal  as there will be a mode of the system that is non-visible in the power spectrum. Therefore all non-real poles of $\Xi_-(s)$ may be identified. 
 %there will be a mirror pole-zero pair of \overline{\Xi^{(2)}_{-}(s) }$ at $(\overline{\lambda}, \overline{\lambda})$ and from \eqref{form2} there will be a mirror pole-zero pair of \overline{\Xi^{(2)}_{+}(s) }$ at $(\overline{\lambda}, \overline{\lambda})$
 A similar argument ensures that all poles of $\Xi_+(s)$ are visible in the power spectrum. 
 \item
 If $\lambda$ is real, then $\Xi_-(s)$ must have a zero at $-\lambda$ for it not to be visible in \eqref{blue}. The symmetries of the zeros  in \eqref{form1} would suggest that there is another zero at $\lambda$. However this  would cancel our original pole. Therefore, there must be a second pole at $\lambda$ in $\eqref{form1}$ (and thus we have a fictitious pole-zero pair in $\Xi_-(s)$). In summary $\Xi_-(s)$ has a term like \eqref{topple1}.
 Also,  $\Xi_+(s)$ must also have an arrangement of poles and zeros as in \eqref{tipple1},    otherwise 
 $|\Xi_-(-i\omega)|^2-|\Xi_+(-i\omega)|^2=1$ could not hold. Hence we have a contradiction to global minimality.    
 \end{itemize}
 Therefore we conclude that all poles of $\Xi_\pm(s)$ can be identified from the power spectrum, and we focus next on the zeros.
 % then in order to use the same argument as in the non-real case  we need to be careful and check that the denominators of \eqref{form1} and \eqref{form2} also have a second root at $\lambda$ (the first cancels with the same term in the ). If this is the case then $\Xi_-(s)$ and $\Xi_+(s)$ will have terms of the form \eqref{topple} and \eqref{tipple} and the result will follow. Suppose to the contrary that there is no second zero at $\lambda$, then by using  the pattern of zeros in \eqref{topple} and \eqref{tipple}  leads to a contradiction in $|\Xi_-(-i\omega)|^2-|\Xi_+(-i\omega)|^2=1$.
 Unlike the case of poles, it is not clear whether a given zero in any of these plots belongs to the factor on the left or the factor on the right in each of these equation (i.e., to $\Xi_-(s)$ or $ {\Xi_{-}(-\overline{s})}^{\#}$ in \eqref{blue}, etc).

Since the poles of $\Xi_-(s)$ and $\Xi_+(s)$ may be different due to cancellations in \eqref{form1} and \eqref{form2}, it is convenient here to add in ``fictitious''  zeros into the plots \eqref{blue}, \eqref{red} and \eqref{yellow} so that 
 $\Xi_-(s)$ and $\Xi_+(s)$ have the same poles.  Note that these fictitious poles and zeros would have been present in \eqref{form1} and \eqref{form2} before simplification. From this point onwards, the zeros in \eqref{blue}, \eqref{red} and \eqref{yellow} will refer to this augmented list which includes the additional zeros.

% Next we identify the real zeros of $\Xi_-(s)$ and $\Xi_+(s)$. 

 \underline{Real zeros:}

In general the real zeros of $\Xi_-(s)$ and $\Xi_+(s)$ come in pairs $\pm \lambda$ (see equations \eqref{form1},  \eqref{form2}),  unless a pole  and zero (or more than one)  cancel on the negative real line. Our task here is to distinguish these two cases from the plots \eqref{blue} \eqref{red} and \eqref{yellow}.  $\Xi_-(s)$ has either
\begin{itemize}
\item  i)  zeros at $\pm \lambda$, or
\item  ii) a zero at $\lambda>0$ but not at $-\lambda$. 
\end{itemize}
In case i) \eqref{blue} will have a double zero at each $\pm\lambda$, whereas in case ii) \eqref{blue} will have a single zero at $\pm\lambda$.  We need to be careful here in discriminating cases i) and ii)  on the basis of the zeroes of \eqref{blue}. For example, a double zero at $\lambda$ in \eqref{blue} could be a result of one case i) or two case ii) in $\Xi_-(s)$. More generally, we could have an $n$th order zero at $\lambda$ and as a result even more degeneracy is possible. A similar problem arises for the zeros of $\Xi_+(s)$ in \eqref{yellow}.

%[[[NEED SKETCH HERE]]] 

Our first observation here is that it is not possible for both $\Xi_-(s)$ and $\Xi_+(s)$ to have zeros at $\pm\lambda$ (taking $\lambda>0$ without loss of generality). If this were possible then by using the symplectic condition $|\Xi_-(-i\omega)|^2-|\Xi_+(-i\omega)|^2=1$ and the fact that we are assuming that  $\Xi_-(s)$ and $\Xi_+(s)$ have the same poles tells us that $\Xi_-(s)$ and $\Xi_+(s)$ must both  have had double poles at $-\lambda$. The upshot is that $\Xi_-(s)$ and $\Xi_+(s)$ will have terms of the form \eqref{topple1} and \eqref{tipple1}, which is a contradiction. 

Now, suppose \eqref{blue} has $n$ zeros at $\lambda>0$ and \eqref{yellow} has $m$ zeros at $\lambda>0$. Then we know that $\Xi_-(s)$ must have $\frac{n-p}{2}$ zeros at $-\lambda$ and $\frac{n+p}{2}$ zeros at $\lambda$. Also, $\Xi_+(s)$ must have $\frac{m-q}{2}$ zeros at $-\lambda$ and $\frac{m+q}{2}$ zeros at $\lambda$. The goal here is to find $p$ and $q$ because if these are known then it is clear that there must be    $\frac{n-p}{2}$ ($\frac{m-q}{2}$) type i) zeros and $p$ (q) type ii) zeros in $\Xi_-(s)$ ($\Xi_+(s)$).

By the observation above it is clear that either $p=n$ or $q=m$. Also, in \eqref{red} there will be $\frac{n+m+p-q}{2}$ zeros at $\lambda$ and $\frac{n+m+q-p}{2}$ zeros at $-\lambda$. Hence $q-p$ is known at this stage. Finally, it is fairly easy to convince ourselves that if $p=n$ but one concludes that $q=m$ (or vice-versa) and using the value of $q-p$ leads to a contradiction. Hence $p$ and $q$ can be determined uniquely. For example, if $n=2$, $m=5$, $q=2$ and $p=3$ so that $q=n$ and $q-p=-1$. Then assuming wrongly that $p=5$ and using $q-p=-1$ it follows that $q=4$ and so $n$ must be 6, which is incorrect.  

 Having successfully identified all real zeros,  we now show how to identify the zeros of $\Xi_-(s)$ and $\Xi_+(s)$ away from the real axis.

\underline{Complex (non real) zeros:}

Comparing the zeros of \eqref{blue} with those of \eqref{red} we find two cases in which the zeros can be assigned directly
\begin{itemize}
\item Case 1: Let $z$ be a zero of \eqref{blue} that is not a zero of \eqref{red}. Then $z$ must be a zero of  ${\Xi_{-}(-\overline{s})}^{\#}$. Hence $-\overline{z}$ is a zero of $\Xi_-(s)$.
\item Case 2: Let $w$ be a zero of \eqref{red} that is not a zero of \eqref{blue}. Then $w$ must be a zero of  ${\Xi_{+}(-s)}^{\#}$. Hence $-w$ is a zero of $\Xi_+(s)$.
\end{itemize}
%Continue this procedure, removing along the way the zeros from all plots  for which $z$ or $w$ responsible for at each step. 
The question now is whether this procedure enables one to identify all zeros? Suppose that there is a zero $v$ that is common to both of these plots. Then $-\overline{v}$ must also be a zero of \eqref{blue}. Now, if $-\overline{v}$ is not a zero of \eqref{red} then $v$ is identifiable  as belonging to $\Xi_-(s)$.
%\begin{itemize}
%\item  $-\overline{v}$ must also be a zero of \eqref{blue}. 
%\item if $-\overline{v}$ is not a zero of \eqref{red} then $v$ is identifiable  as belonging to $\Xi_-(s)$.

%Also $-\overline{v}$ must be a zero of \eqref{red}, otherwise $v$ is identifiable  as belonging to $\Xi_-(s)$. This is because if $-\overline{v}$ were not a zero of \eqref{red} then we fall back into Case 1 above and ultimately a zero at $v$ in $\Xi_-(s)$ must be responsible for all of these zeros. 
%\end{itemize}
Therefore we can restrict our attention to the case that the zero pair $\{v,-\overline{v}\}$ is common to both plots. Note that in this instance the list of zeros of \eqref{yellow} will also contain $\{v,-\overline{v}\}$.  
Assume without loss of generality that $v$ is in the right half complex plane. 
Note that there cannot be a second zero pair $\{u,-\overline{u}\}$ such that $u=\overline{v}$. If this were the case then either $\{v, -v\}$ will be zeros of $\Xi_-(s)$ and  $\{-\overline{v}, \overline{v}\}$ will  be zeros of   $\Xi_+(s)$, or  $\{u, -u\}$ will be zeros of $\Xi_-(s)$ and  $\{-\overline{u}, \overline{u}\}$ will  be zeros of   $\Xi_+(s)$. In either case by using the condition $|\Xi_-(-i\omega)|^2-|\Xi_+(-i\omega)|^2=1$ for all $\omega$ and the fact that   $\Xi_-(s)$ and $\Xi_+(s)$ have the same poles by assumption, it follows that $\Xi_-(s)$ and $\Xi_+(s)$ will have terms of the form \eqref{topple1} and \eqref{tipple1}, which contradicts global minimality. Finally, under the assumptions that the zero pair $\{v,-\overline{v}\}$ is common to both \eqref{red} and \eqref{blue} with no second pair at $\{u,-\overline{u}\}$ such that $u=\overline{v}$, then we can conclude that $v$ must be a zero of $\Xi_-(s)$. For if this were not the case and so $-\overline{v}$ were a zero of $\Xi_-(s)$ then there must be another zero of $\Xi_-(s)$ at $\overline{v}$ (since pole-zero cancellation cannot occur in the right-half plane). Also from \eqref{red} this would require that $\Xi_+(s)$ has a zero at $-v$ (hence also $v$). Therefore we have a contradiction to the fact that there is no second pair at $\{u,-\overline{u}\}$ such that $u=\overline{v}$. 

%[[[NEED PLOT HERE]]]

Therefore we have  successfully identified all  zeros of the transfer function away from the real axis, which completes the proof. 
\end{proof}

In light of this Theorem two globally minimal SISO systems are related by a symplectic transformation as described in Theorem \ref{symplecticequivalence}.
%The result also gives a constructive    method to check global minimality. 
Further it enables one to construct the transfer function of the systems globally minimal part. From this, one can then construct a system realisation of this globally minimal restriction, using the results from Sec. \ref{pond} or \ref{indirect}. We call this realisation method indirect because one first finds a transfer function fitting the power spectrum before constructing the system realisation.

\begin{Corollary} 
Let   $(C, \Omega)$ be a SISO QLS with pure  input ${V}(N,M)$. Then one can construct a globally minimal realisation, $(C', \Omega')$ \textbf{indirectly} from the power  spectrum generated by the QLS $(C, \Omega)$. The realisation  $(C', \Omega')$ will be  unique up to the symplectic equivalence in Theorem \ref{symplecticequivalence}.
\end{Corollary}

\section{Power Spectrum Identification of General QLSs}\label{dogs}

We now give an alternative argument for the identifiability result in the previous subsection. The argument holds for all QLSs, rather than just the generic SISO class in the last subsection. Our method uses the work in Sec. \ref{altos} to reduce the power spectrum identifiability problem to an equivalent (yet simpler) transfer function identifiability problem.

\begin{thm}\label{main}
Let $\left(C_1, \Omega_1\right)$ and $\left(C_2, \Omega_2\right)$ be two globally minimal  and stable QLSs for input $V_{\mathrm{vac}}$,
% which are assumed to be generic in the sense that their system matrix $A_i:=-iJ\Omega_i-\frac{1}{2}C_i^{\flat}C_i$ has distinct eigenvalues. 
%
then
$$
\Psi_{1}(s)=\Psi_{2}(s) \,\, for~all~s \quad \Leftrightarrow \quad \Xi_1(s)=\Xi_2(s)\,\, for~all~s 
$$
\end{thm}

\begin{proof}\label{sofa}
To prove this result, we use the results of Sec. \ref{altos} to write our globally minimal power spectrum identification problem as a minimal transfer function identification problem. That is, 
by Theorem \ref{games} the system \eqref{cask} is minimal. Therefore, from the classical literature  TFE systems are related via 
\begin{align}
\label{c1}&\tilde{A}'=T\tilde{A}T^{-1}, \quad\tilde{B}'=T\tilde{B}, \quad\tilde{C}'=\tilde{C}T^{-1}, \quad \tilde{D}'=\tilde{D}.
\end{align}
Moreover, its observability and controllability matrices, $\mathcal{O}$ and $\mathcal{C}$, will have full rank. 
Additionally, by Lemma \ref{hud} such a similarity transformation must be lower block triangular. 

Now writing $T$ as 
$$\left(\begin{smallmatrix} T_1 & 0\\ T_3 & T_4\end{smallmatrix}\right),$$ to complete the proof it remains to show that (a) $T_3=0$, (b) $T_1=T_4$, (b) $T_1^{\flat}T_1=1$ and (d) $T_1$ is doubled up. 
This is sufficient because it  tells us that  the equivalence classes of  the power spectrum are related via symplectic similarity transformations (and they are the same  gauge transformations  as those obtained  from the transfer function (see Ch. \ref{T.F})). 
The outline of how we show (a)-(d) is given in the following  three steps. The complete proof can be found in Appendix  \ref{p1}.

\begin{enumerate}
\item[(1)] \label{tog1}Firstly, using the pattern in the $\tilde{A}$, $\tilde{B}$ and $\tilde{C}$ matrices defined above, we show that the following holds:
$$\mathcal{O}=\mathcal{O}\left(\begin{smallmatrix} T_4^{\flat} &0\\ -T_3^{\flat}& T_1^{\flat}\end{smallmatrix}\right)T.$$ 
And so  because $\mathcal{O}$ has full rank, we have
$$\left(\begin{smallmatrix} T_4^{\flat} &0\\ -T_3^{\flat} &T_1^{\flat}\end{smallmatrix}\right)T=1.$$
\item[(2)] \label{tog2}We will then 
 show that:
$$\mathcal{O}\left(\begin{smallmatrix} T_4^{\flat}-T_1^{\flat}\\-T_3^{\flat}\end{smallmatrix}\right)=0.$$
This implies that $T_3=0$ and $T_1=T_4$.
\item[(3)] \label{tog3} Combing Steps (1) and (2) it is clear that $T$ must be of the form  $$T=\left(\begin{smallmatrix}T_1&0\\0&T_1\end{smallmatrix}\right)$$ with $T_1^{\flat}T_1=1.$
Finally we show that $T_1$ is doubled-up. 
\end{enumerate}
\end{proof}

\begin{remark}
For general input $V=S_0V_{\mathrm{vac}}S_0^{\dag}$, clearly PQLSs of the form $\mathcal{G}=\left( S_0\Delta(\left(C_-,0\right), \Delta\left(\Omega_-, 0\right)\right)$ will have trivial power spectrum. 
Theorem \ref{mainresult} says that these are the only such systems (up to symplectic equivalence in Theorem \ref{symplecticequivalence}). 
\end{remark}

\section{Identification Method}\label{256}

Suppose that we have constructed the power spectrum  from the input-output data, for instance by treating it as a transfer function and using one of the techniques of \cite{Ljung1}. Here we outline a method to construct a globally minimal system realisation from the power spectrum.  This method will provide us with a  system realisation \textit{directly}, rather than indirectly  via the transfer function (see Sec. \ref{jute}).
The realisation is obtained  by first finding a non-physical realisation and then constructing a physical one from this by applying a criterion developed in \cite{Zhou1}. The identification method is  similar to the one   used in Sec. \ref{indirect} for the transfer function realisation problem.

We have seen many times that the power spectrum may be treated  as if it were a transfer function. Therefore, let $\left(\tilde{A}_0, \tilde{B}_0, \tilde{C}_0, V_{\mathrm{vac}}\right)$ constitute a minimal realisation of $\Psi(s)$, i.e., 
$$\Psi(s)J=V_{\mathrm{vac}}+ \tilde{C}_0\left(s-\tilde{A}_0\right)^{-1}\tilde{B}_0.$$
Further, let us  assume that $\tilde{A}_0, \tilde{B}_0, \tilde{C}_0$ are of the form 
$$\tilde{A}_0=\left(\begin{smallmatrix} -{A}^{\flat}_0&0\\0&{A}_0\end{smallmatrix}\right)\quad
\tilde{B}_0=\left(\begin{smallmatrix}B_1\\B_2\end{smallmatrix}\right) \quad
\tilde{C}_0=\left(\begin{smallmatrix}C_1\\C_2\end{smallmatrix}\right),$$
with, $A_0$, $B_1$ and $C_2$ doubled up and $A_0$ is stable. For example, in Appendix \ref{nog} such a realisation is found for an $n$-mode globally minimal system, with matrices $(A,C)$, possessing $2n$ distinct poles each with non-zero imaginary part.

Now, by minimality, any other realisation of the transfer function can be generated by the similarity transformation 
\begin{equation}\label{create}\tilde{A}=T\tilde{A}_0T^{-1}\quad
\tilde{B}=T\tilde{B}_0\quad
\tilde{C}=\tilde{C}_0T^{-1}.\end{equation}
The problem here is that in general these matrices may not describe a genuine quantum system in the sense that from a given $\left(\tilde{A}, \tilde{B}, \tilde{C}\right)$ one cannot reconstruct the pair $(\Omega, C)$ describing the power spectrum. Our goal is to find a special transformation $T$ mapping $(\tilde{A}_0 , \tilde{B}_0 , \tilde{C}_0)$ to a triple $(\tilde{A}, \tilde{B}, \tilde{C})$ that is physical.

Firstly, as $\tilde{A}_0$ and the physical $\tilde{A}$ we seek are both proper LBT, then by Lemma 
\ref{hud} we may restrict  $T$ to be of the  form
$$T=\left(\begin{smallmatrix}T_1&0\\T_2&T_3\end{smallmatrix}\right), \quad T^{-1}=\left(\begin{smallmatrix}T^{-1}_1&0\\-T_3^{-1}T_2T_1^{-1}&T_3^{-1}\end{smallmatrix}\right).$$
 Using this together with \eqref{create} and    $(\tilde{A}, \tilde{B}, \tilde{C})$ in \eqref{cask} gives:
\begin{align}
\label{one}A^{\flat}=T_1A_0^{\flat}T_1^{-1}\quad\mathrm{and}\quad -C^{\flat}=T_1B_1\\
\label{three}A=T_3A_0T_3^{-1}\quad\mathrm{and}\quad C=C_2T_3^{-1}\\
\label{two}C^{\flat}V_{\mathrm{vac}}C=-T_2{A}^{\flat}_0T_1^{-1}-T_3A_0T_3^{-1}T_2T_1^{-1}\\
\label{six}-V_{\mathrm{vac}}C=C_1T_1^{-1}-C_2T_3^{-1}T_2T_1^{-1}.\\
\label{five}-C^{\flat}V_{\mathrm{vac}}=T_2B_1+T_3B_2
\end{align}
For $(A, C)$ to correspond to a quantum system it must satisfy the physical realisability conditions:
$A+A^{\flat}+C^{\flat}C=0$ (Sec. \ref{MTDR}).
Applying this condition to     \eqref{one} and \eqref{three}  gives:
\begin{align}
\label{nine}
A_0^{\flat}\left(T_1^{\flat}T_1\right)^{-1}+\left(T_1^{\flat}T_1\right)^{-1}A_0+B_1B_1^{\flat}=0\\
\label{eight}\left(T_3^{\flat}T_3\right)A_0+A_0^{\flat}\left(T_3^{\flat}T_3\right)+C_2^{\flat}C_2=0.
\end{align}
Next as quantum system is  stable,  $A_0$ must be  Hurwitz (because it is similar to $A$). Therefore  \eqref{nine} and \eqref{eight} have unique solutions, given by
\begin{align}
\label{eleven}
&\left(T_1^{\flat}T_1\right)^{-1}=\int^{\infty}_0 J\left(B_1^{\dag}Je^{A_0 t}\right)^{\dag}J\left(B_1^{\dag}Je^{A_0t}\right)dt\\
\label{twelve}
&\left(T_3^{\flat}T_3\right)=\int^{\infty}_0 J\left(C_2e^{A_0 t}\right)^{\dag}J\left(C_2e^{A_0t}\right)dt.
\end{align}
Moreover, these solutions will necessarily be of doubled-up form due to the fact $A_0, B_1$ and $C_2$ were.
Therefore, using Lemma \ref{peterf} from Sec. \ref{indirect} we can find doubled-up $T_1$ and $T_3$ from these uniquely (up to the non-identifiable symplectic equivalence class in Theorem \ref{main}).

The upshot of these results is that we may ultimately write down a realisation of the system $(A, C)$ using \eqref{one}  or alternatively from \eqref{three}. By Theorem \ref{main} both solutions are guaranteed to coincide (bar any unidentifiable symplectic matrix) and give a unique (up to such a symplectic transformation) realisation of the power spectrum, hence we are done.

For completeness we may write down the unique solution $T_2$ given the solutions $T_1$ and $T_3$ so to obtain the full realisation \eqref{cask} of the power spectrum. To this end, 
suppose that the solutions $T_1$ and $T_2$ from \eqref{eleven} and \eqref{twelve} lead to (physical) realisations $(A, C)$ and $(\hat{A}, \hat{C})$ that differ by an (unidentifiable) symplectic. That is, $A=S\hat{A}S^{\flat}$ and $C=\hat{C}S^{\flat}$. Then from \eqref{two} we have
$$S\hat{C}^{\flat}V_{\mathrm{vac}}\hat{C}+\left(T_2T_1^{-1}S\right)\hat{A}^{\flat}+\hat{A}\left(T_2T_1^{-1}S\right)=0,$$
which has been obtained by substituting \eqref{one} and \eqref{three} into  \eqref{two}.
 This solution $\left(T_2T_1^{-1}S\right)$ can be found uniquely, hence $T_2$ can be found uniquely from this. Note that $T_2$ will not be of doubled-up type, which is to be expected.

 %[[[Explain why works even though TbT is unique not T; i.e diff between TbT and T still lies in equivalence class.,]]]

 \section{Realisation Example}\label{god1}
We now give an example of the identification method above. Suppose that we have a one-mode QLS characterised by the matrices:
$$\left(\Omega,C\right)= \left( \Delta(2,i), \Delta(7,-1)       \right)            $$
We will now construct a realisation of this system directly from the power spectrum. That is, we will pretend that we didn't have this realisation beforehand and are given only the following power spectrum:
$$\Psi(s)J=\frac{1}{16s^4+1464s^2+40401}\left(\begin{smallmatrix} 16s^4+1464s^2+53084&s^2(-448+48i)-22176-13484i\\s^2(448+48i)+22176-13484i&-12688\end{smallmatrix}\right).$$

\underline{Finding the classical realisation}: The power spectrum may be expanded as in \eqref{gut1} where $n=1$ $\lambda_1=-24+1.732i$ and
$$I_1=\left(\begin{smallmatrix} -0.0278-0.3849i &1.0496-0.2582i\\0.1051+0.0516i & 0.0278+0.3849i \end{smallmatrix}\right) \quad K_1=\left(\begin{smallmatrix} -0.0278+0.3849i&-0.1051+0.0516i\\-1.0496-0.7182i&0.0278-0.3849i\end{smallmatrix}\right).$$
$$T_1=\left(\begin{smallmatrix} 0.0278+0.3849i&+0.1051-0.0516i\\1.0496+0.7182i&-0.0278+0.3849i\end{smallmatrix}\right)\quad W_1=\left(\begin{smallmatrix} 0.0278+0.3849i &-1.0496+7182i\\-0.1051-0.0516i & -0.0278-0.3849i \end{smallmatrix}\right).$$
Therefore, we can write $\tilde{A}_0$, $\tilde{B}_0$ and $\tilde{C}_0$ as 
$$\tilde{A}_0=\mathrm{Diag}( 24+1.732i,24-1.732i ,-24+1.732i,-24-1.732i)
$$ 
$$\tilde{B}_0=\left(\begin{smallmatrix} 1&1.6603+2.8463i\\1.6603-2.8463i&1\\ 0.0278+-0.3849i&0.1051-0.0516i \\-0.1051-0.0516i &-0.0278+-0.3849i    \end{smallmatrix}\right)$$
$$\tilde{C}_0=\left(\begin{smallmatrix} -0.0278-0.3849i &-0.1051i+0.0516i&1&  -1.6619-2.8463i\\  0.1051+0.0516i & 0.0278-0.3849i&-1.6619+2.8463i&1 \end{smallmatrix}\right).$$
Note that $B_1$ and $C_2$ as defined in the previous subsection are doubled-up and further observe that $B_1=-C_2^{\flat}$.

\underline{Finding the quantum realisation}
From \eqref{eleven} and \eqref{twelve} we have $$\left(T_1^{\flat}T_1\right)^{-1}=T_3^{\flat}T_3=\left(\begin{smallmatrix}-0.2054&0\\0&-0.2054\end{smallmatrix}\right).$$
Therefore, we can let $$T_3=-\left(T_1^{\flat}\right)^{-1}= \left(\begin{smallmatrix}0&-0.4532i\\0.4532i&0\end{smallmatrix}\right).$$ In particular, the choice $T_3=-\left(T_1^{\flat}\right)^{-1}$ and the condition $B_1=-C_2^{\flat}$ ensures that the physical realisations coming from \eqref{one} and \eqref{three} are the same (rather than differing by a symplectic). Therefore from \eqref{two} we have 
$$T_2=\left(\begin{smallmatrix} 0.6604+0.3243i&-2.4312i\\-0.2238i&-0.6604+0.3243i\end{smallmatrix}\right)$$.

In summary, the realisation of the power spectrum is given by
$$A=\left(\begin{smallmatrix} -24-1.732i &0\\0 &-24+1.732i\end{smallmatrix}\right) \quad\mathrm{and}\quad C=\left(\begin{smallmatrix} 6.2803-3.6669i&-2.2065i\\2.2065i&6.2803+3.6669i\end{smallmatrix}\right).$$

\section{Scattering and Squeezing}\label{sq+sc}

 We now show how to extend our power spectrum identifiability results to  allow for scattering and squeezing in the field.  As in Sec. \ref{altos} we can represent the power spectrum  $\Psi(s)J$  of the system $(S,C,\Omega)$ as a (classical) cascaded system: 
 \begin{equation}\label{cask1}
\left(\tilde{A}, \tilde{B}, \tilde{C}, \tilde{D}\right):=
\left(\left(\begin{smallmatrix} -A^{\flat}&0\\C^{\flat}SV_{\mathrm{vac}}S^{\flat}C&A\end{smallmatrix}\right), \left(\begin{smallmatrix} -C^{\flat}\\-C^{\flat}SV_{\mathrm{vac}}S^{\flat}\end{smallmatrix}\right), \left(\begin{smallmatrix} -SV_{\mathrm{vac}}S^{\flat}C& SC\end{smallmatrix}\right), SV_{\mathrm{vac}} S^{\flat}      \right)
\end{equation}
(see \eqref{cask}). Note that if $S$ is a purely scattering transformation, i.e, it is a unitary symplectic matrix, then system \eqref{cask1} reduces to system \eqref{cask}.

Now, Theorem \ref{games} also holds in the case of non-trivial $S$, i.e. system \eqref{cask1} is minimal iff $(S,C,\Omega)$  is globally minimal. The proof is very similar to the proof of Theorem \ref{games}.    
 
We can also prove a modified version of  Theorem \ref{main}. 
 
\begin{thm}\label{main7}
Let $\left(S_1, C_1, \Omega_1\right)$ and $\left(S_2, C_2, \Omega_2\right)$ be two globally minimal  and stable QLSs for input $V_{\mathrm{vac}}$,
% which are assumed to be generic in the sense that their system matrix $A_i:=-iJ\Omega_i-\frac{1}{2}C_i^{\flat}C_i$ has distinct eigenvalues. 
%
then
$$\Psi_1(s)=\Psi_2(s)~for ~ all~s 
\Leftrightarrow
S_1V_{\mathrm{vac}}S_1^{\dag}=S_2V_{\mathrm{vac}}S_2^{\dag},~
C_1=C_1T^{\flat},~
J\Omega_1=TJ\Omega_2T^{\flat}
$$
for some symplectic matrix $T$.
\end{thm}
\begin{proof}
Firstly, the condition $S_1V_{\mathrm{vac}}S_1^{\dag}=S_2V_{\mathrm{vac}}S_2^{\dag}$ follows by choosing $s=-i\omega$ and taking the limit $\omega\mapsto\infty$. By following the method with identical steps to those in    Theorem \ref{main} the remainder of the proof may be obtained. The proofs of each step are almost identical to those in Appendix \ref{p1}, except for keeping track of the matrix $S$, which offers little additional complication.  
\end{proof}
Therefore the consequence  of allowing for a non-trivial scattering is  an extra condition on the classes of equivalent systems, namely $S_1V_{\mathrm{vac}}S_1^{\dag}=S_2V_{\mathrm{vac}}S_2^{\dag}$.
An equivalent formulation of this result can be stated in the following Corollary.

\begin{Corollary}
Suppose that the system $\mathcal{G}=\left(1, C, A\right)$ has power spectrum $\Psi(s)$, then all other systems with the same power spectrum  are given by 
$$\mathcal{G}':=\left(S', C', A'\right):=\left(S, CT^{\flat}, TAT^{\flat}\right),$$
where $T$ is symplectic and $S$ is unitary and symplectic. 
\end{Corollary}
This should be seen as the main result of this chapter. Notice the subtlety here that  the transfer function of the systems $\mathcal{G}$ and $\mathcal{G}'$ differ slightly, that is, $\Xi_{\mathcal{G}'}(s)=\Xi_{\mathcal{G}}(s)S$. Therefore, one is able to recover the transfer function from the power spectrum uniquely up to a passive transformation on the field.

Finally, an identification method similar to the one in Sec. \ref{256} may be developed to include non-trivial $S$, but we do not discuss this any further here.

%
%
%
%
%
%

%\section{Unknown inputs}
%Another interesting problem to consider is the situation where not only the system is unknown but also the input (but is still assumed to satisfy the stationary Gaussian properties from  the beginning of this chapter). We would like to understand what can be identified about the system and indeed the input

%NOT SURE THIS REALLY MAKES SENSE AS NEEDS TO BE RELATIVE TO SOMETHING

%NEED TO discuss implications of $S\neq1$, i.e transfer function non-unique. Kind of discussing two problems here i.e scattering and/or unknown inputs. 

\section{Passive Quantum Linear Systems}\label{pass5}
In this section we restrict our attention to   \emph{passive} QLSs (see Sec. \ref{honk}). As we mentioned earlier,  we  drop the doubled-up notation for PQLSs and work with the triple $(S_-, C_-, \Omega_-)$. For simplicity we only consider the case $S_-=1$ here and we drop the subscript minus for convenience.
If the input state is vacuum then the power spectrum is trivial ($\Phi_{V}=V$) and the only globally minimal systems are the trivial ones (zero internal modes).   For this reason we consider \emph{squeezed} inputs, with pure input $V(N,M)$. 

We  first discuss SISO PQLSs in Sec. \ref{SISO1} and then MIMO PQLSs in Sec. \ref{MIMO1}.

\subsection{SISO}\label{SISO1}
The main identifiability problems that we have considered throughout this chapter turn out to have a much simpler solution for SISO PQLSs.

Now,  the transfer function is given by
$$ 
\Xi(s)=1 -C (s{1}_n-A )^{-1}C^{\dag} = 
\frac{\mathrm{Det}\left(s{1}_n + {A}^{\#} \right)}{\mathrm{Det}\left(s{1}_n-A \right)}
$$
where $A= -i\Omega -\frac{1}{2}C^\dagger C$ and its spectrum is $\sigma(A) :=\{ \lambda_1,\dots ,\lambda_n\}$. 
It is a monic rational function in $s$ (see Nomenclature), with poles $p_i= \lambda_i$ in the left half plane, and zeros  
$z_i= -\overline{p_i}= -\overline{\lambda_i}$ in the right half plane.

\begin{thm}\label{sisogm}
Consider  a general SISO PQLS $\mathcal{G}=(C,\Omega)$  with pure input ${V}(N,M)$, such that $M\neq 0$. 

%
%Collect all of the eigenvalues that are either real or come in complex conjugate pairs in a set,  $\mathcal{P}$. Let $n_p$ be the cardinality of $\mathcal{P}$ (so that $|\mathrm{Spec}(A)\backslash \mathcal{P}|=n-n_p$).

%\begin{enumerate}
1) The following are equivalent:
 \begin{itemize}
 \item[i)] the system is globally minimal
 \item[ii)]  the stationary state of the system is fully mixed
 \item[iii)] $A$ and $A^\dagger$ have different spectra, i.e. $\sigma(A) \cap \sigma (A^\dagger) =\emptyset $
 \item[iv)] $A$ does not have real, or pairs of complex conjugate eigenvalues.
 \end{itemize}

2) Let $\mathcal{P}$ be the set of all eigenvalues of $A$ that are either real or come in complex conjugate pairs. A globally minimal realisation of the system is given by the series product of one mode systems 
$\mathcal{G}_{m,i}= ( c_i=\sqrt{2|\mathrm{Re} \lambda_i}|, \Omega_i=-\mathrm{Im}\lambda_i )$ for indices $i$ such that 
$\lambda_i\notin \mathcal{P}$. 

%
%\[\mathcal{G}_{G.M}=\mathcal{G}_{n-n_p}\triangleleft\mathcal{G}_{n_p-1}\triangleleft...\triangleleft\mathcal{G}_{1},\] with  $\mathcal{G}_i=(c_i=\sqrt{2|\mathrm{Real}(p_i)|}, \Omega_i=\mathrm{Imag}(p_i))$ for $i=1,2...,n_p$. 
%Here $p_i$ are the elements of the set $\sigma(A) \backslash \mathcal{P}$.Let $\mathcal{P}$ be the set of all eigenvalues of $A$ that are either real or come in complex conjugate pairs, and let $n_p = |\mathcal{P}|$. %The dimension of the non-identifiable part is $n_p$.

%\end{enumerate}
\end{thm}

\begin{proof}

1) For passive SISO systems the only non-trivial contribution to the power spectrum is from off-diagonal element 
%Moreover, if $M=0$ then the power spectrum is trivial ($\Phi_{V}=V$) and the system is globally minimal if and only if it has zero internal modes.  For this reason we restrict our attention to non-vacuum inputs, i.e. $M\neq0$. 
\begin{eqnarray*}
\Xi(s)\Xi(-s)&=&\frac{\mathrm{Det}\left(s{1}_n + A^\dagger\right)}{\mathrm{Det}\left(s{1}_n-A\right)}
\times\frac{\mathrm{Det}\left(s{1}_n- A^\dagger\right)}
{\mathrm{Det}\left(s{1}_n+A\right)}\\
&=& \prod_{i=1}^n \frac{s+ \overline{\lambda}_i}{s-\lambda_i} \cdot \frac{s- \overline{\lambda}_i}{s+\lambda_i}.
\end{eqnarray*}

In the above expression, zero-pole cancellations occur if and only if $\sigma(A) \cap \sigma (A^\dagger) \neq \emptyset $, or equivalently if $A$ has a real eigenvalue or a pair of complex conjugate eigenvalues (see Fig. \ref{pz}).

 \begin{figure}
\centering
\centering\includegraphics[scale=0.14]{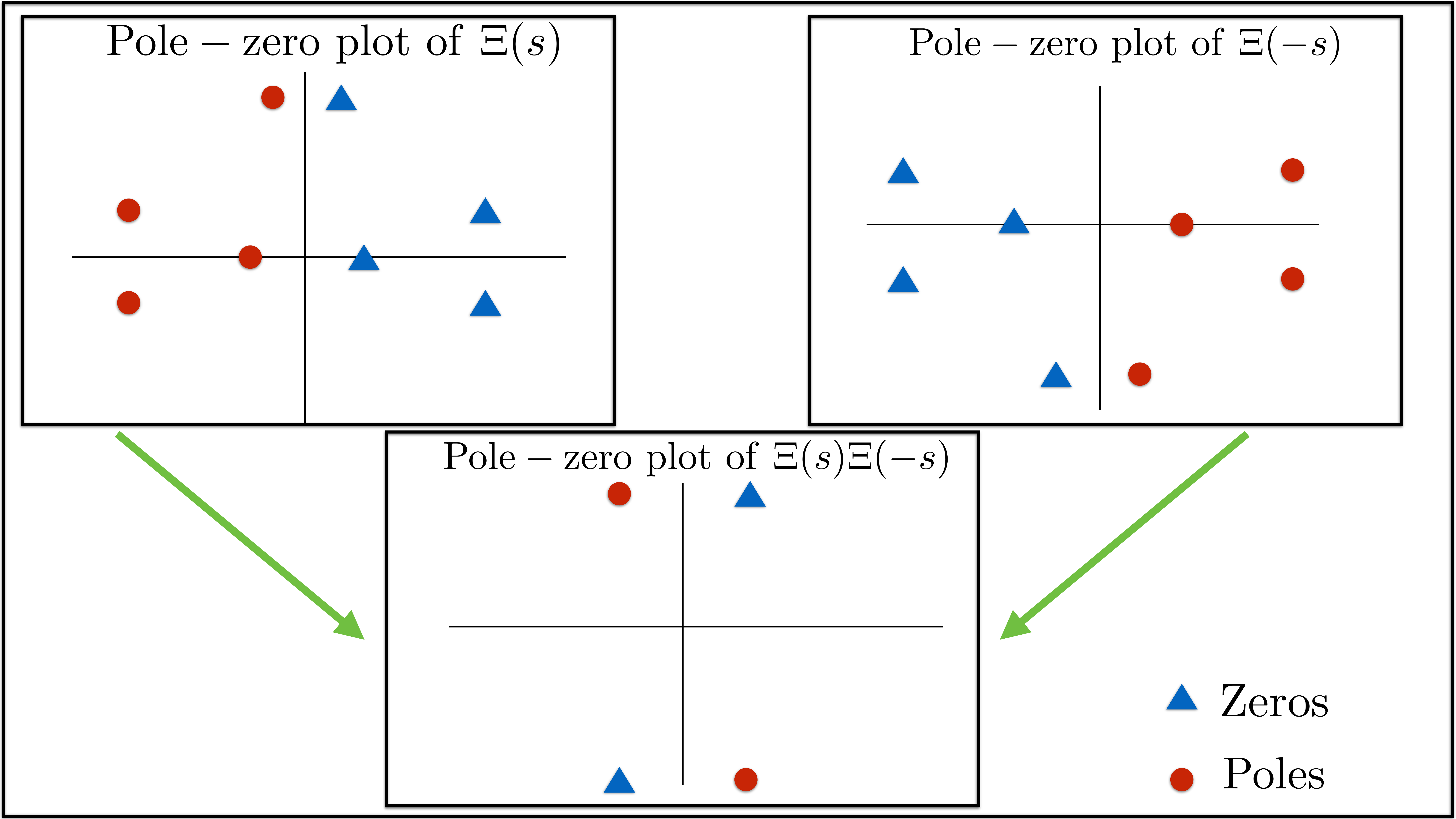}
\caption{There are two types of cancellations in $\Xi(s)\Xi(-s)$. Either (i) when $\Xi(s)$ has a real pole  or (ii) when there is a pole,  $p$, and zero, $z$, of $\Xi(s)$ such that $p=-z$. Both are illustrated here.
\label{pz}}
\end{figure}

If no zero-pole cancellations occur, then $\sigma(A)$ can be identified from $\Xi(s)\Xi(-s)$ and the transfer function can be reconstructed. In this case the system is globally minimal. 

If cancellations do occur then this happens in one of the two types of situations: 

a) real eigenvalue: if $\lambda_i \in \mathbb{R}$ then the corresponding term in the above product cancels

b) complex conjugate pairs: if $\lambda_i = \overline{\lambda}_j$ then the $i$ and $j$ terms in the product cancel against each other. 

In both cases, the remaining power spectrum has the same form, and can be seen as the power spectrum of a series product of one dimensional passive systems, with dimension smaller than $n$, and therefore the system is not minimal.  

This shows the equivalence of i), iii) and iv) while the equivalence of i) and ii) follows from Theorem \ref{equivalence}.

%Therefore, if one of the zeros of  $\mathrm{Det}\left(s{1}_n-A \right)$ cancel with a zero of $\mathrm{Det}\left(s{1}_n- \bar{A}\right)$ (and hence by symmetry  the zeros of $\mathrm{Det}\left(s +A \right)$ will cancel with the zeros of $\mathrm{Det}\left(s{1}_n +\bar{A}\right)$) then the system is not globally minimal. 
%%Or equivalently, one of the eigenvalues of $A=-i\Omega-\frac{1}{2}C^{\dag}C$ should be the same as one of the eigenvalues of $\tilde{A}=+i\Omega-\frac{1}{2}C^{\dag}C$. 
%This statement is clear by the definition of global minimality, as the system will have the same power spectrum as a lower dimensional system. The converse of this statement is trivial, hence we have proved equivalence between the statements (1i), (1iii) and (1iii). The equivalence between (1i) and (1ii) is obtained via Theorem \ref{equivalence}.
%%Moreover, it is clear that the dimension of the globally minimal subspace is given by the number of eigenvalues of $A$ that are not also eigenvalues of $\tilde{A}$.

2) The discussion so far shows that the transfer function factorises as the product $\Xi(s)=\Xi_{\mathrm{m}}(s)\Xi_{\mathrm{p}}(s)$ of a part corresponding to eigenvalues $\lambda_i\in\mathcal{P}$, which has trivial power spectrum due to zero-pole cancellations, and the part corresponding to the complement which does not exhibit any cancelations. A system with transfer function $\Xi(s)$ can be realised as series product $\mathcal{G}_{m}\triangleleft\mathcal{G}_{p}$ of two separate passive systems 
with  transfer functions $\Xi_{\mathrm{m}}(s)$ and $\Xi_{\mathrm{p}}(s)$. 
As argued before, $\mathcal{G}_{p}$ has a pure stationary state which is uncorellated to $\mathcal{G}_{m}$ or the output, while 
$\mathcal{G}_{m}$ has a fully mixed state which is correlated to the output.

Since $\mathcal{G}_{p}$ does not contribute to the power spectrum, a globally minimal realisation is provided by $\mathcal{G}_{m}$
\begin{equation}\label{partm}
\Xi_{m}(s)= \prod_{i \notin \mathcal{P}}  \frac{s+ \overline{\lambda}_i}{s-\lambda_i}
\end{equation}
Each fraction in \eqref{partm} represents a bona-fide PQLS $\mathcal{G}_{m,i}$ with Hamiltonian  and coupling parameters 
$\Omega_i=-\mathrm{Im}\lambda_i $ and $1/2 |c_i|^2=-\mathrm{Re} \lambda_i $.
\end{proof}

With this Theorem it is now possible to construct a globally minimal realisation of the PQLS \textit{directly} from the power spectrum. Moreover, global minimality of PQLSs  may be completely understood in terms of the spectrum of the system matrix $A$, just as was the case for minimality, stability, observability and controllability \cite{Guta2, Gough2}. 
An immediate corollary of this is the following:
\begin{Corollary}\label{pj}
A SISO  PQLS $\mathcal{G}=(C,\Omega)$, with pure input ${V}(N, M)$ has a pure stationary state if and only if either holds
\begin{enumerate}
\item The input is vacuum
\item The eigenvalues of $A$ are real or come in complex conjugate pairs. 
\end{enumerate}
\end{Corollary}

From Theorem \ref{sisogm} there are two types of `elementary' systems that are non-identifiable from the power spectrum for 
arbitrary input ${V}(N,M)$. Written in the doubled up notation, these are either:
\begin{enumerate}
\item[a)]\label{tyr1} one mode systems of the form $\mathcal{G}_1=\left(\Delta(c,0),0\right)$
\item[b)]\label{tyr2} two mode systems of the form \\$\mathcal{G}_2=\left(\Delta(c,0),\Delta(\Omega_-,0)\right)\triangleleft\left(\Delta(c,0),\Delta(-\Omega_-,0)\right)$.  
\end{enumerate}
The system $\mathcal{G}_2$ having trivial power spectrum could be interpreted as destructive interference between the first and the second system. That is, the first system is cancelling or absorbing the second.  We discuss this idea further in Ch. \ref{DUAL},  where we develop the notion of  \textit{quantum absorbers}.

Now it is not immediately obvious that these systems are consistent with the non-identifiable systems in Theorem \ref{mainresult}.  As an example we will show that this is indeed the case in the case for 
$\mathcal{G}_1$ ($\mathcal{G}_2$ is similar). 
\begin{exmp}
Consider system $\mathcal{G}_1$ for input $V(N, M)$, which  has trivial power spectrum $V(N, M)$. Viewed in the vacuum basis of the field the system will be 
\begin{equation}\label{none}
\tilde{\mathcal{G}}_1=\left(S^{\flat}_{\mathrm{in}}\Delta(c,0),0\right)
\end{equation}
 (see Sec. \ref{QHO})) and the power spectrum will be vacuum. 
As $S^{\flat}\Delta\left(c_-,0\right)=\Delta\left(c_-,0\right)S^{\flat}$, it follows that $\tilde{\mathcal{G}}_1$ must be TFE to the system $\left(\Delta(c,0),0\right)$ in the vacuum basis. Therefore, because this system is passive, we have consistency with  Theorem \eqref{mainresult}.

In fact we can even see that \eqref{none} is passive by directly computing its transfer function. One can check that
\[\Xi_-(s)=\frac{s-|c|^2/2}{s+|c|^2/2} \,\,\,\mathrm{and}\,\,\,\Xi_+(s)=0.\]
%Firstly, perform the basis change from input ${V}(N,M)$ to vacuum, as per Theorem \ref{mainresult}. It follows that this new system, which we call $\tilde{\mathcal{G}_1}$ is given by $\left(C=\Delta(\sqrt{N+1}c,\sqrt{N}),\Omega=0\right)$ (we have assumed that $M,N$ are real here for simplicity, so that $M=\sqrt{N(N+1)}$ by purity). One can check that the transfer function of this system is given by 
%\[\Xi_-(s)=\frac{s-|c|^2/2}{s+|c|^2/2} \,\,\,\mathrm{and}\,\,\,\Xi_+(s)=0.\]
%That is, the system $\tilde{\mathcal{G}_1}$ is (transfer function-)equivalent to the system $\left(C=\Delta(c,0),\Omega=0\right)$. For vacuum input, this system clearly has trivial power spectrum, hence the theorems are consistent. A way to interpret this is that one can perform a change of basis on the field and system so that the squeezed stationary state and input become vacuum stationary state and vacuum. 
\end{exmp}

Finally, it seems that the assumption of global minimality seems to be not very restrictive; we illustrate this in the form of an example.

\begin{exmp}
Consider the following SISO PQLS with two internal modes: \[\mathcal{G}=\left((0, 2\sqrt{2}), \frac{1}{2}\left(\begin{smallmatrix} 4+x &4-x\\4-x&4+x\end{smallmatrix}\right)\right),\]
where $x\in\mathbb{R}$. We examine for which values of $x$ the system is globally minimal for squeezed inputs. One can first check that the system is minimal if and only if $x\neq 4$. In Fig. \ref {gm2} we plot the imaginary parts of the eigenvalues of $A$ and $A^\dagger$.  
\begin{figure}
\centering
\includegraphics[scale=0.25]{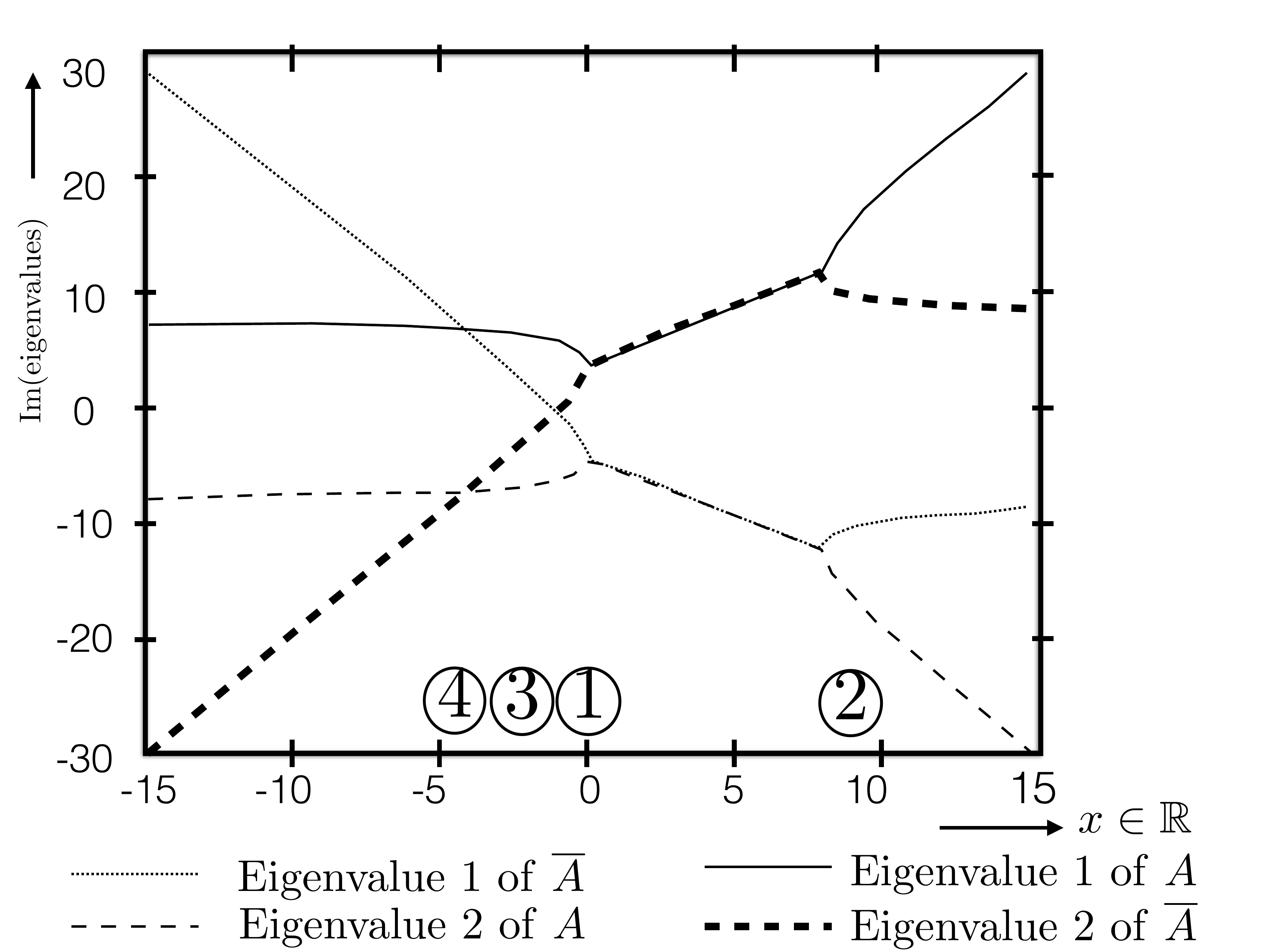}
\caption{Eigenvalues of $A$ and $A^\dagger$ as function of $x$. \label{gm2}}
\end{figure}
By Theorem \ref{sisogm}, the system is non-globally minimal if any of the lines representing the eigenvalues of $A$ intersect those of 
$A^\dagger$. 
There are 4 points of interest that have been highlighted in the figure 
\begin{itemize}
\item[\textcircled{1}] $x=0$: crossing of eigenvalues of $A$ but not with eigenvalues of $A^\dagger$; system is globally minimal.
\item[\textcircled{2}] $x=8$: crossing of eigenvalues $A$ but not with eigenvalues of $A^\dagger$; system is globally minimal.
\item[\textcircled{3}] $x=-1$: An eigenvalue of $A$ coincides with one of $A^\dagger$. Therefore the dimension of the pure component is 1. This occurs when one eigenvalue is real.
\item[\textcircled{4}] $x=-4$: Both eigenvalues of  $A$ coincide with those of $A^\dagger$, and form a complex conjugate pair. 
Therefore the dimension of the pure space is 2.
\end{itemize}
In summary, there were only two values of $x$ for which the system is non-globally minimal.
\end{exmp}
The diagrammatical method used in this example is rather neat, as it not only highlights clearly whether or not a system is globally minimal but also the size of the globally minimal subsystem that is transfer function-identifiable.

\subsection{MIMO}\label{MIMO1}

We now extend our understanding of global minimality to \textit{MIMO} PQLSs. 
Recall from Sec. \ref{dogs} that under global minimality PSE systems are related via a symplectic transformation on the system. Therefore, given  a pure input $V(N,M)$ and a system $(C, \Omega)$ it remains to understand when it is globally minimal.

The MIMO case is more involved than the SISO case because the input correlations between channels become important, which wasn't the case for SISO. Consequently we  don't get such a simple condition on the system matrix, $A$, determining global minimality (as in the SISO case). 
%Recall, in the SISO case that so long as $M\neq0$ we were able to remove any discussions of the input and the only non-trivial element of the power spectrum was the single off-diagonal (squeezing) element. In MIMO we have the additional complication that the thermal part of the power spectrum is non-trivial. 
The information available about the system is given by the quantities
\begin{align}
\label{th1} &\Xi(-i\omega)(N^T+1)\Xi(-i\omega)^{\dag}\\
&\label{sq1} \Xi(-i\omega)M\Xi(i\omega)^T.
\end{align}
Note that \eqref{th1} wasn't present in the SISO case.

We will investigate global minimality  first for $n=1$ modes, then $n=2$ and finally for arbitrary $n$. The reason for doing it this way is because, as we shall see, the non-globally minimal part turns out to be decomposable into one and two mode blocks connected in series.   For $n=1$ mode we have the following result:

\begin{Lemma}\label{onemode}
A MIMO PQLS, $\left(C, \Omega\right)$, with one internal mode is non-globally minimal for field  input, ${V}(N,M)$, if and only if either of the following conditions hold:

(i) $\Omega=0$ and $\left[CC^{\dag}, N^T\right]=0$ and $CC^{\dag}M-M\left(CC^{\dag}\right)^T=0$, or

(ii) $CC^{\dag}M=0$

\end{Lemma}
\begin{proof}
For a one-internal mode system to be non-globally minimal, its power spectrum must be trivial. This implies  the following:
\[\Xi(-i\omega)N^T=N^T\Xi(-i\omega) \,\, \mathbf{and} \,\, \Xi(-i\omega)M=M\overline{\Xi(+i\omega)}\]
from \eqref{th1} and \eqref{sq1}.
In particular, the second of these requires that either $CC^{\dag}M=0$ or $\Omega=0$, otherwise the poles on the left-hand side and the right-hand side will differ. Note that if $CC^{\dag}M=0$ then $CC^{\dag}M$ is also zero, hence the power spectrum is trivial. The remaining condition in case (i) follows from the thermal part of the power spectrum.
\end{proof}
Note that case (i) in Lemma \ref{onemode} is the MIMO extension of the one mode SISO elementary system  from Sec. \ref{SISO1}.    In fact, by an appropriate change of basis a one-mode MIMO system can be viewed as a one mode SISO system (with $m-1$ ancila channels). In this basis, case (i) corresponds to an input supported only on the SISO channel, whereas case (ii) may be interpreted  as an input  supported only on the ancilla channels. 

We have the following result for  $n=2$  modes.

%[[[FOLLOWING THEOREM IS JUST DESCRIBING ALL WAYS THAT COMBINED TWO mode system will be passive when viewed in vacuum basis.]]]
\begin{thm}\label{twomode}
A MIMO PQLS, $\left(C, \Omega\right)$, with two internal modes  $(n=2)$ is non-globally minimal for  input field, ${V}(N,M)$, if and only if there exists a TFE cascaded system\footnote{Note that a MIMO PQLS can always be realised as a cascade (see Sec. \ref{seriesp}). For a two mode PQLS there will be two ways to reorder the modes, corresponding to reordering the elements in the Schur decomposition of the system matrix, $A$ \cite{Petersen2}.},   $\mathcal{G}=\left(d,\Omega_2\right)\triangleleft\left(c,\Omega_1\right)$, such that any of the following conditions hold:
\begin{enumerate}
\item $\Omega_2=0$ and $\left[dd^{\dag}, N^T\right]=0$ and $dd^{\dag}M-M\left(dd^{\dag}\right)^T=0$,
\item $dd^{\dag}M=0$

\item $\Omega_1=-\Omega_2\neq0$ and $c^{\dag}c=d^{\dag}d$, and
$$cc^{\dag}N^T=N^Tcc^{\dag}\quad dd^{\dag}N^T=N^Tdd^{\dag}\quad cc^{\dag}dd^{\dag}N^T=N^Tcc^{\dag}dd^{\dag}$$. 
$$cc^{\dag}M=M\overline{cc^{\dag}}\quad dd^{\dag}M=M\overline{dd^{\dag}} \quad cc^{\dag}dd^{\dag}M=M\overline{cc^{\dag}dd^{\dag}}$$.
\item $\Omega_1=-\Omega_2=0$, and 
$$\left(cc^{\dag}+dd^{\dag}\right)N^T=N^T\left(cc^{\dag}+dd^{\dag}\right)$$
$$\left(cc^{\dag}dd^{\dag}-\frac{1}{2}c^{\dag}cdd^{\dag}-\frac{1}{2}d^{\dag}dcc^{\dag}\right)N^T=N^T\left(cc^{\dag}dd^{\dag}-\frac{1}{2}c^{\dag}cdd^{\dag}-\frac{1}{2}d^{\dag}dcc^{\dag}\right)$$
$$\left(cc^{\dag}+dd^{\dag}\right)M=M\overline{\left(cc^{\dag}+dd^{\dag}\right)}$$
$$\left(cc^{\dag}dd^{\dag}-\frac{1}{2}c^{\dag}cdd^{\dag}-\frac{1}{2}d^{\dag}dcc^{\dag}\right)M=M\overline{\left(cc^{\dag}dd^{\dag}-\frac{1}{2}c^{\dag}cdd^{\dag}-\frac{1}{2}d^{\dag}dcc^{\dag}\right)}.$$
\end{enumerate}
\end{thm}

The proof of this theorem is given in Appendix \ref{nike}.
Let us understand the interpretation of the conditions required for non-global minimality from Theorem \ref{twomode}. Firstly, case (1) and (2) are the same as Lemma \ref{onemode}, being one mode cancellations. Let us understand  cases (3) and (4)  in the following example. 
%Cases (3) and (4) correspond to two mode cancellations and are the MIMO analogues of the two-mode elementary SISO system  from section \ref{SISO1}.

\begin{exmp}
Suppose that we have two input channels and work in the  field basis where there are two  independent squeezed modes (see Sec. \ref{Bosonic}). That is, 
$$N=\left(\begin{smallmatrix}N_1&0\\0&N_2\end{smallmatrix}\right) \quad \mathrm{and}\quad M=\left(\begin{smallmatrix}M_1&0\\0&M_2\end{smallmatrix}\right).$$
Let $c=(c_1, c_2)^T$ and $d=(d_1, d_2)^T$ in this diagonal basis

Let us look at case (3).  If $N_1\neq N_2$  then conditions (3) imply that $c_2=d_2=0$ or $c_1=d_1=0$. That is, we have a SISO series product of two one-mode systems satisfying the conditions of Corollary \ref{pj}, and an ancilla channel. Now, if $N_1=N_2$, and taking $M_1=M_2$ (in general $M_1$ and $M_2$ differ by a phase, which can be undone by changing basis of the input), then conditions (3) imply that $cc^{\dag}=dd^{\dag}=\overline{cc^{\dag}}$. Therefore, $CC^{\dag}=2cc^{\dag}$ and so $CC^{\dag}$ is rank one. The upshot is that this system  is also a  SISO system with ancilla, this time viewed in a different (rotated) basis. Following this analysis we can conclude that case (3) is the MIMO analogue of the two-mode elementary SISO system  from Sec. \ref{SISO1}.

Case (4), where $\Omega_1=-\Omega_2=0$,  reduces to the following subcases:
\begin{itemize}
\item $N=N_1{1}$ and taking $M_1=M_2$, then $cc^{\dag}=\overline{cc^{\dag}}$, $dd^{\dag}=\overline{dd^{\dag}}$.
\item $CC^{\dag}$ is a multiple of identity and either i)$|c_1|^2=|d_1|^2$ and $|c_2|^2=|d_2|^2$ or ii)$|c_2|^2=|d_1|^2$ and $|c_1|^2=|d_2|^2$.
\end{itemize}
The first subcase consists of two SISO one-mode cancellations, i.e the stationary state is pure and separable. The second is a concatenation of two identical single mode SISO systems. 
% [To do: Something like to see this rotate basis so that couples with first mode only. Then second mode must be on other one because diag. ]]]
In summary in this choice of field basis where there is no entanglement between the channels, all of the cases of non-global minimality in Theorem \ref{twomode} reduce to SISO cancellations (Sec. \ref{SISO1}). 
\end{exmp}

Note that in Theorem \ref{twomode} and Lemma \ref{onemode} we didn't use the purity assumption.

%[[[THEOREM extending to many modes plus interpretation. Not done as yet because may need absorber stuff...have some extra comments on 10/10/16 sheet...but need to find proof that i did a few weeks ago]]]

\begin{thm}\label{doned}
Suppose that an $n$ mode MIMO PQLS, $(C, \Omega)$, with pure input field ${V}(N,M)$ has a pure stationary state. Then there exists a TFE cascade realisation
$\mathcal{G}=\left(c_1,\Omega_1\right)\triangleleft\left(c_2,\Omega_2\right)...\triangleleft\left(c_n,\Omega_n\right)$ such that either 
\begin{itemize}
\item $\left(c_1,\Omega_1\right)$ has trivial power spectrum ($\Psi(\omega)=V$), or
\item $\left(c_1,\Omega_1\right)\triangleleft\left(c_2,\Omega_2\right)$ has trivial power spectrum ($\Psi(\omega)=V$).
\end{itemize}
By repeating this statement iteratively, then a system, $\left(c_i,\Omega_i\right)$, within the cascade   either has trivial power spectrum or  $\left(c_{i+1},\Omega_{i+1}\right)\triangleleft \left(c_i,\Omega_i\right) $ does. 
\end{thm}
\begin{proof}
For any PQLS there exists TFE cascaded system \cite{Nurdin3}. There are two cases; either the stationary state of the first system in the cascade is pure, or it is mixed. If it is pure then case the first bullet point in the theorem follows immediately. If it is mixed case then the second bullet point follows; the proof of this  requires  more theory, which we  postpone  until  Ch. \ref{DUAL}. In particular, the proof  
 follows as a direct application of a Theorem \ref{parfg}. 
\end{proof}

This result says that a PQLS with a pure stationary state can be decomposed into smaller systems of one or two modes each with a pure stationary state, connected via the series product.  Note that if there is to be a two-mode pure stationary state then this theorem says that the two modes must be adjacent in (one of) the cascade realisation(s). It also enables us to write the following corollary, which shows that determining whether or not a system is global minimality can be answered by considering only the first two modes in (every) cascade realisation. 
%For example, a three mode entangled pure stationary state  is not possible. 

\begin{Corollary}\label{final}
A MIMO PQLS, $\left(C, \Omega\right)$, with $n$ internal modes is non-globally minimal for pure input field, $\mathbf{V}(N,M)$, if and only if there exists a  cascade realisation, $\mathcal{G}=\left(c_1,\Omega_1\right)\triangleleft\left(c_2,\Omega_2\right)...\triangleleft\left(c_n,\Omega_n\right)$ such on the first two modes any of the conditions in Theorem \ref{twomode} hold.
\end{Corollary}
\begin{proof}
Firstly, by Theorem \ref{equivalence} the system may be realised as a series product of a system with a pure stationary state and one with a mixed stationary state. Applying Theorem \ref{doned} to the pure part, it can be realised as a series of one or two mode systems, each with pure stationary state. Finally, applying Theorem \ref{twomode} to the first  two mode system in the cascade   gives the result. 
\end{proof}

We now give an algorithm enabling one to find the globally minimal restriction of a given MIMO PQLS $(c, \Omega)$ with $n$ internal modes. We assume that its system matrix $A$ has discrete spectrum for simplicity, 

\underline{Algorithm:}
\begin{enumerate}
\item
Calculate the matrix $A$ and its eigenvalues. 
\item
Perform a \textit{Schur diagonalisation}, that is, find a unitary $U$ and lower triangular matrix $A'$ such that $A'=UAU^{\dag}$ using known algorithms. Note that the eigenvalues of $A$ lie along the main diagonal.
\item \label{goof}
Check the conditions in Theorem \ref{twomode} for the first two systems of  the cascaded system $(C', A'):=\left(CU^{\dag}, UAU^{\dag}\right)$.
If any  are satisfied, then the system is not globally minimal and that particular subsystem has a pure stationary state. 
If they aren't then move to step (\ref{goofs}).
 \item Remove the non-globally minimal subsystem and repeat step (\ref{goof}).
\item\label{goofs}
Repeat steps (2) and (3) with a different pair of eigenvalues in the first two slots of $A'$ (order matters).
\item When all orders of modes have  been exhausted, \textbf{stop} and conclude that the remaining system is globally-minimal. 
\end{enumerate}

\begin{remark}
Finding an equivalent cascade realisation of a QLS requires performing a Schur decomposition  on the system matrix, $A$. This algorithm uses the fact that the transformation to the lower triangular matrix in the Schur diagonalization, $A'$,  is unique up to  a  diagonal unitary matrix for a given fixed order of elements on the main diagonal of $A'$. 
\end{remark}

\begin{exmp}
Consider the two-mode  PQLS $$(C, A)=\left(    \left(\begin{smallmatrix} -4&8\\-3&3\end{smallmatrix}\right),   \left(\begin{smallmatrix} -\frac{45}{4}-2i    &\frac{65}{4}+2i\\\frac{17}{4}+2i&-\frac{73}{4}-2i\end{smallmatrix}\right) \right) $$
with input $V(N, M)$, where $N, M$ are diagonal. 

The two TFE  cascade realisations of this system are those given in example \ref{beat}; they correspond to the  the two possible orderings of the eigenvalues in the Schur decomposition of $A$.
Now,  one can verify that the system $(\tilde{C}_2, \tilde{\Omega}_2)\triangleleft (\tilde{C}_1, \tilde{\Omega}_1)$ doesn't satisfy any of the conditions in Theorem \ref{twomode}. However, the system 
 $(C_1, \Omega_1)$
in the cascaded system $(C_2, \Omega_2)\triangleleft (C_1, \Omega_1)$ satisfies condition (1) in Theorem \ref{twomode}. Hence the globally minimal restriction is given by $(C_2, \Omega_2)$. 
\end{exmp}

%[[[Comment about why was so easy in MIMO case: In the SISO case it is only necessary to check one Schur decomposition as all the modes in the Series product commute.]]] 

\section{Entangled Inputs}\label{EI}

Here we show that by using an additional ancillary channel with an appropriate design of input makes it  possible to identify the transfer function from the power spectrum for \emph{all} minimal systems.

Consider the set-up in Fig. \ref{entangledr}, where a pure entangled input state is fed into a QLS with $m$ channels and concatenate with an  additional $m$ ancillary channels. We assume that the input is non-vacuum and is characterised by $V(N, M)$,
which has $2m\times 2m$ blocks 
\[N=\left(\begin{smallmatrix} N_1&N_2\\{N_2}^{\dag}&N_3\end{smallmatrix}\right)\,\,\, M=\left(\begin{smallmatrix} M_1&M_2\\M^T_2&M_3\end{smallmatrix}\right).\] 
Each $N_i$ and $M_i$ are of size $m\times m$. The doubled-up transfer function is given (in $n\times n$ blocks) by
\begin{equation}\label{fud}
\Xi(s)=\left(\begin{smallmatrix}\Xi_{-}(s)&0&\Xi_{+}(s)&0\\0&1&0&0\\{\Xi_{+}(\overline{s})}^{\#}&0&{\Xi_{-}(\overline{s})}^{\#}&0\\0&0&0&1\end{smallmatrix}\right).
\end{equation}
Now calculating the $(2,1)$ and $(1,4)$ blocks of the power spectrum using (\ref{powers}), we obtain: 
$$
N^T_2{\Xi_{-}(s)}^{\dag}+M^T_2{\Xi_{+}(s)}^{\dag}:=\alpha(s)
$$ 
and 
$$\Xi_{-}(s)M_2+\Xi_{+}(s)N_2:=\beta(s).$$
To be clear, $\alpha(s)$ and $\beta(s)$  are known at this stage from the power spectrum. 
 Equivalently we may write these in matrix form as 
\[\left(\begin{smallmatrix}\Xi_{-}(s)&\Xi_{+}(s)\end{smallmatrix}\right)
\Delta(M_2, N_2^{\#})=\left(\begin{smallmatrix}\beta(s)&\alpha(s)^{\dag}\end{smallmatrix}\right)
.\]
Hence if  we choose $N_2$ and $M_2$ such that the matrix $\Delta(M_2, N_2^{\#})$ is non-singular we may identify the transfer function of our  system uniquely. For example, such a choice of input would be $N=x{1}$ and $M=\left(\begin{smallmatrix}0&y1_n\\y1_n&0\end{smallmatrix}\right)$ with $x(x+1)=|y|^2$ (the purity assumption). As one can see there are no requirements on the actual QLS other than minimality. 

\begin{figure}
\centering
\includegraphics[scale=0.35]{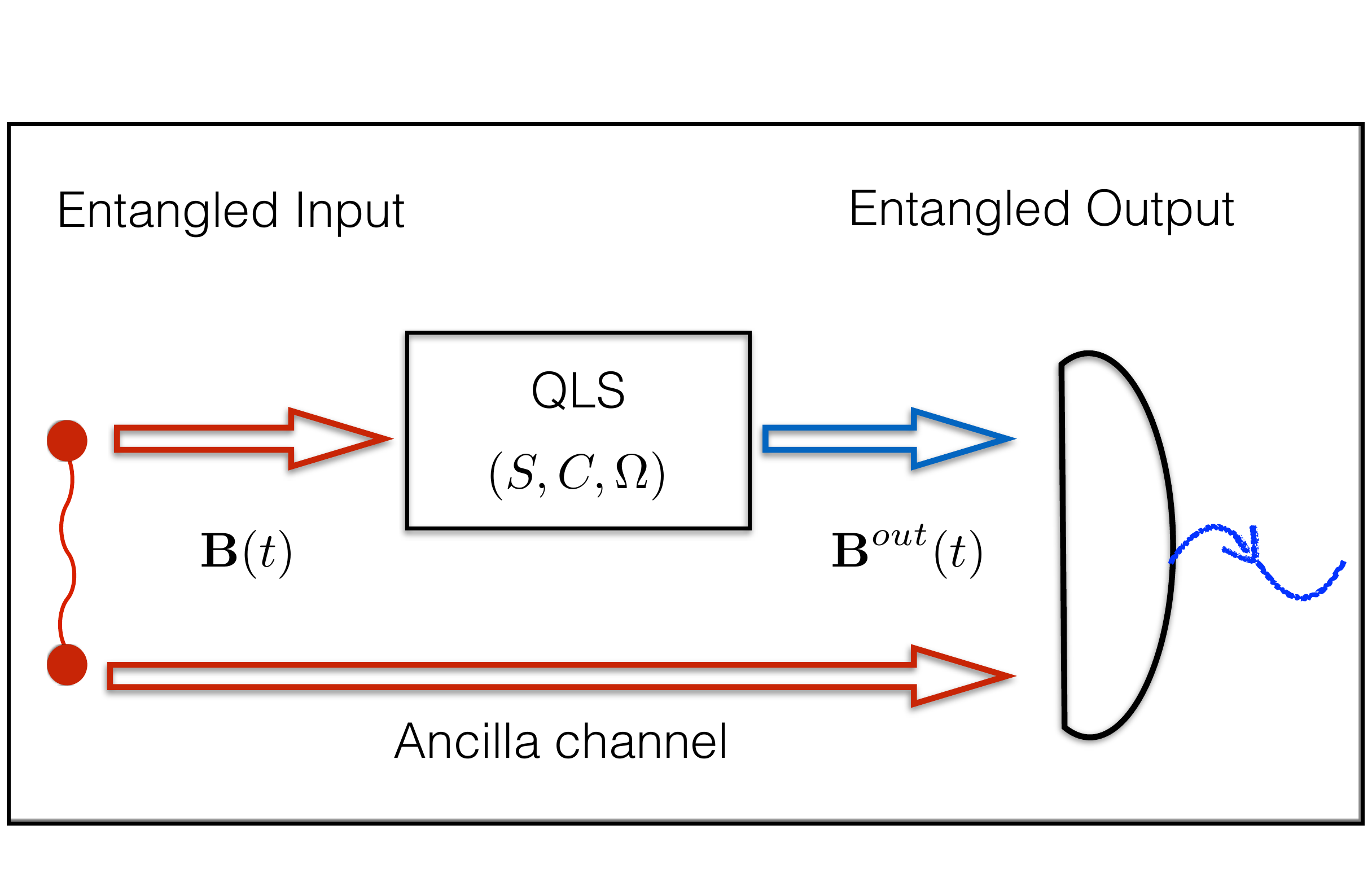}
\caption{Entangled setup discussed in Sec. \ref{EI}. There are two channels, which are our PQLS and an additional ancilla channel. Inputs are entangled over the two channels.  \label{entangledr}}
\end{figure}

\begin{remark}
Recall from the previous subsections that the maximum amount of information we may
obtain about a PQLS from the power spectrum without the use of ancilla is that of the restriction to its globally
minimal subspace.
However, we have seen here that it is possible to construct a globally minimal pair, and hence obtain the whole transfer function simply by embedding the system in a larger space. The crucial point is that we have used a different input, which matters for the power spectrum. 
%To be clear here, there is no
%contradiction because the transfer function we are attempting to identify is the one in Eq. (\ref{fud}) rather than  $\Xi(s)$. 
\end{remark}

From the analysis above, purity was not necessary to identify the transfer function. In fact, one didn't  even require any squeezing. We investigate relaxing the assumption of purity in more detail in the next subsection.

\section{Thermal Inputs}\label{thermy}
An interesting open question is whether the identifiability results of this chapter may be extended to mixed inputs.
 In this section we study identifiability for the  subclass of PQLSs with mixed, and in particular  thermal inputs 
(see Sec. \ref{Bosonic}).  This problem is particularly interesting because neither the input nor the system require any squeezing.  For simplicity, we shall assume that $S=1$  throughout this subsection.

Consider a general PQLS, which has coupling matrix $C$ and system Hamiltonian $\Omega$. Suppose we probe the system with known input $V(N, 0)$ ($N\geq0$ in order to be physical). The power spectrum of this system is
$$\Psi(s)=\left(\begin{smallmatrix} \Xi_-(s)(N^T+1)\Xi_-(-\overline{s})^{\dag}&0\\0& \overline{\Xi_-(\overline{s})}N\Xi_-(-{s})^{T}\end{smallmatrix}\right),$$
where $\Xi_-(s)=1-C\left(s-A\right)^{-1}C^{\dag}$ and $A=-i\Omega-\frac{1}{2}C^{\dag}C$.
Therefore our basic identifiability problem for the power spectrum reduces to identifiability of the quantity 
$ \Xi_-(s)N^T\Xi_-(-\overline{s})^{\dag}=:\Upsilon(s)$. 

Firstly observe that if the PQLS is SISO then  the power spectrum is always trivial because the transfer function is unitary. This should not be too surprising because, as the system is passive, the effect of the system on a given frequency mode is to rotate the input covariance in the $(\mathbf{X}, \mathbf{P})$ phase space. Since a one mode thermal state is centred and circularly symmetric in phase space,  such a rotation will not be visible, i.e., the input and output will appear the same. 
However, in the MIMO case one can choose a thermal input   so that it is not 
 symmetrical with respect to the different channels. We shall see that this allows for identifiability.

%However, as we shall see, this is not necessarily true for MIMO systems, which suggests that entanglement between the channels is aiding identifiability. 

Before answering our identifiability problem for this scenario, let us consider a situation where we are free to modulate the input. That is, suppose we have access to the power spectrum for all input covariances in a small neighbourhood (rather than for a specific input). This assumption is of course quite strong, but it is still nevertheless an interesting starting point for identifiability. 

\begin{thm}
Let $(C_1, \Omega_1)$ and $(C_2, \Omega_2)$ be two PQLSs which are globally minimal for all noise covariances in a neighbourhood of $V=\left(\begin{smallmatrix} N^T+1&0\\0&N\end{smallmatrix}\right)$. If $\Upsilon_1(s)=\Upsilon_2(s)$ for all thermal covariances $V'$ in a neighbourhood of $V$, then the  systems are TFE: $\Xi_1(s)=\Xi_2(s)$.
\end{thm}
\begin{proof}
If $\Upsilon_1(s)=\Upsilon_2(s)$ for all $V'$ in a neighbourhood of $V$ then $\Xi(s)=e^{\phi(s)}\Xi_2(s)$ for some $\phi(s)$. This is because knowing  the action on all $V'$ means that  you know the action on all matrices.  Two actions are the same if and only if $\Xi_1(s)\Xi_2(s)^{\dag}$ commutes with all $V'$, so they must differ by a phase.
Finally as the transfer function is rational and monic, then the phase must be trivial. 
\end{proof} 

Now back to our original identifiability problem. Just like in Sec. \ref{altos}, we can treat $\Upsilon(s)$ as a transfer function realised by the resultant cascaded system:
\begin{equation}\label{cask2}
\left(\tilde{A}, \tilde{B}, \tilde{C}, \tilde{D}\right):=
\left(\left(\begin{smallmatrix} -A^{\dag}&0\\C^{\dag}NC&A\end{smallmatrix}\right), \left(\begin{smallmatrix} -C^{\dag}\\-C^{\dag}N\end{smallmatrix}\right), \left(\begin{smallmatrix} -NC& C\end{smallmatrix}\right), N     \right).
\end{equation}
Notice that A has $2n$ eigenvalues, rather than $4n$, and is proper LBT (it should be  understood that each block is of size $n\times n$).

\begin{thm}\label{main8}
Let $\left(C_1, \Omega_1\right)$ and $\left(C_2, \Omega_2\right)$ be two PQLS with minimal representations \eqref{cask2}    for input $V(N,0)$,
% which are assumed to be generic in the sense that their system matrix $A_i:=-iJ\Omega_i-\frac{1}{2}C_i^{\flat}C_i$ has distinct eigenvalues. 
%
then
$$
\Upsilon_{1}(s)=\Upsilon_{2}(s) \,\, for~all~s \quad \Leftrightarrow \quad \Xi_1(s)=\Xi_2(s)\,\, for~all~s 
$$
\end{thm}

\begin{remark}
Notice that unlike Theorem \eqref{main}, the stable assumption is not required as stablility and minimality are equivalent for PQLSs \cite{Guta2}.  
\end{remark}
\begin{proof}
The proof of this statement can be obtained almost identically to that of Theorem \ref{main}. That is, by firstly using the properties of proper LBT matrices to reduces the admissible set of equivalent transformations to lower block triangular similarity transformations. We then show the following two steps
\begin{enumerate}
\item Firstly, 
$$\left(\begin{smallmatrix}T_4^{\dag}&0\\-T^{\dag}_3&T_4^{\dag}\end{smallmatrix}\right)T=1.$$
\item Finally, 
$$\left(\begin{smallmatrix}T_4^{\dag}-T^{\dag}_1\\-T_3^{\dag}\end{smallmatrix}\right)=0.$$
\end{enumerate}
Notice that step (3) in Theorem \ref{main} is not required. The proofs of steps (1) and (2) here are identical to those in Appendix \ref{p1}, except that $K$ is now replaced with $\left(\begin{smallmatrix}1&0\\0&-1\end{smallmatrix}\right)$ and $V_{\mathrm{vac}}$ is replaced with $N$.
\end{proof}
This theorem says that if the cascaded system \eqref{cask2} is minimal then the transfer function of the QLS is identifiable. However, what does minimality mean here? Recall how we saw in Sec. \ref{altos} that minimality of the analogous cascaded system  was equivalent to global minimality of the QLS (recall that  we had Theorem \ref{games}). We would like to investigate whether this holds here. In fact, proving this equivalence is much trickier here because  Theorem \ref{games} was derived via an intermediate result that, i.e., that  global minimality is equivalent to the system possessing a fully mixed stationary state. Such an intermediate result does not hold in the case of mixed inputs. 
That is, although the statement [global minimality implies (full)-mixed stationary state] is true (we do not give the proof here),  the converse statement [(full)-mixed stationary state implies global minimality] no longer holds. In fact, because the  term $C^{\flat}V\left(C^{\flat}\right)^{\dag}>0$ in the Lyapunov equation \eqref{eq.Lyapunov} (with $V(N,M)$ replaced with $V_{\mathrm{vac}}$) for all mixed inputs, then it follows from \cite[Theorem 3.18]{Zhou1} that $P>0$; hence the stationary state is always fully mixed.

%longer holds. In fact, the stationary state is always fully mixed in this case. To see this  first observe that the covariance matrix of the system's stationary state satisfies the Lyapunov equation \eqref{eq.Lyapunov} with $V(N,M)$ replaced with $V_{\mathrm{vac}}$. 
%Since the term  $C^{\flat}V\left(C^{\flat}\right)^{\dag}>0$   for all mixed inputs, then it follows from \cite[Theorem 3.18]{Zhou1} that $P>0$.

Despite the problems above, can we still prove the equivalence between minimality of the cascaded system and global minimality of the PQLS?
Firstly, the statement [\eqref{cask2} minimal implies $(C, \Omega)$ globally minimal] is trivial. However, we are more interested in whether  the reverse statement is true. Let us now understand when this is the case. The requirement \eqref{cask2} be minimal,  entails it to be both observable and controllable. 

\begin{Lemma}\label{ghjk}
The cascaded system \eqref{cask2} is observable iff it is controllable.
\end{Lemma}
\begin{proof}
Suppose that \eqref{cask2} is observable, which is equivalent to $(\tilde{A}^{\dag}, \tilde{C}^{\dag})$ controllable. It therefore remains to show that $(\tilde{A}^{\dag}, \tilde{C}^{\dag})$ controllable is equivalent to $(\tilde{A}, \tilde{B})$ controllable. 
Firstly, suppose that $(\tilde{A}^{\dag}, \tilde{C}^{\dag})$ is controllable. This is equivalent to the statement: for any (left)-eigenvector and eigenvalue, $X, \lambda$, of $\tilde{A}^{\dag}$, then $X\tilde{C}^{\dag}\neq0$. Equivalently, 
$$\left(\begin{smallmatrix}X_1 &X_2\end{smallmatrix}\right)\left(\begin{smallmatrix}-A&C^{\dag}NC\\0&A^{\dag}\end{smallmatrix}\right)=\lambda\left(\begin{smallmatrix}X_1 &X_2\end{smallmatrix}\right)~\mathrm{implies}~ 
-X_1C^{\dag}N+X_2C^{\dag}\neq0.$$
This in turn is equivalent to the statement:
$$\left(\begin{smallmatrix}X_2 &-X_1\end{smallmatrix}\right)\left(\begin{smallmatrix}-A^{\dag}&0\\C^{\dag}NC&A\end{smallmatrix}\right)=\lambda\left(\begin{smallmatrix}X_2 &-X_1\end{smallmatrix}\right)~\mathrm{implies}~ 
X_2C^{\dag}+(-X_1)C^{\dag}N\neq0,$$
which is equivalent to $(\tilde{A}, \tilde{B})$ controllable.
\end{proof}

In light of Lemma \ref{ghjk}, understanding when \eqref{cask2} is minimal reduces to understanding when \eqref{cask} is observable. Observability is equivalent to the statement: "for any eigenvector and eigenvalues $y, \lambda$ of $\tilde{A}$, then $\tilde{C}y\neq0$". Therefore, by using the definition of $\tilde{A}$ and $\tilde{C}$ above,  the system being not observable reduces to the following cases:
\begin{itemize}
\item[(1)]  There exists a pair of (right)-eigenvectors of $A^{\dag}$, $y_1, y_2$, with the same eigenvalue such that 
$$NCy_1=Cy_2$$ 
and $Cy_1\neq0$.
\item[(2)] There exists a (right)-eigenvalue $\left(\begin{smallmatrix}0\\y\end{smallmatrix}\right)$ of $A^{\dag}$, where $y$ is of size $n$ such that $Cy\neq0$. 
\end{itemize}
The second case may be excluded if one assumes our system $(C, \Omega)$ is minimal. Let us assume for simplicity that the eigenvalues of $\tilde{A}$ are distinct so that in case (1) above we must have  $y_1=\mu y_2$ for some $\mu\neq0$. Therefore, $Cy_1$ is an eigenvector of $N$ and must  necessarily have a real eigenvalue (as $N\geq0$). In summary, for the cascaded system to be non-observable  the minimal system $(A, C)$ must have an input $N$ such that one of the eigenvectors of $A^{\dag}$, $y_1$, is such that $Cy_1$ is an eigenvector of $N$.
This condition is quite generic. Moreover, it can be shown that  when this non-observability condition holds, then  the system is non-globally minimal (see Appendix \ref{noobs}). Therefore,  when $A$ has  distinct eigenvalues    global minimality implies identifiability, which   extends  Theorem \eqref{main} to this subclass of mixed inputs. 

The interpretation of the non-identifiable (sub-)systems (2) can be seen to be essentially a SISO system viewed in a different field basis. In this canonical basis the inputs also  must  not be entangled, otherwise the system would be identifiable as in Sec. \ref{EI}.     Hence these systems are non-identifiable due to the perfect symmetry of the input.

%[[[Don't need to impose any additional restriction, such as ancilla channels. Get these for free here?- as as $c=0$ on these.]]]

%[[[Suggests purity of input may not be key feature. ALSO squeezing was required in the pure case (ref the theorem) so interesting in that sense]]]

%[CONJECTURE: When is system G.M?; when one of pure states in mixture is]]]

%[[[MIXED INPUTS WORK FINE in SISO i.e equiv classes same and which systems are g.m are same...also active SISO too...could discuss this]]]

\section{Quantum Fisher Information}\label{JUKKA1}

We  now find an explicit expression for the QFI at stationarity and show that if one assumes global minimality, the zeros of it exactly correspond to the set of gauge transformation of the power spectrum that we saw in Theorem \ref{main}. The proof of Theorem \ref{main} was obtained by using mainly classical system theoretic concepts; the results here give an alternative proof, this time using   quantum mechanical arguments.

Throughout this section we work with the raw (pure) field input, rather than performing the trick in Sec. \ref{G.Ms} and treating the field as vacuum. The reason for this is that  we shall be working with 
stochastic integrals in the following, where one must a little 
careful performing static squeezing operations on the field \cite{Gough1}. 
%However, observe for the moment that 
%because the input is pure   it may be written
%
%\begin{equation}\label{vtrick}
%V=\left(\begin{smallmatrix} \sqrt{N^T+1}& M\frac{1}{\sqrt{N^T+1}}^{\#}\\ M^{\#}\frac{1}{\sqrt{N^T+1}}&\sqrt{N^T+1}^{\#}\end{smallmatrix}\right) \Bigg(\begin{smallmatrix} 1&0\\0&0\end{smallmatrix}\Bigg) 
%\left(\begin{smallmatrix} \sqrt{N^T+1}&\frac{1}{\sqrt{N^T+1}}M\\\frac{1}{\sqrt{N^T+1}}^{\#}M^{\#}&\sqrt{N^T+1}^{\#}\end{smallmatrix}\right)
% \end{equation}
 %
\subsection{Preliminaries}

We need a few preliminaries in this section. First of all denote the Heisenberg evolved system operator by ${j}_t(\mathbf{X})=\mathbf{U}^{\dag}(t)\left(\mathbf{X}\otimes \mathds{1}_{\mathrm{field}}\right)\mathbf{U}(t)$. It follows that ${j}_t(\mathbf{X})$ satisfies the QSDE
$$d{j}_t(\mathbf{X})=\sum_i\left(  {j}_t\left(\left[\mathbf{X}, \mathbf{L}_i\right]\right)d\mathbf{B}^{\dag}_i(t)+ {j}_t\left(\left[\mathbf{L}^{\dag}_i, \mathbf{X}\right]\right)d\mathbf{B}_i(t)\right)+{j}_t(\mathcal{L}(\mathbf{X})),$$
where $\mathcal{L}(\cdot)=-i[(\cdot),\mathbf{H}]+\sum_i\left( \mathbf{L}^{\dag}_i(\cdot)\mathbf{L}_i-\frac{1}{2}\mathbf{L}^{\dag}_i\mathbf{L}_i(\cdot)
-\frac{1}{2}(\cdot)\mathbf{L}^{\dag}_i\mathbf{L}_i\right)$
is called the \textit{Lindblad generator} and $\mathbf{L}_i$ are the elements of the coupling operator \eqref{couplef}. Also, define $T_t(\mathbf{X}):=\left<\xi|j_t(\mathbf{X})|\xi\right>$, which is the Heisenberg evolution of $\mathbf{X}$ restricted to the system. The operator $T_t(\mathbf{X})$ satisfies the following properties \cite{Guta4, Frigerio1}:
\begin{itemize}
\item Firstly, $dT_t(\mathbf{X})=T_t(\mathcal{L}(\mathbf{X}))$ and so $T_t(\cdot)$ is a completely positive trace preserving semigroup with generator $\mathcal{L}(\cdot)$.
\item Also, $\lim_{t\to\infty} T_t(\mathbf{X})=   \lim_{t\to\infty}\frac{1}{t}\int^t_0T_s(\mathbf{X})ds=\left<\mathbf{X}\right>_{\rho_{ss}}\mathds{1}$, where $\rho_{ss}$ is the stationary state of the system and $\left<\cdot\right>_{\rho_{ss}}$ is the (quantum) expectation on the system with respect to the state $\rho_{ss}$.
%\item The limit $  \lim_{t\to\infty}\int^t_0T_s(\mathbf{X})ds$ exists. 
\end{itemize}

We can also  calculate the evolution of the operators $T_t(\breve{\mathbf{a}})$ and $T_t(\breve{\mathbf{a}}^{\dag}X\breve{\mathbf{a}})$ for some matrix $X$.  
Firstly,  from Eq. \eqref{lan2} the Heisenberg system evolution for the system operator $\breve{\mathbf{a}}$
has solution $T_t(\breve{\mathbf{a}})=e^{At}\breve{\mathbf{a}}$. 
Notice that $\lim_{t\to\infty}T_t(\breve{\mathbf{a}})=0$, as the system is assumed to be Hurwitz.
Also, it follows from the  Ito rules (see Appendix \ref{cherry}) that
\begin{align*}
dT_t(\breve{\mathbf{a}}^{\dag}X\breve{\mathbf{a}})&=\bra{\xi}|\breve{\mathbf{a}}^{\dag}(t)X\left(A\breve{\mathbf{a}}(t)dt\ket{\xi}-C^bd\mathbf{B}(t)\ket{\xi}\right)   \\
& +\left(\bra{\xi}\breve{\mathbf{a}}^{\dag}(t)A^{\dag}dt-\bra{\xi}C^{\#}d\breve{\mathbf{B}}^{\dag}(t)\right)A\breve{\mathbf{a}}(t)\ket{\xi}+J(t)\mathds{1}dt\\
&=    \bra{\xi}    \breve{\mathbf{a}}^{\dag}(t)\left(XA      +    A^{\dag}X  \right)\breve{\mathbf{a}}(t)\ket{\xi}dt +J(t)\mathds{1}dt,       
\end{align*}
for some matrix $J(t)$, which we do not specify here.
%That is, 
%\[\frac{d}{dt}T_t(\breve{a}^{\dag}X\breve{a})=\frac{d}{dt}\braket{\xi|\breve{a}^{\dag}(t)X\breve{a}(t)|\xi}+J(t)\mathds{1}.\]
Therefore
%Finally, using the observation that $\braket{\xi|d\breve{A}(s)|\xi}=0$ it follows that 
\begin{equation}\label{lan4}
T_t(\breve{\mathbf{a}}^{\dag}X\breve{\mathbf{a}})=\breve{\mathbf{a}}^{\dag}e^{A^{\dag}t}Xe^{At}\breve{\mathbf{a}}+K(t)\mathds{1}
\end{equation}
for some $K(t)$. The  particular form of $K(t)$  is not important, but 
%Just observe that $K(t)$ contains two contributions; the first from the integral of $J(t)$ and second from the  $d\breve{A}^{\dag}(s)d\breve{A}(s)$ term coming from the product of (\ref{lan2}).
just observe  that $\lim_{t\to\infty}K(t)=\left<  \breve{\mathbf{a}}^{\dag}X\breve{\mathbf{a}}\right>_{\rho_{ss}}$.   

%because is the stationary mean because in the limit $t\to\infty$ the first term vanishes.

\subsection{QFI Calculation}

We consider a quantum statistical model over a parameter space $\Theta\in\mathbb{C}$, which is a family of QLSs $(C_{\theta},\Omega_{\theta})$ indexed by an unknown parameter $\theta\in\Theta$.
We denote the  dependence on $\theta$ in  the unitary operator $\mathbf{U}(t)$  by $\mathbf{U}_{\theta}(t)$. 
Recall from Ch. \ref{plm} that the most information that one can hope to obtain from any measurement is given by the \textit{quantum Fisher information} (QFI). 
In this subsection we calculate the QFI for our  QLS in the stationary approach. We assume for the moment that we have full access to the (pure) output state, which is given by
$\ket{\psi_{\theta}(t)}:=\mathbf{U}_{\theta}(t)\ket{\phi \otimes \xi}$ where $\ket{\phi}$ is the initial state of the system, $\ket{\xi}$ is the pure (stationary) input corresponding to $V(N,M)$ on the field and $\mathbf{U}_{\theta}(t)$ is the unitary operator describing the joint evolution of the system and field (see Eq. \eqref{eq.unitary.cocycle}).
Note that the reduced state of the field, $\rho^{\mathrm{out}}(t)$, may be obtained by taking the partial trace:
$$\rho^{\mathrm{out}}(t)=\mathrm{Tr}_{sys}\left( \ket{\psi_{\theta}(t)}\bra{\psi_{\theta}(t)}\right).$$

\begin{thm}\label{Q.F.I}
The QFI   for pure state $\mathbf{U}_{\theta}(t)\ket{\phi\otimes \xi}$   scales linearly with time, with asymptotic rate constant: 
\[f_{\theta}:=\lim_{t\to\infty}\frac{F_\theta}{t}.\] Moreover, the QFI rate is given by 
\begin{equation}\label{QFIMAIN}
f_{\theta}=4  \sum_{i=1}^m\,\left<   \left(\dot{\tilde{\mathbf{L}}}_i-i\left[\tilde{\mathbf{L}}_i, \mathbf{W}\right]\right)^{\dag}   \left(\dot{\tilde{\mathbf{L}}}_i-i\left[\tilde{\mathbf{L}}_i, \mathbf{W}\right]\right) \right>_{\rho_{ss}}
\end{equation}
where 
 $\left<\cdot\right>_{\rho_{ss}}$ is the quantum expectation of the system at stationary state $\rho_{ss}$ and
$$\mathbf{W}:=\int^{\infty}_0T_s(\mathbf{R}_0)\mathrm{d}s,$$
  \[ \mathbf{R}_0:= \dot{\mathbf{H}}+\mathrm{Im}\sum_{i=1}^m\dot{\tilde{\mathbf{L}_i}}^{\dag}\tilde{\mathbf{L}_i}  -\left<\dot{\mathbf{H}}+\mathrm{Im}\sum_{i=1}^m\dot{\tilde{\mathbf{L}_i}}^{\dag}\tilde{\mathbf{L}_i}\right>_{\rho_{ss}} \mathds{1}  \]
 for modified coupling operator 
 $\tilde{\mathbf{L}}:=\sqrt{N^T+1}\mathbf{L}-\left(\sqrt{N^T+1}\right)^{-1}M\mathbf{L}^{\#}$. Note that by definition $\mathbf{R}_0$ is hermitian. 
\end{thm}
The proof of this theorem follows the same method as the proof for non-linear systems from \cite{Guta4} and is given in Appendix \ref{cherry}.
The  reader may have noticed that so far we have assumed access to both system and field. Typically, in QLS theory we are of course only allowed  access to the field. So the result is actually  weaker than we require.     However, we can argue that Eq. \eqref{QFIMAIN} represents the QFI for the field only because as the system reaches stationarity, the rate at which information content is extracted from the system goes to zero, which is evidenced by the fact that power spectrum contains no terms from the system. This of course is not a proof; to prove that  Eq. \eqref{QFIMAIN} represents the QFI for the field only we can employ the following argument: consider the joint system-output state 
$\ket{\psi_{\theta}(t)}$ as above. By performing a \textit{Gaussian Schmidt Decomposition} \cite{Wolf1}, we can write 
$$\ket{\psi_{\theta}(t)}=\sum_i \sqrt{\lambda_i}\ket{\phi_i}\otimes \ket{\eta^{out}_i(t)},$$
where $\ket{\eta^{out}_i(t)}$ and $\ket{\phi_i}$ represent eigenbases of states of the output and stationary state, respectively 
(see Fig. \ref{spane}).
The output components $\ket{\eta^{out}_i(t)}$ are orthogonal (approximately, for large times) and their mixture is the output state. 
The proof that they are almost orthogonal should follow from the gap properties of the coupling operator $\mathbf{L}$ (or equivalently the eigenvalues of the system matrix $A$), but  we do not have  this yet. However, it has been proven in \cite{Guta4} for the non-linear setup, so we expect it to hold here 
too. The fact that the $\ket{\eta^{out}_i(t)}$ are almost orthogonal (even when you take different local parameters) means that you can distinguish them in the output without destroying the pure state $\ket{\eta^{out}_i(t)}$. 
Then each of these components has QFI with rate given by \eqref{QFIMAIN} and we are back in the case where we had access to the system and output. 
%However, it is expected that this expression 
%i%t was shown in \cite{} for finite-dimensional systems that the QFI of the output is the same as th
%However, in \eqref{QFIMAIN} the expression between $\bra{\phi}$ and $\ket{\phi}$, which is ordinarily an operator on the system is in fact a multiple of identity at stationarity (see  Eqs. (\ref{qfi1}) and (\ref{qfi2}) for why this is the case).
 %Thus we can remove the expectation on the system and interpret the quantity 
%$$4 \sum_{i=1}^m\,\,\mathbb{E}_{\rho_{ss}}\left[   \left(\dot{\tilde{\mathbf{L}}}_i-i\left[\tilde{\mathbf{L}}_i, \left(\int^{\infty}_0T_s(\mathbf{R}_0)\right)\right]\right)^{\dag}   \left(\dot{\tilde{\mathbf{L}}}_i-i\left[\tilde{\mathbf{L}}_i, \left(\int^{\infty}_0T_s(\mathbf{R}_0)\right)\right]\right) \right]$$
%as the QFI for the field at stationarity. The physical justification for  this is that  as the system approaches the stationary state the rate at which information is extracted from the system alone is zero.

%[[[DO WE ACTUALLY NEED THIS QUANTITY HERE?]]]

%[[[NEED TO ASK MADALIN ABOUT THIS ARGUMENT. When we have hand wavy argument we need to make sure intro to section is correctbecuase at moment it sounds like the caveat there is also true in rest of section.]]]

We are now in a position to prove the claim at the beginning of this section.

\subsection{Unidentifiable Directions in the Tangent Space}

\begin{thm}\label{activeclass}
If a one-dimensional family of QLSs $(C_{\theta},\Omega_{\theta})$ has QFI equal to zero, then the components of the tangent vector $\mathcal{T}:=(\dot{C}_{\theta},\dot{\Omega}_{\theta})$
satisfy 
\begin{equation}\label{rud1}
\dot{C}_{\theta}=-iC_{\theta}JR \,\,\,\mathrm{and}\,\,\, \dot{A}_{\theta}=i\left[JR,A_{\theta}\right]
\end{equation}
where $R$ is some Hermitian matrix of the form
$$R= \left(\begin{matrix} R_1&R_2\\R_2^{\dag}&R_1^{T}\end{matrix}
\right),$$
with $R_1=R_1^{\dag}$ and $R_2=R_2^T$ \cite{Xiao1}.
\end{thm}
\begin{proof}
The reverse implication can be straightforwardly verified, so it remains to prove the forward implication. We drop the subscript $\theta$ for ease of notation.

Firstly, as the system is globally minimal $\rho_{ss}$ must be of full rank, the zero QFI rate implies 
\begin{equation}\label{rud}
\dot{\tilde{\mathbf{L}}}_i= i\left[\tilde{\mathbf{L}}_i, \left(\int^{\infty}_0T_s(\mathbf{R}_0)\right)\right]
\end{equation} 
or all $i$.

Now, let us calculate $\mathbf{W}=\int^{\infty}_0T_s(\mathbf{R}_0)ds$. 
%Note that we have used the notation $\mathcal{W}^{-1}$ because, aside from any technicalities emanating from the fact we are working in infinite dimensions, $\lim_{t\to\infty}\int^t_0T_s(\mathcal{E}_0(\dot{D}))ds$ is essentially  the inverse of the Lindblad generator of $\mathcal{E}_0(\dot{D})$.
%
We can write $\mathbf{R}_0$ as 
 $$\mathbf{R}_0=\breve{\mathbf{a}}^{\dag}X\breve{\mathbf{a}}-     \left<\breve{\mathbf{a}}^{\dag}X\breve{\mathbf{a}}\right>_{\rho_{ss}}\mathds{1},$$ where the  hermitian matrix  $X$ is given by
\begin{align*}
X=\Delta&\left(   \frac{ \dot{\Omega}_-}{2}+\frac{1}{4i}\left(\left(\dot{C}_-^{\dag}C_-   -C_-^{\dag}\dot{C}_-\right)+ \left(\dot{C}_+^{\dag}C_+-C_+^{\dag}\dot{C}_+^{}\right)^T\right),\right.\\
&\left. \frac{ \dot{\Omega}_+}{2} +\frac{1}{4i}\left(     \left(\dot{C}_-^{\dag}C_+ -C_-^{\dag}\dot{C}_+\right)    +     \left(\dot{C}_-^{\dag}C_+ -C_-^{\dag}\dot{C}_+\right)^T             \right)\right).
\end{align*}
Note that here and in the following $C_-$ and $C_+$ are the modified coupling coefficients (i.e. $\tilde{\mathbf{L}}=C_-\mathbf{a}+C_+\mathbf{a}^{\dag}$). In fact one can easily verify that $X=\frac{\dot{\Omega}}{2}+\frac{1}{2}\mathrm{Im}\left(\dot{C}^{\dag}JC\right)$.
Using  \eqref{lan4} we can show  that   %
\begin{align*}
 T_t(\mathbf{R}_0)&=    T_t(\mathbf{R})-\left<\breve{\mathbf{a}}^{\dag}X\breve{\mathbf{a}}\right>_{\rho_{ss}}\mathds{1}\\
 &=\breve{\mathbf{a}}^{\dag}e^{A^{\dag}t}Xe^{At}\breve{\mathbf{a}}+\left<\breve{\mathbf{a}}^{\dag}X\breve{\mathbf{a}}\right>_{\rho_{ss}}\mathds{1}-\left<\breve{\mathbf{a}}^{\dag}X\breve{\mathbf{a}}\right>_{\rho_{ss}}\mathds{1}\\
 &=\breve{\mathbf{a}}^{\dag}e^{A^{\dag}t}Xe^{At}\breve{\mathbf{a}}.
 \end{align*}
The second line follows if $t$ is taken to be sufficiently large, which is a valid assumption as we are working at stationarity. 
Notice that  one is not required to calculate the stationary mean here. Hence,
\begin{equation}\label{hasbro}
\int^{\infty}_0T_s(\mathbf{R})ds=\breve{\mathbf{a}}^{\dag}\left(\int_0^{\infty}e^{A^{\dag}t}Xe^{At}dt\right)\breve{\mathbf{a}}:=     \breve{\mathbf{a}}^{\dag}B\breve{\mathbf{a}}     
\end{equation}
Observe that $B$ satisfies the properties of a local symplectic transformation (i.e it is hermitian and doubled-up (see \eqref{rud1})  because  $X$ does.   
 
 %\begin{remark}\label{remark}
%There are certain freedoms in the choice of the matrix $X$. For example, the choice 
%\[X'=  \left(\begin{matrix}\dot{\Omega}_1 &   \frac{\dot{\Omega}_+}{2}\\\frac{\dot{\Omega}^{\#}_+}{2}&0\end{matrix}\right)
%+\frac{1}{2i}     
%  \left(\begin{matrix}
%\left(\dot{C}_-^{\dag}C_-   -C_-^{\dag}\dot{C}_-\right)
%&
%\left(\dot{C}_-^{\dag}C_+ -C_-^{\dag}\dot{C}_+\right)    
%\\
%- \left(\dot{C}_-^{T}C_+^{\#} -C_-^{T}\dot{C}_+^{\#}\right)^T
%&
% \left(\dot{C}_+^{\dag}C_+-C_+^{\dag}\dot{C}_+\right)
%\end{matrix}
%\right)\]
%would also have worked. As a result $\mathbf{R}$ with $X$ or $X'$ differ by  $F(t)\mathds{1}$, where $F(t)$ is a matrix. The only effect of this is to alter the constant matrix in \eqref{lan4}, which as we have seen is not important. Our particular choice of $X$ was simply chosen as a matter of convenience. 
% \end{remark}
 
Now $\mathbf{L}_i=C_{-_i}\mathbf{a}+C_{+_i}\mathbf{a}^{\dag}$ where $C_{\pm_i}$ is the $i$th row of $C_{\pm}$. Define the following matrix 
  \[C_i=\left(    \begin{matrix}C_{-_i}& C_{+_i}   \\  C_{+_i}^{\#} &C_{-_i}^{\#} \end{matrix}      \right),\]
 which describes the coupling to the $i$th field.  Now from Eq. \eqref{rud}, by 
 calculating the RHS and equating coefficients of $\mathbf{a}$ and $\mathbf{a}^{\dag}$,  it follows that 
 \[\dot{C}_i=-iC_iJ2B,\]
  hence
   \begin{equation}\label{gauge1}\dot{C}=-iCJ2B.\end{equation}
   Eq. (\ref{gauge1}) is the desired gauge transformation on the coupling matrix. 
   
   To find the  Hamiltonian relation, consider $B=  \int_0^{\infty}e^{A^{\dag}t}Xe^{At}$. This  is formally  equivalent to the Lyapunov equation \cite{Ljung1}:
 \begin{equation}\label{Lyap}
 A^{\dag}B+BA+X=0.
 \end{equation}
As $A= -iJ\Omega-\frac{1}{2}JC^{\dag}JC$, it follows that $A^{\dag}=-JAJ-C^{\dag}CJ$. Thus taking the Lyapunov equation (\ref{Lyap}) and multiplying by $J$ one obtains
\begin{align*}
0&=JA^{\dag}B+JBA+JX\\
&=\left[JB, A\right]-JC^{\dag}JCJB+JX\\
&=\left[JB, A\right]+\frac{i}{2}\dot{A}.
\end{align*}
The last step here is understood by using the substitution $\dot{C}=-iCJ2B$ into $JX=\frac{1}{2}J\dot{\Omega}+\frac{1}{2}\mathrm{Im}\left(J\dot{C}^{\dag}JC\right)$.
Hence
\begin{equation}
\dot{A}=i\left[J2B, A\right]
\end{equation}
as required.

\begin{remark}
Eq. \eqref{gauge1} is written in terms of the modified coupling matrix. One can obtain the result for the unmodified coupling matrix by multiplying both sides of \eqref{gauge1} by $S^{\flat}$ (defined in Eq. \eqref{vtrick}). As the input is independent of any unknown parameter, the result will be the same as in  the modified situation.
\end{remark}
\end{proof}

Hence the infinitesimal \textit{gauge transformations} \eqref{rud1} are  indistinguishable from any measurement of the field. 
Notice that these  transformations can be generated from the unidentifiable   transformations,  
$C\mapsto CS^{\flat}, \Omega\mapsto S\Omega S^{\flat}$.  Therefore, this result says that the infinitesimal transformations in   tangent space given by zero QFI rate, are exactly those generated by the unidentifiable symplectic transformations in Theorem \ref{main}.

\begin{remark}
Notice that the transformations  \eqref{rud1}, which have been derived directly from unidentifiable directions in the QFI, don't leave the system unchanged. 
 This  further justifies  our argument   that equation \eqref{QFIMAIN} is the quantum Fisher information rate for the field only (and not the field and system).
\end{remark}

\subsection{Equivalent Expressions for the QFI in the  Stationary Approach}\label{ounce1}
The following   expression for the QFI rate can be derived from  \eqref{QFIMAIN}:
\begin{align}\label{gun1}
f_{\theta|\mathrm{time}}=4 \mathbb{E}_{\rho_{ss}}\left[   \breve{\mathbf{a}}^{\dag}D^{\dag}JV(N,M)JD\breve{\mathbf{a}}   \right],
\end{align}
where $D:=\dot{C}-iCJ2B$ with $B$ as in eq. \eqref{hasbro}. Notice that $D=0$ corresponds to the  unidentifiable directions in Theorem \ref{activeclass}. We denote the QFI rate here by $f_{\theta|\mathrm{time}}$ to signify that it was derived from the time-domain. 

We can also obtain an expression for the QFI in the frequency domain. Since all frequency modes are independent at stationarity, the QFI per unit time is given by the following Parseval's theorem type result:
\begin{equation}\label{gun2}
f_{\theta|\mathrm{freq}}=\frac{1}{2\pi}\int^{\infty}_{-\infty}f_{\theta}(\omega)d\omega,
\end{equation}
where $f_{\theta}(\omega)$ is the QFI per unit time in frequency $\omega$, i.e.,   the QFI of the Gaussian state with covariance matrix   $\Psi(\omega)$. To obtain the QFI rate on frequency $\omega$, we use the known  Gaussian state result  \cite{Ziang1}, which gives the  QFI   in terms of its moments: 
 \begin{equation}
f_{{\theta}}(\omega)=-\frac{1}{4}\mathrm{Tr}\left(J\dot{\Psi}(\omega)J\dot{\Psi}(\omega)\right).
\end{equation}
Clearly \eqref{gun1} and \eqref{gun2} must be equivalent. We show this in  an example.

\begin{exmp}
Consider the SISO quantum cavity as in example \ref{opticalcavity}, parameterised by (passive) coupling $c\in\mathbb{R}$ and Hamiltonian $\Omega$ and suppose that we would like to estimate $\Omega$.

First let us calculate $f_{\Omega|time}$ at $\Omega=0$. In this case $D|_{\theta=0}$ is given by $D=i/c$. Also, the covariance of the  stationary state of the system, $P:=\mathbb{E}_{\rho_{ss}}\left[\breve{\mathbf{a}} \breve{\mathbf{a}} ^{\dag}\right]$, can be obtained from the Lyapunov equation \eqref{eq.Lyapunov}. It turns out that 
$P=V(N, M)$
at $\theta=0$.
Notice that the stationary state of the system in this case is pure because at $\theta=0$ the system is no longer globally minimal (see Theorem \ref{sisogm}).
Therefore, 
$$f_{\Omega=0|\mathrm{time}}=\frac{16N(N+1)}{c^2}$$
(where we have use the purity condition $N(N+1)=|M|^2$ here).

On the other hand we can calculate $f_{\Omega=0|freq}$. It is a simple exercise to obtain
$$f_{\Omega=0|\mathrm{freq}}=\frac{N(N+1)}{\sqrt{2\pi}}\int^{\infty}_{-\infty} \frac{8c^4}{\left(\omega^2+\frac{1}{4}c^4\right)^2}d\omega.$$
Finally, by performing the change of variables $\omega=\frac{1}{2}\mathrm{tan}\alpha$ or otherwise, we obtain $f_{\Omega=0|freq}=\frac{16N(N+1)}{c^2}$.
\end{exmp}

%[[[COULD DO PASSIVE CALCULATIONS SEPERATELY]]]

\section{Conclusion}

Our main result is that under global minimality and pure stationary inputs the power spectrum contains as much information as the transfer function, i.e., their classes of equivalent systems are the same in both functions. Therefore, no information is lost by utilising stationary inputs rather than time-dependent inputs.  It is interesting to note that this equivalence between the power spectrum and transfer function is a consequence of the unitarity and purity of the input state, and does not hold for  general CLSs (\cite{Anders1, Glover1}). 

We also extended these results to a class of mixed inputs, that is, thermal inputs for PQLSs.
Identifiability in the case of general mixed input states (for general QLSs) remains an open question. 
The difficulty arises due to the failure of the [fully mixed stationary state $\iff$ globally minimal] result. 
%Firstly, the equivalence between global minimality and (full)-mixedness of the stationary state (i.e. Thm [REF]) no longer holds.
Nevertheless, we  expect that the transfer function can be recovered uniquely from the power spectrum generally. 
Not only does our analysis with thermal states for PQLSs supports this, but also because  the inputs here are  mixtures of pure states, on which the result holds. This problem will be a focus for us in future works.

Given that we now understand what is identifiable, the next step is to understand how well parameters can be estimated. This requires a two step approach. Firstly, one finds the best input state giving the largest QFI in \eqref{gun1} or \eqref{gun2} in terms of some resource (e.g. time or energy). Secondly, is there a simple measurement choice that enables one to attain the optimal estimation precision from the QFI? In these problems there could be one or many unknown parameters. We discuss one such estimation problem in Sec. \ref{FEDERER}.

 Lastly, it would be interesting to consider these identifiability problems in the more realistic scenario of noisy QLSs. That is the analogous problem to the one in Sec. \ref{noisek} for time-dependent inputs. 
 %In a QLS noise may be modelled by the inclusion of additional channels that cannot be monitored. Understanding what can be identified here will likely be more challenging.

%[[[TO PUT IN: Note that the work here also extends a result in \cite{Naoki}. There, conditions were derived to determine when the stationary state of the linear system is pure. Here, by means of the previous theorem, we have established a test to determine if there is a \textbf{subsystem}  with a pure stationary state. ]]]

\chapter{Quantum Enhanced Estimation of PQLSs}\label{QEEP}

The goal of this chapter is to  understand how well we can estimate the parameters governing the 
dynamics in a  PQLSs.  We consider the time-dependent approach in this chapter  (see Sec. \eqref{ghu}).
The first crucial question that one must ascertain before any statistical estimation 
can begin, is to understand what one could possibly hope to identify. 
We saw in Sec. \ref{idenf}  that two minimal systems with parameters $(\Omega,C,S)$ and $(\Omega', C',S')$ are 
equivalent if and only if their parameters are related by a unitary transformation, 
i.e. $C'=CT$ and $\Omega'=T\Omega T^{\dag}$ for some $n\times n$ unitary 
matrix $T$, and $S= S'$. Therefore, without any additional information, we can only identify the equivalence 
class of systems related by a unitary transformation as above. 
We consider now that some prior information is available, which is encoded in a 
parametrisation $\theta\to (\Omega_\theta, C_\theta, S_\theta)$ in terms of 
an unknown parameter $\theta\in\Theta\subset\mathbb{R}^d$. 
Under this assumption, the system is identifiable if each equivalence class 
contains at most one element $(\Omega_\theta, C_\theta, S_\theta)$ of the 
model (see Definition \ref{identy}).

Our objective reduces to performing a measurement on the output and 
finding an estimator $\hat{\theta}$ of $\theta$ based on the measurement outcome. 
The optimization is two-fold; over the choice of input state and the measurement. 
We identify three properties that are desirable for a `good' probe states and measurements pairs:
\begin{enumerate}

\item{Realistic states:} 
the input states can be prepared with current technology;

\item{Sensitivity:} 
the input states are sensitive to the change $\theta$;			

\item{Simple measurements:}
The scaling from (ii) may be achieved with realistic output measurements.
\end{enumerate}
The setup is similar  to the standard 	quantum metrology setup (see Sec. \ref{class}), where  one also has access to an input and performs measurements 
on the output to infer information about the black-box. 
However, the added difficulty in our setup is that we work in continuous time 
and therefore the output quantum signal can be measured to produce a 
(continuous) classical stochastic process, whereas in traditional quantum 
metrology one works in discrete time. 
This requires a different analysis of the behaviour \cite{Jacobs1}. 
As a consequence, 
%unlike the classical case, 
 the theory of \textit{quantum enhanced system identification} 
%a systematic methodology of quantum metrology 
has not yet been fully developed. 
%for PQLSs or indeed for more general \textit{active} systems. 
The purpose of chapter is to strive towards this by building on the results 
of \cite{Guta2}. 
%Guta and Yamamoto \cite{Guta}, who are two of the authors of this work. 

In the first half of this chapter (Sec. \ref{Toes1}) we provide a realistic scheme to identify a {\it single} 
unknown parameter of a PQLS at the Heisenberg limit. 
We consider an interferometric approach, where the essential idea is to detect 
an unknown variation between two almost identical systems. 
The technique used here is the same  as that used to detect gravitational 
waves \cite{LIGO1, LIGO2}. 
We show that it is possible to achieve optimal scaling using squeezed input states 
in terms of the \textit{quantum Fisher-information} \cite{Petz1}. 
The action of the interferometer is to displace a squeezed state by an amount 
proportional to the unknown parameter in a known direction in quadrature 
phase space. 
Then performing a simple homodyne measurement of this quadrature provides 
an estimate of the unknown parameter at the fundamental limit. 
The scheme is physically realistic as those non-classical states may be prepared 
with current technology. 
A physical example of a MIMO PQLS is given for estimating entanglement in two atomic ensembles.

In  Sec. \ref{Toes2} we consider the multiple parameter problem. The optimal states in  the single parameter case in terms of fixed photon number 
are (shown to be) of a single frequency. 
However, for SISO multi-parametric models the situation is more complicated as 
it is impossible to identify more than one parameter using a monochromatic 
input state. 
We discuss how to extend the interferometric approach using squeezed-coherent states to the case of multiple parameters. We see that Heisenberg scaling is also possible in this case and a simple measurement presents itself, which again comes down to detecting an unknown displacement in a known direction using homodyne measurements. 
%In the second half of the chapter (Sec. \ref{Toes2})  we consider a particular multi-parameter parameter. 
We also show that a result from \cite{Humphreys1}, i.e., that there is an $\mathcal{O}(d)$ enhancement to be had for states with a fixed number of photons by estimating the  multiple parameters 
using entangled states, can be extended to  SISO PQLSs,  Here, $d$ is the number of parameters.

\section{Quantum Enhanced Estimation of PQLSs; Single Parameter Estimation}\label{Toes1}
In this section 
let $(S,C,\Omega) = (S_\theta,C_\theta, \Omega_\theta)$ be a PQLS, whose 
dynamics depend on a one-dimensional parameter $\theta\in \mathbb{R}$. We describe the estimation precision for several choices of input state.

\subsection{Previous Results for PQLSs}\label{prevol}
\subsubsection{Product States} \label{tonga} 
Suppose that the input is given by a  pure state $|\psi\rangle$, and the energy 
is approximately spread over a finite number of frequencies, 
$\omega_1, \dots, \omega_p$ so that $|\psi\rangle$ can be represented as 
a tensor product over the chosen frequency modes (all other modes are in the 
vacuum state):
\begin{eqnarray}
   \ket{\psi}=\ket{\psi_1; \omega_1}\otimes\ket{\psi_2; \omega_2}\otimes
        \dots \otimes \ket{\psi_p;\omega_p},
\end{eqnarray}
subject to the energy constraint 
\begin{eqnarray}
E=\sum_{i=1}^{p}{}\Braket{ \psi_i;  \omega_i |  \mathbf{b}^{\dag}(\omega_i)\mathbf{b}(\omega_i)|  \psi_i; \omega_i}.
\end{eqnarray}
Due to linearity, each input frequency mode is affected by the PQLS independently 
of the others. 
Considering the optimization over input state first, the objective is to find 
a state with a large QFI, which will provide a low QCRB for the mean square 
error. 
Since the QFI is additive for product states, it follows that 
${I(\theta)}=\sum_{i=1}^p{I_i(\theta)}$, where ${I_i(\theta)}$ is the QFI 
for frequency $\omega_i$. 
Combining these facts, it follows that for one-dimensional parameter and 
product input states, the optimal input is a monochromatic signal. 
Note that the optimal frequency may depend on the unknown parameter. 
In practice one could use an adaptive procedure to tune the input frequency 
in order to obtain better estimation precision \cite{Ljung1}. 
Based on this argument, we can ignore the optimal frequency's dependence 
on $\theta$ in the limit of large input energies.

\begin{exmp}[Coherent state input \cite{Guta2}] \label{ex1}
Suppose that we probe the system with a (monochromatic) coherent state 
of amplitude $\alpha\in\mathbb{C}^m$ and frequency $\omega$, i.e. 
$\ket{\psi} = \ket{\alpha; \omega}$, with energy, $E$, given by 
$E = \|\mathbf{\alpha}\|^2$. 
The output state is obtained by rotating the amplitude vector $\alpha$ 
by the $\theta$-dependent transfer function, i.e. it is the coherent state 
$\ket{\Xi_{\theta} (-i\omega)\alpha ;\omega}$. 
Using Eq. \eqref{unit}, it follows that the QFI  is given by 
\cite{Guta2}
\begin{eqnarray}\label{tonga2}
    F(\theta|\alpha)
      = \left\|\frac{d\Xi_{\theta}(-i\omega)}{d\theta} 
            \alpha \right\|^2.
\end{eqnarray}
This is maximized when $\alpha$ is the eigenvector of the selfadjoint operator 
$L(\omega, \theta):= i\Xi_{\theta}^*(-i\omega)\cdot \mathrm{d}\Xi_{\theta}(-i\omega)/\mathrm{d}\theta$ whose eigenvalue 
has the largest absolute value. Thus the optimum QFI over frequencies and 
coherent amplitudes is  
$$
     F(\theta) 
         = E \cdot \sup_\omega \left\|
                  L(\omega, \theta)
                        \right\|^2.
$$
Since the QFI scales linearly in $E$, this is the \textit{standard scaling} or 
``shot noise" regime. 
Furthermore, it can be shown that by measuring an appropriate output 
quadrature one can achieve the above scaling for the optimal frequency, 
essentially by following the techniques of \cite{Wiseman3}.
\end{exmp}

%%%%%%%%%%%%%%%%%%%%%%%%%%
\subsubsection{Non-Gaussian States.}\label{cats}
%%%%%%%%%%%%%%%%%%%%%%%%%%

We now replace the coherent state of the previous example by more general state supported by a single monochromatic input mode of frequency $\omega$. In the space of input modes $\mathbb{C}^m$ we choose the mode corresponding to the eigenvalue of 
$L(\omega, \theta)$ with the largest absolute value. We prepare this mode in a ``cat state'' \cite{Monroe1} , i.e. a coherent superposition of Fock states of energy $E$
$$
|CAT; \omega \rangle:=\frac{|0;\omega \rangle + |N ;\omega\rangle}{\sqrt{2}}, \qquad N= 2E
$$ 
The QFI is  \cite{Guta2}
$$
       F(\theta) 
            = 4 E^2 \cdot \left\|
                    L(\omega, \theta)
                       \right\|^2,
$$
which exhibits a \textit{quadratic or Heisenberg scaling} in energy. 
Therefore these states are optimal for monochromatic inputs with a fixed 
number of photons.
% (by using a similar argument to that in section \ref{class} 
%with generator  $G(i\omega)=\mathbf{a}^*(i\omega)\mathbf{a}(i\omega)$). 
 There are several (physical and mathematical) problems associated with this choice of input states. Firstly, it is not clear how to achieve the Heisenberg scaling with a realistic measurement.
Secondly, the output signal has a period of $2\pi/N$ with respect to $\theta$. A consequence of this is that one will only be able to determine the phase modulo $2\pi/N$. In order for the result to be unambiguous, it would appear that the phase must already be localised within a $2\pi/N$ interval beforehand (i.e. the variance of the prior distribution is of order $\mathcal{O}(1/N^2)$). 
%In other words, in order to get $\mathcal{O}(1/N^2)$ precision one must already have  $\mathcal{O}(1/N^2)$ precision. 
This situation may be resolved by adaptive procedures based on varying the number of photons in the input state \cite{Berry1}. Finally, these highly non-classical states are difficult to create in practice; with present technology they are limited to very small $N$ \cite{Afek1}.

\subsubsection{Interferometric Approach}\label{inter1}

We now briefly discuss an alternative setup to the one in Sec. \ref{cats} for the standard metrology protocol, 
as it will be useful for our estimation method. 
Consider the setup in Fig. \ref{c26} where we estimate the phase difference between the two arms of the interferometer. For a fixed number of particles $N$, the optimal states to use are  \textit{N00N states} 
\cite{Bollinger1} for which the QCRB reaches the fundamental limit on precision given by the Heisenberg limit $1/N^2$. 
Moreover, it can be shown that the QCRB is attained by performing 
photon counting measurements \cite{Berry1}.  Just like ``cat states'', the N00N states have a large QFI is because of the 
large number variance between the two interferometer modes, but are similarly difficult to prepare in practice. 

\begin{figure}
\centering
\includegraphics[width=80mm, height=50mm]{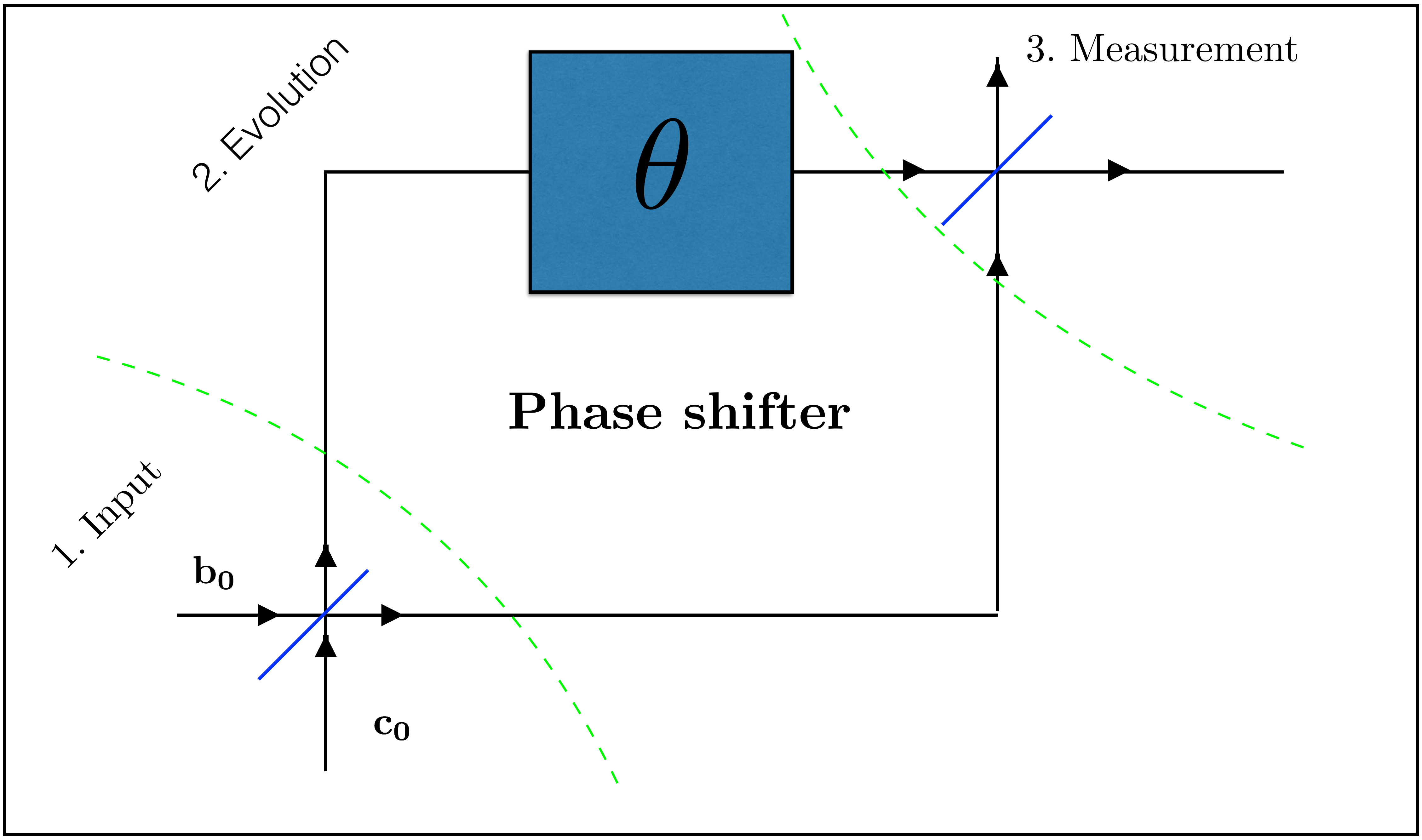}
\caption{
This figure shows the quantum metrology interferometric set-up. 
Two input channels are fed into the interferometer, acquiring a phase 
difference between the two beam-splitters. 
This phase difference is the result of a SISO PQLS. There are three stages to the interferometer, an input stage, an evolution 
stage and a measurement stage.
\label{c26}}
\end{figure}

%It is worth mentioning a parallel setup to the one in Section \ref{cats} for the standard metrology protocol as it will be useful for our estimation method. Consider the setup in Figure \ref{c26} where we have are trying to estimate the phase difference between the two arms of the interferometer. It turns out that     
%for a fixed number of particles, $N$, the fundamental limit on precision with 
%respect to the QFI figure of merit for phase shift, has been proven to be 
%$N^2$  (Heisenberg limit) just like the CAT state. Moreover, the optimal states to use are  \textit{N00N states} 
%\cite{Bollinger1}. 
%It can be shown that the QCRB is attained by performing 
%photon counting measurements \cite{Berry1}.  Just like CAT states, the reason that N00N states  have such a large QFI is because of the 
%large number variance between the two interferometer modes.
%However, also like the CAT states, the creation of N00N states is very difficult. 
%are the interferometric analogue of cat-states and are thus 
%problematic for the same reasons as those were not (except for the measurement). 

%%%%%%%%%%%%%%%%%%%%%%%%%%%%%%%%%%%%%%%%% 
\subsubsection{States Entangled in the Frequency Domain?}
%%%%%%%%%%%%%%%%%%%%%%%%%%%%%%%%%%%%%%%%%

Let us discuss briefly  the possibility of using input states entangled 
over many frequencies. 
We have seen above that monochromatic states are optimal for coherent 
inputs.  In Appendix \ref{frequent} we argue for fixed photon input states that  there is little advantage in considering frequency-entangled inputs. 
Based on this evidence and the fact that this paper focuses on realistic metrology, 
we will assume that all inputs are monochromatic for one parameter models. 
Clearly, multiple frequencies must be considered in finding optimal 
states in the multi-parameter case (cf. Sec. \ref{Toes2}).

%%%%%%%%%%%%%%%%%%%%%%%%%
\subsubsection{Indefinite Photon-Number States}\label{indefinite}
%%%%%%%%%%%%%%%%%%%%%%%%%

Here we consider using  indefinite photon-number states \cite{Rafal1}
 as a resource for quantum 
metrology. 
The difference with the above is that we fix only the average 
photon number (or equivalently average energy) rather than demand that 
 photon number is fixed.

%Consider a SISO PQLS, which for a fixed frequency input is equivalent to the  standard metrology protocol (Sec. \ref{plm}). 
%Any advantage in estimation precision from relaxing the 
 %the available class of states to allow for indefinite photon-numbers will come from coherences between 
%sectors having different numbers of particles. 
%Therefore, photon-number measurements will be of no use in this case as
%such coherences would be destroyed. 
%Alternatively, one can employ homodyne measurements; however, finding the optimal input state while 
%including the reference beam  (\textit{local oscillator}) in the resource counting proves to be difficult \cite{Rafal1}.
% use an additional reference 
%beam (called a \textit{local oscillator}); this beam ought to be included when counting resources. 
%It is using this reference beam that leads to the difficulties in determining the 
%optimal indefinite photon-number state (see \cite{Rafal1}). 
%%Suffice to say finding the optimal indefinite photon-number state is a difficult 
%%challenge. 
%It therefore emerges more practical to analyse experimentally accessible 
%Gaussian states. 
%Furthermore, in \cite{Berry2} the hypothesis of whether or not it is possible to achieve 
%sub-Heisenberg (better than $\mathcal{O}(1/N^2)$) scaling using indefinite 
%photon-number states for quantum metrology was investigated, and it was found that the maximum 
%improvement that any indefinite photon-number state would bring is a constant 
%factor. 

\subsection{Estimation of a One-Dimensional Parameter in a PQLS}
\label{sec.1D}

In this section we will describe an  interferometric approach  
for identifying the unknown parameter in our system $(S_\theta,C_\theta, \Omega_\theta)$. Crucially our choice of input state will have all of  the desirable features suggested above. 
  The scheme is inspired by a phase estimation scheme proposed in \cite{Caves1}. We will show that the performance of these states in terms of sensitivity are as good as the best fixed particle number strategy (N00N states). We first treat the case of SISO systems after which we extend to the case of MIMO systems.

\subsubsection{The SISO Case}\label{bon}

Consider the  interferometric setup as in Sec. \ref{inter1}, where we have replaced the standard phase shifter with our SISO PQLS $(S_\theta,C_\theta, \Omega_\theta)$. Consider also feeding two (quasi)-monochromatic light pulses of frequency $\omega$  into the input channels, whose frequency-domain modes are denoted by $\mathbf{b}_0$ and $\mathbf{c}_0$; the first pulse is a coherent state 
$\ket{\alpha}= \exp(i (\alpha \mathbf{c}_0+ \bar{\alpha} \mathbf{c}^*_0)|\Omega\rangle$ of amplitude 
$\alpha:=|\alpha| e^{i\psi}\in\mathbb{C}$ while the second is a squeezed state 
$\ket{S^{\phi}_r} := \exp(i\phi \mathbf{b}^*_0 \mathbf{b}_0)\exp(r (\mathbf{b}_0^2- \mathbf{b}_0^{*2})/2)|\Omega\rangle$, 
where $\phi$ and $r$ represent the angle and respectively magnitude of squeezing. The (average) energy of the input is given by 
$E:=
\langle S^\phi_r  | \mathbf{b}_0^*\mathbf{b}_0|  S^\phi_r \rangle +  
\langle \alpha | \mathbf{c}_0^* \mathbf{c}_0 | \alpha \rangle=\sinh^2(r)+ |\alpha|^2$. 

\begin{remark}
%Recall that  N00N states have enhanced sensitivity because of the 
%maximal particle entanglement between the arms of the interferometer 
%\cite{Rafal1}. 
It can be seen \cite{Giov1, Hofmann1} that the $N^{\mathrm{th}}$-particle component of the output of the first beam splitter for these squeezed/ coherent input states have a very large N00N component, which is essentially why they turn out to be so advantageous for quantum metrology.
\end{remark}

Firstly, these states are  physically realistic. The creation and use of such states for metrology purposes has been demonstrated in \cite{LIGO1, LIGO2, Pitkin1, Rafal1}.           
Rather than producing a squeezed state of a given frequency, a more 
realistic approach is to send a squeezed state containing a continuum of 
frequencies and post-process (take the Fourier-transform) the output signal 
of the measured time-domain momentum quadrature. 
Choosing a larger bandwidth of frequencies will result in a shorter experiment 
time but will decrease the energy of a given frequency. 
The frequency domain profile of the input state can be shifted so 
that it centres on a particular frequency \cite{Gardner1}.
One is also able to tune the degree and direction of squeezing in the frequency 
domain \cite{Lvovsky1}. 

Since we deal with a SISO system, the transfer function at frequency $\omega$ is a complex phase  
$ \Xi_{\theta}(-i\omega):=e^{-i\lambda(\omega, \theta)}$. Taking into account the action of the two beamsplitters we find that the input-output map in the Schr\"{o}dinger picture is (see for example \cite{Rafal1}) 
\begin{equation}\label{scar}
|S^\phi_r \otimes \alpha \rangle \mapsto 
e^{-i\lambda(\omega,\theta) 
\left(\mathbf{b}_0^* \mathbf{b}_0 -\mathbf{c} _0^*   \mathbf{c}_0\right)/2 } | S^\phi_r \otimes \alpha\rangle
=:\ket{\psi_{\theta}}
\end{equation}
Now, using  formula \eqref{unit} we find that the QFI depends only on the difference $\psi-\phi$, so without loss of generality we may set $\phi=0$; in this case, the QFI is maximized when $\psi=0$ and is explicitly given by 
${F(\theta)}=|\partial\lambda(\omega, \theta)/\partial\theta|^2
\left(\alpha^2e^{2r}+\sinh^2{r}\right)$ \cite{Rafal1}. 
This situation corresponds to momentum-squeezing in one arm and position 
displacement in the other. We note that putting all of the energy into one of the two arms 
(i.e. either $\alpha=0$ or $r=0$) will result in ${F(\theta)}\propto E$, which 
is the standard precision scaling. 
The optimal QFI subject to the energy constraint $E=|\alpha|^2+\sinh^2(r)$ is achieved when equal energy is put into both arms: $\alpha^2=\sinh^2(r)=E/2$; then 
\begin{equation}\label{put1235}
      {F(\theta)}=\Big|\frac{\partial\lambda(\theta)}{\partial\theta}\Big|^2
                              \Big( \frac{E}{2}\cdot e^{2r} + \frac{E}{2} \Big)
                  \approx 
                  \Big|\frac{\partial\lambda(\theta)}{\partial\theta}\Big|^2E^2. 
\end{equation}
Therefore, this strategy is (asymptotically) as good as the one using N00N states \cite{Rafal1}.

Finally we will see that this QFI scaling can be achieved experimentally by using a homodyne 
measurement. Since we are interested in the asymptotic regime of large energy, we assume that the parameter 
$\theta$ is known up to at least the standard uncertainty $E^{-1/2}$, so we can write $\theta= \theta_0 + h/\sqrt{E}$ 
where $\theta_0$ is known and  $h$ is an unknown ``local parameter". The reference parameter 
$\theta_0$ could be obtained by using iterative adaptive procedures while using a small proportion of the input energy \cite{Gill1}.
Using the knowledge of $\theta_0$, we slightly modify the set-up by adding a second linear system in 
the other arm of the interferometer, denoted by $\theta_0$ in Fig. \ref{daf}. This can be a static phase rotation element with phase $e^{-i\lambda(\omega, \theta_0)}$, or another PQLS with matrices 
$(S,C,\Omega) = (S_{\theta_0},C_{\theta_0}, \Omega_{\theta_0})$, and its purpose is to ``balance" the action of the PQLS and simplify the structure of the final measurement as shown below. Later, we see another instance of balancing the action of the QLS when we discuss \textit{quantum absorbers}  in Ch. \ref{DUAL}.
\begin{figure}
\centering
\includegraphics[width=90mm, height=55mm]{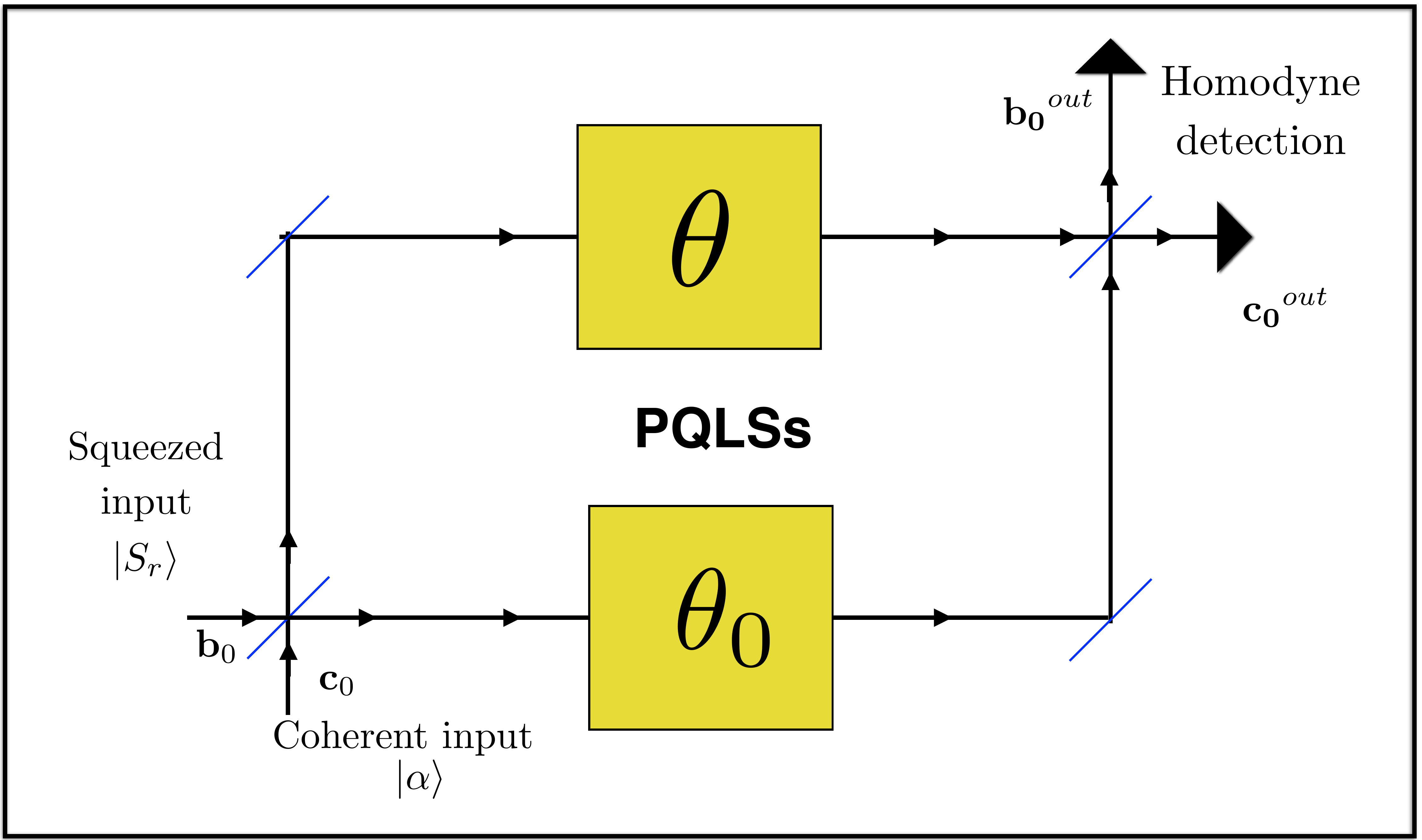}
\caption{Figure showing the interferometric set-up for SISO PQLS with 
a squeezed state $\ket{S_r}$ incident upon one arm of the interferometer 
and a coherent state $\ket{\alpha}$ in the other. 
The set-up is the same as before except that a phase shifter or a second known PQLS has been 
placed in the other interferometric arm.
\label{daf}}
\end{figure}
Now, the output operators of the 
interferometer  are given by
\begin{eqnarray*}
      &\mathbf{b}^{out}_0
            %:=\mathbf{b}_2
            =\frac{1}{2}\Big(\mathbf{b}_0(e^{-i\lambda(\omega,\theta)}+e^{-i\lambda(\omega,\theta_0)})
                            +\mathbf{c}_0(e^{-i\lambda(\omega,\theta)}-e^{-i\lambda(\omega,\theta_0)})\Big),
\\\ &
      \mathbf{c}^{out}_0
            %:=\mathbf{c}_2
            =\frac{1}{2}\Big(\mathbf{b}_0(e^{-i\lambda(\omega, \theta)}-e^{-i\lambda(\omega,\theta_0)})
                             +\mathbf{c}_0(e^{-i\lambda(\omega,\theta)}+e^{-i\lambda(\omega, \theta_0)})\Big).
\end{eqnarray*}
%
%Recall that the transfer function contains the maximum amount of information about the PQLS, therefore, it is clear that one cannot observe $\theta$ directly but instead $\lambda(\theta)$ (and then an estimate of $\theta$ may be obtained from this). 
%
By  expanding $e^{-i\lambda(\omega, \theta)} = e^{-i\lambda(\omega,\theta_0)} (1-i h\lambda^\prime_\theta(\omega,\theta_0)/\sqrt{E}) + O(E^{-1})$ 
we find that in the optimal setting of equal energy in the input channels ($|\alpha|^2=E/2$) the action of the output annihilator is 
$$
\mathbf{b}^{out}_0\ket{S_r}\otimes\ket{\alpha}\approx
e^{-i\lambda(\omega, \theta_0)}
\left(\mathbf{b}_0-i h\lambda^\prime_\theta(\omega,\theta_0) \frac{\alpha}{2\sqrt{E}}\mathds{1}\right)\ket{S_r}\otimes\ket{\alpha}.
$$ 
 Hence when the energy is large we have
$
       \mathbf{b}^{out}_0\approx e^{-i \lambda(\omega,\theta_0)}
           \left( \mathbf{b}_0-i\frac{h\lambda^\prime_\theta(\omega, \theta_0)\alpha}{2\sqrt{E}}\mathbf{1}\right)
$
and the action of the Mach-Zehnder interferometer is to first displace the squeezed mode 
$\mathbf{b}_0$ along the momentum axis by $h\lambda^\prime_\theta(\omega, \theta_0)\alpha/ \sqrt{2E}$, followed by a 
$e^{-i \lambda(\omega, \theta_0)}$ phase rotation. %(see Figure \ref{spk}). 
Let $X$ denote the outcome of measuring the following quadrature of the mode $\mathbf{b}_0$:
\[
      \mathbf{X}_{\pi/2+\lambda(\omega,\theta_0)}
%           :=\cos\left(\pi/2+\lambda(\theta_0)\right)\mathbf{Q}
%                     +\sin(\pi/2+\lambda(\theta_0))\mathbf{P}
            =-\sin(\lambda(\omega,\theta_0))\mathbf{Q}+\cos(\lambda(\omega,\theta_0))\mathbf{P}.
\]
As $\mathbb{E}_\theta\left[ X\right]= - h\lambda^\prime_\theta(\omega,\theta_0) / {2}$ we choose the (locally) unbiased estimator of 
$\theta$ given by  
$$
\hat{\theta}=\theta_0 -\frac{2 }{\lambda^\prime_\theta(\omega,\theta_0) \sqrt{E}}  X
$$ 
and obtain the mean square error
\begin{equation}\label{tofu}
      \mathbb{E}\left[ (\hat{\theta} -\theta)^2\right]
          =\frac{4}{E\lambda^\prime_\theta(\omega,\theta_0)^2}\mathrm{Var}\left( X\right)
          =\frac{2
       }{E \lambda^\prime_\theta(\omega,\theta_0)^2 e^{2r}}\approx\frac{1}{ \lambda^\prime_\theta(\omega,\theta_0)^2 E^2},
\end{equation}
hence the Heisenberg-limited estimation is achieved.

%%%%%%%%%%%%%%%%%%%%%%%%%%%%%%%%%%%%%%%%%%%
%%%%%%%%%%%%%%%%%%%%%%%%%%%%%%%%%%%%%%%%%%%
\subsubsection{The MIMO Case}
\label{boni}

In this section we extend the SISO case to  MIMO PQLSs . 
The strategy that we use for MIMO is a direct extension of the SISO 
strategy and is illustrated in Fig. \ref{lok}. The  
input channels $\mathbf{b}=[\mathbf{b}_1, \ldots, \mathbf{b}_m]^T$ are prepared in a monochromatic 
multi-mode squeezed state \cite{Adesso1} of frequency 
$\omega$; these channels are mixed by means of separate 50:50 beamsplitters with the corresponding ancilla channels
$\mathbf{c}=[\mathbf{c}_1, \ldots, \mathbf{c}_m]^T$ prepared in coherent states of the same frequency. 
One of the outputs of each beam splitter is then passed through the PQLS and the other through a second known PQLS (to be specified shortly) before 
recombining at  another 50:50 beam-splitter.

\begin{figure}
\centering
\centering\includegraphics[scale=0.20]{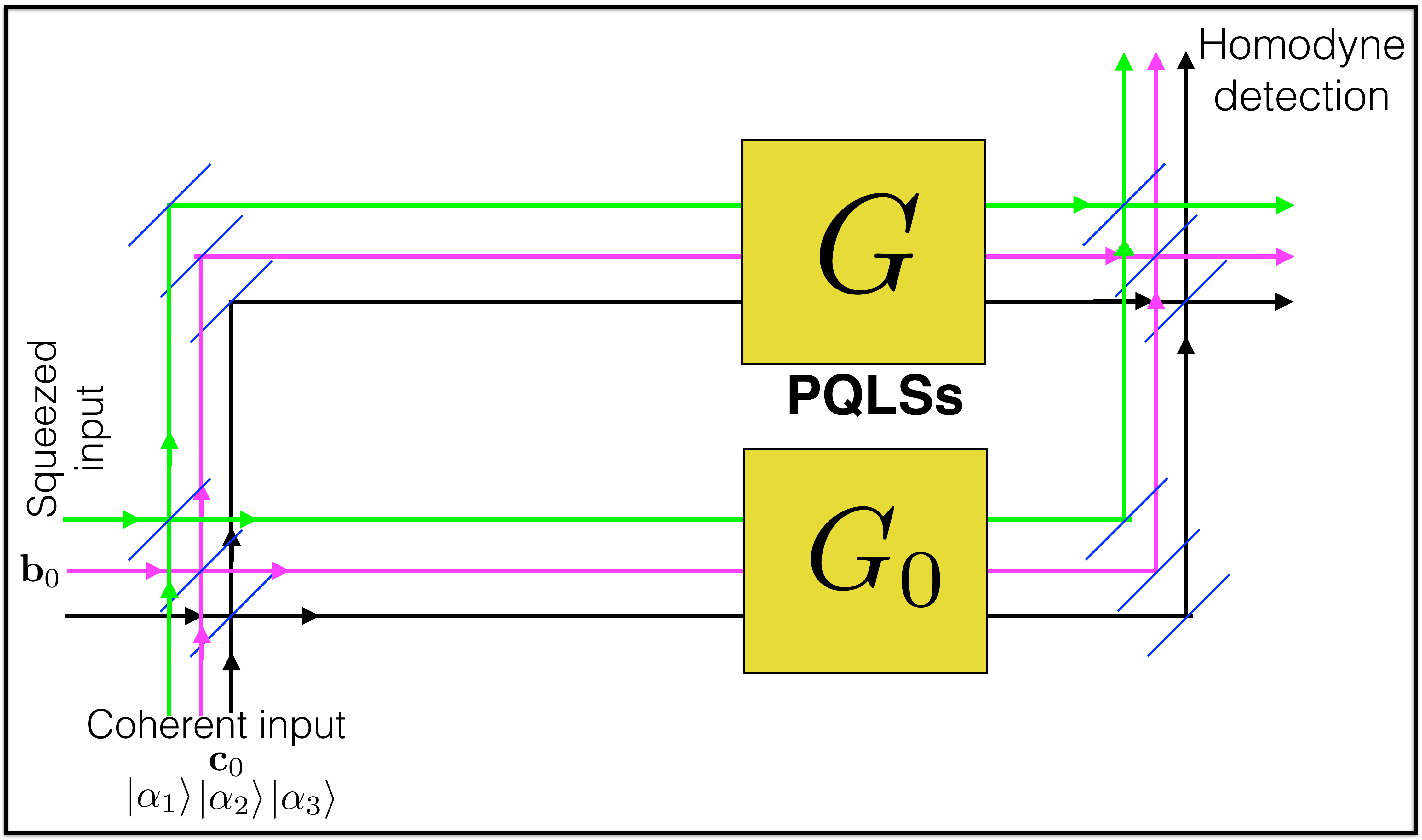}
\caption{This figure shows the  MIMO PQLS system identification setup. 
In this diagram the number of field modes, $m$, is equal to three.  
\label{lok}}
\end{figure}

Just as in the SISO case, we assume that we already computed a rough estimate $\theta_0$ with standard deviation of order 
$1/\sqrt{E}$, such that the parameter can be written as 
$\theta = \theta_0 + h/\sqrt{E}$ for some local parameter $h$. This could be achieved for instance by performing appropriate 
homodyne measurements on the output for a monochromatic input. 
Then the transfer function can be expanded as 
$$
\Xi_{\theta}(-i\omega) = \Xi_{\theta_0}(-i\omega) \left(\mathds{1}_m-ihL(\omega,\theta_0)/\sqrt{E}\right)+\mathcal{O}\left(E^{-1}\right)
$$ 
%Could explain expansion in modes
where $L(\omega,\theta_0)$ is a $m\times m$ Hermitian matrix.
The second system to be placed in lower arm of the interferometer (see Fig. \ref{lok}) is either the PQLS 
$\left(S_{\theta_0}, C_{\theta_0}, \Omega_{\theta_0}\right)$ or a static configuration of beam splitters giving the phase element 
$\Xi_{\theta_0}(-i\omega)$. It follows that the $2m$ input channels transform according to
$$
    \left[\begin{array}{c}
       \mathbf{b}_0^{out}\\
       \mathbf{c}_0^{out} 
     \end{array}\right] \approx
     \left[\begin{array}{cc}
           \Xi_{\theta_0} & 0\\
           0  &  \Xi_{\theta_0} 
      \end{array}\right]
      \left\{\mathds{1}_{2m}-\frac{i h }{2\sqrt{E}}\left[\begin{array}{cc}
           L(\omega,\theta_0)&  L(\omega,\theta_0)\\
          L(\omega,\theta_0)  &  L(\omega,\theta_0)
      \end{array}\right] \right\}
        \left[\begin{array}{c}
          \mathbf{b}_0\\
          \mathbf{c}_0
        \end{array}\right].
$$
where $\mathbf{b}_0^{out} =\mathbf{b}^{out}(\omega)$ and $\mathbf{c}_0^{out} =\mathbf{c}^{out}(\omega)$ are the frequency modes we are interested in.
As in the SISO case we prepare the probe mode ${\bf b}_0$  in a squeezed state and ${\bf c}_0$
in a coherent state 
with amplitude 
$\mbox{{\boldmath $\alpha$}}=[\alpha_1, \ldots, \alpha_m]^T$ such that $\alpha_i=\mathcal{O}(\sqrt{E/2m})$.  Then
\begin{equation}\label{pfg}
     \mathbf{b}_0^{out}\ket{S_r}\otimes\ket{\mathbf{\alpha}}\approx \Xi_{\theta_0}\left\{\mathbf{b}_0
        - \frac{ih}{2\sqrt{E}} L(\omega, \theta_0)\mbox{{\boldmath $\alpha$}}\mathds{1}
%        \left[\begin{array}{c}
%    L_{11}\alpha_1 + \ldots + L_{1m}\alpha_m\\
%   
%       \vdots\\
%    L_{m1}\alpha_1 + \ldots + L_{mm}\alpha_m
%      \end{array}\right]
%      
      \right\}\ket{S_r}\otimes\ket{\mathbf{\alpha}}. 
\end{equation}
which means that the local parameter $h$ is imprinted into the output state via a displacement of the squeezed mode ${\bf b}_0$.
%%%%%%%%%%%%%%%%%%%%%%%%%%%%%%

Now, let us find the optimal state of squeezed-coherent type. Since $\Xi_{\theta_0}$ is known, we will ignore its effect which can be undone by performing a phase transformation on the output modes. The output state is therefore a displaced squeezed state whose mean is given by 
$$
\langle {\bf Q} \rangle = \frac{h}{\sqrt{2E}} {\rm Im} \left[ L(\omega,\theta_0)\mbox{{\boldmath $\alpha$}}\right] =\frac{h}{\sqrt{E}} \mu_q, \quad 
\langle {\bf P} \rangle = \frac{-h}{\sqrt{2E}} {\rm Re} \left[L(\omega,\theta_0)\mbox{{\boldmath $\alpha$}}\right] =\frac{h}{\sqrt{E}}  \mu_p
$$
where ${\bf Q}$ and ${\bf P}$ are the vectors of canonical coordinates of the $m$ modes. We denote by $V$ the covariance matrix 
$$
V := {\rm Re}\left(
\begin{array}{cc}
\langle {\bf Q}   {\bf Q}^T \rangle & \langle {\bf Q} {\bf P}^T \rangle\\
\langle {\bf P}{\bf Q}^T \rangle & \langle {\bf P}{\bf P}^T \rangle
\end{array}
\right)
$$
The QFI for the parameter $\theta= \theta_0 + h/\sqrt{E}$ is given by \cite{Monras1}
$$
I(\theta_0) =\mu^T V^{-1}\mu, 
\qquad 
\mu := \left(
\begin{array}{c} 
\mu_q\\
\mu_p
\end{array}
\right)= 
\frac{1}{\sqrt{2}}
 \left(
\begin{array}{c} 
{\rm Im} \left[ L(\omega,\theta_0)\mbox{{\boldmath $\alpha$}} \right]\\
{\rm Re} \left[ L(\omega,\theta_0)\mbox{{\boldmath $\alpha$}} \right]
\end{array}
\right)
$$
%We write $-iL(\omega,\theta_0)\mbox{{\boldmath $\alpha$}}/2$ as $[\tilde{L}_1, ... \tilde{L}_m]^T$ and denote the (Schr\"{o}dinger picture) output state by $\ket{\psi_{\theta}}$. Firstly, since the output state is a pure displaced-squeezed Gaussian state where all  information about $\theta$ is contained in the displacement, it follows from \cite{Monras} that 
%$$
%F(\theta_0) =2 \mu^T V^{-1}\mu,
%$$
%where $\mu$ is the displacement in (\ref{pfg}), given by
%\begin{eqnarray}
%d:&&=\left[ \bra{\psi_{\theta}}\mathbf{X}_1  \ket{\psi_{\theta}},  \bra{\psi_{\theta}}\mathbf{P}_1\ket{\psi_{\theta}},...,  \bra{\psi_{\theta}}\mathbf{X}_m \ket{\psi_{\theta}},  \bra{\psi_{\theta}}\mathbf{P}_n \ket{\psi_{\theta}}\right]^T
%\\&&=\left[\mathcal{R}(\tilde{L}_1), \mathcal{I}(\tilde{L}_1),..., \mathcal{R}(\tilde{L}_m), \mathcal{I}(\tilde{L}_m)\right]^T,
%\end{eqnarray}
%and $V$ is the covariance matrix of the output, given by
%$$V=\bra{\psi_{\theta}}
%\left(  \begin{array}{c} X_1\\P_1\\\vdots\\X_m\\P_m\end{array}\right)    \left( \begin{array}{ccccc} X_1&P_1&...&X_m&P_m\end{array}\right)  \ket{\psi_{\theta}}.  
%$$

We will now argue that the best strategy is to squeeze in one mode (or quadrature) only. Since Fisher information is independent of 
the chosen basis in the space of modes, we will choose the latter to be the eigenbasis of the selfadjoint matrix $L(\omega, \theta_0)$ such that 
$(1,0,\dots , 0)$ is the eigenvector whose eigenvalue has the largest absolute value equal to $\|L(\omega, \theta_0)\|$. Then the following inequality holds
$$
I(\theta_0) \leq \|V^{-1}\| \cdot \|\mu\|^2 \leq \|V^{-1}\|  \cdot \|L(\omega, \theta_0)\|^2 \cdot \| \mbox{{\boldmath $\alpha$}} \|^2/2.
$$
The second inequality becomes equality by setting $\mbox{{\boldmath $\alpha$}} = \sqrt{E_c} (1,0\dots , 0)$ where $E_c$ is the energy of the coherent state.  The first inequality is saturated by choosing $V$ to be the covariance of a squeezed state in which the first mode is squeezed along the $\mathbf{Q}$ quadrature, while all other modes are in the vacuum. In the leading order in the energy of the squeezed state we have $\|V^{-1}\| = 8{E_{sq}}$, and by imposing the constraint $E_c+ E_{sq} = E$ we find that the optimal energy distribution is 
$E_c= E_{sq} = E/2$. Therefore the maximum value of the QFI is 
\begin{equation}\label{put14}
I(\theta_0) = {E^2} \|L(\omega,\theta_0)\|^2 
\end{equation}
Finally, the left hand side can be further optimised over the frequency $\omega$ to obtain the highest QFI in this setting. As in the SISO case, the QFI can be achieved by measuring the displaced squeezed quadrature. 

\begin{remark}
The expression \eqref{put14} is the QFI from the modes $\mathbf{b}^{out}_0$ only (rather than $\mathbf{b}^{out}_0$ and $\mathbf{c}^{out}_0$), however it is equal to   \eqref{put1235}, which is the QFI from both channels. Note also  the consistency between the expressions \eqref{tofu} and  \eqref{put14}, corresponding to the SISO and MIMO cases.
\end{remark}
\subsection{Example: Two Atomic Ensembles}\label{Macroscopic}
%%%%%%%%%%%%%%%%%%%%%%%%%%%%%%%%%%%%%%%

Consider the two coupled atomic ensembles setup from  example \ref{ensemble}. In particular, assume that we would like to investigate $\theta$, which we assume to be small.

First, the transfer function matrix of this system, at $s=-i\omega$, is 
given by 
\begin{eqnarray}
& & \hspace*{-2em}
      \Xi(-i\omega) 
                 = I - C(-i\omega I-A)^{-1}C^\dagger
                 = I - \kappa \sqrt{Y} \Big( -i\omega I 
                               + \frac{\kappa}{2} Y \Big)^{-1} \sqrt{Y}
\nonumber \\ & & \hspace*{1.25em}
        = \frac{1}{\kappa^2/4 -i\omega\kappa \cosh(2\theta) -\omega^2}
           \left[ \begin{array}{cc}
             -\omega^2 - \kappa^2/4 & -i \omega \kappa \sinh(2\theta) \\
             -i\omega \kappa \sinh(2\theta) & -\omega^2 - \kappa^2/4 \\
           \end{array} \right].
\nonumber
\end{eqnarray}
If $\theta=0$, we have 
\[
      \Xi_0(-i\omega) 
        = \frac{i\omega + \kappa/2}{i\omega - \kappa/2}
           \left[ \begin{array}{cc}
             1 & 0 \\
             0 & 1 \\
           \end{array} \right],
\]
which can be represented as 
\[
      \Xi_0(-i\omega) = e^{-iG_0(\omega)},~~
      G_0(\omega)
        =  \left[ \begin{array}{cc}
               g_0(\omega) & 0 \\
               0 & g_0(\omega) \\
            \end{array} \right],
       \]
with $g_0(\omega)=2\arctan(-2\omega/\kappa)-\pi$. 
This corresponds to the transfer function matrix of two independent cavities 
coupled to probe fields with strength $\kappa$.
Note that the small unknown parameter $\theta$ brings a small deviation 
from $\Xi_0$ and yields $\Xi$ in the form 
\[
      \Xi(-i\omega) = \Xi_0(-i\omega) \tilde{\Xi}(-i\omega).
\]
By calculating $\Xi_0^{-1}\Xi$, we obtain 
$\tilde{\Xi}(-i\omega)=e^{-i\tilde{G}(\omega)}\approx I-i\tilde{G}(\omega)$ 
and find 
\[
     \tilde{G} = -f(\omega)\theta 
            \left[ \begin{array}{cc}
               0 & 1 \\
               1 & 0 \\
            \end{array} \right],~~~ 
     f(\omega)=\frac{2\kappa \omega}{\omega^2 + \kappa^2/4}. 
\]

To identify $\theta$ we use the interferometric technique presented above, 
illustrated in Fig. \ref{lok}. 
In our context the nominal system is given by the independent pair of 
cavities with generator $G_0$. 
Now the input fields into the interferometer are given by 
$\mathbf{b}=[\mathbf{b}_1, \mathbf{b}_2]^T$ and 
$\mathbf{c}=[\mathbf{c}_1, \mathbf{c}_2]^T$; 
as described before, $\mathbf{b}_1~(\mathbf{b}_2)$ and 
$\mathbf{c}_1~(\mathbf{c}_2)$ are combined at the 
first beam splitter, hence they have the same frequencies. 
The output fields of the interferometer are related to $\mathbf{b}$ and 
$\mathbf{c}$ as 
\[
      \left[ \begin{array}{c}
               \mathbf{b}^{out} \\
               \mathbf{c}^{out} \\
            \end{array} \right]
       = e^{-iG_0}\Big\{ 
            I_2 - \frac{i}{2}
              \left[ \begin{array}{cc}
               \tilde{G} & \tilde{G} \\
               \tilde{G} & \tilde{G} \\
            \end{array} \right] \Big\}
          \left[ \begin{array}{c}
               \mathbf{b} \\
               \mathbf{c} \\
            \end{array} \right]. 
\]
Thus by setting the $\mathbf{c}$ modes to coherent fields with amplitudes 
$\mbox{{\boldmath $\alpha$}} = [\alpha_1, \alpha_2]^T$, in good 
approximation we have 
\[
     \mathbf{b}^{out} = e^{-i G_0}
        \Big( \mathbf{b} - \frac{i}{2}\tilde{G}\mbox{{\boldmath $\alpha$}}\mathds{1} \Big). 
\]
Now because $\tilde{G}$ has elements only in the off-diagonal terms, we 
should set: 
\begin{itemize}
\item
$\mathbf{b}_1$ to be a momentum-squeezed vacuum field with squeezing level $r$, 
while $\mathbf{b}_2$ to be a vacuum, 
\item
$\mathbf{c}_2$ to be a coherent field with amplitude $\alpha_2=\alpha$, while 
$\mathbf{c}_1$ to be a vacuum, 
\item
and measure $\mathbf{b}_1^{out}$. 
\end{itemize}
 Interestingly, in this example the best input is one where energy is inputted into only one arm of each of the two interferometers. 
Actually we then have 
\[
     \mathbf{b}_1^{out} 
       = e^{-ig_0(\omega)} \Big(\mathbf{b}_1 + \frac{i}{2} f(\omega) \alpha \theta\mathds{1} \Big), 
\]
implying that we can estimate $\theta$ by measuring the phase-shifted 
$\mathbf{P}$-component of $\mathbf{b}_1^{out}$. 
But before writing the output signal equation for identification we should note 
that $f(\omega)^2$ takes the maximum value $f(\omega^{\rm opt})^2=4$ 
at $\omega^{\rm opt}=\pm \kappa/2$; 
this is the optimal measurement frequency we should take. 
Thus the optimal homodyne measurement is given by setting the phase of 
the local oscillator to $g_0(\omega^{\rm opt})$, and it generates the following 
signal for estimating $\theta$: 
\[
\hspace*{-4em}
     y = \frac{\sqrt{2}}{f(\omega^{\rm opt})\alpha} \mathbf{P}_1^{(g_0)}
        = \frac{\sqrt{2}}{f(\omega^{\rm opt})\alpha} \cdot
           \frac{e^{ig_0(\omega^{\rm opt})} \mathbf{b}_1^{out} - 
                       e^{-ig_0(\omega^{\rm opt})} (\mathbf{b}_1^{out})^*}{\sqrt{2}i}
        = \frac{\sqrt{2}}{f(\omega^{\rm opt}) \alpha} \mathbf{P}_1 + \theta,
\]
where $\mathbf{P}_1=(\mathbf{b}_1-\mathbf{b}_1^*)/\sqrt{2}i$.
\footnote{
Note that this is not a frequency-dependent homodyne measurement, 
i.e. the so-called variational measurement found in the proposal for gravitational 
wave detection, which requires implementing additional filter cavities. 
Also note that $f(\omega)$ takes zero at $\omega=0$, i.e. the center 
frequency, in which case clearly we cannot obtain any information about 
$\theta$. 
Optimization of $\omega$ is essential in these sense. 
}
Clearly the expectation of $y$ yields $\theta$ with variance 
\[
     {\rm Var}(y) 
       = \frac{2}{f(\omega^{\rm opt})^2}\cdot \frac{1}{2\alpha^2 e^{2r}}
       \approx \frac{1}{f(\omega^{\rm opt})^2E^2}, 
\]
which is the Heisenberg scaling.

\begin{figure}
\centering
\includegraphics[scale=0.34]{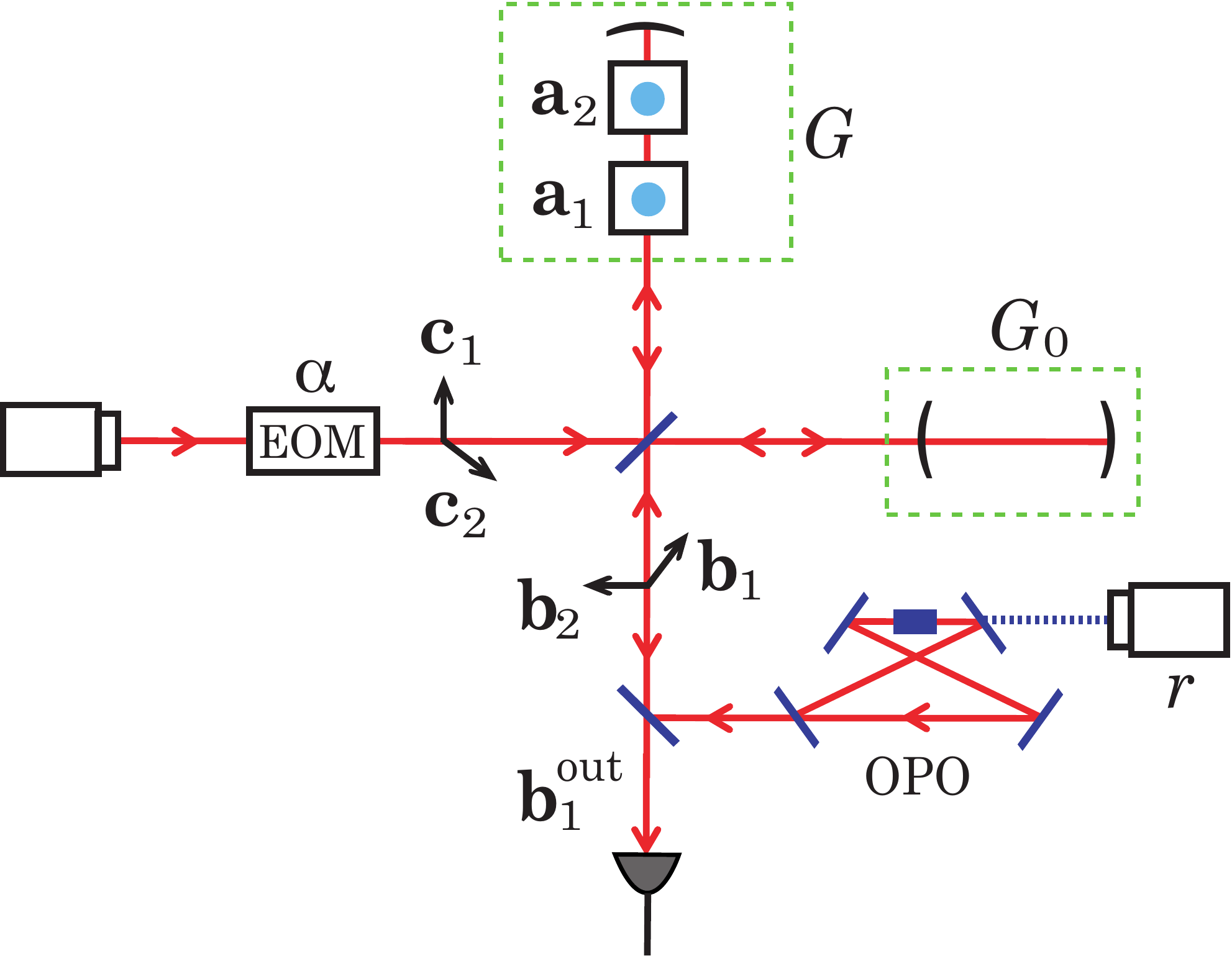}
\caption{\label{2I2Oexample}
Configuration of the one-parameter system identification method for a 
MIMO system composed of atomic ensembles denoted by $G$. 
$G_0$ is a Fabri-Perot cavity having the same decay rate as that of $G$. 
}
\end{figure}

A possible experimental setup is depicted in Fig.~\ref{2I2Oexample}. 
We here consider a Michelson-type interferometer composed of two 
optical paths. 
The input field $\mathbf{b}$ is injected from the bottom arm, while $\mathbf{c}$ 
comes from the left arm. 
Although in \cite{Muschik1} $\mathbf{b}_1$ and $\mathbf{b}_2$ are 
taken as two bosonic modes with different frequencies, here they are represented 
as those with different polarizations (same as for $\mathbf{c}$). 
The system $G$ is placed in the upper part of the interferometer; as seen 
in \cite{Muschik1} this system is composed of the cascade of two 
chambers containing atomic ensembles. 
The nominal system $G_0$, on the other hand, is now given by an optical cavity 
having the same coupling strength $\kappa$, and it is now placed in the 
right end of the interferometer. 
As described above, our identification strategy is that we create a squeezed 
state in the mode $\mathbf{b}_1$ and a coherent state in the mode $\mathbf{c}_2$; 
the former can be realized by constructing an optical parametric oscillator 
(OPO) and the latter can be realized by placing an electric optical modulator 
(EOM) along the input path of $\mathbf{c}_2$.

\section{Estimation of  Multi-Dimensional Parameters in a PQLS}\label{Toes2}

In this section we discuss how to extend the interferometric approach using squeezed-coherent states to the case of multiple parameters.

%A squeezing-based practical scheme is also given that has nearly the same performance as the photon- entanglement scheme. The frequency-comb entangled light fields used here are advantageous as they can be generated with current technology, and require only a single squeezing operation (compared with the d squeezing operations in the independent strategy).

\subsection{General Multiple Parameter Setup}
Suppose that we have a  PQLS with $m$ input channels and that we would like to estimate multiple parameters $\theta_1, \hdots, \theta_d$. 
Consider  the metrology setup with inputs supported on $d/m\in\mathbb{Z}$ frequencies and a MIMO PQLS with $m$ channels. By  thinking of the different  frequencies as parallel channels, we have $d$ input channels on which to support our input state. 
Based on the successes of the interferometric setup using squeezed-coherent state in the single parameter case we consider a generalisation of this here. 
 That is, the  
input channels $\mathbf{b}=[\mathbf{b}_1, \ldots, \mathbf{b}_d]^T$ are prepared in a  
multi-mode squeezed state \cite{Adesso1}; these channels are mixed by means of separate 50:50 beamsplitters with the corresponding ancilla channels
$\mathbf{c}=[\mathbf{c}_1, \ldots, \mathbf{c}_d]^T$ prepared in coherent states. 
One of the outputs of each beam splitter is then passed through the PQLS and the other through a second known PQLS (to be specified shortly) before 
recombining at  another 50:50 beam-splitter. We assume that the frequencies $\omega_1, ...\omega_{d/m}$ have been chosen beforehand and      we drop the reference to frequency in the following       (we will consider the problem of varying these in Sec. \ref{varfreq}).

As in the single parameter case, we assume that we already computed a rough estimate of the parameters each with standard deviation of the order $1/\sqrt{E}$, such that the parameters can be written as $\theta_i=\theta_{0i}+h_i/\sqrt{E}$ for some local parameter $h_i$. Then the transfer function (of the frequency concatenated system) may be expanded as 
$$
\Xi_{\theta} = \Xi_{\theta_0} \left(\mathds{1}_m-i\sum_i \frac{h_iL^{(i)}(\theta_0)}{\sqrt{E}}\right)+\mathcal{O}\left(E^{-1}\right)
$$ 
%Could explain expansion in modes
where $L^{(i)}(\omega,\theta_0)$ are  $d\times d$ Hermitian matrices and $\theta_0=\left(\theta_{01}, \hdots, \theta_{0d}\right)$.
Now, place the PQLS 
$\left(S_{\theta_0}, C_{\theta_0}, \Omega_{\theta_0}\right)$ in the lower  interferometer arms. 
Also, prepare the probe mode ${\bf b}_0$  in a squeezed state and ${\bf c}_0$
in a coherent state 
with amplitude 
$\mbox{{\boldmath $\alpha$}}=[\alpha_1, \ldots, \alpha_d]^T$ such that $\alpha_i=\mathcal{O}(\sqrt{E/2m})$. We will assume that   $\alpha_i\in{\mathbb{R}}$ for all $i$.
It follows that 
\begin{equation}\label{pfg}
     \mathbf{b}_0^{out}\ket{S_r}\otimes\ket{\mathbf{\alpha}}\approx \Xi_{\theta_0}\left\{\mathbf{b}_0
        -i \frac{       \sum_i h_iL^{(i)}(\omega,\theta_0)         }{2\sqrt{E}} \mbox{{\boldmath $\alpha$}}\mathds{1}
%        \left[\begin{array}{c}
%    L_{11}\alpha_1 + \ldots + L_{1m}\alpha_m\\
%   
%       \vdots\\
%    L_{m1}\alpha_1 + \ldots + L_{mm}\alpha_m
%      \end{array}\right]
%      
      \right\}\ket{S_r}\otimes\ket{\mathbf{\alpha}}. 
\end{equation}
Since $\Xi_{\theta_0}$ is known we will ignore its effect like in the single parameter case.
%We will consider states of squeezed-coherent type supported on $m/d\in\mathbb{Z}$ frequencies. 
Therefore the local parameter $h_i$ is imprinted on the output via a displacement of the squeezed mode; on mode $\mathbf{b}_j$ the displacement is along the (unnormalized) quadrature
$$\mathbf{X}^{(i)}:=\mathbf{b}_j     \sum_{k}\left(iL^{(i)}_{jk}\alpha_k\right)    -\mathbf{b}^{\dag}_j   \sum_{k}\left(i\overline{L^{(i)}_{jk}}\alpha_k\right).$$
Therefore the generator belonging to $\theta_i$ is thus $\mathbf{G}^{(i)}:=\sum_{jk}\left(\mathbf{b}_jL^{(i)}_{jk}+\mathbf{b}_j^{*}\overline{L^{(i)}_{jk}}\right)\alpha_k$. 
The difference with the single parameter case is that here there are $d$ parameters each displacing the input additively in \textbf{different} directions. Further, notice that the generators of the displacement for each parameter, $\mathbf{G}^{(i)}$, depend on the coherent input, and change according to variations in this input. Finding the optimal scaling or even whether Heisenberg scaling is possible in the multi-parameter case is more involved. 

In order to infer the $d$ parameters  with high accuracy from measuring quadratures, it is clear that we must measure at least $d$ quadratures. Moreover these quadratures $\mathbf{Y}_1, ...\mathbf{Y}_d$ must satisfy the following properties:
\begin{itemize} 
\item[(1)]\label{fond1} They must contain sufficient information about all parameters. This can only happen if at least one of them does not commute with each displacement generator. That is, for each parameter $\theta_k$, there exists an $i$ such that $[\mathbf{G}^{(k)}, \mathbf{Y}_i]\neq0$.
This condition is necessary because if it didn't hold then measurements of the output will not contain any information about $\theta_k$ on average because in the Heisenberg picture $\left<\mathbf{Y}_i\right>=\left<e^{i\theta_k\mathbf{G}^{(k)}}\mathbf{Y}_ie^{i\theta_k\mathbf{G}^{(k)}}\right>$ for all $i$.
%[[[STUFF WITH EXPECTATIONS...on average can't get info about parameter]]]]
\item[(2)]\label{fond2}  To get Heisenberg scaling we must be able to measure at least $d$ quadratures with small variance. If these quadratures are to be measured then they must all commute, that is,   $[\mathbf{Y}_i, \mathbf{Y}_j]=0$ for all $i$ and $j$.
\end{itemize}
If the system is identifiable (see Def. \ref{identy}),  we see that there always exists a coherent input such that conditions (1) and (2) hold. We then  show how to achieve the Heisenberg scaling in this case.
% [[[WE WILL ALSO OUTLINE OPTIMALITY PROBLEM...check done this for both cases]]]

\subsection{SISO PQLS}\label{sausage1}
Consider a SISO system with inputs supported on $d$ frequencies (i.e $m=1$). There is a simplification in this case because the generators for  parameters $\theta_k$ are given by $\mathbf{G}^{(k)}=\mathbf{Q}_1L^{(k)}_{11}\alpha_1+\hdots + \mathbf{Q}_dL^{(k)}_{dd}\alpha_d$.  Therefore, as the $\alpha_i\in\mathbb{R}$ and $L^{(k)}$ are Hermitian and diagonal, the action of all generators is a displacement in the same direction, i.e each mode  $\mathbf{b}_i$, is translated along the momentum axis.

 Consider a squeezed input consisting of a product of $d$ momentum-squeezed one mode states over the modes $\mathbf{b}_0$ each with energy $E_i=\mathcal{O}(E/2d)$ and the coherent state as above. Since the action of the system is a displacement along the momentum axis (as this direction is canonically  conjugate to the generator), we will measure the momentum quadrature of each mode. Clearly conditions (1) and (2) hold under identifiability.  
 We will now  show that Heisenberg scaling is achievable under the assumption that the parameters are identifiable. The (classical) Fisher information  in this case is given by \cite{Monras1}
 \begin{align*}
 F_{mn}(\theta_0)&=8\mu^T_m\mathrm{Diag}(E_1,...E_m)\mu_n,
\quad
\mu_m=\frac{1}{\sqrt{2}}\sum_i L_{ii}^{(m)}\alpha_i.
 \end{align*}
Hence
 \begin{equation}\label{tofu1}
 F(\theta_0)=4\sum_{i}E_i\alpha_i^2\left(\begin{smallmatrix} L^{(1)}_{ii} \\ \vdots\\ L^{(d)}_{ii}\end{smallmatrix}\right)\left(\begin{smallmatrix}L^{(1)}_{ii}&\hdots&L^{(d)}_{ii}\end{smallmatrix}\right).
 \end{equation}
Therefore the MSE, which is equal to $\mathrm{Tr}\left(F(\theta_0)^{-1}\right)$ (see Eq. \eqref{rtvz}),   will scale quadratically with $E$ iff the  vectors $v_i:=\left(\begin{smallmatrix}L^{(1)}_{ii}&\hdots&L^{(d)}_{ii}\end{smallmatrix}\right)^T$ form a basis. Writing the transfer function as a phase $\Xi_{\theta}(-i\omega_i)=e^{-i\lambda_i}$, then this condition is equivalent to the invertibility of the Jacobean of the  map
$f:\mathbb{R}^d\mapsto\mathbb{R}^d$
$$
     f: (\theta_1, \dots , \theta_d) \mapsto (\lambda_1, \lambda_2, \ldots, \lambda_d). 
$$
Assuming  that the parameters $\theta$ are identifiable as in Def \ref{identy}, then the Jacobian must be invertible and we have Heisenberg scaling.

\subsection{MIMO PQLS}
 Proving that Heisenberg scaling is possible for a  MIMO PQLS is not as straightforward     because the the (displacement) generators of each parameter could potentially be along different directions, rather than all in the same direction. Nevertheless it is possible, as we see now.

Consider a squeezed input consisting of a product of $d$  mode states each with energy $E_i=\mathcal{O}(E/2d)$, squeezed in the quadratures 
$$\mathbf{Y}_i=e^{i\phi_i}\mathbf{b}_i+ e^{-i\phi_i}\mathbf{b}^{\dag}_i,$$
where $\phi_i$ will be chosen later, and the coherent state as above. Our measurement will be the  quadratures 
$\mathbf{Y}_i.$ 
Clearly condition (2) is satisfied here. Let us see that there exist a choice of $\alpha$ such that condition (1) is also.
We have
\begin{align}
\nonumber [\mathbf{G}^{(k)}, \mathbf{Y}_i]&=\sum_j\left(e^{-i\phi_i}\overline{L^{(k)}_{ij}}- e^{i\phi_i}L^{(k)}_{ij}\right)\alpha_j\\
:=\sum_j\beta^{(k)}_{ij}\alpha_j \label{aimed}.
\end{align}
For a given $k$ we require that there exists at least one $\mathbf{Y}_i$ such that \eqref{aimed} is non-zero. Since we have freedom over the choice of $\alpha_i$s then it is clear that \eqref{aimed} can be made  to be non-zero if at least one of the terms $\beta^{(k)}_{ij}$ is non-zero, which, given that we have choice over $\phi_i$, is always possible if one of the $L^{(k)}_{ij}$ is non-zero for a given $j$. Therefore the only way that we will not be able to   to choose a measurement and input, $\alpha$, such that condition (1) is satisfied would be  if for a given  $k$, $L^{(k)}_{ij}=0$ for all $i$ and $j$, which is a contradiction to  the fact that $\theta_i$ is identifiable. 

Now let us show that Heisenberg scaling is possible here and show how to achieve it.  The Fisher information of this measurement is thus \cite{Monras1}
\begin{align*}
 F_{mn}(\theta_0)&=8\mu^T_m\mathrm{Diag}(E_1,...E_m)\mu_n,
\quad
\mu_m=\frac{1}{\sqrt{2}}\mathrm{Re}(ie^{i\phi}L^{(m)}\alpha).
 \end{align*}
Hence
 \begin{equation}\label{coins1}
 F(\theta_0)=4\sum_{i}E_i\left(\begin{smallmatrix} K^{(1)}_{i} \\ \vdots\\ K^{(d)}_{i}\end{smallmatrix}\right)\left(\begin{smallmatrix}K^{(1)}_{i}&\hdots&K^{(d)}_{i}\end{smallmatrix}\right),
 \end{equation}
where $K^{(i)}_j=\mathrm{Re}\left(e^{i\phi_j}\sum_l L^{(i)}_{jl}\alpha_l\right)$. Because $\alpha_i^2$ are of order $E$, then clearly if $F(\theta_0)$  is invertible then the MSE, which is equal to $\mathrm{Tr}\left(F(\theta_0)^{-1}\right)$ (see Eq. \eqref{rtvz}), will scale quadratically with $E$. It therefore remains to show that there exists a choice of $\alpha$ such that this matrix is invertible. 
Suppose that $m=d$ for simplicity. The Fisher information matrix is invertible iff the vectors in Eq. \eqref{coins1} are linearly independent. Suppose that for a particular choice of $\alpha$ and $\phi_i$ they are not linearly independent. That is, 
\begin{equation}\label{starz}
\left(\begin{smallmatrix} \mathrm{Re}\left(e^{i\phi_j}\sum_l L^{(1)}_{il}\alpha_l\right)\\\vdots\\\mathrm{Re}\left(e^{i\phi_j}\sum_l L^{(d)}_{il}\alpha_l\right)\end{smallmatrix}\right)=\gamma\left(\begin{smallmatrix} \mathrm{Re}\left(e^{i\phi_j}\sum_l L^{(1)}_{jl}\alpha_l\right)\\\vdots\\\mathrm{Re}\left(e^{i\phi_j}\sum_l L^{(d)}_{jl}\alpha_l\right)\end{smallmatrix}\right)
\end{equation}
for some $i$ and $j$ and constant $\gamma$. Since we have freedom over the choice of $\alpha$ and $\phi_i$s, then  it is always possible to   vary these so that  the two vectors in Eq. \eqref{starz} are not  multiples of one another if and only if $L^{(i)}\neq a L^{(j)}$ for some constant $a$, which is true under identifiability.

Therefore in conclusion we can always identify $d\leq m$  parameters of a PQLS  at Heisenberg level using this method iff the parameters are identifiable. The trick is in choosing the $\alpha$ and $\phi_i$ such that we have at least $d$ linearly independent vectors $\left(K_i^{(1)},..., K_i^{(d)}\right)^T$.

%[[[COULD comment about other point of view that if any of channels are mirrored then will be fatal in case $m=d$.]]]

\subsection{Optimisation Problem}
Still unanswered is what is the optimal squeezed-coherent state (and measurement) giving the smallest MSE? 
That is, we would like to minimise the quantity $\mathrm{Tr}(F^{-1})$, where $F(\theta_0)$ is the (classical) Fisher information given by,
$$F_{ij}(\theta)=8\mu_i^T\mathrm{Diag}(E_1, \hdots, E_m)\mu_j, \quad
\mu_m=\frac{1}{\sqrt{2}}\mathrm{Re}\left(UL^{m}\alpha\right).$$
Here $U$ is a $m\times m$ unitary matrix characterising the direction of squeezing for our squeezed input state and as a result we are measuring the quadratures 
$U\mathbf{b}+\overline{U}\mathbf{b}^{\dag} $
here.
Thus the problem is a minimisation over all choices of $U, E_1,... E_m$ and $\alpha_1, ...\alpha_m$ subject to total energy constraint $E\approx E_1+...+E_m +|\alpha_1|^2+...+|\alpha_m|^2$. This minimisation proves to be very difficult, even  in the SISO case. Moreover, it is not even even clear how many channels should be used; recall we saw in the one parameter case that the optimal number is one, so perhaps here it is $d$? In the one extreme consider a  ``broadband squeezing'' in the frequency domain, which  would make sense practically and metrologically. Some simple calculations seem to indicate that squeezing  lots of frequency modes independently decreases the  accuracy since you need to spread energy over all of them. This would be an interesting question for future work; perhaps it may be that it's optimal to squeeze only a number of modes of the order of the number of parameters.
 Furthermore, is there any advantage in considering entanglement between the squeezed modes? We consider a particular instance of this problem in Sec. \ref{simultaneous}. Creating frequency entangled squeezed states can be achieved using a  \textit{synchronously pumped optical parametric oscillator} (SPOPO) \cite{Roslund1, Medeiros1}.

%[[[could add in comments from madalins emails about other ways of showing identifiability....]]]

%Suppose that we would like to estimate parameters $\theta_1, \hdots, \theta_d$ by considering . 

%%%%%%%%%%%%%%%%%%%%%%%%%%%%%%%%%%%%%%%%%
\section{Using Entanglement Between Frequencies}\label{simultaneous}
%%%%%%%%%%%%%%%%%%%%%%%%%%%%%%%%%%%%%%%%%

In this section we extend the work of \cite{Humphreys1}, which deals with the 
problem of estimating $d$-dimensional parameters. 
For a SISO PQLS we  show 
 there is an $\mathcal{O}(d)$ advantage to be had by 
 using probe states with entanglement between frequencies, over $d$ repetitions of optimal single 
parameter identification experiments.

\subsection{The Analogous Quantum Metrology Result}

We begin by  reviewing the result \cite{Humphreys1}.  
\begin{figure}
\centering
\centering\includegraphics[scale=0.15]{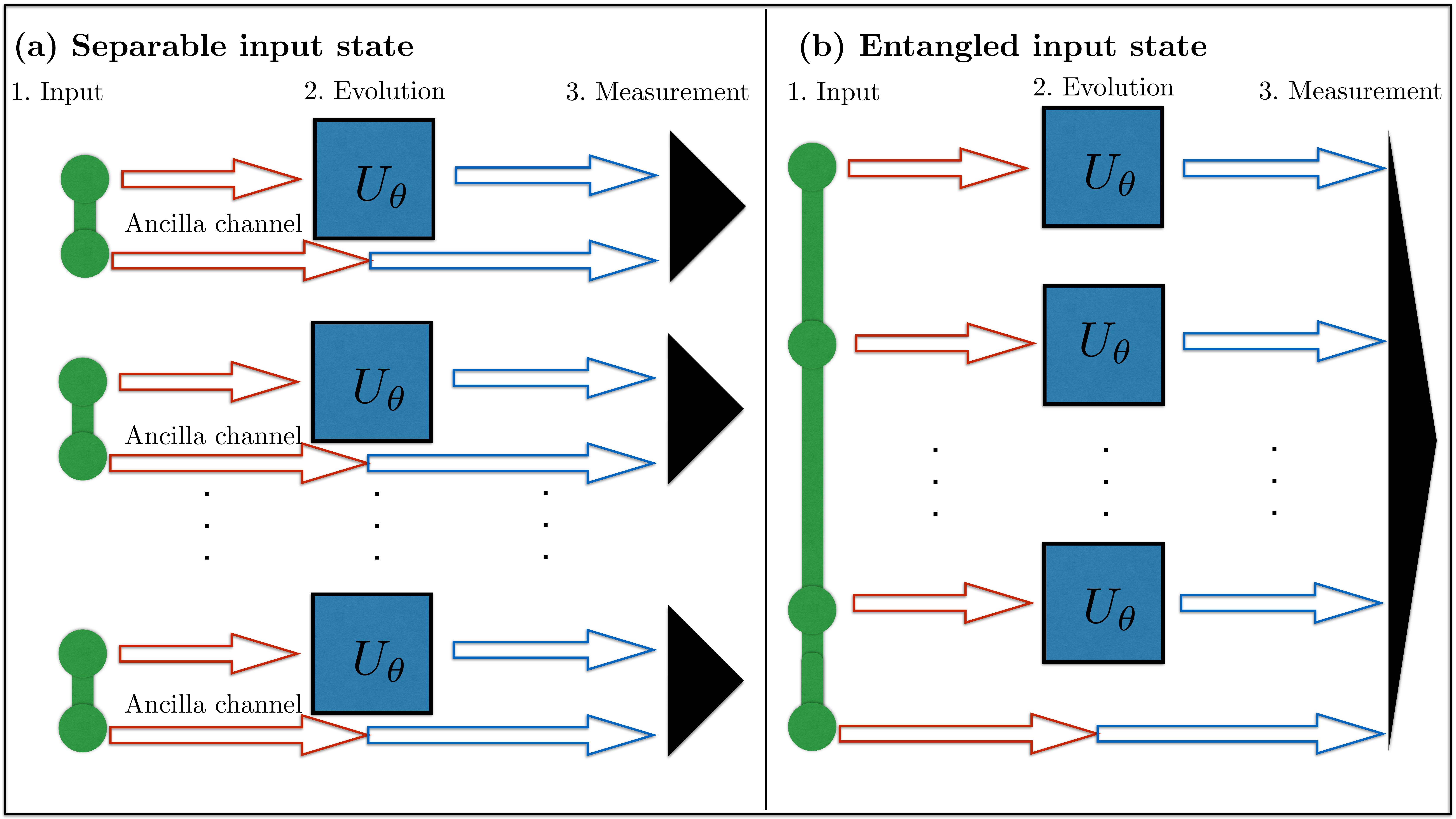}
\caption{A multi-parameter model with a local Hamiltonian structure and $d$ independent phases to estimate. 
Strategy (i): estimate the phases separately. Strategy (ii): estimate the phases using   entangled states between the phase channels and a non-local measurement. 
%There are three stages to this quantum metrology protocol; an input, 
%an evolution and a measurement stage. 
Note that the setup is different from that of Fig. \ref{classquant}, as there each channel evolves with the same unitary.
\label{c2ty1}}
\end{figure}
Consider the quantum metrology set-up in Fig. \ref{c2ty1}, 
which consists of a $(d+1)$-mode interferometer with $d$ independent 
phases $\theta_1, \theta_2, \dots , \theta_d$ and an ancilla channel. 
The total unitary  operator is given by 
$\mathbf{U}:=\mathds{1}\otimes\mathrm{exp}\left(
\sum_{i=1}^di\mathbf{N}_m\theta_m\right)$, where 
$\mathbf{N}_m$ is the number operator of the $m$th mode.

The following strategies using states with the same fixed photon number (hence energy) were compared in 
 \cite{Humphreys1}:
\begin{enumerate}
\item Estimating the $d$ phases using $d$ separate interferometers, each with 
a $N/d$-photon N00N state, or equivalently using a single product of N00N states for each mode.
\item Estimating $d$ phases using an $N$-photon $(d+1)$-mode entangled state.
\end{enumerate}

%Using $N/d$ photons in each, the QFI scales as $d^2/N^2$ for each phase. Hence the total phase variance is $d^3/N^2$.

We know that the best $N$-photon probe state for estimating a single phase 
is the N00N state. In light of this we    consider using $d$ repetitions of these states to estimate the parameters, which turns out to be the   optimal fixed photon number state of type (1) \cite{Proctor1}. Moreover, 
 using the results from the previous subsection the QFI scales as 
$N^2/d^2$ for each phase. 
Combining these, the total mean square error is $d^3/N^2$.

On the other hand, based on the intuition that N00N states are the best for single phases, the authors considered the following 
$(d+1)$-mode generalization 
\begin{align*}
     \Ket{\psi}
         &=\alpha \Ket{0} \otimes \Big(\Ket{N, 0, \ldots} + 
                    \Ket{0, N, 0, \ldots, 0} + \cdots + 
                        \Ket{0, \ldots, 0, N} \Big) \\
                        &
                        + \beta\Ket{N,0, \ldots, 0}.
\end{align*}
It was found that the optimal Cramer-Rao bound for the mean square error 
%(saturation of the QCRB implies a saturation of the mean-squared error) 
over these states is $\frac{(1+\sqrt{d})^2d}{4N^2}$ corresponding to $\alpha=1/\sqrt{d+\sqrt{d}}$. 
This shows a possible $\mathcal{O}(d)$ improvement by using entangled states. A similar result was found  in \cite{Liu1} by 
using entangled coherent states (ECSs).  
%This improvement with the number of phases is rather surprising. 
%As Strategy~(i) uses $d$ repetitions of the optimal single phase strategy, 
%this enhancement is present over any individual estimation strategy. 
The caveat here is that it is not clear whether the Cramer-Rao bound in strategy (ii) is achievable   \cite{Proctor1}. While it is true that a certain ``commutativity of the SLDs in expectation'' condition holds, this only guarantees the achievability of the CRB in the asymptotic sense where a large number of identically prepared copies of the state are available. However, by using independent probes, the Heisenberg scaling in energy is lost. Therefore, a more careful analysis of the scaling is necessary, but is beyond the scope of this thesis. We will restrict ourselves to the optimisation of the QCRB here.

We will now see that this sort of simultaneous estimation strategy is 
well-suited to to estimating multiple parameters simultaneously in PQLSs.

%Notice that we are dealing with local Hamiltonians and because each mode 
%contains information about only one phase, the symmetric logarithmic 
%derivatives of all phases commute for all fixed photon number states in this setup. 
%Thus the QCRB is in principle attainable if one measures the SLD (this may not be a simple measurement though). 
%

\subsection{$\mathcal{O}(d)$ Enhancement Using Frequency-Entangled States}\label{jam}

Suppose that we have a SISO PQLS whose transfer function $\Xi_\theta(s)$ 
depends on unknown parameters ${\theta}=(\theta_1, \theta_2, ..., \theta_d)$. 
In particular, for a fixed frequency $s=-i\omega$ it can be represented as 
$\Xi_\theta(-i\omega)=e^{-i\lambda(\theta)}$. 
%That is, we are now generalising our system identification problem to allow 
%for the possibility of a $d$-dimensional parameter space. 
Firstly, a fairly trivial but nonetheless important fact is that, for a SISO PQLS, 
it is impossible to identify more than one parameter using only a monochromatic 
frequency input. 
That is, the inverse problem is under-constrained. 
We now give an example of this fact for the case of two unknown parameters. 

\begin{exmp}\label{fgty0}
Consider the  set-up in Fig. \ref{c2ty}, where a N00N state 
of fixed frequency $\omega$ is used as a probe for our two-parameter SISO PQLS. The state acts on two channels: one with the system and another ancilla channel. Denote the   two unknown parameters by $\theta_1$ and 
$\theta_2$. 
The action of the PQLS on the input is thus 
\[
      \frac{1}{\sqrt{2}}\left(\ket{0,N}+\ket{N,0}\right)
         \mapsto\frac{1}{\sqrt{2}}\left(\ket{0,N}
              +e^{-i\lambda(\theta_1, \theta_2)}\ket{N,0}\right).
\]
Using Eq. (\ref{qfiformula}), given earlier for pure unitary families, the QFI is 
given by
\[
       F\left(\theta_1, \theta_2\right)_{l,m}
          =N^2\frac{d\lambda}{d\theta_l}
                   \frac{d\lambda}{d\theta_m}.
\] 
Recall that the MSE is bounded below by the trace of the inverse of the QFI 
matrix. 
However, here the QFI is singular. 
Therefore we are unable to identify both parameters. 
From this example it follows that we must allow for inputs supported over more 
frequencies if we are to identify parameters in a multi-parametric model.
\begin{figure}
\centering
\centering\includegraphics[scale=0.15]{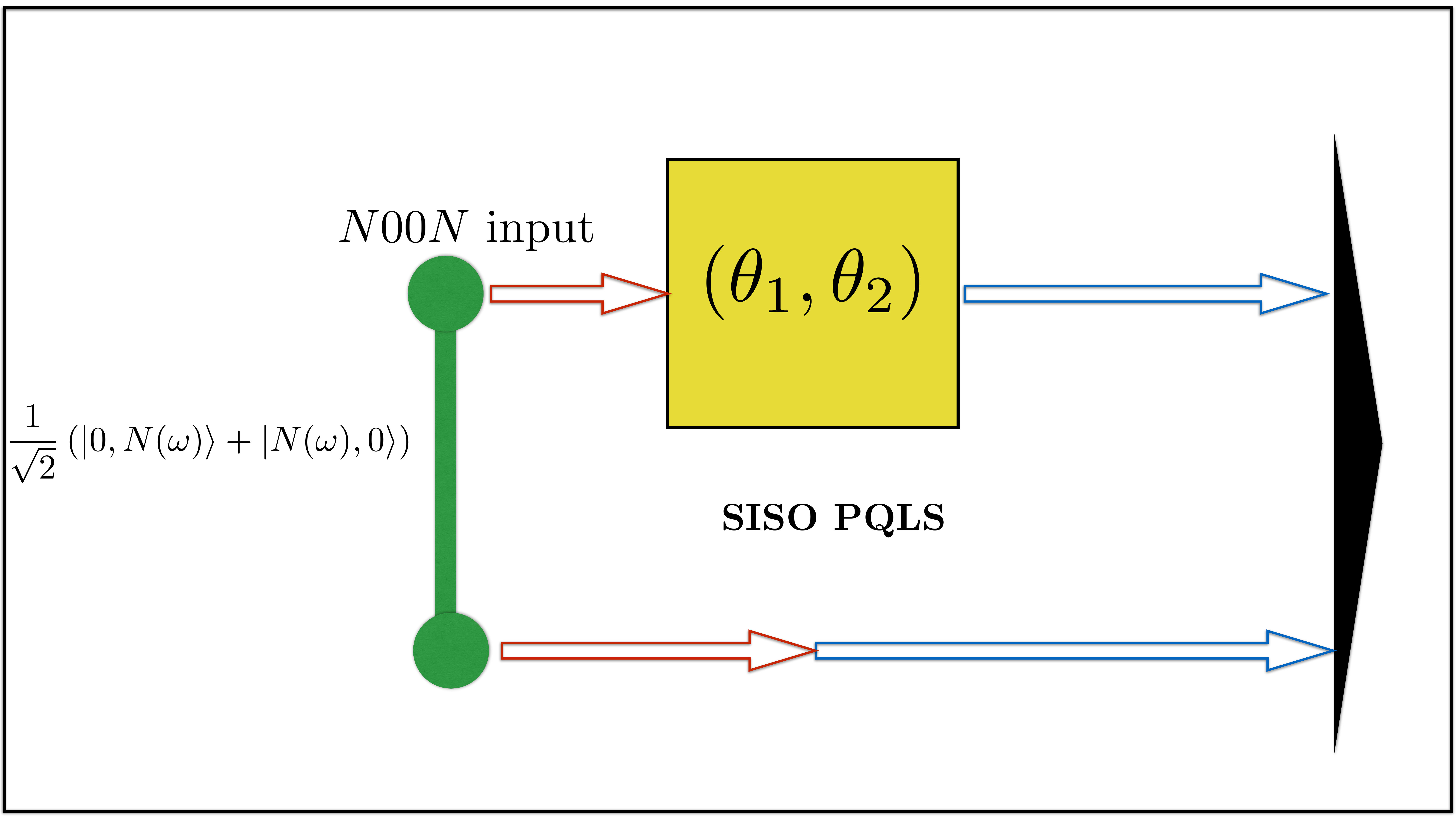}
\caption{This figure shows the set-up in Example \ref{fgty0}, where a N00N 
state is passed through a two parameter SISO PQLS. 
We consider an interferometric set up, where the two parallel channels meet 
at a beam-splitter. 
\label{c2ty}}
\end{figure}
\end{exmp}

Consider the set-up in Fig. \ref{c2ty2} where an input state, supported over 
(fixed) frequencies $\omega_1, \omega_2, \ldots, \omega_d$, is passed through 
a SISO PQLS. 
Note that, due to the linear nature of PQLSs, each input with frequency $\omega_i$ 
acts independently and so can be thought of as separate phase rotation  channels   with phase 
$\Xi_{{\theta}}(-i\omega_i):=e^{-i\lambda_i(\theta)}$. 
%Recall from earlier that, in the frequency domain, white-noise operators 
%corresponding to two different frequencies commute. 
Therefore, the relevant total unitary operator is given by
\begin{eqnarray*}
    \mathbf{U}:=\mathds{1}\otimes \bigotimes_{i=1}^d 
       \mathrm{exp}\left(-i\lambda(\omega_i, \theta)\mathbf{b}^*(\omega_1)\mathbf{b}(\omega_1)\right)
%          \otimes
%       \mathrm{exp}(-i\lambda_2\mathbf{b}^*(i\omega_2)\mathbf{b}(i\omega_2))
%         \otimes \cdots \otimes
%       \mathrm{exp}(-i\lambda_d\mathbf{b}^*(i\omega_d)\mathbf{b}(i\omega_d)). 
       \end{eqnarray*}
 where $\mathds{1}$ is the identity on the ancillary system.
\begin{figure}
\centering
\centering\includegraphics[scale=0.18]{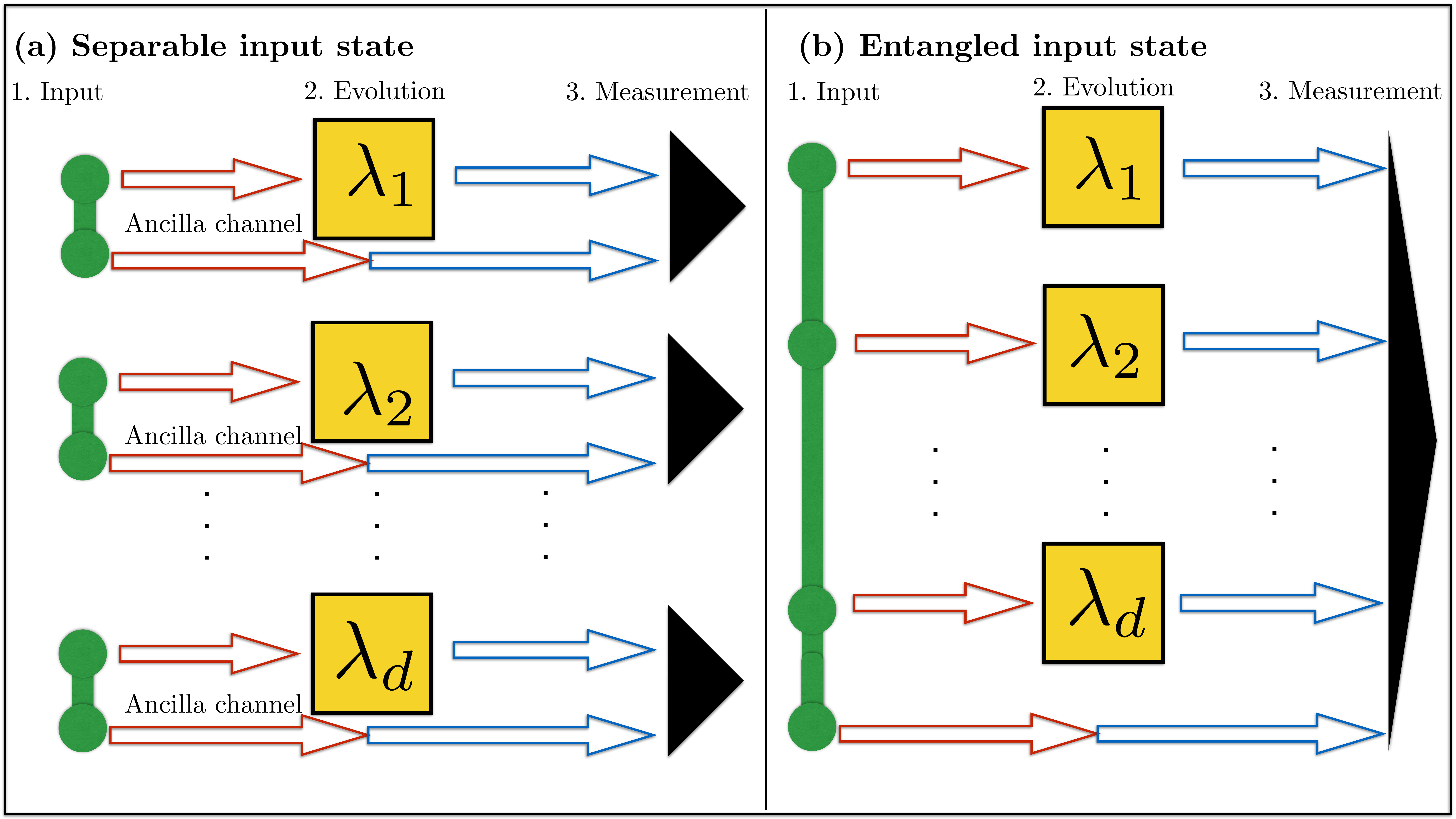}
\caption{We consider a $d$-parameter model and pass an input state supported over $d$ frequencies through our PQLS. We compare here two strategies: Strategy (i) corresponds to performing $d$ separate single frequency experiments and  Strategy (ii) is a simultaneous estimation scheme using inputs entangled over the frequencies.
\label{c2ty2}}
\end{figure}
As in the Sec.  \ref{simultaneous} we compare the following two strategies which uses the same energy resources:
\begin{enumerate}
\item 
Estimating the $d$ parameters using $d$ separate interferometers, each with $N/d$-photon N00N state.
\item Estimating the $d$ unknown parameters using an 
$N$-photon $(d+1)$-mode frequency-entangled state. 
\end{enumerate}
In both strategies we assume for simplicity that the input frequencies, $\omega_1, \ldots, \omega_d$, 
have been chosen beforehand and the same resources have been given to each channel (system) channel. 

This problem is similar to the metrology one from 
Sec. \ref{simultaneous} with the difference that the   phase of each 
interferometer arm depends on \textit{all} of the unknown parameters, rather than a single one. Note that 
we could choose to estimate the $\theta$'s by first estimating the $\lambda$'s and then calculating $\theta=f(\lambda)$ for some continuously differentiable function $f$. However, this method will not generally provide an optimal estimate for $\theta$. Nevertheless it can be verified that the ``commutativity in expectation'' condition for the SLDs holds and the  QCRB is achievable in the asymptotic sense discussed  in Sec. \ref{simultaneous}. Our task will be to analyse the upper bound to the MSE given by  $\mathrm{Tr}\left(F(\theta)^{-1}\right)$, where $F(\theta)$ is the QFI for a specific input state.

%The reason for this is that the SLDs here will be linear combinations of the commuting ones from Section \ref{simultaneous}. 
%Therefore,  the QCRB is still saturated here. Saturation of the QCRB implies a saturation of Equation \ref{rtvz}  and so for the rest of this subsection we shall be interested in minimising $\mathrm{Tr}\left(I(\theta)^{-1}\right)$.

\subsubsection{Performance of Strategy 1}
Recall from Sec. \ref{inter1} that the optimal state for estimating a single parameter was a single frequency N00N state.  
For the multiple parameter case it is therefore natural to consider the input state: 
\[
\hspace{-2cm}
     \Ket{\psi} = 
     \frac{1}{\sqrt{2}}\left(\Ket{0,\frac{N(\omega_1)}{d}}
                                               +\Ket{\frac{N(\omega_1)}{d}, 0} \right)
     \otimes \ldots \otimes
     \frac{1}{\sqrt{2}}\left(\Ket{0, \frac{N(\omega_d)}{d}}
                                               +\Ket{\frac{N(\omega_d)}{d}, 0}\right),
\]
which corresponds to $d$ separate ``frequency channels''
 with $N/d$ particle N00N states as inputs for each of the $d$ channels. 
Let us now, understand how accurate this strategy is by calculating $\mathrm{Tr}\left(I(\theta)^{-1}\right)$. Defining $\frac{\partial\Ket{\psi_{\theta}}}{\partial\theta_l}
:=\ket{\partial_l\psi_{\theta}}$, as in Sec. \ref{plm}, we may firstly obtain 
\begin{eqnarray}
& & \hspace*{0em}
     \Braket{\psi_{\theta}| \partial_l\psi_{\theta}}
         =\frac{iN}{2d}\left(\frac{\partial\lambda_1}{\partial\theta_l}
            +\frac{\partial\lambda_2}{\partial\theta_l} 
               + \ldots + \frac{\partial\lambda_d}{\partial\theta_l}\right),
\nonumber \\ & & \hspace*{0em}
     \Braket{\partial_l\psi_{\theta}| \partial_m\psi_{\theta}}
        =\frac{N^2}{d^2}\left(\frac{1}{2}\sum_{i=1}^d
            \frac{\partial\lambda_i}{\partial\theta_l}
            \frac{\partial\lambda_i}{\partial\theta_l} 
            + \frac{1}{4}\sum_{\scriptstyle i,j=1 \atop\scriptstyle i\ne j}^{d}
            \frac{\partial\lambda_i}{\partial\theta_l}
            \frac{\partial\lambda_j}{\partial\theta_l} \right).
\nonumber
\end{eqnarray}
Then using \ref{qfiformula}, it follows that the QFI is given by 
\begin{eqnarray}
\label{loft}
      F({\theta})_{l,m}
         =\frac{N^2}{d^2}\sum_{i=1}^d
              \frac{\partial\lambda_i}{\partial\theta_l}
              \frac{\partial\lambda_i}{\partial\theta_m}\frac{N^2}{d^2} (J(\theta)^TJ(\theta))_{l,m}
\end{eqnarray}
where $J(\theta)$ is the Jacobian of the map
$f:\mathbb{R}^d\mapsto\mathbb{R}^d$
$$
     f: (\theta_1, \dots , \theta_d) \mapsto (\lambda_1, \lambda_2, \ldots, \lambda_d),
$$
Assuming the the parameter $\theta$ are identifiable, the Jacobian must be invertible and therefore
\begin{equation}\label{bad232}
{\rm Tr} (F({\theta})^{-1}) = \frac{d^2}{N^2} {\rm Tr} (J(\theta)^{-1}(J(\theta)^{T})^{-1}) = 
\frac{d^2 H(\theta) }{N^2 | J (\theta)|^2}
\end{equation}
where  $| J (\theta)|$ is the determinant of $J(\theta)$, and $H(\theta):= \| P(\theta)\|_2^2$ with $P(\theta)$ the cofactor matrix of the Jacobian $J(\theta)$. 
%
%Finally, we may obtain 
%%
%\begin{equation}
%\label{bad232}
%      \mathbb{E} \| \hat{\theta}- \theta\|^2\geq\mathrm{Tr}(F({\theta})^{-1})
%         =\frac{d^2 H({\theta})}{N^2\left(\mathrm{Det}J (\theta)\right)^2},
%\end{equation} 
%%
%where $J(\theta)$ is the Jacobian of the function 
%$f:\mathbb{R}^d\mapsto\mathbb{R}^d$
%%
%$$
%     f: (\theta_1, \dots , \theta_d) \mapsto (\lambda_1, \lambda_2, \ldots, \lambda_d),
%$$
%%
%and $H({\theta})=\sum_{i,j=1}^d(P_{i j}(\theta))^2$, where $P(\theta)$ is the cofactor 
%matrix of the Jacobian $J(\theta)$. 
%%Importantly $H({\theta})$ is clearly greater than or equal to zero, as should 
%%be expected with it providing a lower bound for the mean-squared error. 
%The details of this calculation may be found in \ref{appendixA}. 
Equation (\ref{bad232}) provides the lower bound to the MSE for strategy 1 which features the Heisenberg scaling with respect to $N$.

\subsubsection{Performance of Strategy 2}
Since N00N states are optimal for single phases, we consider the following  generalisation:
\begin{eqnarray}
\label{bluet}
 %\hspace*{-4.5em}
              \Ket{\psi}&=& \Ket{0} \otimes \Big( \alpha\Ket{N(\omega_1), 0, \ldots, 0}
                   % +\alpha_2\Ket{0, N(\omega_2), 0, \ldots, 0} 
                   + \cdots 
                       + \alpha\Ket{0, \ldots, 0, N(\omega_d)} \Big)
\nonumber \\ 
& &
           +\beta \Ket{N(\omega_A)} \otimes \Ket{0, \ldots, 0}.
\end{eqnarray}
where $\alpha, \beta$ are real and $d \alpha^2+\beta^2=1$. 
%
%For simplicity,  we restrict to a symmetrised 
%problem in which $\alpha_1= \ldots =\alpha_d:=\alpha$. 
%Note that in both strategies we made the following basic assumptions:
%the input frequencies, $\omega_1, \ldots, \omega_d$, 
%have been chosen beforehand and the same resources have been given to each channel (system) channel. 
%%Also assume without loss of generality that $\alpha$ and 
%%$\beta$ are real numbers. The normalization condition is thus 
%%$\alpha^2d+\beta^2=1$. 
%%Note that for the choice of the ancilla frequency, $\Omega_A$, any frequency would suffice.
%%These assumptions have a sound basis as they represent the situation of 
%%maximum prior ignorance. 
%Of course, one may improve the results by performing an optimisation over the amplitudes of the different frequency components. 
%However, this is difficult and detracts from the main result, so is left for 
%future consideration. 
%
%Let us now analyse  the performance of these states by calculating the QFI matrix, finding 
%$\mathrm{Tr}\left(I(\theta)^{-1}\right)$ and minimising with respect to $\alpha$ ($\beta$ is uniquely 
%determined by the normalization condition). 
The QFI matrix is given by 
\begin{eqnarray}
     F({\theta})_{l,m}&=&4N^2\Bigg[\alpha^2\sum_{i=1}^d
       \frac{\mathrm{d}\lambda_i}{\mathrm{d}\theta_l}\frac{\mathrm{d}\lambda_i}
        {\mathrm{d}\theta_m}-\alpha^4 \left(\sum_{i=1}^d\frac{\mathrm{d}\lambda_i}
        {\mathrm{d}\theta_l}\right)\left(\sum_{i=1}^d\frac{\mathrm{d}\lambda_i}
        {\mathrm{d}\theta_m}\right) \Bigg],\nonumber\\
        &=&4N^2\alpha^2J(\theta)^T\left\{{1}-d \alpha^2 Q \right\}J(\theta)
        \label{cools}
\end{eqnarray}
 where $Q$ is the orthogonal projection onto the vector $\mathbf{1}/\sqrt{d} =(1,\dots, 1)/^T\sqrt{d}$. As in the case of strategy 1 we can compute the 
 Cram\'{e}-Rao bound as 
\begin{eqnarray}
{\rm Tr}( F(\theta)^{-1}) 
&=&
 \frac{1}{4N^2 \alpha^2} 
 {\rm Tr}\left(J(\theta)^{-1} \left({1} + \frac{d\alpha^2}{1-d\alpha^2} Q\right) (J(\theta)^T)^{-1} \right)\nonumber \\
 &=&
 \frac{1}{4N^2\alpha^2} {\rm Tr}\left(J(\theta)^{-1} (J(\theta)^T)^{-1} \right)+  
 \frac{d}{4N^2(1-d\alpha^2)}  {\rm Tr}\left(J(\theta)^{-1}  Q  (J(\theta)^T)^{-1} \right)
 \nonumber \\
 &=&
 \frac{1}{4N^2\alpha^2} \frac{H(\theta)}{|J(\theta)|}+  
 \frac{d}{4N^2(1-d\alpha^2)}  {\rm Tr}\left(J(\theta)^{-1}  Q  (J(\theta)^T)^{-1} \right)
 \nonumber \\
 &=&
\frac{H(\theta)-\alpha^{2} K({\theta})  }
                     {4N^2 |J(\theta)|^2 (1-d\alpha^2)\alpha^{2}},   
                     \label{cooler}
 \end{eqnarray}
where 
$$
H(\theta) := \|P(\theta)\|_2^2, \quad K({\theta})=\frac{1}{2}\sum_{j=1}^d\sum_{l,m=1}^d(P(\theta)_{j,l}-P(\theta)_{j,m})^2
$$ 
with $P(\theta)$ the cofactor matrix of the Jacobian of $f({\theta})$.
%This details of this calculation may be found in  \ref{appendixB}. 
The minimum over $\alpha$ of the lower bound  in Eq. (\ref{cooler}) can be evaluated as
\begin{eqnarray*}
\min_{\alpha}\mathrm{Tr}(F(\theta)^{-1})%:=\mathrm{Tr}(F(\theta)^{-1})_{opt}
=\frac{\left(2dH(\theta)- K(\theta) \right)+\sqrt{4dH (\theta) \left(dH(\theta) -  K(\theta) \right)  } }{4N^2| J(\theta)|^2},
\end{eqnarray*}
and is attained at 
\begin{eqnarray}\label{coolerer} 
\alpha^2_{opt}=\frac{2dH(\theta)+\sqrt{4dH(\theta)\left(dH(\theta) -  K(\theta)\right)  } }{2dK(\theta) }.
\end{eqnarray}
%
%Thus the best input state is given by Eq. (\ref{bluet}) with 
%$\alpha_i=\alpha_{opt}$, $\forall i$. 
By using the definitions of $H(\theta)$ and $K(\theta)$ one can check that 
$dH(\theta) - K(\theta)=\sum_{j=1}^d(\sum_{i=1}^dP_{ji})^2\geq0$, and further that $2dH(\theta)+K(\theta)\geq\sqrt{4dH(\theta)(dH(\theta)-K(\theta))}$. This implies
%$$2dH+K\geq\sqrt{4dH(dH-K)}\,\,\,\, \mathrm{or \,\,equivalently} \,\,\,\,K(K+8Hd)\geq0,$$
%it follows that 
%\begin{eqnarray*}
%    \min_{\alpha}\mathrm{Tr}(F(\theta)^{-1})
%        & =&\frac{4dH(\theta)+\sqrt{4dH(\theta)\left(dH(\theta) -  K(\theta)\right)  } -2dH(\theta)-K(\theta)}
%                    {4N^2\left(\mathrm{Det}J(\theta)\right)^2}
%\nonumber \\ 
%&\leq & 
%          \frac{dH(\theta)}{N^2(\mathrm{Det}J(\theta))^2}.
%\end{eqnarray*}
$$
    \min_{\alpha}\mathrm{Tr}(F(\theta)^{-1})
    \leq 
     \frac{dH(\theta)}{N^2 |J(\theta)|^2}.
$$

\subsubsection{Comparison of  Strategies 1 and 2}

We summarise the two results on the quantum Cramer-Rao lower bounds obtained with the two strategies
\begin{itemize}
\item[-] with  strategy 1: 
$ \mathrm{Tr}(F(\theta)^{-1})=\frac{d^2H(\theta)}{N^2 |J(\theta)|^2}$
\item[-] with strategy 2:
$ \mathrm{Tr}(F(\theta)^{-1})\leq\frac{dH(\theta)}{N^2 |J(\theta)|^2}$. 
\end{itemize}
This indicates that by using entangled probe states one may be able to reduce the MSE by a factor $d$ compared with product states. %This result is rather surprising 
%considering that each frequency acts independently in a PQLS. 
This result extends that of \cite{Humphreys1} (see Sec. \ref{class}). which can be regarded as a special case when $H(\theta)=d$, 
$K(\theta)=d^2-d$ and $ |J(\theta)|=1$. 
%We also note that the same scaling with $N, d$ holds if if the parametrisation with $\theta$ is ``regular'' in the sense that
%$$
%c {1}  \leq J^TJ \leq C {1}
%$$
%for some constants $0<c\leq C$.
%Note that the metrology result that we discussed in section \ref{class} may be interpreted as exactly our linear system but estimating the $\lambda$s rather that the $\theta$s. Moreover our results can be seen to be more general because the one parameter per channel metrology setup is a special case of our results with the values $H=d$, $K=d^2-d$ and $\mathrm{Det}J=1$.

% Of course the POVM measurement realising the Type-(2) bound may depend on the parameter and for this (and other) reason(s) may not be physical. 

Generalising the above results to MIMO systems is more involved. Entangling between modes is more difficult to analyse because these channels are  not independent. Alternatively, using only one of these channels for each frequency causes problems as the setup behaves like the noisy setup as information is lost in the remaining channels. The most interesting question here is whether it is  possible to use entanglement between the channels to gain an estimation improvement.

%%%%%%%%%%%%%%%%%%%%%%%%%%%%%%%%%%%%%%%
\subsection{Squeezing-Based Realization}
\label{realise}
%%%%%%%%%%%%%%%%%%%%%%%%%%%%%%%%%%%%%%%

Because the N00N states that we have used here aren't very practical (for the reasons discussed in Sec. \ref{prevol}), we would instead like to develop a realistic scheme 
to achieve this $\mathcal{O}(d)$  improvement in performance, which exhibits the properties of a good choice of probe state (outlined in Sec. \ref{sec.1D}). 
 We consider here  the possibility of using     squeezed-coherent states like the ones in Sec. \ref{Toes1} to achieve this goal.

Firstly, a squeezing-based realization of the type-(1) strategy would need to have
squeezing and displacement operations in $d$ experiments to estimate $\lambda_1, \ldots, \lambda_d$. Therefore, let each experiment have power $E/d$, split equally between the squeezing and displacement operations (i.e., $E/2d$).
Hence through the $d$ experiments the total amount of energy 
is $(E/2d + E/2d)\times d = E$. 
Now,  the variance $\left<\Delta \lambda_1^2 + \ldots + \Delta\lambda_d^2\right>$ 
is of the order of 
\[
      \frac{1}{\alpha^2 e^{2r}}\times d = \frac{1}{(E/2d) (4\times E/2d)} \times d
       = \frac{d^3}{E^2}. 
\]
Therefore, the precision is of the order $d^3/E^2$.

On the other hand, if there is an  ${\cal O}(d)$ improvement to be had we must consider using  entanglement between the frequency channels. Such entanglement, which is a non-linear property  can be achieved with current technology using a  {\it frequency-comb entangled state} \cite{Medeiros1, Roslund1}. A special OPO  called the  SPOPO can be used to generate  a multiple set of squeezed fields. In the squeezing-based realization of the type-(2) strategy the goal is to estimate the $d$ displacement operations $\lambda_1, \ldots, \lambda_d$ (see Sec. \ref{Toes2}) using a single squeezing operation. Suppose that  the squeezing operation has power  $E/2$ while $E/2d$ power is given to each displacement operation, so that 
the total amount of energy is $E/2 + (E/2d)\times d = E$. 
The variance $\left<\Delta \lambda_k^2\right>$ will then be of order  
$1/\alpha^2 e^{2r} = 1/(E/2d) (4\times E/2)$, hence the total variance 
$\left<\Delta \lambda_1^2 + \ldots + \Delta\lambda_d^2\right>$ is of the 
order of 
\[
      \frac{1}{\alpha^2 e^{2r}}\times d = \frac{1}{(E/2d) (4\times E/2)} \times d
       = \frac{d^2}{E^2}.
\]
Hence there will be an  ${\cal O}(d)$ improvement over the type-(1) strategy. However  the difficulty here is in developing a squeezed state with the property that it can be used to estimate $d$ displacements each with variance $E/2$ under total energy $E/2$. 
Whilst investigating this problem we became aware of  the work \cite{Proctor1}.
The essential point  there was  that if one goes beyond the class of fixed photon number input states  (and take average energy as the resource constraint) then  for any type-(2) strategy there will always exist a type-(1) strategy at least as good as it (in terms of MSE). This suggests that such  an ${\cal O}(d)$ enhancement in the class of squeezed and coherent states is not possible. It also raises questions over the efficacy   of the result \cite{Humphreys1},  as well as our results in Sec. \ref{simultaneous}. Frequency entangled states are seemingly only advantageous within the class of fixed photon number states.
The critical question is therefore: is it physically relevant to
consider only fixed total particle number probe state? In quantum optics indefinite photon number states are very natural (e.g coherent states), therefore the answer to this question is most likely negative. 

%[[[CAVEAT, the noon ones are not best for multiple parameter...unequal superposition]]]

%We would argue that, for optical sensing, the answer to this is no. This is because indefinite photon number states are the norm in optics (e.g., coherent states). It is also possible to consider optimizing over probe states from some other, perhaps more physically well-motivated, sub- space. For example, Gagatsos et al. [16] consider only Gaussian probe states. Interestingly, they conclude that SE is of limited benefit under these conditions.

\section{Open Problem: Optimising the Input Over Frequency}\label{varfreq}

A problem for future research is the general multiple parameter set-up and how  to optimise over the frequency of input states.  In the one-parameter case, it was understood that the best strategy was a monochromatic probe state. For $d$-parameters, this problem is much more subtle. Although it is expected that it $d$-frequency probe state is the best strategy it is not clear how to choose the frequencies.  For example, for each unknown parameter there is an optimal frequency in the single-parameter model but when $d$ of these are combined into  a $d$-parameter model it may be the case that a different choice of frequencies may indeed work better.  
We see an example of this now. 
 
 \begin{exmp}
 Consider the one-mode SISO PQLS characterised by unknown parameters $(c, \Omega)$, which we would like to estimate. Consider also the setup from Sec. \ref{sausage1} where the input is supported on  two frequency modes, 
with frequencies  $\omega_1$ and $\omega_2$.
Also, as in Sec. \ref{sausage1} suppose that the coherent inputs have real amplitudes and the squeezed inputs are momentum-squeezed. We assume for simplicity that the energy of the input supported on each frequency  is the same (i.e. $E/2$) and $E_{\mathrm{coherent}}=E_{\mathrm{squeezed}}$ on each channel. 
Let us optimise over the  frequencies  $\omega_1$ and $\omega_2$ under a fixed energy constraint using a momentum measurement on each mode.

Firstly, the optimal frequencies for the single parameter experiments to estimate either $c$ or $\Omega$ are given by $\omega=\Omega\pm\frac{1}{2}|c|^2$ or $\omega=\Omega$, respectively (i.e. minimising the MSE \eqref{tofu} with respect to a total energy constraint  as  in Sec. \ref{sec.1D}). 

On the other hand the Fisher Information for the setup here is given by 
$$F(|c|^2, \Omega|\omega_1, \omega_2)=\frac{E^2}{4}\left(\begin{smallmatrix}p^2_1+p_2^2&p_1q_1+p_2q_2\\p_1q_1+p_2q_2&q_1^2+q_2^2\end{smallmatrix}\right),$$
 where $$p_i:=\frac{-|c|^2}{(\Omega-\omega_i)^2+\frac{1}{4}|c|^4} \quad\mathrm{and}\quad q_i:=\frac{(\Omega-\omega_i)}{(\Omega-\omega_i)^2+\frac{1}{4}|c|^4}.$$
 We have also chosen for convenience to estimate $|c|^2$, which can be done without loss of generality because the phase is not identifiable from the transfer function (see Example \ref{probes} or Corollary \ref{colo}).
 It follows that the MSE in this case is given by 
 \begin{align*}
&M(|c|^2, \Omega|\omega_1,\omega_2):= \mathrm{Tr}(F(|c|^2, \Omega)^{-1})\\
&=\frac{4}{E^2}\times
 \frac{\left(|c|^4+(\Omega-\omega_1)^2\right)   \left(\frac{1}{4}|c|^4+(\Omega-\omega_2)^2\right) +\left(|c|^4+(\Omega-\omega_2)^2\right)   \left(\frac{1}{4}|c|^4+(\Omega-\omega_1)^2\right)}{  \left(\frac{1}{4}|c|^4+(\Omega-\omega_1)^2\right)^2 \left(\frac{1}{4}|c|^4+(\Omega-\omega_2)^2\right)^2}.
 \end{align*}
 Suppose that the true values of the parameters are  $|c|^2=1$ and $\Omega=3$. A plot of $M(1,3|\omega_1, \omega_2)$ is given in Fig. \ref{chevy} (a) as a function of the frequencies $(\omega_1, \omega_2)$. It is clear from the diagram that the minimum values of the MSE is not at $(3,3\pm0.5)$, as would be the case for the single parameter experiments. In fact, using a numerical solver  a global minimum can be  found  at $(3.4278, 2.5722)$ (note that another can be found by switching $\omega_1\longleftrightarrow\omega_2$). 
Furthermore, suppose that we fix the parameter $\omega_1=3$, which recall  is the optimal choice in the single parameter experiment for estimating $\Omega$. A plot of the  function $M(1,3|3,\omega)$ as a function of frequency $\omega$ is given in  Fig. \ref{chevy} (b). Observe that  the minima deviate  from $\omega=3+0.5$; they  are actually given by $\omega=3\pm2^{-3/4}$.

\begin{figure}
\centering
\centering\includegraphics[scale=0.22]{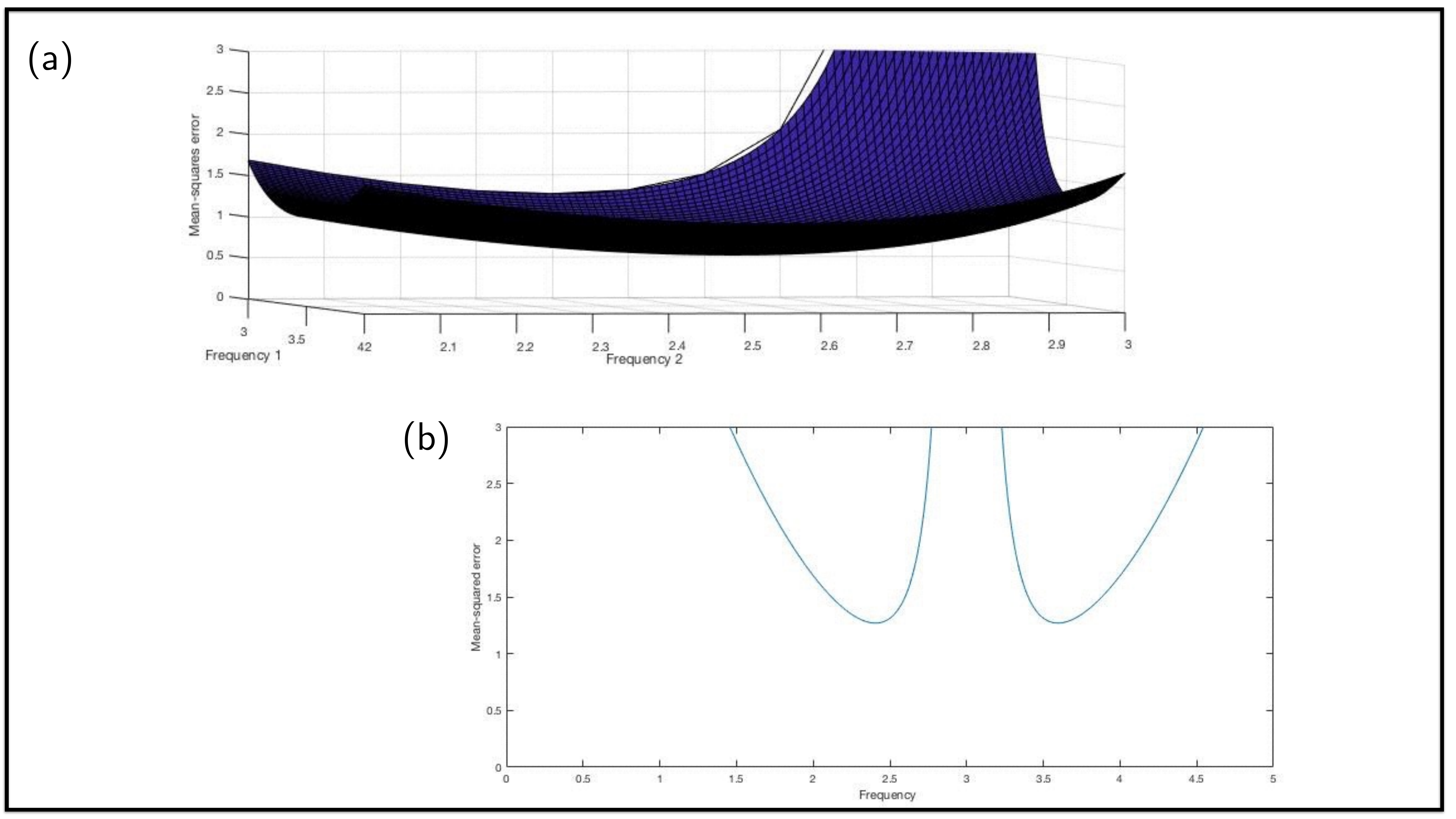}
\caption{(a): Plot of the MSE as a function of frequencies $(\omega_1, \omega_2)$. (b) Plot of MSE with one of frequencies fixed at $\omega=3$.
\label{chevy}}
\end{figure}

\end{exmp}

%could be more complicated how to weight the inputs also 

 \section{Conclusion}

In this chapter  we have developed a method allowing us to identify  unknown parameters (i.e., single and multiple parameters) of a given PQLS at the Heisenberg limit using squeezed states. The emphasis was on realistic Heisenberg detection schemes. In order to achieve such large estimation precision one exploits the decreased uncertainty in one of the quadratures that squeezed states possess. The action of the PQLS is a displacement, which can be detected by performing homodyne detection at the output. We are not currently aware of any other strategy in the literature achieving the Heisenberg limit for general  PQLSs (other than the single parameter SISO case, which is essentially equivalent to the the standard metrology scheme). We found the optimal state of squeezed-coherent type for the single parameter case. 
%For SISO systems the result is optimal, whereas for MIMO systems Heisenberg scaling is achieved and only misses the quantum Cramer-Rao lower bound by a constant factor. In fact it seems highly plausible that this constant factor may be improved adaptively; by first measuring one of the elements of the transfer function, thus achieving a rough estimate of the value of the unknown parameter, before diagonalising the PQLS (rotating by a given unitary) at this estimated value of $\theta$ and applying the technique in Section \ref{ft} in order to achieve an even better estimate. At both stages of this procedure one is guaranteed Heisenberg scaling so the combined procedure must also exhibit Heisenberg scaling. 
Also for the  the single parameter case a MIMO example for two atomic ensembles was given, where a particular parameter determines whether or not a macroscopic system possesses  entanglement. The parameter values are typically very small and, as a zero value corresponds to no entanglement, distinguishing it from zero is crucial. Therefore high precision is required. A possible experimental set-up was also given.

%For the one parameter case it was justified that it was best to use monochromatic probe states. Moreover, the optimal choice of frequency depends on the unknown parameter. However, the scaling with $E$ does not rely on the frequency choice of the input and so in particular any other choice will only affect the scaling by a constant. As mentioned earlier, in practice one adaptively tunes the frequency to improve the precision.

% In order to identify parameters in a multi-parameter SISO PQLS one must consider inputs of more than one frequency. One can achieve scaling of the order $E^2$ simply by repeating the single parameter strategy at the post-processing stage for a sufficiently broad range of frequencies. Perhaps more interestingly than this, if one exploits frequency entanglement, then this single experiment has an $\mathcal{O}(d)$ advantage over $d$ optimal (in terms of fixed photon number) single parameter experiments. We discussed in Sec. \ref{rotten} the possibilities for extending these results to  MIMO multi-parameter QLSs.
%One may improve this further by up to a factor of two if one uses least twice as many frequencies as unknown parameters. The difficulties in extending this result to MIMO systems were highlighted.

% result also works the same for MIMO PQLSs. A corollary of these results is that for multi-parameter MIMO systems, probe states entangled between the field channels surprisingly offer an advantage in estimation precision. NOT SURE ON LAST PART...NEED TO DISCUSS

We have extended the result of \cite{Humphreys1} to SISO PQLSs, by showing there is an $\mathcal{O}(d)$ advantage by estimating $d$ parameters using frequency entangled states versus the alternative of frequency separable states. We pointed out that their results and ours are only valid under the assumption of fixed photon inputs, which ultimately led to our failure in developing an analogous  realistic scheme using Gaussian states. It still remains an open problem whether frequency entangled states offer any advantage from a metrological point of view, even in the SISO case. By the results of \cite{Proctor1} it is clear that simply estimating the phases at different frequencies (i.e the $\lambda(\omega|\theta)$s rather than the $\theta$s) there is no advantage. However this does not necessarily imply that there cannot be an advantage in estimating the $\theta$s.

%pointed out the problems with this in practice...highlighted that may not be possible in practice

%Much of the analysis for multi-parameter models was done using N00N states, which as we have seen, are not very realistic. Instead a squeezing-based realisation using a frequency-comb entangled light field was given, that has nearly the same performance as the photon-entanglement scheme. However, its great advantage is that it can be generated using only a \textit{single} special type of OPO called an SPOPO, and so saves on the number of squeezing operations. A simple measurement choice is also presented that realises the Heisenberg QFI scaling. Note, that there are also good measurement choices for the photon based realisation, see \cite{Humphreys1}, hence the advantage here is in the physical realisability of the state. 

We outlined various open problems in Secs. \ref{realise} and \ref{varfreq} concerned with the optimisation problem in the multiple parameter case and allowing for optimisation over frequency, respectively. Finally another interesting open problem is to extend this work to active quantum linear systems (AQLS), where system identification remains an open problem to-date.

%[[[USE THE FOLLOWING PARA to explain that not clear whether better to use entanglement in SISO case. Proof used it in that paper.
%``This inequal- ity is perhaps unsurprising as Fkk is the QFI for the
%  one-parameter problem that the multi-parameter prob- lem reduces to if all of the other parameters are known [42]. It is intuitively clear that the optimal setting for estimating a parameter is when the values of all other variables are known exactly.''
%COMBINE IT WITH 
%``whilst in the metrology problem one is able 
%to extract information about a single parameter from a single two-mode 
%interferometric experiment, the PQLS parallel of this, being able to extract 
%information about a single parameter using a single frequency experiment 
%is not possible. Despite this difference the SLDs commute for fixed photon inputs.'' ]]

%Naoki dissipation....conditions on dark state explanation..my result also extends this as can create any arbitrary mixed state...also replicates naoi's result because create one mode squeezed state and can take arbitrary symplectic of it-williamson's theorem...can probably do same with one-mode thermal stares.

\chapter{Feedback Control Methods for Parameter Estimation in QLSs}\label{FEDERER}

In Ch. \ref{QEEP} we considered parameter estimation for PQLS under an energy resource constraint. In this chapter, we  
 take time as our main resource  and consider the same estimation goals. 
 
In general, information about a QLS (or a parameter therein) is obtained at a linear rate with respect to time, as evidenced  by the result in Sec. \ref{JUKKA1}. However, we shall see that when the eigenvalues of the system matrix $A$ (or equivalently the poles of the transfer function) are close to the imaginary axis, so that the system destabilises, the QFI is enhanced and scales quadratically with the observation time. Being more precise the QFI scales as $T_{\mathrm{tot}}^{2(1-\epsilon)}$ and is valid for observation times of the order $T_{\mathrm{tot}}=\tau^{\frac{1}{1-\epsilon}}$, where $\tau$ is the correlation time (or stabilisation time). The constant $\epsilon>0$  essentially ensures that the system can be considered stationary during the experiment (see Sec. \ref{toy}). For times much longer than the correlation time the linear scaling will be restored.

In our setup the enhancement in precision arises from  the internal system being \textbf{almost} decoherence-free. More generally, a  \textit{decoherence-free subsystem} \cite{Yamamoto1, Kasia1} is part of the system that   doesn't feel the effect of the environment (noise) and evolves as a closed system. The study of DFSs has been a fruitful subject in recent years; there are a vast array of applications, such as  \textit{quantum computation} \cite{Beige1}, \textit{quantum memories} \cite{Yamamoto4, Kiepinski1} and \textit{quantum metrology} \cite{Kasia2, Dorner1} (see also \cite{Yamamoto1} for a study in the context of QLSs).
Our results  highlight further the importance of DFSs as a resource for quantum metrology.  
This Heisenberg-level scaling with respect to time was also observed in \cite{Kasia1, Kasia2} for non-linear quantum systems when the system undergoes a \textit{dynamical phase transition} (DPT). A DPT is a singular change in the dynamics associated with the vanishing of the spectral gap. In particular, the system displays an intermittent behaviour, switching between periods of high and low emission rates. We discuss this in greater detail in Sec. \ref{dico}. 

Firstly, we consider using  time-dependent inputs in PQLSs and give an adaptive procedure to obtain the `Heisenberg level' scaling when one (or more) of the poles are very close to the imaginary axis (Sec. \ref{piles}). 
An adaptive procedure is necessary because at each stage the optimal frequency of the input depends on the unknown parameter. We  give two feedback methods enabling one to destabilise a PQLS. The feedback methods are based on either isolating one mode in the system or even the entire system (so that all poles are close to the imaginary axis in this case). The method of isolating the entire system is the most straightforward and is achieved by increasing the reflectivity of the mirrors coupling the field(s) to the system using beam-splitters, which should be achievable in practice. The method of destabilising one mode  is perhaps more ambitious and is  facilitated with coherent feedback. The final ingredient in a metrological protocol is a suitable choice of measurement. We show that simple (adaptive) homodyne detection works here.

Finally, we consider using stationary inputs as a resource for quantum metrology for  PQLSs in Sec. \ref{dfsw}. We show again that an unstable PQLS  results in  Heisenberg time scaling with respect to the QFI.

  %[[[SAY MORE: here we think more generally with time scaling. Energy not energy density was our resource there....also there we are using adaptive procedure to estimate phase rather than specific parameters (i.e argument in diff between $e^{i\theta-i\theta_0}$ and $e^{if(\theta-i\theta_0)}$]]].

%[[[COMMENT ABOUT NOT INTERESTED IN S because is in field so behaves as if closed.

%[[[COMMENT THAT WILL DISCUSS IN BOTH TIME AND STATIONARY REGIME]]]

%[[SUMMARY OF EXACTLY WHAT WE DO]]]

\section{Problem Formulation and Preliminary Investigation}\label{problem}

We saw in Ch. \ref{WEEP} that in linear systems theory one can only access the system  using probe states and measurements via the field. Therefore, any metrological scheme will be \textit{indirect} in this sense. 
 This situation is contrary to the quantum metrology setup in Sec. \ref{class}, where one is able to \textit{directly} estimate system parameters with states and measurements on the system (see Fig. \ref{CS}). Let us now undertake a preliminary investigation to  draw comparisons between direct and indirect metrology.

\subsection{Direct Metrology}\label{DIREC} 

Firstly we look at the direct metrology protocol. 

%[[[However, note that the point of this work is not to suggest that we should use indirect over direct ]]]

 \begin{figure}
\centering
\includegraphics[scale=0.17]{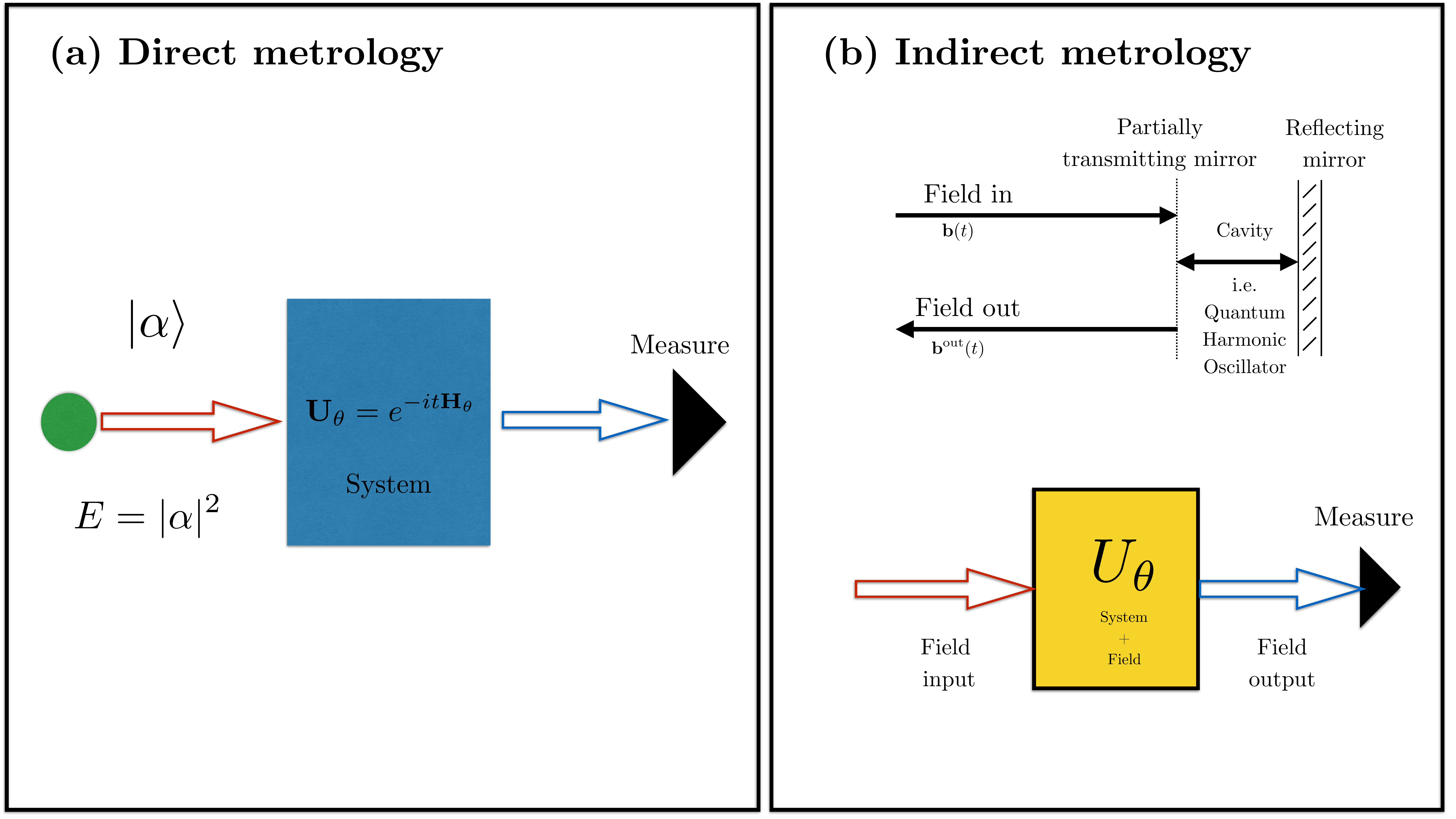}
\caption{ Comparison between direct and indirect metrology setups. In the direct setup, (a), the state evolution is closed and evolves unitarily. In the indirect metrology setup, (b), the system in (a) is coupled with an ancillary system, which can be prepared in an arbitrary state. As the information about $\theta$ leaks into the ancilla it can be measured.  Here  the ancilla system is a Bosonic bath.  \label{CS}}
\end{figure}

Consider the following closed system evolution: suppose that we have a continuous variables system, which is represented by a mode $\mathbf{a}$ satisfying the usual CCR $\left[\mathbf{a}, \mathbf{a}^{\dag}\right]=1$. Also, assume that system evolves under the Hamiltonian 
$\mathbf{H}_{\theta}:=\theta \mathbf{a}^{\dag}\mathbf{a}$
for a period of time $t$. In other words the evolution of the system is described by the unitary operator
\[\mathbf{U}_{\theta}=e^{-i\theta t\mathbf{a}^{\dag}\mathbf{a}}.\]
Suppose that the  parameter $\theta$ is unknown and we would like estimate it (see Fig. \ref{CS}).

Let us probe the system with a coherent state of amplitude $\alpha$, hence energy $E=|\alpha|^2$. After time $t$, the state evolves to  a coherent state of amplitude $\alpha e^{-i\theta t}$, i.e, the action is a rotation in the complex plane. How sensitive this rotation is to changes in $\theta$ characterises the eventual performance of any measurement. We can calculate this performance using the  QFI (see Sec. \ref{plm1}), hence by  Eq.  \eqref{qfiformula} we have 
%and letting it evolve under the unitary $\mathbf{U}_{\theta}$ for a time $t$ and then measuring one can obtain an estimate for $\theta$. The maximum amount of information that one can hope to obtain from any measurement is given by the quantum Fisher-information (QFI), which (by using (Need reference) for unitary rotational families) is given by 
\begin{align}
\nonumber F_{d}(\theta)&=4\mathrm{Var}\left(t\mathbf{a}^{\dag}\mathbf{a}\right)\\
& \label{bencH}=4 t^2 E.
\end{align}
Notice that $F_{d}(\theta)$  scales quadratically with time\footnote{We are not  interested in the scaling with energy here, but just note that this precision can be further improved by using CAT states for example (see Sec. \ref{class})). 
%  Indeed all results can be repeated for these Heisenberg regimes also [CAN THEY?].} 
}.
Furthermore, adaptive homodyne measurements can be shown to attain the QFI bound asymptotically. %[[[This needs researching as to whether such a measurement is possible because of phase degeneracy]]].

%LEADING ON SENTENCE>>>This should provide us with a benchmark for the QLS setup problem because it is expected that a direct metrology protocol will always be better than any indirect scheme. 

\subsection{Indirect Metrology}\label{indo}

If instead we perform an indirect measurement by coupling with a Bosonic field and performing  measurements on the field to estimate $\theta$, then it is expected the the QFI in this case will be worse. That is, not all information is expected to be transferred to the field and we get information loss. 
The restriction that we must measure indirectly, rather than directly, is one that is encountered generally in linear quantum systems theory as direct access to the system is not possible. 
In fact in this instance we have a passive QLS (under the usual rotating wave and Markov approximations) and in particular a cavity (see Example \ref{opticalcavity}).

Now let us show explicitly that the indirect strategy is worse than the direct one. 
Suppose that we have a  QLS that is characterised by coupling parameter $c\in\mathbb{R}$ and Hamiltonian $\Omega\in\mathbb{R}$. Note that both parameters are identifiable in the sense of Definition \ref{identy}. Suppose that we would like to estimate $\Omega$ (with $c$ assumed to be known). 
Using a coherent probe state of one frequency\footnote{Recall from Sec. \ref{tonga} that one frequency  is optimal \cite {Guta2}} then   the QFI  is given by (see Eq. \eqref{tonga2})
\begin{align}
\nonumber F_{ind}(\Omega| \omega_{\mathrm{opt}})&=4E\max_{\omega}\left|\frac{d\Xi(-i\omega)}{d\Omega}\right|^2  \\
\label{fedy8}&=4E\max_{\omega} \frac{c^4}{\left(\left(\omega-\Omega\right)^2+\frac{1}{4}c^4\right)^2} \\
\label{fedy1}&=64E\frac{1}{c^4},\end{align}
which is less than \eqref{bencH} in general\footnote{Note that in order to obtain this large QFI one must run the experiment for a time $t\gg1/c^2$  so that the system stabilises (see Sec. \ref{hastings}). On this long time scale the input-output dynamics become unitary in the frequency domain.}.  We now discuss a toy model that suggests that it may be possible to do considerably better than \eqref{fedy1} using an adaptive approach facilitated by feedback.

\subsection{Toy Model}\label{toy}

Consider again  the  one mode passive SISO system above. Suppose that we have a fixed energy resource, $E$ and time resource, $T_{\mathrm{tot}}$, for metrology. 
Further, suppose that  we are also allowed to choose the value of $c$. Although this may seem artificial, later we will give two feedback schemes designed to implement this. 

Choose $c$ such that  $c^2=\mathcal{O}(T_{\mathrm{tot}}^{-{(1-\epsilon)}})$, and consider running the experiment for a time $T_{\mathrm{tot}}$. Observing the experiment over the time interval $[T_{\mathrm{tot}}^{1-\epsilon},T_{\mathrm{tot}}]$ (note here we are taking $T_{\mathrm{tot}}$ to be sufficiently large so that the width of the interval $[T_{\mathrm{tot}}^{1-\epsilon},T]$ is much greater than the width of the interval $[0,T_{\mathrm{tot}}^{1-\epsilon}]$ and the system has reached stationarity  over $[T_{\mathrm{tot}}^{1-\epsilon},T_{\mathrm{tot}}]$ (see Sec. \ref{hastings}). It follows from Sec. \ref{indo} that the optimal QFI 
is proportional to  $E T_{\mathrm{tot}}^{2(1-\epsilon)}$. The upshot is that one achieves  (almost) the same scaling for estimation of the Hamiltonian parameter as when you have direct access to the closed system. The advantage here is that  it achieves this scaling indirectly (via the field).

A potential caveat here is that we need to prepare an input state (of coherent-type) on a frequency that is unknown.  We can overcome this difficulty  by using an adaptive procedure. 

\textbf{Adaptive procedure:}
\begin{itemize} 
\item Step 1 (Initial experiment): Run a prior experiment using the QLS to obtain a rough estimator for our unknown parameter using time resources $T_1=\frac{T_{\mathrm{tot}}}{3}$. That is, probe the system with a monochromatic coherent state  (of any frequency)  with total energy $E_1:=\tilde{E}T_1$, where $\tilde{E}$ is the energy density. It follows from above that the QFI is given by
$$
F^{(1)}(\Omega)\propto \tilde{E}T_{\mathrm{tot}}  $$
Perform a measurement to obtain an estimator of $\Omega$ with error  $\mathbb{E}\left[(\Omega-\hat{\Omega}_1)^2\right]=\mathcal{O}\left(\frac{1}{{T_{\mathrm{tot}} }}\right)$.
\item Step 2 (Feedback 1): Now run a second experiment using the QLS for a time $T_2=\frac{T_{\mathrm{tot}}}{3}$. Choose the 
 coupling parameter  to be $\hat{c}_2$ so that $|\hat{c}_2|^2=\mathcal{O}\left(\frac{1}{\sqrt{T_{\mathrm{tot}}}}\right)$. Probe the system with a monochromatic coherent state of frequency $\hat{\Omega}_1$ with total energy $E_2:=\tilde{E}T_2=E_1$. Since $\hat{\Omega}_1$ varies form the true value of $\Omega$ by  $\mathcal{O}\left(\frac{1}{{T_{\mathrm{tot}} }}\right)$, then the QFI at this stage is
\begin{align*}
F^{(2)}(\Omega)&=4E_2\left.\left|\frac{d\Xi(-i\omega)}{d\Omega}\right|^2\right|_{\omega=\hat{\Omega}_1}\\
&\propto \tilde{E}\cdot T^2_{\mathrm{Tot}}\end{align*}
%
%Now the total time after Step 2 is $T_{\mathrm{tot}}=T_0+T_1=\mathcal{O}(T_1)$. Also assuming that the total energy in step 2 is much larger than that of step 1, which equates to $\tilde{E}_2>>\sqrt{E}_0$, it follows that 
%
%$$F^{(2)}_{\mathrm{opt}}(\Omega)\propto E_{\mathrm{tot}}\cdot T^{2(1-\epsilon)}_{\mathrm{tot}}$$. 
%
%Note that $\Omega-\hat{\Omega}_1=\mathcal{O}\left(\frac{1}{\sqrt{E_1\cdot E_0}}\right)$.
%
%
%
%\item Step 3 (Feedback 2):
%Choose the coupling parameter in this step to be $\hat{c}_2$ so that $\hat{c}_2=\mathcal{O}\left(\frac{1}{ \left(F^{(2)}_{\mathrm{opt}}(\Omega)^{1/4}\right)    }\right)$. The minimum stabilisation time of this experiment is therefore $T_2=\mathcal{O}\left(   \left(F^{(2)}_{\mathrm{opt}}(\Omega)^{1/2}\right)    \right)$. It follows that the optimal QFI after this stage has precision
%\begin{align*}
%F^{(3)}_{\mathrm{opt}}(\Omega)&=4E_2\max_{\omega}\left|\frac{d\Xi(-i\omega)}{d\Omega}\right|^2\\
%
%&\propto E_2\cdot T_2^2\end{align*}
%where $E_2$ is the total energy in Step 3, which is assumed to be much larger than the energy in Step 2. The optimal choice of $\omega$ in this stage was $\omega=\hat{\Omega}_1$. Now the total time, $T_{\mathrm{tot}}$, after Step 3 is $T_{\mathrm{tot}}=T_0+T_1+T_2=E_0+\sqrt{E_0}+\sqrt{E_0\cdot E_1}  \approx T_2$, provided $E_1\geq E_0$.     Hence $F^{(3)}_{\mathrm{opt}}(\Omega)=E\cdot T_{\mathrm{tot}}^2$. Note that $\Omega-\hat{\Omega}_2=\mathcal{O}\left(\frac{1}{\sqrt{E_2\cdot T_2^2}}\right)$.
%
%
Perform a measurement to obtain an estimator of $\Omega$ with error $\mathbb{E}\left[(\Omega-\hat{\Omega}_2)^2\right]=\mathcal{O}\left(\frac{1}{{T^2_{\mathrm{Tot}} }}\right)$.
\item Step 3 (Feedback 2): Now run a third experiment using the QLS for a time $T_3=\frac{T_{\mathrm{tot}}}{3}$. Choose the 
 coupling parameter  to be $\hat{c}_3$ so that $|\hat{c}_3|^2=\mathcal{O}\left(      \frac{1}{   \left(      \sqrt{      T^2_{\mathrm{Tot}}      }       \right)^{1-\epsilon}           }                   \right)  $. Probe the system with a monochromatic coherent state of frequency $\hat{\Omega}_2$ with total energy $E_3:=\tilde{E}T_3=E_1=E_2$. Since $\hat{\Omega}_2$ varies form the true value of $\Omega$ by  $\mathcal{O}\left(\frac{1}{{T^2_{\mathrm{Tot}} }}\right)$, then the QFI at this stage is
\begin{align*}
F^{(3)}(\Omega)&=4E_3\left.\left|\frac{d\Xi(-i\omega)}{d\Omega}\right|^2\right|_{\omega=\hat{\Omega}_2}\\
&\propto E_3\cdot T^{2(1-\epsilon)}_{\mathrm{Tot}}\end{align*}\\
Perform a measurement to obtain an estimator of $\Omega$ with error $\mathbb{E}\left[(\Omega-\hat{\Omega}_3)^2\right]=\mathcal{O}\left(\frac{1}{E_3{T^2_{\mathrm{Tot}} }}\right)$.
\end{itemize}
%In this procedure we have split the total time into three experiments of equal lengths. In each experiment we   choose the coupling so that it is proportional to the square root of the error in the previous step. 
%In each step we are able to adaptively tune the frequency of the input to increase the sensitivity of the current step. 

%Crucially this procedure doesn't require the (circular) assumption about prior knowledge of the unknown parameter that was required at the start of the subsection; at each stage the sensitivity is improved by a polynomial factor of $T_{\mathrm{tot}}$. 
%The reader may be wondering whether this procedure can be continued to any power of $T_{\mathrm{tot}}$. However, it is a simple exercise to check that the improvement factor terminates if one is to perform another step, which not surprising following our discussion regarding the performance of  indirect metrology versus direct metrology. Indeed any further  stage would require a smaller coupling and as a result any advantage is lost because the experiment takes longer to stabilise.

In this procedure we have split the total time into three experiments of equal lengths. In each experiment we   have chosen the coupling so that $|c^2|$  is proportional to the square root of the MSE from the previous step (we discuss this relationship shortly). 
By  adaptively tuning the frequency in each step, one is  able to  increase the sensitivity of the current step; at each stage the sensitivity is improved by a polynomial factor of $T_{\mathrm{tot}}$.
Crucially this procedure doesn't rely on the (circular) assumption about prior knowledge of the unknown parameter  at the start of the subsection.
The reader may be wondering whether this procedure can be continued to any power of $T_{\mathrm{tot}}$. However,  any further  stage would require a smaller coupling to have any hope of improving precision, which in turn would require a longer time to stabilise (see Sec. \ref{hastings} for the definition of stabilisation time) and so any improvement would  therefore be lost (see Sec. \ref{dico} for a further discussion). Indeed, in  each stage the coupling is decreased until in stage 3 when it reaches the  smallest possible level so that the experiment still stabilises in time $T_{\mathrm{tot}}$ (recall that the system is stationary over the interval $[T_{\mathrm{tot}}^{1-\epsilon}, T_{\mathrm{tot}}]$). It shouldn't be too surprising that we cannot improve the precision any further given our discussion on indirect metrology versus direct metrology.
Finally, the reason why the coupling is chosen so that $|c^2$|  is proportional to the square root of the MSE from the previous step is because 
$\left|\frac{d\Xi(-i\omega)}{d\Omega}\right|^2$ 
is largest when $\omega$ is chosen to that $(\omega-\Omega)$ and $|c|^2$ are of the same order (see Eq. \eqref{fedy1}).

In summary, we have seen one example of an  enhancement in precision arising from  the internal system being \textbf{almost} decoherence-free. 
This method will work provided there exists a feasible method to vary the coupling and that a measurement can be found at each step realising the QFI scaling (we see both in the following).

%[[[IMPORTANT THAT KEPT ENRGY DENSITY CONSTANT THROUGHOUT]]]

\section{Feedback Method for PQLSs; Time-Dependent Approach}\label{piles}

In the following, we will extend the adaptive approach from Sec. \ref{problem} to general PQLSs. We will design two feedback methods allowing implementation of this toy model and the one we develop for PQLSs in a physically meaningful way. Finally, we show that a  homodyne measurement will enable us to achieve the desired bounds.

\subsection{Adaptive Procedure}

Recall that in the example given in Sec. \ref{problem} the main idea was to destabilise the system and create a system that is almost isolated from the field. We now discuss two possible generalisations of this to arbitrary PQLSs.
\subsubsection{Isolating the Entire System}

One was to generalise this is the following: consider the system 
\begin{equation}\label{fedr1}
 \mathcal{G}=      \left(   C=\left(\delta_1, ..., \delta_{n}\right), \Omega\right),
 \end{equation}
where all parameters depend implicitly on one unknown parameter,   $\theta$, that we would like to estimate (see Fig. \ref{WCS} (a)). We assume that all column vectors $\delta_i$ are sufficiently small (in a sense to be defined later). The result is that all poles of the transfer function of this system will be close to the imaginary axis.

 \begin{figure}
\centering
\includegraphics[scale=0.18]{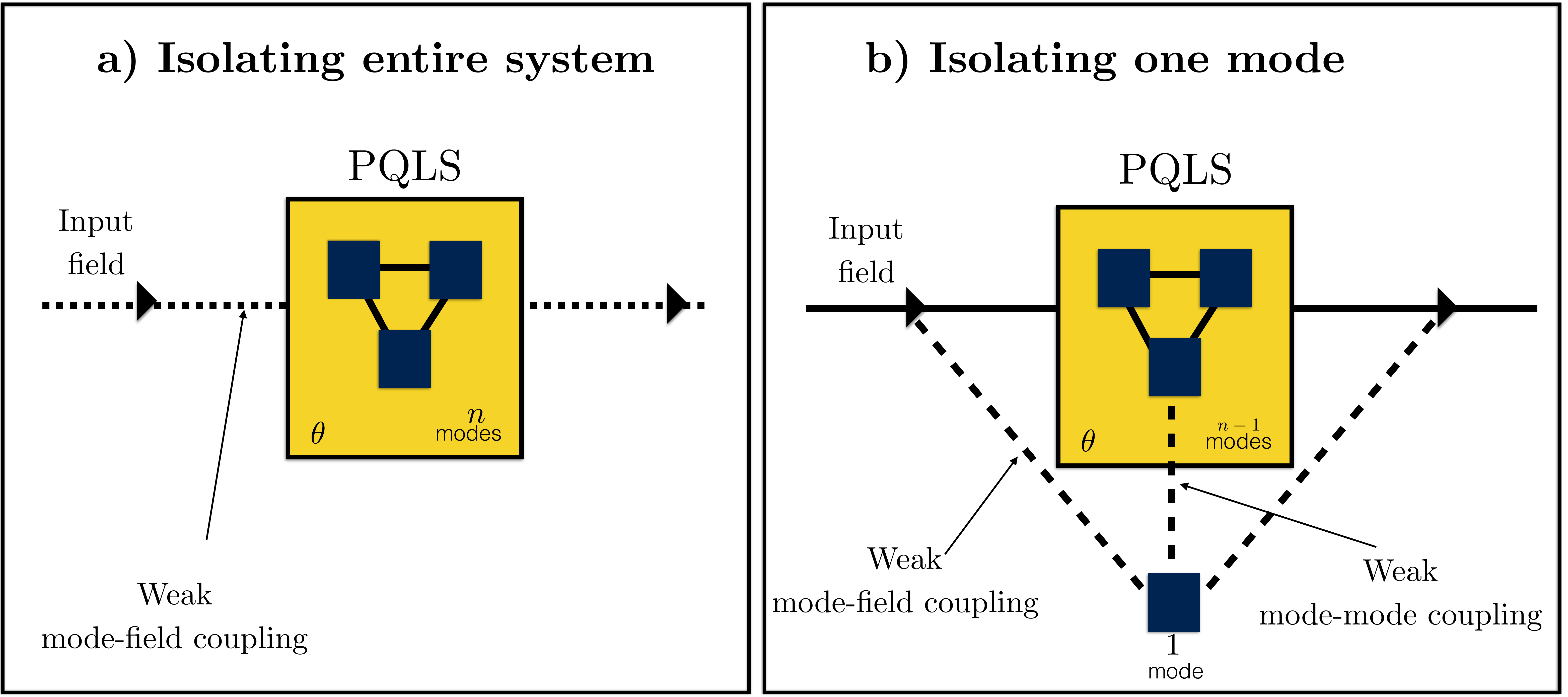}
\caption{ In a) the entire system has a small coupling with the field, whereas in b) one mode is weakly coupled with the field. \label{WCS}}
\end{figure}

\subsubsection{Isolating One Mode}

Alternatively, 
a more interesting setup  is to isolate and destabilise  one mode within the subsystem. Consider the setup in Fig. \ref{WCS} (b). We assume that one of the $n$ modes has both small direct (mode-field) and indirect (mode-mode) couplings.  The notion of  `small' will be made more precise later.   Such a system is characterised by 
  \begin{equation}\label{fedr2}
  \mathcal{G}=      \left(   C=\left(c_1, ..., c_{n-1}, \delta_1\right):=\left(c, \delta_1\right), \Omega=\left(\begin{smallmatrix}\Omega_1&\delta_2\\ \delta_2^{\dag}&\Omega_2\end{smallmatrix}\right)\right),
  \end{equation}
where all parameters depend implicitly on one unknown parameter,   $\theta$, that we would like to estimate. Here $\Omega_1, \delta_2, \Omega_2, \delta_1$ and $c_i$ are $n-1\times n-1, n-1\times 1, 1\times 1, m\times 1$ and $m\times1$ matrices, respectively. Note the similarities with the independent oscillator canonical form in \cite{Gough2}. 

In both of these setups the key feature is that there is a quasi-DFS within the system \cite{Yamamoto1}. We saw one example  in Sec. \ref{problem} where this sort of scenario was advantageous for estimation; we see shortly that this is also true in general.

%[[[Comment that neglecting scattering because  is in field. Therefore  direct metrology which is boring.]]] 

\subsection{Feedback Method 1}\label{dorm1}

We now design a feedback method enabling us to isolate the whole system, that is, drive it to the form \eqref{fedr1}. Our method is basically  to increasing  the reflectivity of the mirrors coupling the field(s) to the system using beam-splitters.
 \begin{figure}
\centering
\includegraphics[scale=0.13]{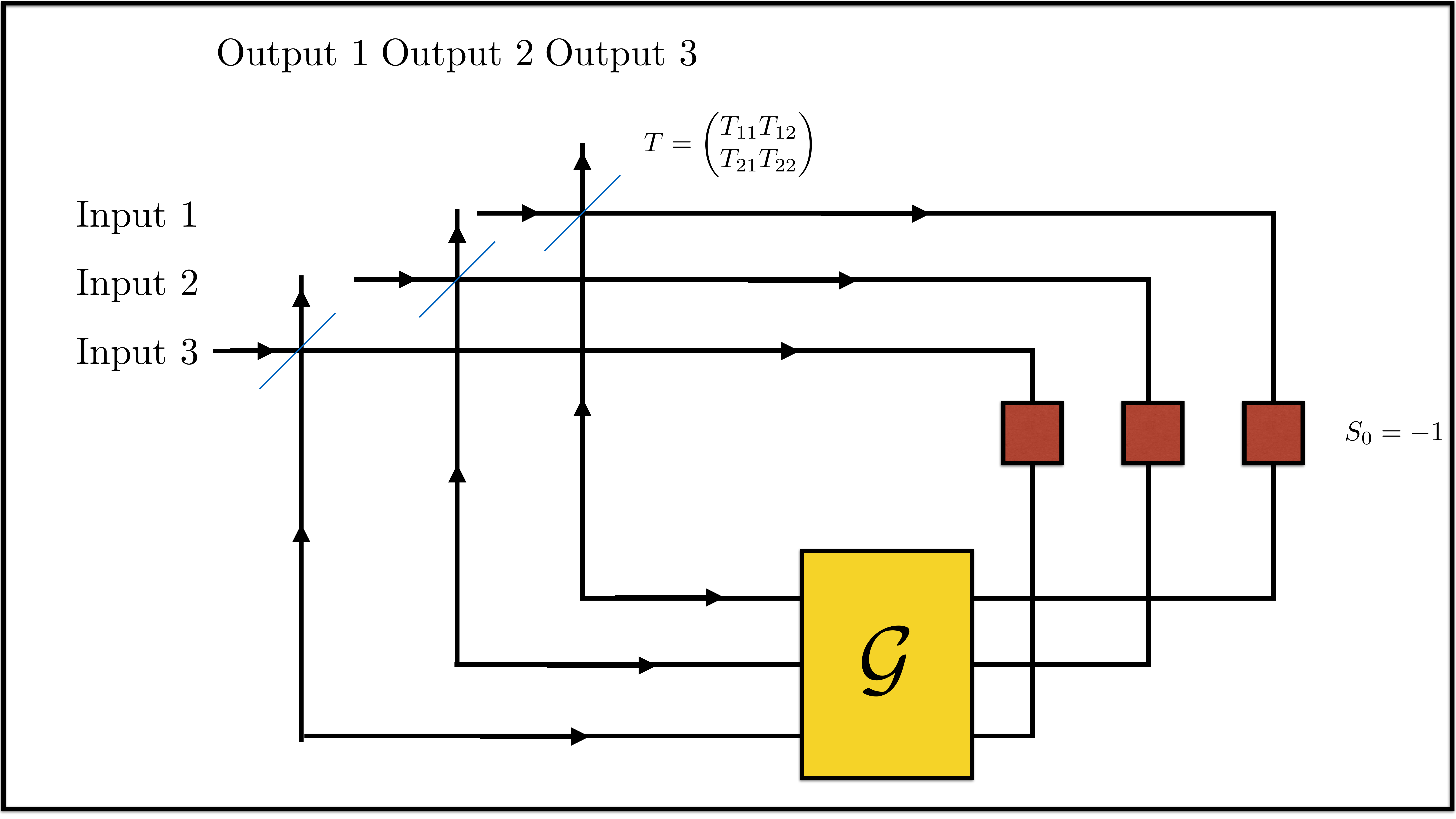}
\caption{ The setup in Sec. \ref{dorm1} where the PQLS (in this figure shown with three input channels) is placed in a circuit with beamsplitters in reduce the overall coupling
\label{FL}}
\end{figure}

Consider the setup in Fig. \ref{FL} where each input and corresponding output are one of the outputs and inputs respectively of a beam-splitter. Also there is a phase transformation on each channel as in the diagram. The result of adding these beam-splitters is the transformation 
\begin{equation}\label{cs}
C\mapsto C\times \left(T_{12}(T_{22}+1)^{-1}\right)
\end{equation}
\begin{equation}\label{SSS}
S=1\mapsto \left(T_{11}-T_{12}(T_{22}+1)^{-1}T_{21}\right)1
\end{equation}
and $\Omega$ remains unchanged. We have assumed here that all $m$ beamsplitters are identical and are described by a $2\times2$ unitary $(T_{ij})$ and all phases are given by $S_0=-1$ see \cite[Chapter 5]{Gough5}.
To calculate this, one uses the \textit{Feedback reduction rule} \cite{Gough6} (see \cite{Gough5} for the       calculation in the SISO case).

Now letting 
\[T=\left(\begin{matrix}\cos(\phi) \sin(\phi)\\-\sin(\phi) \cos(\phi)\end{matrix}\right)\]
it follows that in the limit of $\phi$ small the overall result of the beam-splitters is $C_{ij}\mapsto C_{ij}\times \phi/2$. Hence we are free to make $C$ as small as we like. 

\begin{remark}
As a result of the feedback there is a nuisance phase scattering in the field (see Eq. \ref{SSS}). However, since it is known (as it is the choice of the experimenter) it can be removed with a phase shifter in  the field. 
\end{remark}

%[[[NEED TO ASSUME IDENTIFIABLE SO ONE PARAMETER IN EACH EQUIV CLASS]]]

Now let us show how the adaptive procedure in Sec. \ref{toy} can be used with this feedback method. For simplicity, let us assume that the system is SISO (if it isn't then one can obtain a SISO system by setting $\phi=0$ in $n-1$ of the beam-splitters). 
As in Sec. \ref{toy},  split the time into three equal parts and choose  $\phi$ so that the all coupling parameters are  proportional to the fourth root of the MSE from the previous step. Therefore, after step 1 we have $\mathbb{E}\left[\left(\theta-\hat{\theta}_1\right)^2\right]=\mathcal{O}\left(\frac{1}{T_{\mathrm{tot}}}\right)$ and we choose the coupling in step 2 as $C_2=\mathcal{O}(T_{\mathrm{tot}}^{-1/4})$. 
Writing the transfer function in terms of the poles as
\begin{equation}\label{qpi}
\Xi_{\theta}(-i\omega)=\frac{   \left(-i\omega+\overline{z}_1\right)   }{\left(-i\omega-z_1\right)}\times..\times\frac{\left(-i\omega+\overline{z}_n\right)}{\left(-i\omega-z_n\right)},
\end{equation}
 this means that the 
 $n$ poles of the transfer function are at distances $\mathcal{O}\left(\frac{1}{\sqrt{T_{\mathrm{tot}}}}\right)$ from the imaginary axis. 
Furthermore, given that our estimate of $\theta$ is known with error variance $\mathcal{O}\left(\frac{1}{T_{\mathrm{tot}}}\right)$ it follows each of the pole locations are known (or can be estimated using our current estimate of $\theta$) to  the same error; we denote such estimates by $\hat{z}_i$. 

Now to calculate the QFI we need to evaluate $\frac{d\Xi(-i\omega)}{d\theta}$ at some particular choice of input frequency. If we take the input frequency to be $\mathrm{Im}(-\hat{z}_1)$ it follows that   
\begin{equation}\label{aaron2}
		\left.
\frac{d\Xi(-i\omega)}{d\theta}
\right|_{\omega=\mathrm{Im}(-\hat{z}_i)}=\mathcal{O}	\left(\sqrt{    {T}_{\mathrm{tot}}     }\right).
\end{equation}
This result follows immediately from the observation that the term
\begin{equation}\label{aaron1}
\frac{d}{d\theta}  \left( \frac{   -i\omega+\overline{z}_1   }{-i\omega-z_1} \right)=\frac{\dot{\overline{z_1}} (-i\omega-z_1)     +   \dot{{z_1}} (-i\omega+\overline{z}_1)      }{(-i\omega-z_1)^2},
\end{equation}
in the derivative is of order $\mathcal{O}	\left(\sqrt{    {T}_{\mathrm{tot}}     }\right)$, which follows because 
$$\left.\left(-i\omega-z_1\right)\right|_{\omega=\mathrm{Im}(-\hat{z}_1)}=\left(   \mathrm{Im}(\hat{z}_1)-\mathrm{Im}({z}_1)\right)-\mathrm{Re}(z_1)=\mathcal{O}	\left(\frac{1}{\sqrt{    {T}_{\mathrm{tot}} }    }\right).$$

			The QFI is thus identical to step 2 in Sec. \ref{toy}. Step 3 can be developed  similarly. Note that  evaluating these functions at $\omega=\mathrm{Im}(-\hat{z}_1)$ 				corresponds to a coherent state of frequency $\mathrm{Im}(-\hat{z}_1)$ and this choice of frequency is crucial here. Note that in calculating Eq. \eqref{aaron1} it is important that $\dot{z}_1\neq0$, otherwise Eq. \eqref{aaron2} will be $\mathcal{O}(1)$ and the adaptive enhancement is lost.

\subsection{Feedback Method 2}\label{iles}
We now design a feedback method enabling us to isolate one mode within the system; that is, drive it to the form \eqref{fedr2}. 

 \begin{figure}
\centering
\includegraphics[scale=0.18]{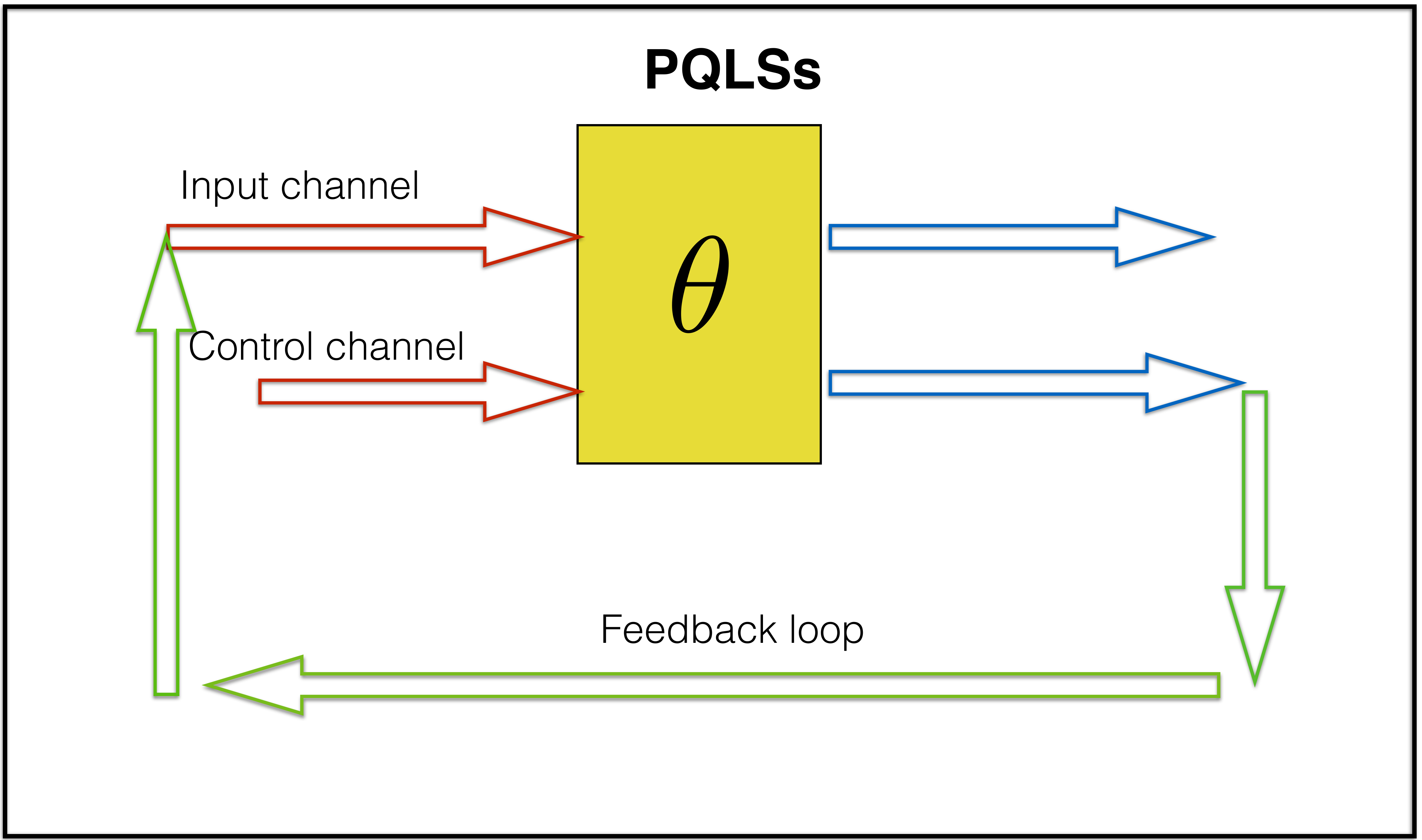}
\caption{ Feedback setup  for method 2.\label{FLP}}
\end{figure}

\subsubsection{SISO}

Consider an arbitrary SISO PQLS $(C, \Omega)$ in the setup in Fig. \ref{FLP}. We have adapted this into a 2I2O model by adding a new control channel and added a feedback loop (which is assumed to take  negligible time). Assume that in this additional control channel  we are able  to choose the  parameters.   A remark that the feedback loop is necessary as it ensures that the resultant model remains SISO.      
Note that this is an example of a coherent feedback scheme \cite{Peter1, James1} as all connections here are quantum. In particular, there is no measurement on the output of the second channel, but rather it is connected to the first input channel. Coherent feedback strategies such as the one here have been discussed in \cite{James1}.

%Feedback scheme with 2 different systems wouldn't work as just adds more poles.

Using the feedback reduction rule \cite{Gough5} the resultant model after connecting the feedback loop is given by 
\[\left(   C+D, \Omega+\frac{1}{2i}\left(D^{\dag}C-C^{\dag}D\right)\right),\]
where $D$ is the $1\times n$ matrix of control parameters corresponding to the extra channel. 
In our method, the goal  is to choose the $n$  control parameters in $D$ so that the resultant SISO model is of the form \eqref{fedr2}. This means that both
\begin{align}
\label{tr} \Omega_{jn}+\frac{1}{2i}\left(\overline{d}_jc_n-\overline{c}_jd_n\right)& \,\,\, \mathrm{for} \,\,\, 1\leq j\leq n-1,\\
&\label{tr2}c_n+d_n
\end{align}
need to be small.

Let us understand how the adaptive method in Sec. \ref{toy} can be combined with this feedback method to obtain enhanced scaling. As in  Sec. \ref{toy} split time into three equal parts. Step 1 is the same as Sec. \ref{toy}, meaning that we can estimate $\theta$ with MSE $\mathcal{O}\left(\frac{1}{{T_{\mathrm{tot}}}}\right)$. Consider the following choice of feedback:
\begin{align}
&d_n=-\hat{c}_n\\
&  \overline{d}_j=\frac{\hat{\overline{c}}_j(-\hat{c}_n)-2i\hat{\Omega}_{jn} }{\hat{c}_n}    \,\,\, \mathrm{for} \,\,\, 1\leq j\leq n-1,
\end{align}
where $\hat{c}_n$ and $\hat{\Omega}_{jn}$ are estimators of the $c_n$ and $\Omega_{jn}$, which describe the direct and indirect coupling to the $n$th mode. Also label  $\hat{c}_i$ for $i\neq n$ as the estimators of the couplings $c_i$ to the other modes. These can be estimated to a  MSE $\mathcal{O}\left(\frac{1}{{T_{\mathrm{tot}}}}\right)$ because they are determined by $\theta$. Recall  our adaptive procedure for the toy example where it was best to choose $|c|^2$ and $(\omega-\Omega)$ to be of the same order (rather than $c$ and $(\omega-\Omega)$). In light of this, in step 2 of the adaptive procedure   we  choose the 
 direct and indirect  coupling to the $n$th mode (eq. \eqref{tr} and \eqref{tr2}) to be proportional to   $\mathcal{O}\left(\frac{1}{{T^{1/4}_{\mathrm{Tot}}}}\right)$.
 This entails  using our estimates $\hat{c}_n$ and $\hat{\Omega}_{jn}$  (for $j\neq n$), which have  MSE
$ \mathcal{O}\left(\frac{1}{{T_{\mathrm{tot}}}}\right)$, and modifying them as 
 $\hat{c}_n \mapsto \hat{c}_n+\frac{1}{T^{1/4}_{\mathrm{Tot}}},      \hat{\Omega}_{jn} \mapsto  \hat{\Omega}_{jn}+\frac{1}{T^{1/4}_{\mathrm{Tot}}}$. The upshot is that  $\hat{c}_n$ and $\hat{\Omega}_{jn}$ will have
 MSE 
  $\mathcal{O}\left(  \frac{1}{    \sqrt{   T_{\mathrm{tot}}   }    }\right)$ and the couplings  \eqref{tr} and \eqref{tr2} will be small and of the order $\mathcal{O}\left(\frac{1}{{T^{1/4}_{\mathrm{Tot}}}}\right)$.
%As a result the coupling to the weak mode following the feedback, given by \eqref{tr} and \eqref{tr2}, is of the order $\mathcal{O}\left(\frac{1}{{T^{1/4}_{\mathrm{Tot}}}}\right)$.

Now, consider the transfer function of the resultant system given by \eqref{fedr2}, which is
\[\Xi_{\theta}=\frac{\mathrm{Det}\left(-i\omega+i\Omega-\frac{1}{2}C^{\dag}C\right)}{\mathrm{Det}\left(-i\omega+i\Omega+\frac{1}{2}C^{\dag}C\right)}:=  \frac{\mathrm{Det}({A})}  {\mathrm{Det}({B})}       .\] This has derivative:
\[\frac{d\Xi_{\theta}}{d\theta}=\frac{   \mathrm{Det}({A})' \mathrm{Det}({B})  -\mathrm{Det}({A})\mathrm{Det}({B})'   }{     \mathrm{Det}({B})^2}.\]
Observe that   $\mathrm{Det}\left(-i\hat{\omega}+i\Omega\pm\frac{1}{2}C^{\dag}C\right)=\mathrm{const}\times\frac{1}{\sqrt{T_{\mathrm{tot}}}}$, where $\hat{\omega}$ is the estimator of $\Omega_2$ in Eq. \eqref{fedr2}, which has MSE $\mathcal{O}\left(\frac{1}{{T_{\mathrm{tot}}}}\right)$ from step 1 (see remark \ref{hod1e}). This
 can be seen by writing 
\begin{align}
\nonumber  \mathrm{Det}\left(-i\omega+i\Omega\pm\frac{1}{2}C^{\dag}C\right)&= \left( -i\omega+i\Omega_2\pm\frac{1}{2}\delta_1^{\dag}\delta_1\right)\times
\\
\label{jesus1}&\mathrm{Det}\left(-i\omega+i\Omega_1\pm\frac{1}{2}c^{\dag}c-\frac{\left(i\delta_2\pm\frac{1}{2}c^{\dag}\delta_1\right)\left(i\delta_2^{\dag}\pm\frac{1}{2}\delta^{\dag}_1 c\right)}{ -i \omega+i\Omega_2\pm\frac{1}{2}\delta_1^{\dag}\delta_1}   \right).       
 \end{align}
Therefore it follows that $\left|\frac{d\Xi_{\theta}}{d\theta}(-i\hat{\omega})\right|^2=\mathcal{O}\left(T_{\mathrm{tot}}\right)$\footnote{Provided that $\Omega_2$ depends on $\theta$.}, as in step 2 of the adaptive procedure. Step 3 can be developed  similarly. 

\begin{remark}\label{hod1e}
As a result of the feedback there is a second term contributing to the Hamiltonian. That is, we have $\Omega_2=\Omega_{nn}+\frac{1}{2i}\left(\overline{d}_nc_n-\overline{c}_nd_n\right)$. Therefore, our  estimator should be chosen as $\hat{\Omega}_{nn}+\frac{1}{2i}\left(\overline{d}_n\hat{c}_n-\overline{\hat{c}}_nd_n\right)$, rather than simply $\hat{\Omega}_{nn}$. As $c_n$ and $\Omega_{nn}$ can be estimated with a MSE $\mathcal{O}\left(\frac{1}{T_{\mathrm{tot}}}\right)$ from step 1, therefore      the term $\left( -i\omega+i\Omega_2\pm\frac{1}{2}|\delta_1|^2\right)$ will be of order $\mathcal{O}\left(\frac{1}{\sqrt{ T_{\mathrm{tot}}  }}\right)$, as required. 
\end{remark}

%Could assume real coupling instead of this remark

\begin{remark}\label{aaron3}
This method of feedback can also be used to implement the toy example in Sec. \ref{toy}. In that case the extra channel is used to modify the coupling only.
\end{remark}

\subsubsection{MIMO}

The above theory also holds in the case of MIMO systems, but there are a few subtleties to be aware of.

 Consider a  PQLS $(C, \Omega)$ coupled to $m$ fields. The coupling matrix to the $n$th mode is now a column vector of size $m$, so there are now $m$ direct coupling parameters and $n-1$ indirect coupling parameters to control. To account for the increase in direct parameters we  use $m$ control channels (rather than one).

Using the feedback reduction rule \cite{Gough5}, the  resultant model following feedback is given by 
\[\left(   C+D, \Omega+\frac{1}{2i}\left(D^{\dag}C-C^{\dag}D\right)\right),\]
where $D$ is an $m\times n$ matrix of control parameters. This is identical to the SISO case except that the number of rows in $D$ has increased to  $m$.
The goal, as before,  is to make the (direct and indirect) coupling to the $n^{\mathrm{th}}$ mode small. 
Choosing the feedback parameters to be 
\begin{align}\label{est1}
{d}_{jn}=-\hat{c}_{jn}& \,\,\,\mathrm{for}\,\,\, 1\leq j\leq m\\
\label{earth}\overline{d}_{1j}=\frac{\sum_{k=1}^m\hat{\overline{c}}_{kj}\left(-\hat{c}_{kn}\right)-2i\hat{\Omega}_{jn}}{\hat{c}_{1n}} & \,\,\,\mathrm{for}\,\,\, 1\leq j\leq n-1,
\end{align}
it follows that the direct coupling to the weak mode will be given by 
\begin{align}\label{keen1}
c_{jn}-\hat{c}_{jn} \,\,\,\mathrm{for}\,\,\, 1\leq j\leq m
\end{align} and the indirect coupling is given by 
\begin{align}\label{keen2}
\Omega_{jn}-\hat{\Omega}_{jn}\frac{c_{1n}}{\hat{c}_{1n}}+\frac{1}{2i}\left(c_{1n}\left(\overline{c}_{1j}-\hat{\overline{c}}_{1j}\right)+\overline{c}_{1j}\left(\hat{c}_{1n}-c_{1n}\right)    
+\sum_{k=2}^m\hat{c}_{kn}\left(\overline{c}_{kj}-    \hat{\overline{c}}_{kj}\frac{c_{1n}}{\hat{c}_{1n}}\right)   \right).
\end{align}
The adaptive procedure is then similar to the SISO case (see Appendix \ref{HUDE12}).

\subsubsection{Physical Meaning of this Feedback Setup}
One question that remains is to understand what it means physically to add an extra control channel? Consider the simplest possible physical example of a linear system, which is an optical cavity. It consists of two  mirrors; one that is partially transmitting and one that is perfectly reflecting. Between these mirrors a trapped electromagnetic mode is set up, whose frequency depends on the separation of the mirrors \cite{Gough3}. This can be characterised in the SLH model by a coupling parameter corresponding to the reflectivity of the mirror and a Hamiltonian representing the Hamiltonian of the internal cavity system. This is nothing more than a damped QHO \cite{Gardner1}. Now, replacing the reflective mirror with a second
 partially transmitting mirror  where  we are able to control it's reflectivity, corresponds mathematically to adding an extra (control) channel. It is expected that a similar procedure would work for more than one mode.

\subsection{Other Methods for Synthesising DFSs}
Other  methods have been given for synthesising  DFSs. For instance, in \cite{Pan1} they consider using Hamiltonian control or system-environment coupling control. Further DFSs are considered in detail in \cite{Yamamoto5} and in particular the use of similar coherent feedback methods to ours for driving the system to a DFS; it is shown that for a given QLS there exists a coherent feedback controller achieving the task.

\subsection{Measurement}

Consider a SISO PQLS, we will now show that a homodyne measurement enables one to realise the QFI level of scaling in steps 2 and 3 of our adaptive procedure (homodyne measurements also work  for step 1  \cite{Guta2}). 
Let us look at step 2 (step 3 is similar). Working in the time domain,   our     coherent input of frequency $\hat{\Omega}_1$ has time-dependent amplitude $\alpha(t)=\alpha e^{-i\hat{\Omega}_1 t}$, where $\alpha$ is a constant. It follows from Eq. \eqref{langevin3} and \eqref{lan2} that the output in the time-domain is a coherent state with amplitude 
\begin{equation}\label{mediate}
\alpha_{\mathrm{out}}(t):=\Xi(-i\hat{\Omega}_1)\alpha(t)+Ce^{At}\alpha_0-Ce^{At}(A+i\hat{\Omega}_1)^{-1}C^{\dag}\alpha.
\end{equation}
Here $\alpha_0$ is the initial state of the cavity. Consider  performing a homodyne measurement on the integrated mode 
$$\frac{1}{\sqrt{T_{\mathrm{tot}}-T_1}}\int^{T_{\mathrm{tot}}}_{T_1}e^{i\hat{\Omega}_1t}\mathbf{b}^{\mathrm{out}}(t)dt,$$  which is equivalent to a homodyne measurement of the $\hat{\Omega}_1$ frequency mode.  On this mode there is a coherent state with amplitude $\frac{1}{\sqrt{T_{\mathrm{tot}}-T_1}}\int^{T_{\mathrm{tot}}}_{T_1}e^{i\hat{\Omega}_1t}\alpha_{\mathrm{out}}(t)dt$. Now choosing  $T_1$ to be $T_{\mathrm{tot}}^{1-\frac{\epsilon}{2}}$, then the second and third terms in \eqref{mediate} will offer little contribution to the integral. The meaning of this is that the system  has stabilised and, as a result, the output is 
\begin{equation}\label{Diego1}
\frac{1}{\sqrt{T_{\mathrm{tot}}-T_1}}\int^{T_{\mathrm{tot}}}_{T_{\mathrm{tot}}^{1-\frac{\epsilon}{2}}}\alpha_{\mathrm{out}}(t)dt=\Xi(-i\hat{\Omega}_1)\alpha\sqrt{   T_{\mathrm{tot}}-T_{\mathrm{tot}}^{1-\frac{    \epsilon}{2}    }    }\approx \Xi(-i\hat{\Omega}_1)\alpha\sqrt{T_{\mathrm{tot}}}.
\end{equation}
The action of the system on this mode is thus a unitary rotation, so that homodyne detection would be suitable \cite{Wiseman3}. Note that we are monitoring the system over the interval $[T_1, T_{\mathrm{tot}}]$ rather than the entire time interval $[0,T_{\mathrm{tot}}]$. Therefore, we are not using all available information, which begs the question whether it may  be possible to do even better; we discuss this shortly. 
Finally, note that homodyne measurements can also be  used in the  MIMO case, although we do not discuss this here.

Let us now understand why decoherence-free subsystems lead to enhanced scaling. For this we consider the simplest possible model, that is, the cavity setup from Sec. \ref{indo}. Let us simplify things further by considering the non-adaptive procedure, which we saw in Sec. \ref{toy}. That is, 
consider a coherent probe with frequency equal to the one unknown parameter $\Omega$ and choose the coupling as  $|c|^2=T_{\mathrm{tot}}^{-(1-\epsilon)}$ (so that the system has (approximately) stabilised  for $t\in[T_{\mathrm{tot}}^{1-\frac{\epsilon}{2}},T_{\mathrm{tot}}]$). For time  $t\in[T_{\mathrm{tot}}^{1-\frac{\epsilon}{2}},T_{\mathrm{tot}}]$, the state of the system is a coherent state with mean amplitude 
\begin{align}
\alpha_{\mathrm{sys}}(t)\approx-2\frac{\alpha}{c}e^{-i\Omega t},
\end{align}
which has been obtained by using Eq. \eqref{lan2}. Here 
 $\alpha$ is the amplitude of the input. 
Therefore by decreasing  $c$ we push the poles of the transfer function closer to the origin, essentially destabilising the system.  Moreover,  as the coupling constant goes to zero we create states of large amplitude in the system. As a result, these increasing amplitude states in the system oscillate for a longer and longer time, so that better information from the signal may be deduced. 
%[[[Beam splitter has effect of continuously pumping system, which explains large amplitude]]]
Furthermore, by writing the transfer function as $\Xi(-i\omega)=e^{i\phi(\Omega)}$ where $\phi(\Omega)=\pi-2\mathrm{arctan}\left(\frac{2(\Omega-\omega)}{|c|^2}\right)$ and performing a Taylor expansion with respect to $\Omega$ and setting  $\omega=\Omega$, we   have $\phi\approx\pi-4\Omega T_{\mathrm{tot}}^{1-\epsilon}$. Therefore, according to \eqref{Diego1} the output over the interval $[T_{\mathrm{tot}}^{1-\frac{\epsilon}{2}},T_{\mathrm{tot}}]$ is a coherent state with amplitude $(\alpha e^{i\pi})\sqrt{T_{\mathrm{tot}}}e^{-4i\Omega T_{\mathrm{tot}}^{1-\epsilon}}=(Ee^{i\pi})e^{-4i\Omega T_{\mathrm{tot}}^{1-\epsilon}}$. Hence the phase of the output is very sensitive to changes in $\Omega$, just like the   direct metrology example from Sec. \ref{DIREC}.

\subsection{Discussion}\label{dico}

Recall our measurement strategy in the previous subsection, where we neglected the outcome of the experiment over the time interval $[0, T_{\mathrm{tot}}^{1-\epsilon/2}]$. Can we improve our precision by considering this interval within the measurement too? More generally, what is the optimum time to run the experiment for?

Consider again estimating the Hamiltonian parameter in a quantum cavity. Let us calculate the QFI over the entire time interval $[0, T_{\mathrm{tot}}]$.  Assuming for simplicity that the system initially starts in vacuum and using \eqref{mediate}, the output state at time $t$ is a coherent state with amplitude 
\begin{equation}\label{gone}
\alpha_{out}(t,\omega)=(\Xi(-i\omega)\alpha e^{i\omega t}-e^{At}|c|^2(A+i\omega)^{-1})\alpha,
\end{equation}
 where $\alpha e^{-i\omega t}$ is the amplitude of the coherent input with frequency $\omega$.
Now, combining \eqref{gone} with the result in Appendix \ref{HUDE} and setting $\omega=\Omega$, 
it follows that  the QFI over $[0, T_{\mathrm{tot}}]$  is 
\begin{equation}\label{idle}
F(\theta)=\int^T_0\left.\left|  \frac{d\alpha_{out}(t, \omega)}{d\Omega}\right|^2\right|_{\omega=\Omega}dt.
\end{equation}
Let us suppose that $|c|^2=\frac{k}{T_{\mathrm{tot}}}$ for some positive $k$. Calculating the QFI \eqref{idle} as a function of $k$ gives  
\begin{equation}\label{idle2}
F(\theta)=64|\alpha|^2T_{\mathrm{tot}}^3\frac{1}{k^2}\left[  1+\frac{1}{k}    \left(e^{-k}(-\frac{1}{4}k^2-\frac{3}{2}k-2.5)-2e^{-\frac{1}{2}k}(-k-4)-\frac{11}{2}\right)              \right].
\end{equation}
A plot of this function is given in  Fig. \ref{Plot1}. Firstly, observe that the choice $k=1$ is optimal and gives a value $F(\theta)\approx93.26|\alpha|^2T_{\mathrm{tot}}^3=93.26ET_{\mathrm{tot}}^2$, which corresponds to an experiment length equal to  the reciprocal of the spectral gap.  Therefore there is an improvement by a factor of $T_{\mathrm{tot}}^{\epsilon}$ over our method.  Can we take advantage of this for more general PQLSs? That is, can we adapt our measurement and adaptive strategies? The problem with this is that when $T_{\mathrm{tot}}=\frac{1}{|c|^2}$, the output state has another contribution (which was not present in our calculations) due to the fact that system has not stabilised (see Eq. \eqref{gone}). Therefore, the input-output map will no longer be a unitary rotation  so it's not clear how to perform the measurement in practice. That is, will Homoodyne measurement still work or do we need to consider  more general general Gaussian measurements, such as Heterodyne \cite{Gardner1},  to realise this level of precision?  This problem is something to consider in  future works. Now,  the QFI in \eqref{idle2} decreases rather quickly with $k$   until 
$$F(\theta)\approx 64ET^2_{\mathrm{tot}}\left(\frac{1}{k^2}-\frac{11}{2}\frac{1}{k^3}\right).$$
In particularly the choice $k=T^{\epsilon}$ recovers the work from the earlier subsections. 

\begin{figure}
\centering
\includegraphics[scale=0.25]{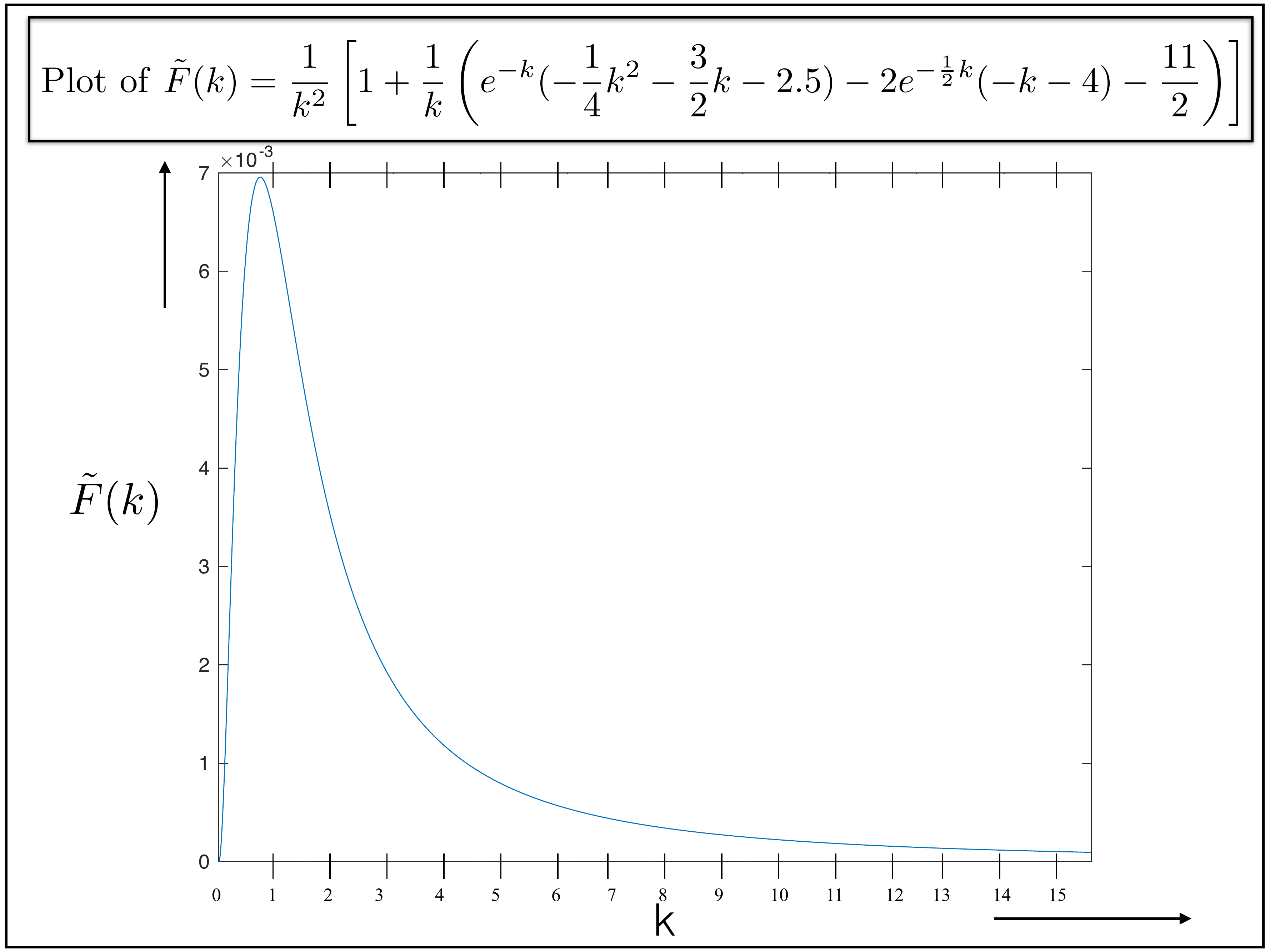}
\caption{This figure shows a plot of the QFI as a function of (rescaled) experimental time for the estimation example in Section \ref{dico}.\label{Plot1}}
\end{figure}

\subsection{Non-Linear DFS Example}\label{nonny}
The use of decoherence-free subsystems (DFSs) as a resource for parameter estimation has been seen recently in a non-linear context \cite{Kasia2}. Let us discuss a very simple example of this. An open quantum system evolves according to the system and field unitary $\mathbf{U}(t)$, which is identical to the one in eq. \eqref{eq.unitary.cocycle}. However, the difference with the linear setup is that the system is no longer constrained to be a QHO (note that the field is generally assumed to be a Bosonic bath as in the linear setup). Suppose that we have a  is a 2-level system with levels $\ket{0}$ and $\ket{1}$. Let the system Hamiltonian be $\mathbf{H}=\theta\ket{0}\bra{0}$ and suppose we have two Bosonic channels with respective coupling operators $\mathbf{L}_1=c\ket{1}\bra{0}$ and $\mathbf{L}_2=c\ket{0}\bra{1}$
 with $c\in\mathbb{R}$. We assume that the field initially starts in vacuum and the initial state of the system is $\rho_{\mathrm{sys}}=\ket{0}\bra{0}$. For more information about the non-linear setup see for example \cite{Rivas1}.

Now, at some time $t_1$ the system state will change from $\ket{0}$ to $\ket{1}$ and simultaneously a photon will be emitted into channel 1. Next, the system will change from $\ket{1}$ back to $\ket{0}$ and simultaneously a photon will be emitted into channel 2. This process will repeat indefinitely. We can now fairly straightforwardly write down the joint system-field state as a superposition of all such events. Assuming that the   experiment is run for  time $T_{\mathrm{tot}}$ equal to the inverse of the spectral gap (i.e $1/c^2$), we have 
\begin{align}
\nonumber\ket{\Phi}^{\mathrm{sys}\otimes\mathrm{field}} &=\ket{0}\otimes  e^{-i\theta T_{\mathrm{tot}}-\frac{1}{2}}  \ket{\mathrm{vac}}\ket{\mathrm{vac}}
\\&\nonumber+\ket{0}\otimes
\sum^{\infty}_{k=1}\frac{1}{T_{\mathrm{tot}}^{k}}e^{-\frac{1}{2}}\int_{\mathbb{R}^{2k}}e^{-i\theta(t_1-t_2+t_3-...+t_{2k-1})}\ket{t_1,t_3,...t_{2k-1}}\ket{t_2,t_4,...t_{2k}}\\
\label{smell}&+\ket{1}\otimes \sum^{\infty}_{k=1}\frac{1}{T_{\mathrm{tot}}^{k-\frac{1}{2}}}e^{-\frac{1}{2}}\int_{\mathbb{R}^{2k-1}}e^{-i\theta(t_1-t_2+t_3-...+t_{2k-1})}\ket{t_1,t_3,...t_{2k-1}}\ket{t_2,t_4,...t_{2k-2}},
\end{align}
where $\ket{t_1,t_3,...t_{2k-1}}$ represents the (unnormalised) state of channel 1, with the $t_i$s indicating the photon emission times,  (similarly for $\ket{t_2,t_4,...t_{2k}}$) and $\ket{\mathrm{vac}}$ indicates vacuum in the field. Note that the integrals here are taken over the times $t_i$. Observe that the state has acquired a phase for the times when the system was in state $\ket{0}$.
% and the integral and sum account for the fact that any number of emissions are possible at any times. 
The factor $e^{-\frac{1}{2}}$   is the normalisation factor. The expression \eqref{smell} is called a \textit{Dyson series expansion} \cite{Dyson1}. 
The result of measuring the two output channels will be a series of clicks (see Fig. \ref{click1}) obeying a Poissonian distribution with unit mean and variance. 

Suppose that we would like to estimate $\theta$. Depending on parity of the output state the internal system state will be known. Let us calculate the QFI given an even number of photons have been detected. The even output state is given by 
$$\ket{\Phi}^{\mathrm{even}}=\frac{1}{\sqrt{p^{\mathrm{even}}}}\sum^{\infty}_{k=1}\frac{1}{T_{\mathrm{tot}}^{k}}e^{-\frac{1}{2}}\int_{\mathbb{R}^{2k}}e^{-i\theta(t_1-t_2+t_3-...+t_{2k-1})}\ket{t_1,t_3,...t_{2k-1}}\ket{t_2,t_4,...t_{2k}},$$
where $p^{even}$ is the probability that there is a non-zero even number of photons emitted
(we have neglected the case of zero photon detections). 
Now since this state is pure we can use formula \eqref{qfiformula} to compute the QFI.
After some straightforward algebra we have
\begin{equation}
F^{\mathrm{even}}(\theta)=4T_{\mathrm{tot}}^2\left[\sum_{k=1}^{\infty}e^{-1}\frac{k(k+1)}{(2k+2)!}- \left( \sum_{k=1}^{\infty}e^{-1}\frac{k(2k)}{(2k+2)!}\right)^2\right].
\end{equation}
In particular the QFI scales quadratically with time, just as in the linear case. 
%Note that the result for the odd even photon number case may be obtained similarly. 

Suppose now that we know that two photons have been detected (one in each channel). The (normalised) two-photon (2p-)output state is therefore given by
$$\ket{\Phi}^{\mathrm{2p}}=\sqrt{\frac{2}{T_{\mathrm{tot}}^2}}\int_{\mathbb{R}^{2}}e^{-i\theta t_1}\ket{t_1}\ket{t_2}.$$
Notice that  the photon detection time of the photon in the second channel (C2) doesn't contain any  information about $\theta$. Therefore, given a photon detection at time $t_2$ the output state on channel one (C1) is thus
$$\ket{\Phi}^{\mathrm{C1}}=\sqrt{\frac{1}{t_2}}\int_0^{t_2}e^{-i\theta t_1}\ket{t_1}\mathrm{d}t_1.$$

We now see that a frequency measurement on C1 enables one to achieve the $T^2$ scaling. To see this we calculate the CFI for this measurement. Since the  probability density function (pdf) for detection at frequency $\omega$ is given by 
$$p(\omega)=\frac{1}{2\pi t_2}\left[\frac{1-\mathrm{cos}((\omega-\theta)t_2)}{(\omega-\theta)^2}\right].$$
it follows (by using \eqref{cod}) that  the CFI  is given by 
$$I^{\mathrm{C1}}(\theta|t_2)=\frac{t_2^2}{6}.$$
Finally, to obtain the full Fisher Information for the measurement we must average over all detection times in C2, corresponding to the photon measurement there, and then weight the total with the probability of obtaining two photons. Hence  
$$I^{\mathrm{C1}}(\theta)=\frac{e^{-1}T_{\mathrm{tot}}^2}{36}.$$

On the other hand   it turns out  that $F^{\mathrm{C1}}(\theta)=\frac{e^{-1}T_{\mathrm{tot}}^2}{9}$, so this measurement is only a factor of 4 worse than the optimal one 
under the assumptions here.
 Moreover, much of the information in the field is contained in the first photon, which is evidenced by the fact that  $F^{\mathrm{C1}}(\theta)/F^{\mathrm{even}}(\theta)\approx 31\%$. Therefore, our measurement choice here seems to be very good and so there is not 
 much to be gained from measuring the other photons\footnote{However, it is possible to compute the CFI for this case because there is a simplification where   for a given number of emissions the distribution of each photon is iid. We don't do this here though.}.

It was also  shown in \cite{Kasia2} that  at a  first-order DPT  the  QFI of a general non-linear system  become quadratic in time.   However, in that context it remains an open problem how to exploit the large QFI scaling (i.e find a physical measurement). 
%Note that in our example above, where a simple measurement choice was available,  the system does not undergo a  DPT.

\begin{figure}
\centering
\includegraphics[scale=0.19]{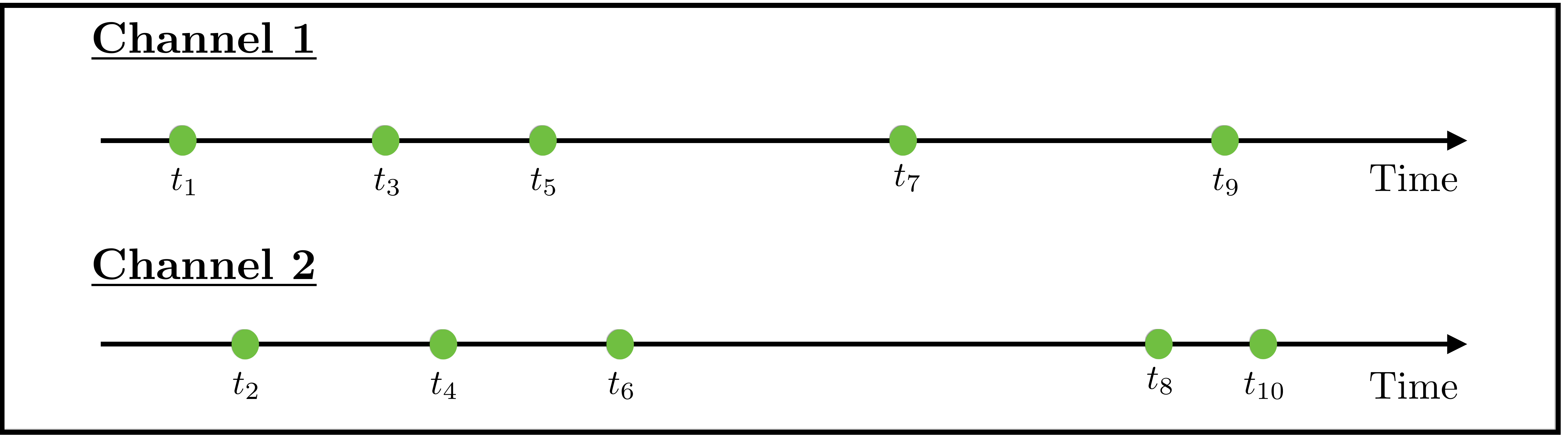}
\caption{Timeline of photon detections is two output channels for example in Sec. \ref{nonny}.
\label{click1}}
\end{figure}

%Perhaps kasha's paper more important when at c2=1/t....explains why get extra terms...

\subsection{Noise}
We briefly discuss an example, which attempts to understand how noise affects the theory in this chapter. 

Consider the quantum cavity from example \ref{probes}, where there is one accessible channel and one noise channel (see Sec. \ref{noisek}). Suppose that we are to estimate $\Omega$ from the accessible channel.
%(note that this parameter  is identifiable in the sense of Def. \ref{identy}). 

Suppose that, as above, we use a monochromatic coherent input state with time dependent amplitude $\alpha(t)=\alpha e^{i\omega t}$ (i.e a  frequency $\omega$ coherent state of  amplitude $\alpha\sqrt{T_{\mathrm{tot}}}$). Then provided that  
$T_{\mathrm{tot}}\gg\left(|c_1|^2 +|c_2|^2\right)$ the output will be a coherent state of frequency $\omega$ and amplitude
$$  \left(1-\frac{|c_1|^2}{-i\omega+i\Omega+\frac{1}{2}\left(|c_1|^2 +|c_2|^2\right)}\right) \alpha \sqrt{T_{\mathrm{tot}}}.$$
Now, we can use the result of Appendix \ref{HUDE} to calculate the QFI:
$$F(\Omega)=E\frac{|c|^4}{\left((\omega-\Omega)^2+\frac{1}{4}(|c|_1^2+|c_2|^2)^2\right)^2},$$
where $E$ is the energy of the input given by $E=|\alpha|^2T_{\mathrm{tot}}$. 
Notice the similarities with the expression for the noiseless case in Eq. \ref{fedy8}.
Therefore, the results in this noiseless case are valid so long as $c_2$ is not too big (i.e so that $T_{\mathrm{tot}}\gg\left(|c_1|^2 +|c_2|^2\right)$) so that  the system is able to stabilise. 
It is expected that this conclusion will be true for arbitrary noisy PQLSs too. Therefore the difficulty will be in designing a feedback method to reduce the   coupling in the noisy channel without having access to it, which would probably need further assumptions on the accessibility of this channel. 

An interesting question to consider is in what circumstances it is possible for noisy PQLS (or more generally QLSs) to achieve Heisenberg scaling. From the above  it seems that in order to get Heisenberg scaling we would need a spectral gap of order $\mathcal{O}\left({T}^{-1}_{\mathrm{tot}}\right)$. This might be true only if the system can be separated as two subsystems, such that one is coupled with the accessible channel  and the second with the noise channel. Then the noisy channel would not affect the dynamics of the first subsystem and we can estimate parameters within to Heisenberg level.  On the other hand, perhaps Heisenberg scaling is possible for a more general class of QLS, such as those  with noise unobservable subspaces (see Sec. \ref{Noisek4})?

\section{Feedback Method for QLSs; Stationary Approach}\label{dfsw}

Decoherence-free subsystems are also  advantageous for estimation in the stationary approach. We shall study this here for SISO PQLSs (we  revisit this in the following chapter for general SISO QLSs).
Assume, by the feedback methods in Sec. \ref{piles} or otherwise, that we have a SISO PQLS with (at least) one eigenvalue that is $\mathcal{O}(\tau^{-1})$ from the imaginary axis, where $\tau=T_{\mathrm{tot}}^{1-\epsilon}$ with $T_{\mathrm{tot}}$ large enough so that the system has reached stationarity. 
%For example, this could be as a result of the feedback methods in the previous subsection. 
%Suppose that we would like to understand how well a parameter may be estimated  using  a fixed time constraint $T_{\mathrm{tot}}$. 

\subsection{QFI Scaling}\label{poke11}

Let us show that there is an enhancement in the QFI in terms of    time resources, $T_{\mathrm{tot}}$, in the stationary approach. Recall that the input to our system is a series of iid Gaussian stationary quantum noise field processes characterised by the covariance, $V(N,M)$ 
(see Eq. \eqref{Ito}). Therefore, the output will also be series of Gaussian stationary quantum noise field processes with mean zero and covariance $\mathbb{E}\left[\breve{\mathbf{b}}_{\mathrm{out}}(t)\breve{\mathbf{b}}^{\dag}_{\mathrm{out}}(s)\right]=V_{\theta}(t,s)$.

To calculate the QFI 
 it is simpler to work in the frequency domain because all frequency modes are independent and therefore the action of the system is a series of rotations on the  squeezed input states. The QFI \textbf{per unit time} from frequency $\omega$ is    
(see Sec. \ref{ounce1})
\begin{align}
\nonumber f_\theta(\omega)&=-\mathrm{Tr}\left(J\dot{\Psi}(\omega)J\dot{\Psi}(\omega)\right)\\
\label{tease}&=|M|^2\left|\frac{d\Xi(-i\omega)\Xi(+i\omega)}{d\theta}\right|^2
\end{align}
assuming ${T}_{\mathrm{tot}}>>\tau$. 
By writing the transfer function in terms of its poles as in Eq. \eqref{qpi} where $-\mathrm{Re}(z_1)=\mathcal{O}(\tau^{-1})$,
 it follows that 
\begin{equation}\label{tease2}
\frac{d\Xi(-i\omega)}{d\theta}=
\begin{cases}
  \mathcal{O}(\tau)     & \quad \text{if } |\omega-\mathrm{Im}(z_1)|\sim\frac{1}{\tau}\\
    \mathcal{O}\left(\frac{1}{\tau(\omega-  \mathrm{Im}(z_1))^2      }\right)  & \quad \text{otherwise}\\
  \end{cases}
\end{equation}
Now the QFI per unit time is obtained by integrating \eqref{tease} over $\omega$.
Using \eqref{tease2} it follows that that the only significant contributions to the integral come from the intervals $|\omega-\mathrm{Im}(z_1)|\sim\frac{1}{\tau}$ and $|\omega+\mathrm{Im}(z_1)|\sim\frac{1}{\tau}$. Therefore $f_{\theta}\propto \tau^2\times\frac{1}{\tau}=\tau$ (see Fig. \ref{picf}).  Hence the QFI scales as $\mathcal{O}(T_{\mathrm{tot}}\tau).$ This suggests that DFSs are advantageous in the stationary regime too. 
 Note that unlike the time-domain approach (Sec. \ref{piles}), no adaptive  strategy is necessary.
%[[[Could explain that in stationary and time-sep the reasons for enhanced scaling are diff]]]

\begin{figure}
\centering
\includegraphics[scale=0.15]{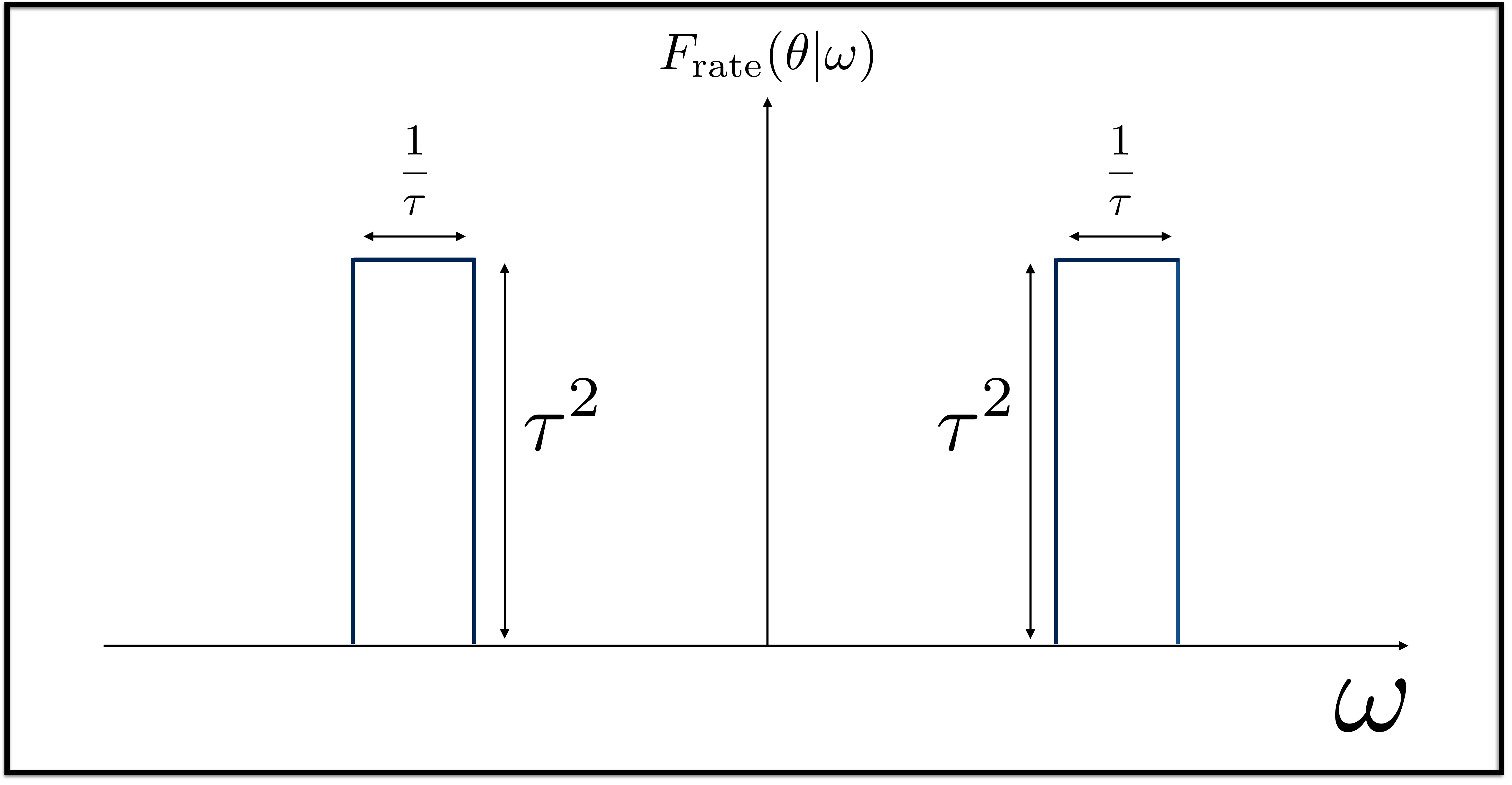}
\caption{There are two frequency bands contributing to the QFI.
\label{picf}}
\end{figure}

\begin{exmp}\label{fug}
Let us calculate the QFI explicitly for estimating $\Omega$ in  the cavity example from   Sec. \ref{indo}. The spectral gap is given by $|c|^2$. As 
$$\frac{d\Xi(-i\omega)}{d\Omega}=\frac{i|c|^2}{i(\Omega-\omega)+\frac{1}{2}|c|^2},$$
then 
\begin{align*}
f_{\Omega}(\omega)&=|M|^2|c|^4\left(\frac{1}{\left((\Omega-\omega)^2+\frac{1}{4}|c|^4\right)^2}+\frac{1}{\left((\Omega+\omega)^2+\frac{1}{4}|c|^4\right)^2}
\right.\\
&+\left.2\frac{1}{   (\Omega-\omega)^2+\frac{1}{4}|c|^4}\frac{1}{(\Omega+\omega)^2+\frac{1}{4}|c|^4}          \right)\\
&\geq |M|^2|c|^4\left(\frac{1}{\left((\Omega-\omega)^2+\frac{1}{4}|c|^4\right)^2}+\frac{1}{\left((\Omega+\omega)^2+\frac{1}{4}|c|^4\right)^2}\right).
\end{align*}
Therefore, by using Eq. \eqref{gun2}, 
$$f_{\Omega}\geq\frac{1}{\pi}\int^{\infty}_{-\infty}  |M|^2|c|^4  \frac{1}{\left(\omega^2+\frac{1}{4}|c|^4\right)^2}d\omega.$$
Finally, using the substitution $\omega=\frac{1}{2}\mathrm{tan}(\theta)$, or otherwise, one obtains $f_{\Omega}\geq\frac{4 |M|^2}{|c|^2}=4 |M|^2\tau$.
\end{exmp}

The next step is to find a measurement enabling this heightened scaling. We shall return to  this problem in the following chapter with the help of \textit{coherent quantum absorbers}.

\section{Conclusion}
In this chapter we have seen that when the system is almost dynamically unstable, i.e., the poles of the transfer function are close to the imaginary axis the sensitivity is increased to a quadratic scaling with time if the total time is of the order of the correlation time. We gave two feedback schemes designed to exploit this enhancement. The common theme in both of these was that the poles of the transfer function were driven to the imaginary axis. The first feedback scheme used a beam-splitter with a large reflectivity to weaken the coupling of the field to the system. The second method used a coherent feedback scheme and an addition of a second (set of) control channel(s) connected in a feedback loop. We explained that the meaning of adding a second control channel was to add a set of reflective mirrors in the cavity that one is  able to control. In all honesty this second scheme could potentially still be very challenging experimentally, as it would require some access to the system. Nevertheless it   is still interesting mathematically. 

In the regime of time-dependent inputs we developed an adaptive procedure to realise the $T^2$ level of scaling. The key step was that the frequency of the input is chosen as an estimator from the previous step and one mode (or the entire system) is destabilised step by step. Note that we worked with coherent states here, hence the scaling with with energy is linear. It is expected that all results could be enhanced by a further factor of $E$, for example by using squeezed-coherent states, provided that a suitable measurement and estimator exists achieving the scaling. 

We also considered the case of stationary inputs and showed that the same enhancement was possible.

%TO DO: DISCUSS THAT FEEDBACK HERE IS NECESSARY i.e. can't get $ET^2$ without!]]]]

\section{Outlook}

There are various directions to extend this work. Firstly, preliminary investigations indicate that these methods would also work for multiple parameters by isolating more modes in the feedback method in Sec. \ref{iles}. Another direction is  to the problem from Sec. \ref{piles} for the case of time-dependent inputs to active QLSs.  It is our expectation that the  ideas would transfer over directly. It would be an interesting problem for the experimental community to confirm the ideas here, i.e., that  Heisenberg scaling is obtainable in practice from an experiment.

It remains an open problem to investigate the general parameter estimation problem (i.e. beyond this quasi-DFS setup) in the stationary inputs approach. That is, what are the optimal inputs and measurements to use with either time or energy (or both) as a resource for quantum metrology. This problem would be particularly interesting in the MIMO case.

\chapter{Quantum Absorbers}\label{DUAL}

Environment (or reservoir) engineering has increased in popularity in recent years \cite{Plenio1, Stan1, Yamamoto2, Yamamoto3}. This innovative technique is a way of designing  a system's master equation in order to drive the system into a desired pure stationary state using dissipative field dynamics \cite{Metcalf1} or performing continuous measurement on the field. It has prominent applications in laser cooling or optical pumping. This topic has also been studied  in QLS theory. Specifically in \cite{Yamamoto2} necessary and sufficient conditions were given for a general QLS to have a unique pure steady state. This leads to a dissipative procedure enabling one to engineer an arbitrary pure gaussian state.

In this work we consider the following environment engineering problem:   given a QLS our goal is to design a second QLS, called the \textit{dual system} so that radiation emission from the QLS is coherently reabsorbed by the dual. That is, the dual system acts as a coherent quantum absorber for the first.    In other words the stationary state of the combined system is pure and therefore the output is equal to the input. We shall use the terms \textit{dual} and \textit{quantum absorber} synonymously in the following. We show here that for a stable QLS it is always possible to find a stable QLS coherent absorber. One reason why this setup is so interesting is because it provides a natural purification of the stationary state. Purifications with the use of larger Hilbert spaces have found widespread use in quantum information theory, as generally they are much easier to work with than mixed states \cite{Nielsen1}.
Obtaining a purification is well-known to be a hard problem and  was the main reason prompting the  study of the parallel problem to ours for  non-linear quantum systems \cite{Stan1}. In particular it was shown there that  such a quantum absorber always exists for SISO non-linear systems. 

In Sec. \ref{purple6} we solve the environment engineering problem for general (MIMO as well as SISO) QLSs. A corollary of our work is given in Sec. \ref{pight}, which    enables the proof of Theorem \ref{doned}, as promised in Ch. \ref{MIMO1}.

Finally, we discuss an  application of quantum absorbers to quantum estimation, and in particular how to devise  optimal measurements and estimators for the Heisenberg level stationary input setup from Sec. \ref{dfsw}. We begin by revisiting the QFI calculation from Sec. \ref{poke11}, which demonstrated the scaling enhancement when the system destabilises. We recalculate the QFI in the time domain, which simplifies considerably with the help of the dual absorber system. The essential point is that the 
 useful information about the system  is contained only within output correlations  between times $t$ and $s$ such that $|t-s|<\tau$, where $\tau$ is the inverse of the spectral gap. 
We then consider  a homodyne measurement on the output. We recognise that  in post-processing the measurement it is useful to devise estimators placing weight only on these correlation lengths. 
%The advantage here is that no adaptive procedure is necessary (unlike the case of time-dependent inputs).
For example, total integrated current is not very useful as the noise and signal contribution will be comparable. That is, although the mean of the estimator will be large, so will the variance and therefore all information about the signal is lost. Using such clever weighted estimators overcomes this problem and ensures that the  variance is not be too large while the mean can still be large (see Sec. \ref{SIDF}). Our estimator has a simple interpretation when viewed in the frequency domain as a frequency band-limited measurement.  Lastly, we consider whether it is possible to devise an optimal strategy realising the QFI  in the simpler case of a  PQLS. That is, we would like to improve the constant in front of the $\tau T_{\mathrm{tot}}$ level of scaling. %Again the impetus is on post-processing.  
It is for this problem that   the use of quantum absorbers  become particularly useful. 
Generally in a PQLS the action of the system in the frequency domain is to rotate the squeezed input by an angle which depends on the unknown parameter. Therefore in general each frequency will rotate by a different phase.  This is potentially problematic because  the optimal quadrature to measure in each frequency is  different. Our trick to overcome this and  obtain optimality is 
 to proceed adaptively by first obtaining a rough estimate of the system (with suboptimal scaling) and then use a quantum absorber. The effect of the absorber is to make all of the rotations in the frequency domain very small. The upshot is that the best quadrature to measure will be the same for all frequencies. This quadrature is the best measurement and realises the optimal QFI scaling following the results in \cite{Monras2}.

We discuss  other potential applications of quantum absorbers  in Sec. \ref{appos}.

\section{Finding the Dual System}\label{purple6}

\begin{figure}
\centering
\includegraphics[scale=0.25]{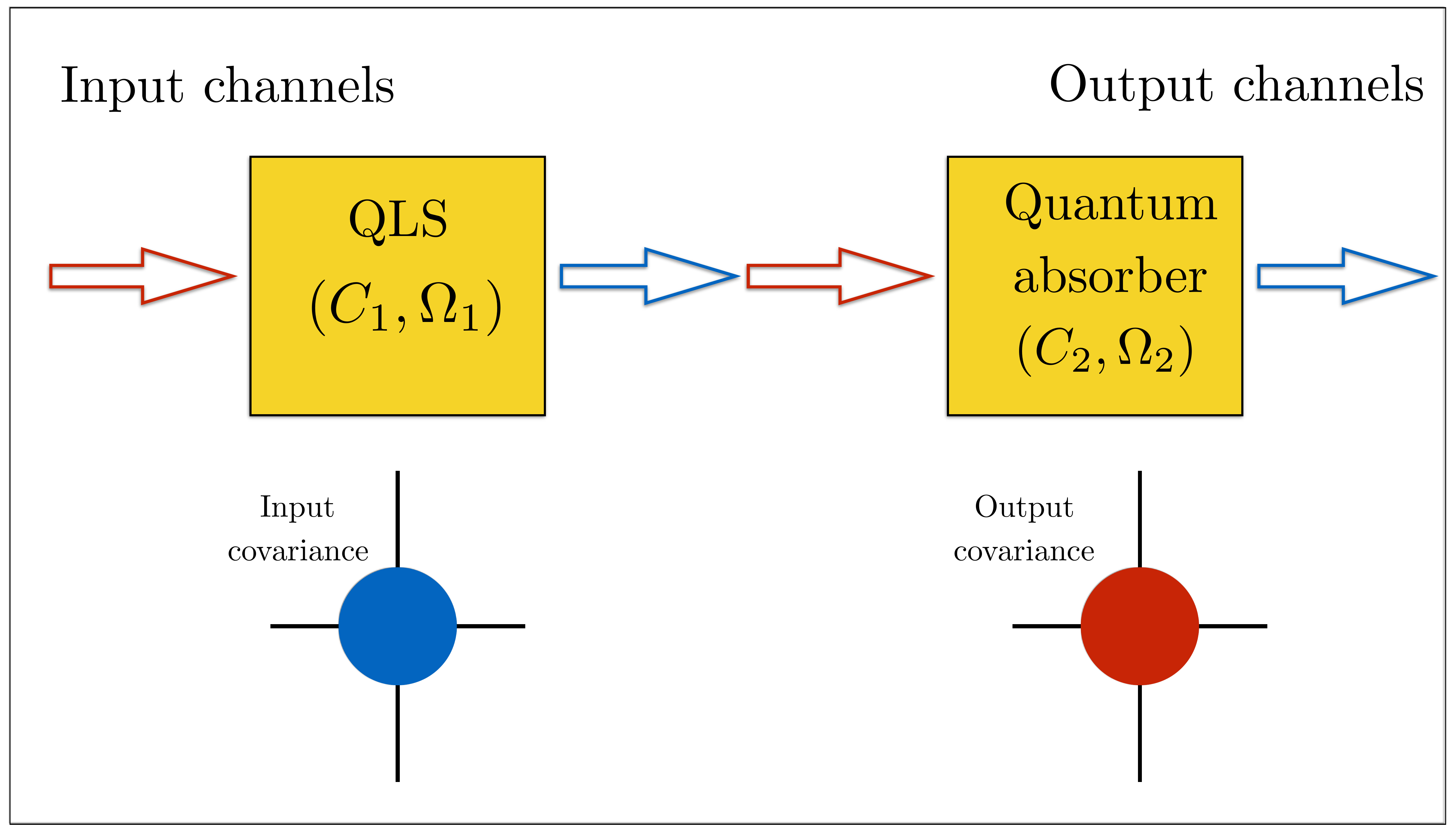}
\caption{QLS $(C_2, A_2)$ is a quantum absorber for QLS $(C_1, A_1)$. The input and output states  are identical  \label{absorber}}
\end{figure}

Consider a globally minimal QLS characterised by $(C_1, A_1)$ for vacuum input, $V_{\mathrm{vac}}$.  Our goal is to find a second system $(C_2, A_2)$ such that the combined system resulting from connecting them in series (see Sec. \ref{greed}) has trivial power spectrum, i.e., $\Psi(s)=V_{\mathrm{vac}}$. Label the system modes of the first system as $\breve{\mathbf{a}}_1$ and the modes of the second system by $\breve{\mathbf{a}}_2$ (see Fig. \ref{absorber}).

The combined system may be represented by\footnote{In this representation the modes are ordered as $\left[\begin{smallmatrix}\breve{\mathrm{a}}_1\\ \breve{\mathrm{a}}_2\end{smallmatrix}\right]$, rather than the doubled-up form $\breve{\mathbf{a}}$ where $\mathbf{a}=\left[\begin{smallmatrix}\mathbf{a}_1\\\mathbf{a}_2\end{smallmatrix}\right]$.}

$$\left(C_{\mathrm{series}}, A_{\mathrm{series}}\right):=\left((C_1, C_2), \left(\begin{smallmatrix}A_1&0\\-C^{\flat}_2C_1&A_2\end{smallmatrix}\right)\right).$$
Let us change basis on our system $(C_1, A_1)$ by applying a symplectic transformation   so that its (reduced) stationary state, $P$, is given by 
$$\left<\breve{\mathbf{a}}_1\breve{\mathbf{a}}_1^{\dag}\right>:=\left(\begin{smallmatrix}\mathrm{Diag}(N_1+1,..., N_n+1)&0\\0&\mathrm{Diag}(N_1,..., N_n)\end{smallmatrix}\right)$$
(recall that $P$ is the solution of the Lyapunov equation for the first system). That is, all modes are independent of each other and are  thermal states (see Sec. \ref{QHO}). Note that  $N_i\neq0$ for all $i$ by global minimality (see Theorem \ref{equivalence}). Such a transformation will alter the matrices $(C_1, A_1)$, but we still denote them $(C_1, A_1)$ so to not convolute the notation.

Now the requirement that   $\left(C_{\mathrm{series}}, A_{\mathrm{series}}\right)$ has a trivial power spectrum is equivalent to that its stationary state be pure (Theorem \ref{equivalence}). In light of this, we consider 
 the  pure extension of $P$ onto the modes $\breve{\mathbf{a}}_2$ given by 
$$\left<\left(\begin{smallmatrix} \breve{\mathbf{a}}_1\\ \breve{\mathbf{a}}_2\end{smallmatrix}\right)
\left(\begin{smallmatrix} \breve{\mathbf{a}}^{\dag}_1& \breve{\mathbf{a}}^{\dag}_2\end{smallmatrix}\right)\right>:=  \left(\begin{smallmatrix}     P&Q\\Q&P\end{smallmatrix}\right),$$  
where 
$$P=\left( 
 \begin{smallmatrix}\mathrm{Diag}(N_1+1,..., N_n+1)&0\\0&\mathrm{Diag}(N_1,..., N_n)\end{smallmatrix}\right):= \left( 
 \begin{smallmatrix}N+1&0\\0&N\end{smallmatrix}\right),$$
 $$ Q=\left(\begin{smallmatrix} 0& \mathrm{Diag}(M_1,..., M_n)   \\ \mathrm{Diag}(M_1,..., M_n)    &0\end{smallmatrix}\right):=
  \left( 
 \begin{smallmatrix}0&M\\M&0\end{smallmatrix}\right)$$
with $N_i, M_i\in\mathbb{R}$, $N_i(N_i+1)=M_i^2$ (or equivalently $P, Q$ are real matrices satisfying $P=QP^{-1}Q$). That is, we have a product of two-mode squeezed states, as  in Theorem \ref{WOLFF}, so that entanglement is localised to one mode at each site.
The solution to our problem  (should one exist) is therefore given by the matrices  $(C_2, A_2)$ satisfying the Lyapunov equation
\begin{equation}\label{purple}
 \left(\begin{smallmatrix}A_1&0\\-C^{\flat}_2C_1&A_2\end{smallmatrix}\right) \left(\begin{smallmatrix} P&Q\\Q&P\end{smallmatrix}\right)+\left(\begin{smallmatrix} P&Q\\Q&P\end{smallmatrix}\right)
  \left(\begin{smallmatrix}A^{\dag}_1&-\left(C^{\flat}_2C_1\right)^{\dag}\\0&A^{\dag}_2\end{smallmatrix}\right) +
  \left(\begin{smallmatrix} C_1^{\flat}V_{\mathrm{vac}}\left(C^{\flat}_1\right)^{\dag} &      C_1^{\flat}V_{\mathrm{vac}}\left(C^{\flat}_2\right)^{\dag} \\        C_2^{\flat}V_{\mathrm{vac}}\left(C^{\flat}_1\right)^{\dag} & C_2^{\flat}V_{\mathrm{vac}}\left(C^{\flat}_2\right)^{\dag} \end{smallmatrix}\right)=0.
  \end{equation}
  Equation \eqref{purple} leads to the following three equations:
\begin{align}
\label{purple1}A_1P+PA_1^{\dag}+C_1^{\flat}V_{\mathrm{vac}}\left(C_1^{\flat}\right)^{\dag}=0\\
 \label{purple2}A_1Q+QA_2^{\dag}-P\left(C_2^{\flat}C_1\right)^{\dag}+C_1^{\flat}V_{\mathrm{vac}}\left(C_2^{\flat}\right)^{\dag}=0\\
 \label{purple3}A_2P+PA_2^{\dag}-C_2^{\flat}C_1Q-Q \left(C_2^{\flat}C_1\right)^{\dag}            +C_2^{\flat}V_{\mathrm{vac}}\left(C_2^{\flat}\right)^{\dag}=0.
 \end{align}
 Now, rearranging \eqref{purple2} and substituting this into \eqref{purple3} gives
 \begin{align*}
 &-C_2^{\flat}V_{\mathrm{vac}}\left(C_1^{\flat}\right)^{\dag}Q^{-1}P+C_2^{\flat}C_1PQP^{-1}-QA_1^{\dag}Q^{-1}P-PQ^{-1}C_1^{\flat}V_{\mathrm{vac}}\left(C_2^{\flat}\right)^{\dag}\\&+PQ^{-1}P \left(C_2^{\flat}C_1\right)^{\dag}
 -PQ^{-1}A_1Q-C_2^{\flat}C_1Q-Q \left(C_2^{\flat}C_1\right)^{\dag}+C_2^{\flat}V_{\mathrm{vac}}\left(C_2^{\flat}\right)^{\dag}=0.
 \end{align*}
Using the condition $P=QP^{-1}Q$ leads to cancellations of terms 2, 5, 7 and 8 in this equation, thus  
\begin{equation}
-C_2^{\flat}V_{\mathrm{vac}}\left(C_1^{\flat}\right)^{\dag}Q^{-1}P -QA_1^{\dag}Q^{-1}P-PQ^{-1}C_1^{\flat}V_{\mathrm{vac}}\left(C_2^{\flat}\right)^{\dag}-PQ^{-1}A_1Q+C_2^{\flat}V_{\mathrm{vac}}\left(C_2^{\flat}\right)^{\dag}=0
\end{equation}
Finally, we can use $P=QP^{-1}Q$ and Eq. \eqref{purple1} on   the second and fourth terms in this expression to obtain
\begin{equation}
\left(PQ^{-1}C_1^{\flat}-C_2^{\flat}\right)V_{\mathrm{vac}}\left(PQ^{-1}C_1^{\flat}-C_2^{\flat}\right)^{\dag}=0,
\end{equation}
hence 
\begin{equation}\label{cheese1}
C_2^{\flat}V_{\mathrm{vac}}=PQ^{-1}C_1^{\flat}V_{\mathrm{vac}}.
\end{equation}
Because $V_{\mathrm{vac}}$ is not invertible, then it may appear that $C_2$ in \eqref{cheese1} is not unique. However, upon closer inspection and in particular using  $V_{\mathrm{vac}}=\left(\begin{smallmatrix}1&0\\0&0\end{smallmatrix}\right)$ allows one to recover $C_{2-}$ and $C_{2+}$ uniquely and, hence from this (using the doubled-up property), $C_2$.
Furthermore, we can combine this solution with \eqref{purple1} and \eqref{purple2} to obtain an expression for $A_2$ in terms of $C_1, A_1, P$ and $Q$. That is, 
\begin{equation}\label{cheese2}
A_2=QP^{-1}A_1PQ^{-1}+C^{\flat}_2C_1PQ^{-1},
\end{equation}
 where $C_2$ is implicitly a function of $C_1, A_1, P$ and $Q$.

\begin{remark}
Our dual system, $(C_2, A_2)$, is unique up to the symplectic equivalence in Theorem \ref{symplecticequivalence}). This follows because the solution  Eqs. \eqref{cheese1} and \eqref{cheese2} was unique in the particular basis we chose (i.e.  the one with the  diagonal structure in  the matrices $P$ and $Q$ for the stationary state)
\end{remark}

Now, at this stage it is not clear whether  our solution, $(C_2, A_2)$, 
is (i) physically realisable  and (ii) stable.  

Firstly, let us show that  $(C_2, A_2)$ is physically realisable. There are two conditions that must hold for this to be the case (see Sec. \ref{MTDR}): 
\begin{itemize}
\item[(a)]  the  matrices are doubled-up 
\item[(b)]  the matrices satisfy \eqref{PR} from Sec. \ref{MTDR}.
\end{itemize}
We can ensure that $C_2$ is doubled-up by construction from \eqref{cheese1}. Let us see that $A_2$ is also doubled-up. Begin by writing $A_2$ in block form as $A_2=\left(\begin{smallmatrix}A^{(11)}_2& A^{(12)}_2\\A_2^{(21)}& A_2^{(22)}\end{smallmatrix}\right)$. Therefore, we require that  $A^{(11)}_2= \overline{A_2^{(22)}}$ and
$A^{(12)}_2=\overline{A_2^{(21)}}$; we verify the first of these here (the proof of the second is similar).
Using \eqref{cheese1} and \eqref{cheese2} $A^{(11)}_2= \overline{A_2^{(22)}}$ is equivalent to the following requirement:
\begin{align*}
&M(N+1)^{-1}A_{1-}(N+1)M^{-1}+NM^{-1}C_{1-}^{\dag}C_{1-}(N+1)M^{-1}-(N+1)^{-1}M^{-1}C_{1+}^T\overline{C}_{1+}(N+1)M^{-1}\\&=
MN^{-1}A_{1-}NM^{-1}+NM^{-1}C_{1-}^{\dag}C_{1-}NM^{-1}-(N+1)^{-1}M^{-1}C_{1+}^T\overline{C}_{1+}NM^{-1}
\end{align*}
Now multiplying this expression by $M$ on both sides and using the property $N(N+1)=M^2$  and the fact that $N$ and $M$ commute since they are diagonal matrices gives: 
$$A_{1-}N=NA_{1-}+NC^{\dag}_{1-}C_{1-}-(1+N)C_{1+}^T\overline{C}_{1+}$$
(after some manipulation). Hence we require that 
$$A_{1-}N-N(A_{1-}+ C^{\dag}_{1-}C_{1-}- C_{1+}^T\overline{C}_{1+})+C_{1+}^T\overline{C}_{1+}=0.
$$
This can be seen to be true by using the physical realisably of $(C_1, A_1)$ (in particular that that $A_{1-}+A^{\dag}_{1-} C^{\dag}_{1-}C_{1-}- C_{1+}^T\overline{C}_{1+}=0$) and condition \eqref{purple1}, hence we are done.
Finally, we prove condition (b) in order to complete the physical realisability proof. Condition (b) is equivalent to
\begin{equation}\label{tho}
A_2^{(11)}+A_2^{(11)\dag}+C_{2-}^{\dag}C_{2-}-C^T_{2+}\overline{C}_{2+}=0
\end{equation}
\begin{equation}\label{tho1}
A_2^{(12)}+A_2^{(12)T}+C_{2-}^{\dag}C_{2+}-C^T_{2+}\overline{C}_{2-}=0.
\end{equation}
We show here that  Eq. \eqref{tho} is true (Eq. \eqref{tho1} is similar). Using  \eqref{cheese1} and \eqref{cheese2} it follows that \eqref{tho} is equivalent to 
$$N^{-1}MA_{1-}M^{-1}N+NM^{-1}A_{1-}^{\dag}MN^{-1} +M^{-1}C_{1+}^T\overline{C}_{1+}M^{-1}
+NM^{-1}(   C^{\dag}_{1-}C_{1-}- C_{1+}^T\overline{C}_{1+})M^{-1}N=0.$$
Multiplying this equation by $M$ on both sides, we obtain the following equivalent expression
$$N(A_{1-}+A_{1-}^{\dag}+C^{\dag}_{1-}C_{1-}- C_{1+}^T\overline{C}_{1+})N+A_{1-}N+NA_{1-}^{\dag}
+C_{1+}^T\overline{C}_{1+}=0.$$
This is true by \eqref{purple1} and the physical realisability of the first system.

To see that the system $(C_2, A_2)$ is stable observe that $A_2$ can be shown to satisfy the following:
$$A_2=QP^{-1}\left(\begin{smallmatrix}N&0\\0&-(N+1)\end{smallmatrix}\right)A_1\left(\begin{smallmatrix}N^{-1}&0\\0&-(N+1)^{-1}\end{smallmatrix}\right)PQ^{-1}.$$
Therefore the eigenvalues of $A_2$ are the same as those of $A_1$, hence the system must be stable.

Let us see the specific case of a  one mode system (with a one mode dual system), where the solution simplifies considerably. Firstly, from the Lyapunov equation \eqref{purple1} we obtain the conditions:
\begin{equation}\label{cheese3}
N_1=\frac{C_{1+}^{\dag}C_{1+}}{C_{1-}^{\dag}C_{1-}-C_{1+}^{\dag}C_{1+}}=\frac{-A_1-C_1^{\dag}C_1}{2A_+}.
\end{equation}
Note that this is not an over-constrained system of equations, instead they are a necessary condition of our particular choice of basis. As $M_1=\sqrt{N(N+1)}$ and using \eqref{cheese1}, \eqref{cheese2} and \eqref{cheese3} it follows  (after some careful algebra)  that 
$$C_{2-}=-\sqrt{\frac{   C_{1-}^{\dag}C_{1-}}{C_{1+}^{\dag}C_{1+}}}C_{1,+}, \quad C_{2+}=-\sqrt{\frac{   C_{1+}^{\dag}C_{1+}}{C_{1-}^{\dag}C_{1-}}}C_{1,-} \quad \mathrm{and} \quad A_2=\Delta\left( \overline{A}_{1-} , -\overline{A}_{1+}\right).$$
Equivalently, $A_2=J\overline{A_1}J$. It is straightforward to check that this system is physical (i.e satisfies PR conditions \eqref{PR}) and is stable (the eigenvalues of $A_1$ and $A_2$ coincide).

\begin{exmp}
We will now  find the dual system for the following two-mode QLS:
$$C=\left(\begin{smallmatrix}5&4&1&-i\\1&i&5&4\end{smallmatrix}\right)\quad A=\left(\begin{smallmatrix}-12-2i&0.5i&1&-2-2.5i\\-20-0.5i&-7.5-6i&-6-7.5i&-2i\\1&-2+2.5i&-12+2i&-0.5i\\-6+7.5i&2i&-20+0.5i&-7.5+6i\end{smallmatrix}\right).$$
The two-mode stationary state of this system has covariance matrix (given by the solution to the Lyapunov equation \eqref{eq.Lyapunov}):
$$P=\left(\begin{smallmatrix} 1.1067 & -0.0799 - 0.1952i & -0.1680 + 0.0636i & -0.3262 - 0.1575i\\
-0.0799 + 0.1952i  & 1.7835 &  -0.3262 - 0.1575i  & 0.8234 - 0.3690i\\
-0.1680 - 0.0636i & -0.3262 + 0.1575i  & 0.1067 & -0.0799 + 0.1952i\\
-0.3262 + 0.1575i  & 0.8234 + 0.3690i & -0.0799 - 0.1952i&   0.7835 \end{smallmatrix}\right).$$
Now, performing the change of basis
$$C\mapsto CS^{\flat} \quad A\mapsto S A S^{\flat} \quad P\mapsto S P S^{\dag},$$
where the symplectic $S$ is given by 
$$S=\left(\begin{smallmatrix}
 1.0132 + 0.0018i & -0.0325 - 0.2069i & -0.1418 - 0.0134i & -0.2018 - 0.0969i\\
-0.0060 - 0.1457i   &1.1105 + 0.0013i  &-0.1657 - 0.0058i  & 0.4286 - 0.2080i\\
0.1418 + 0.0134i & -0.2018 + 0.0969& 1.0132 - 0.0018i & -0.0325 + 0.2069i\\
-0.1657 + 0.0058i  & 0.4286 + 0.2080i&-0.0060 + 0.1457i   &1.1105 - 0.0013i  
\end{smallmatrix}\right),$$
we obtain the  TFE system characterised by  $(C_1, A_1)$, where
$$C_1=\left(\begin{smallmatrix}
4.9055 - 0.3949i &  4.2841 - 1.3610i & -0.2132 - 0.0847i  & 0.6671 - 2.2200i\\
-0.2132 + 0.0847i  & 0.6671 + 2.2200i& 4.9055 + 0.3949i &  4.2841 + 1.3610i 
\end{smallmatrix}\right)$$
$$
 A_1=\left(\begin{smallmatrix}  
 -12.0838 - 3.5321i  & 0.0117 + 1.4452i & -1.0080 - 0.4978i & -0.5957 - 0.4592i\\
 -21.5223 - 3.0091i & -7.4162 - 3.4138i &  4.4199 -10.4105i &  3.4092 - 4.9884i     \\
 -1.0080 + 0.4978i & -0.5957 + 0.4592i& -12.0838 + 3.5321i  & 0.0117 - 1.4452i \\
 4.4199 +10.4105i &  3.4092 + 4.9884i  &-21.5223 + 3.0091i & -7.4162 + 3.4138i 
\end{smallmatrix}\right).$$
 Importantly the two modes of $(C_1, A_1)$ are independent at stationarity and have covariance matrix   given by $P=\mathrm{Diag}\left(1.3623, 1.0022, 0.3623, 0.0022\right)$.
 
 We can now find the dual directly using Eqs. \eqref{cheese1} and \eqref{cheese2}. Specifically, the dual is characterised by
$$C_2=\left(\begin{smallmatrix}
4.5733 + 1.8180i & -1.2936 + 4.3049i & -0.2287 + 0.0184i  &-2.2092 + 0.7018i\\
-0.2287 - 0.0184i  &-2.2092 - 0.7018i&4.5733 - 1.8180i & -1.2936 - 4.3049i 
\end{smallmatrix}\right)
$$ $$A_2=\left(\begin{smallmatrix}   
-12.0838 + 3.5322i  & 0.0412 -21.7310i &  1.0074 - 0.4989i &  8.9136 - 6.9596i\\
-1.4331 + 0.1886i  &-7.4163 + 3.3866i & -0.2533 - 0.6494i & -3.3657 - 5.0183i\\
1.0074 + 0.4989i &  8.9136 + 6.9596i&-12.0838 - 3.5322i  & 0.0412 +21.7310i \\
 -0.2533 + 0.6494i & -3.3657 + 5.0183i&-1.4331 - 0.1886i  &-7.4163 - 3.3866i 
\end{smallmatrix}\right).$$
 One can indeed verify that the dual system $(C_2, A_2)$ is both stable (i.e. the eigenvalues of $A_2$ have negative real part)  and physical (it satisfies Eq. \eqref{PR}). 
\end{exmp}

\section{Reducibility of the Power Spectrum}\label{pight}

We now digress and use some of the results that we have just found  in order to develop a  result that enables us to prove Theorem \ref{doned} earlier. 
%a result required earlier that if a system is not globally minimal then the modes cancel singularly or in pairs in the power spectrum. More precisely the \textit{pure} part of the system in [REF] is TFE to the QLS which is decomposed as  a cascade of one-mode or two mode  QLSs subsystems; each of these subsystem  have trivial power spectrum.

\begin{thm}\label{parfg}
Consider an $n$-mode  cascaded QLS, $\mathcal{G}$, with vacuum input that has a pure stationary state. Also suppose that $\mathcal{G}$ is a cascade of two QLSs, 
$\mathcal{G}=(C_2,A_2)\triangleleft (C_1,A_1)$, 
where: 
\begin{itemize}
\item the QLS $(C_1,A_1)$
has one mode and is globally minimal.
\item the QLS $(C_2,A_2)$ has $n-1$ modes.
\end{itemize}
Then there exists a TFE QLS to $(C_2,A_2)$ given by $(C'_2, A'_2)=(C'_{2_2},A'_{2_2})\triangleleft (C'_{2_1},A'_{2_1})$
such that $(C'_{2_1},A'_{2_1})\triangleleft (C_1,A_1)$ has a pure stationary state (hence  trivial power spectrum), where
\begin{itemize}
\item the QLS $(C'_{2_1},A'_{2_1})$
has one mode and is globally minimal and
\item the QLS $(C'_{2_2},A'_{2_2})$ has $n-2$ modes.
\end{itemize}

The upshot is that the modes in the power spectrum are reducible in pairs.
\end{thm}
\begin{remark}
Note that we have chosen to impose the requirement that $(C_1,A_1)$ be globally minimal, otherwise if it were not, then the mode would be reducible. Therefore this assumption is necessary to rule out the trivial case.
\end{remark}

%[[[NEED to go back and check that use theorem exactly]]]

%[[[is the form of the system important?]]

%ARE WE MISSING A SUBTLETY IN DUAL SECTION THAT  (i.e forgot to mention) any system has a dual because any the equiv system of original system has that same dual?
%I don't think it's stated that every system has a dual

%with vacuum input that  has a pure stationary state, where $(C_1,A_1)$ is a 1 mode globally minimal QLS and $(C_2,A_2)$ is an $n-1$  mode  QLS. Then there exists QLSs $(C'_1,A'_1)$ and $(C'_2, A'_2)=(C'_{2_2},A'_{2_2})\triangleleft (C'_{2_1},A'_{2_1})$ that have the same transfer functions as $(C_1,A_1)$ and 
%$(C_{2},A_{2})$, respectively, and satisfy the condition:
%$$(C'_{2_1},A'_{2_1})=
%\left(\Delta(-\sqrt{\frac{C_{-1}^{'\dag}C'_{-1}}{C_{+1}^{'\dag}C'_{+1} }} C'_{+1},     -\sqrt{\frac{C_{+1}^{'\dag}C'_{+1}}{C_{-1}^{'\dag}C'_{-1} }} C'_{-1}            ), \Delta(\overline{A'}_{-1}, -\overline{A'}_{+1})\right).$$

%The upshot is that the QLS $(C'_{2_1},A'_{2_1})\triangleleft (C'_1,A'_1)$ has a pure stationary state (hence it has trivial power spectrum).
%\end{thm}
%
%
%
\begin{proof}
First, using  Theorem \ref{WOLFF} and considering a bipartition of modes between the systems $(C_1,A_1)$ and $(C_2,A_2)$ there exists TFE systems  $(C'_1,A'_1)$ and $(C'_2,A'_2)$ with combined stationary state
$$\left<\left(\begin{smallmatrix}\breve{\mathbf{a}}_1\\\breve{\mathbf{a}}_{2_1}\\\breve{\mathbf{a}}_{2_2}\end{smallmatrix}\right)\left(\begin{smallmatrix}\breve{\mathbf{a}}^{\dag}_1&\breve{\mathbf{a}}^{\dag}_{2_1}&\breve{\mathbf{a}}^{\dag}_{2_2}\end{smallmatrix}\right)
\right>:=\left(\begin{smallmatrix} P&Q&0\\Q&P&0\\0&0&V\end{smallmatrix}\right)
%=\left(\begin{smallmatrix}N_1+1&0&0&M&0\\0&N_1&M_1&0&0\\
%0&M&N+1&0&0\\M&0&0&N&0\\0&0&0&0&V
%\end{smallmatrix}\right)
,$$
where 
$$P=\left(\begin{smallmatrix}N+1&0\\0&N\end{smallmatrix}\right) \quad Q=\left(\begin{smallmatrix}0&M\\M&0\end{smallmatrix}\right)$$
with $N, M\in\mathbb{R}$ such that $P=QP^{-1}Q$. That is, the mode belonging to $(C'_1,A'_1)$ and one mode of $(C'_2,A'_2)$ have an entangled pure stationary state. The matrix $V$ here is of size $2(n-2)\times2(n-2)$ and is the (pure) covariance matrix of the modes $\breve{\mathbf{a}}_{2_2}$.

Now write the matrices $(C'_2,A'_2)$ as $$\left( (C'_{2_1}, C'_{2_2}), \left(\begin{smallmatrix}  A'_{2_{1}} &A'_{2_{12}}\\A'_{2_{21}}&A'_{2_{2}}\end{smallmatrix}\right) \right).$$ Note that the modes are ordered as  $\left[\begin{smallmatrix}\breve{\mathrm{a}}_1\\ \breve{\mathrm{a}}_{2_1}\\ \breve{\mathrm{a}}_{2_2}\end{smallmatrix}\right]$ in this representation. 

Now, using  the usual  Lyapunov equation \eqref{eq.Lyapunov} for the stationary state, and by writing it as a $3\times3$ block matrix equation with respect to the modes $\left[\begin{smallmatrix}\breve{\mathrm{a}}_1\\ \breve{\mathrm{a}}_{2_1}\\ \breve{\mathrm{a}}_{2_2}\end{smallmatrix}\right]$, it follows that the $(1,1)$, $(1,2)$, $(2,1)$ and $(2,2)$ reduce to the Eq. \eqref{purple} in Sec. \ref{purple6}, hence $(C'_{2_1},A'_{2_1})$ are given by Eqs. \eqref{cheese1} and \eqref{cheese2}; we will use this fact shortly.
%Therefore, $C'_{2_1}$ and $A_{2_{1}}$ are of the form in the Theorem by the analysis in Sec. \ref{purple6}.  

It remains to show that  $A_{2_{12}}=0$ and $A_{2_{21}}=-C^{'\flat}_{2_2}C'_{2_2}$, for then this would imply  that the system $(C'_2,A'_2)$ is a cascade of the QLSs $(C'_{2_1},A'_{2_{1}})$  and $(C'_{2_2},A'_{2_{2}})$ (see Eq. \eqref{concat1}). To see this, begin by observing that the $(1,3)$ and $(2,3)$ block entries of the Lyapunov equation lead to the equations
\begin{equation}\label{feb1}
A'_{2_{21}}P+VA^{'\dag}_{2_{12}}-  C^{'\flat}_{2_2}C'_1Q         +C^{'\flat}_{2_2}V_{\mathrm{Vac}}\left(C^{'\flat}_{2_1}\right)^{\dag}=0
\end{equation}
\begin{equation}\label{feb2}
A'_{2_{21}}Q-    C^{'\flat}_{2_2}C'_{1}P       +C^{'\flat}_{2_2}V_{\mathrm{Vac}}\left(C^{'\flat}_{1}\right)^{\dag}=0.
\end{equation}
Now multiplying Eq. \eqref{feb2} by $Q^{-1}P$ and using \eqref{cheese1} and the condition $P=QP^{-1}Q$ gives
$$
A'_{2_{21}}P-  C^{'\flat}_{2_2}C'_1Q         +C^{'\flat}_{2_2}V_{\mathrm{Vac}}\left(C^{'\flat}_{2_1}\right)^{\dag}=0.
$$
Finally, subtracting this equation from \eqref{feb1} gives $A_{2_{12}}=0$, as required. Finally,  the condition $A_{2_{21}}=-C^{'\flat}_{2_2}C'_{2_2}$ follows from the PR conditions \eqref{PR}.
\end{proof}

%REMARK THAT DPENDS ON BAIS CHOICE COULD BE SINGLE MODE CANCELLATIONS ENTIRELY

The interpretation of  the theorem is that the  system, $(C'_{2_1},A'_{2_1})$ is absorbing the energy emitted by the system, $(C_{1},A_{1})$; or more precisely the overall effect of $(C_{1},A_{1})$ is nullified by $(C'_{2_1},A'_{2_1})$, which is analogous to destructive interference.

%[[[need to show first type of system are passive]]]

%[[[Link to theorem 3...either thm 3 is special case of it or this is more general case of theorem 3....two types of elementary system ]]]

%

%
%
%
%%
%%
%
%
%
%
%
%
%
%

\section{Application 1: Estimation}\label{blog}

Quantum absorbers offer interesting opportunities for quantum estimation within QLSs. In this subsection we discuss one instance of this.
 Suppose that  there are a set of unknown parameter(s), $\theta$, that we would to estimate within a QLS. By using a dual system at parameter $\theta_0$ (representing prior knowledge of the system) enables one to focus on a much smaller  neighbourhood of the unknown parameter space. Essentially the dual system now constitutes part  of the measurement  and increases the  available class of realistic measurements (see Fig \ref{absmet}). Even in the standard metrology setup (see Sec. \ref{class}) adaptive measurement procedures using some of the available resources to obtain rough guesses for the parameter(s) are necessary  \cite{Monras1} because the optimal measurement choice itself may depend on the unknown parameter(s); for example the optimal quadrature to measure in a homodyne measurement.

\begin{figure}
\centering
\includegraphics[scale=0.25]{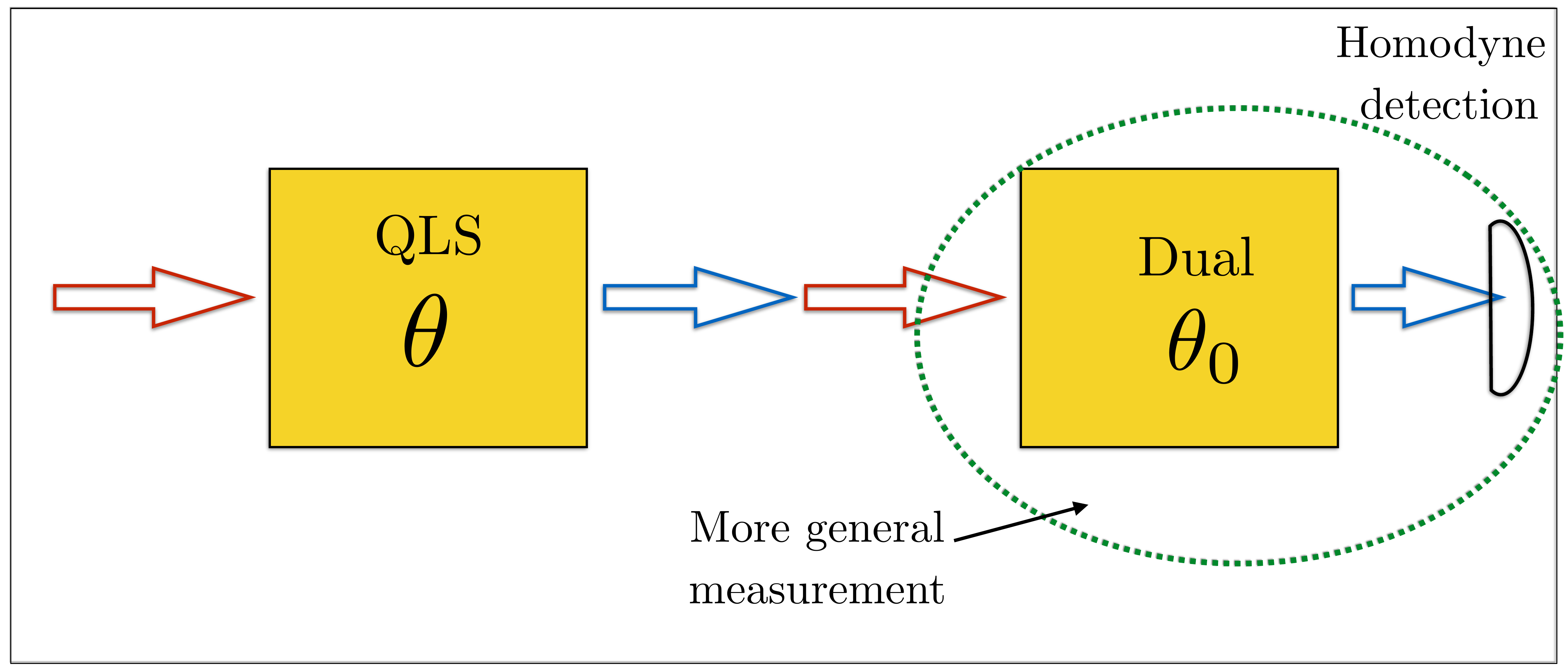}
\caption{Metrology setup with quantum absorber.  \label{absmet}}
\end{figure}

Consider the setup in Sec. \ref{dfsw} where we have a 
SISO QLS (note that in particular  we are going beyond the class of passive systems in Sec. \ref{dfsw}) with (at least) one eigenvalue that is $\mathcal{O}(\tau^{-1})$ from the imaginary axis, where $\tau=T_{\mathrm{tot}}^{1-\epsilon}$ and $T_{\mathrm{tot}}$ is large enough so that the system has reached stationarity at time $T_{\mathrm{tot}}$.
Assume that the 
 input to our system is  characterised by the covariance $V_{\mathrm{vac}}$  (we can assume this WLOG  by using the trick from Sec. \ref{G.Ms}). We also employ  
a quantum absorber system characterised by $\theta_0$, so that  at $\theta=\theta_0$ the output  of the combined system will also be vacuum; we will use this observation many times in the following.

\subsection{QFI Scaling Calculation Revisited}\label{revo}
Recall from Sec. \ref{ounce1} that we can obtain an equivalent definition of the QFI by working in the time domain. By  using a dual system this expression simplifies  and enables us to understand the QFI enhancement  from Sec. \ref{poke11}.
Note that applying a dual system to the output field does not change the QFI.

Let us simplify things by discretising   the process and working with increments of unit length. That is,  consider $\breve{\mathbf{B}}_i=\int^i_{i-1}  \breve{\mathbf{b}}_{\mathrm{out}}(t)dt$ for $i=\{1, 2,...T_{\mathrm{tot}}\}$. Denote the covariance of this Gaussian process by the (quadrature picture) covariance  matrix
$$
V(\theta) := {\rm Re}\left(
\begin{array}{cc}
\langle {\bf Q}   {\bf Q}^T \rangle & \langle {\bf Q} {\bf P}^T \rangle\\
\langle {\bf P}{\bf Q}^T \rangle & \langle {\bf P}{\bf P}^T \rangle
\end{array}
\right),
$$
where $\mathbf{Q}$ and $\mathbf{P}$ are vectors canonical coordinates of the $T_{\mathrm{tot}}$ discrete modes.
 It follows from \cite{Monras1} that the QFI of such a Gaussian model is given by 
\begin{equation}
F(\theta)=\mathrm{Tr}\left(V(\theta)^{-1}\frac{dV(\theta)}{d\theta}V(\theta)^{-1}\frac{dV(\theta)}{d\theta}\right).
\end{equation}
Now let us evaluate the QFI at $\theta=\theta_0$, where the output is vacuum
 ($V(\theta)=\mathds{1}$).
%[[[Say analogous to evaluating Metrology model at $\theta=0$]]].
Therefore the QFI may be written as $F(\theta=\theta_0)=\left.\sum_{i,j}\left(\frac{dV}{d\theta}\right)^2_{i,j}\right|_{\theta=\theta_0}$. Now, $\left.\left(\frac{dV(\theta_0)}{d\theta}\right)_{i,j}\right|_{\theta=\theta_0}$ is significant so long as  $\frac{1}{2}|i-j|<\frac{1}{|\mathrm{Re}(\lambda_{\mathrm{min}})|}:=\tau$ where $\lambda_{\mathrm{min}}$ is the eigenvalue of $A$ closest to the imaginary axis. Therefore $F(\theta)\propto T\tau$. So at opposite ends of the spectrum we have: 
\begin{enumerate}
\item  QFI of the order $\mathcal{O}(T_{\mathrm{tot}})$, when the time is much larger than one over the spectral gap (standard scaling).
\item 
QFI  of the order $\mathcal{O}(T^2_{\mathrm{Tot}})$, when the spectral gap is comparable with $T_{\mathrm{tot}}$ (Heisenberg).  In this case the useful information about the system  is contained only within output correlations  between times, $t$ and $s$ such that $|t-s|<\tau$, where $\tau$ is the inverse of the spectral gap. This observation will be key to developing our measurement and estimator in the following subsection.
\end{enumerate}

\subsection{Signal-to-Noise Ratio (SNR) for Quadrature Measurements}\label{SIDF}

\subsubsection{The Discrete Time Measurement}
We now find a measurement and show how to construct an estimator enabling Heisenberg scaling. Initially it's simpler to understand how to do this if one works in discrete time.     
Consider the measurement 
\begin{align}\label{cof1}
\mathbf{Y}=\sum_{i,j}\alpha_{i,j}\mathbf{X}_i\mathbf{X}_j, 
\end{align}
where $\mathbf{X}_i$
 is the position quadrature  of the Gaussian mode $\breve{\mathbf{B}}_i$.  
The outcome  can obtained by measuring $\mathbf{X}$ and then post-processing by weighting with the factor $\alpha_{ij}$.
Denote the outcomes of  the operators $\mathbf{X}_i$ and $\mathbf{Y}$ by $X_i$ and $Y$, respectively. The  weighting factor $\alpha_{ij}$ will be prescribed later. 
Now   ${X}_1, ..., {X}_{T_{\mathrm{tot}}}$ form a stationary process with $\mathbb{E}\left[{X}_i\right]=0$ and $\mathbb{E}\left[{X}_i{X}_j\right]=W_{ij}$, where the   law of ${X}_1,...{X}_{T_{\mathrm{tot}}}$ depends on $\theta$ via the covariance  $W_{ij}$.  The CFI of the whole process is typically hard to calculate, so we seek a lower bound. 
A metric often used to characterise the performance of precision measurements is the signal-to-noise ratio (SNR) \cite{Escher1}, which is a lower bound for CFI. 
The SNR of ${Y}$ is defined as 
\begin{equation}\label{SNR}
\mathrm{SNR}_{{Y}}=\frac{\left(\frac{d\mathbb{E}\left[{Y}\right]}{d\theta}\right)^2}{\mathrm{Var}[{Y}]}
\end{equation}
The SNR gives the error in using a linear transformation of $Y$ as an estimator for $\theta$.

%[[[COULD MENTION THAT IT's a post processing task]]]

At $\theta=\theta_0$ all of the ${X}_i$s are independent (because we are using an absorber),  whereas at $\theta\neq\theta_0$ the process has long correlation length $\tau$. Therefore,  $\left(\frac{dV(\theta_0)}{d\theta}\right)_{i,j}$ is significant so long as  $|i-j|<\tau$. Hence 
\begin{equation}\label{goats4}
\frac{d\mathbb{E}\left[{Y}\right]}{d\theta}=\sum_{ij}\alpha_{ij}\dot{V}_{ij}\approx \tau T_{\mathrm{tot}}.
\end{equation}
On the other hand evaluating the denominator of the SNR at $\theta=\theta_0$ we obtain
\begin{align*}
\mathrm{Var}[{Y}]&=\sum_{ijkl}\alpha_{ij}\alpha_{kl}\mathbb{E}\left[(X_iX_j-V_{ij}) (X_kX_l-V_{kl})   \right] \\
&=\sum_{ijkl}\alpha_{ij}\alpha_{kl}\mathbb{E}\left[(X_iX_j-\delta_{ij}) (X_kX_l-\delta_{kl})   \right].
\end{align*}
Now, 
\begin{itemize}
\item if $i\neq k$ and $j\neq l$ then the summand equals $\mathbb{E}\left[(X_iX_j-\delta_{ij})\right]^2=0$;
\item  if  $i=k$ but $j\neq l$ then the summand equals $\mathbb{E}\left[(X_iX_j-\delta_{ij}) (X_kX_l-\delta_{kl})   \right]=\mathbb{E}[X_i^2X_jX_l]-\delta_{ij}\delta_{il}=0$;
\item if $i\neq k$ but $j= l$, then the summand equals $\mathbb{E}\left[(X_iX_j-\delta_{ij}) (X_kX_l-\delta_{kl})   \right]=\mathbb{E}[X_iX_kX_j^2]-\delta_{ij}\delta_{kj}=0$. 
\end{itemize}
The upshot is that the only contribution to $\mathrm{Var}[Y]$ is when $(i,j)=(k,l)$. Therefore 
\begin{align}
\nonumber\mathrm{Var}[Y]&=\sum_{ij}\alpha^2_{ij}\mathbb{E}\left[(X_iX_j-\delta_{ij})^2 \right]\\
\nonumber&=\sum_{i}\alpha^2_{ii}\mathbb{E}\left[(X_i^2-1)^2 \right]+\sum_{i,j: i\neq j}\alpha_{ij}\mathbb{E}\left[X_i^2 \right]\mathbb{E}\left[X_j^2 \right]\\
\label{gfy}&=\sum_{i}\alpha^2_{ii}\mathbb{E}\left[(X_i^2-1)^2 \right]+\sum_{ij : i\neq j}\alpha^2_{ij}V_{ii}V_{jj}.
\end{align}
Now, both $\mathbb{E}\left[(X_i^2-1)^2 \right]$ and $V_{ii}$ are constants independent of $i$. Therefore, the variance depends on the shape of $\alpha_{ij}$. Let $\alpha_{ij}=\alpha_{i-j}$ and consider the following two cases: (i) $\alpha_{ij}=1$, (ii) $\alpha_{ij}=\alpha_{i-j}=e^{-\frac{|i-j|}{\tau}}$ (see Fig. \ref{2cases}). 
\begin{itemize}
\item
Case (i) corresponds to a measurement of zero-frequency (or equivalently total integrated current) i.e $\mathbf{Y}=\left(\sum_i \mathbf{B}_i+\mathbf{B}_i^{\dag}\right)^2$. In this case $\mathrm{Var}(Y)\sim T_{\mathrm{tot}}+T_{\mathrm{tot}}^2\sim T_{\mathrm{tot}}^2$, so that $\mathrm{SNR}_Y\sim\frac{(\tau T_{\mathrm{tot}})^2}{T_{\mathrm{tot}}^2}=\tau^2$. 
\item Case (ii) corresponds to placing more weight on mode correlations within the stabilisation time (and less weight on longer autocorrelations) in accordance with the QFI calculation from Sec. \ref{revo}. In this case $\mathrm{Var}(Y)\sim \tau+\tau T_{\mathrm{tot}}\sim \tau T$, so that $\mathrm{SNR}_Y\sim\frac{(\tau T_{\mathrm{tot}})^2}{\tau T_{\mathrm{tot}} }=\tau T_{\mathrm{tot}}$.
\end{itemize}
Therefore, the  SNR is enhanced when $\alpha_{ij}$ has width  order of the correlation length. The essential point is that if $\alpha_{ij}$ is not spread over all $(i,j)$, as in the case of square of total integrated current, then the variance is not too large while the mean can still be large since the important contributions come from $|i-j|<\tau$.

\begin{figure}
\centering
\includegraphics[scale=0.18]{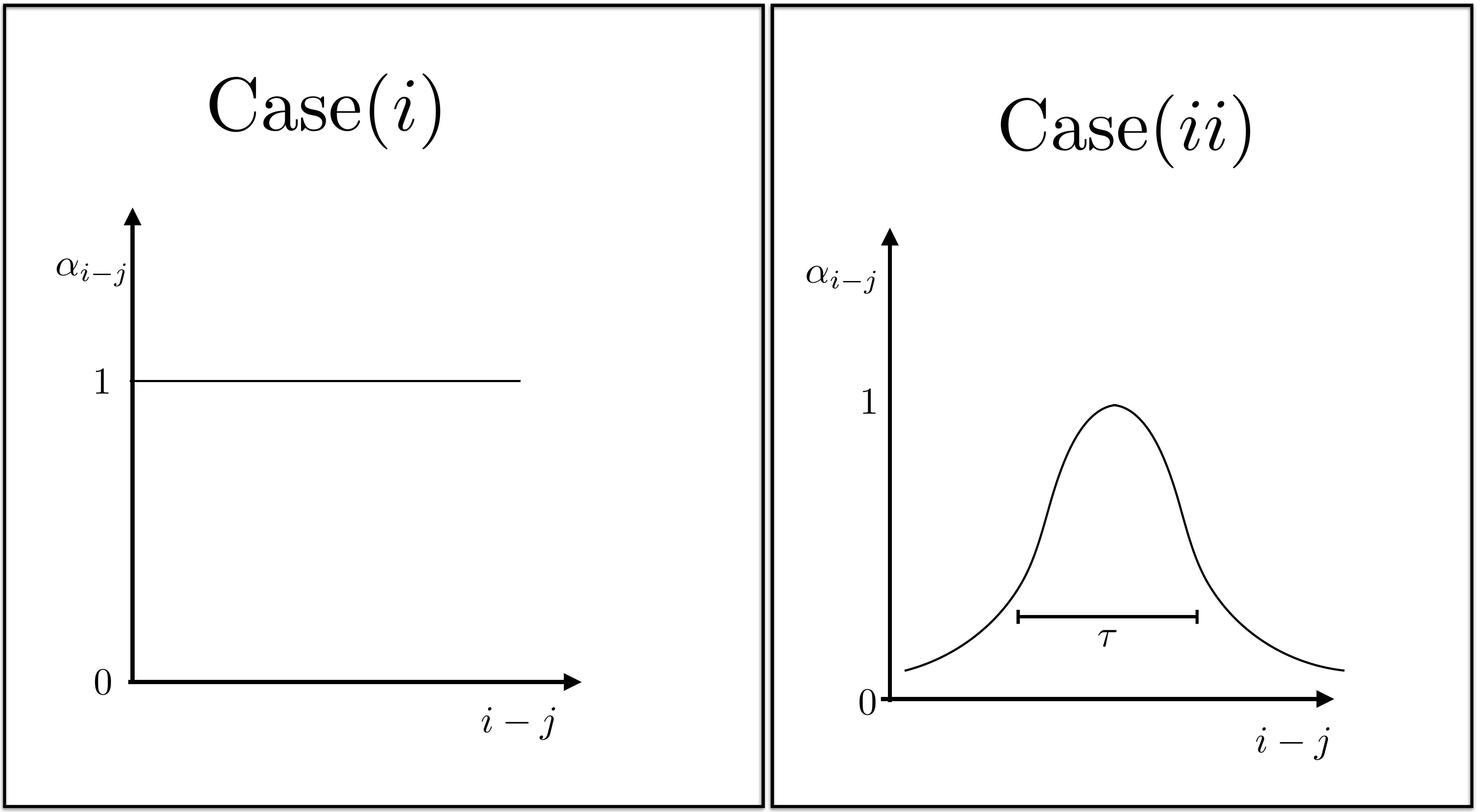}
\caption{The figure compares the two weights, $\alpha_{ij}$, that we consider here.
\label{2cases}}
\end{figure}

\subsubsection{Returning to Continuous Time}

Now returning to continuous time, the analogue of \eqref{cof1}  is given by
\begin{align} \label{cof2}
\mathbf{Y}=\int^{T_{\mathrm{tot}}}_0\int_0^{T_{\mathrm{tot}}}k(t,t')
\mathbf{X}(t)\mathbf{X}(t')dtdt'
\end{align}
where $\mathbf{X}(t)$ is the quadrature measurement $\mathbf{X}(t):=\mathbf{b}_{\mathrm{out}}(t)+\mathbf{b}_{\mathrm{out}}^{\dag}(t)$ and we have weighting factor $K(t,t'):= e^{\frac{-|t-t'|}{\tau}+i\omega_0(t-t')}$.  Notice that we have generalised our estimator by allowing for a frequency shift $\omega_0\in\mathbb{R}$; the reason why we do this will become clear shortly. Denote the outcomes of $\mathbf{X}(t)$ and $\mathbf{Y}$ by $X(t)$ and $Y$, respectively.

Let us consider  explicitly calculating the SNR \eqref{SNR}, where now $Y$ is the continuous-time version of the measurement. 
Firstly, 
 since at $\theta=\theta_0$ the input field is equal to the output field (as we are using the dual), the calculation of the denominator of the SNR is identical to the discrete time version above (i.e., case (ii)). Hence  it is of order $\mathcal{O}\left(\tau{T_{\mathrm{tot}}}\right)$.
Computing the numerator of the SNR    directly entails calculating   $\mathbb{E}\left[{X}(t){X}(t')\right]$ by using  the following expression, which is obtained from \cite{Zhang5}:
\begin{align}\label{cows1}
&\left<\breve{\mathbf{b}}_{\mathrm{out}}(t)\breve{\mathbf{b}}^{\dag}_{\mathrm{out}}(t')\right>=V\delta(t-t')\\&+
\begin{cases}\nonumber
\int^s_0Ce^{A(t-r)}C^{\flat}V\left(C^{\flat}\right)^{\dag}e^{A^{\dag}(t'-r)}C^{\dag}dr -Ce^{A(t-t')}C^{\flat}V   & \quad \text{if }t> t' \\
0 & \quad \text{if }t= t'\\
   \int^t_0 Ce^{A(t-r)}C^{\flat}V\left(C^{\flat}\right)^{\dag}e^{A^{\dag}(t'-r)}C^{\dag}dr -V\left(Ce^{A(t'-t)}C^{\flat}\right)^{\dag}   & \quad \text{if }t< t'.\\
  \end{cases}
\end{align}
Let us see this in an example.

\begin{exmp}\label{lastex}
 Consider a  one mode passive cavity, where this time the input is $V(N,M)$ rather than vacuum, as in example \ref{fug} from  Sec. \ref{dfsw}. Recall that $\tau=1/c^2$ (assuming that $c\in\mathbb{R}$) and let $\tau=T_{\mathrm{tot}}^{1-\epsilon}$ as in example \ref{fug}. For simplicity in the following  calculations, we observe  the experiment   over the interval $t\in[T_{\mathrm{tot}}^{1-\epsilon/2}, T_{\mathrm{tot}}]$ so that $e^{At}\approx 0$ for all $t$ and $T_{\mathrm{tot}}-T_{\mathrm{tot}}^{1-\epsilon/2}\approx T_{\mathrm{tot}}$. 
 We shall explicitly compute the SNR here in the case when $\Omega=0$. Firstly, we can compute Eq.  \eqref{cows1}; doing so and summing over the elements of the resulting matrix we obtain:
\begin{align*}
&\mathbb{E}\left[{X}(t){X}(t')\right]=\left<{\mathbf{X}}(t){\mathbf{X}}^{\dag}(t')\right>=(2N+1+M+\overline{M})\delta(t-t')\\&+
\begin{cases}
-Mc^2\left(\frac{c^2}{2A}+1\right)e^{A({t-t'})}  -\overline{M}c^2\left(\frac{c^2}{2\overline{A}}+1\right)e^{\overline{A}({t-t'})}    & \quad \text{if }t> t' \\
0 & \quad \text{if }t= t'\\
   -Mc^2\left(\frac{c^2}{2A}+1\right)e^{A({t'-t})}  -\overline{M}c^2\left(\frac{c^2}{2\overline{A}}+1\right)e^{\overline{A}({t'-t})}   & \quad \text{if }t< t'.\\
  \end{cases}
\end{align*}
 Note that if $\Omega=0$ we have $\mathbb{E}\left[{X}(t){X}(t')\right]=(2N+1+M+\overline{M})\delta(t-t')$ and in particular the input equals the output, which shouldn't be surprising considering that the system is non-globally minimal in this case (see Sec. \ref{powers34}). The significance of this is that we don't require a dual system in order to simplify the calculation of $\mathrm{Var}[Y]$. In  particular we have $\mathrm{Var}[Y]=\mathcal{O}(\tau T_{\mathrm{tot}})$ here, for the same reason that we had it for the case of vacuum input (see Eq. \eqref{gfy}).
 
 Now, computing  the numerator of the SNR \eqref{SNR}   from this for the case of estimating  $\Omega$ gives
\begin{align}
& \label{vorf} \int^{T_{\mathrm{tot}}}_0\int_0^{T_{\mathrm{tot}}}k(t,t')\frac{d\mathbb{E}\left[{X}(t){X}(t')\right]}{d\Omega}dtdt'\\
\nonumber &=2c^2\mathrm{Re}   \left\{  \frac{-M\dot{A}c^2}{2A^2}\left[\frac{1}{A+u}\left(T_{\mathrm{tot}}-\frac{1}{A+u}\right)+\frac{1}{A+\overline{u}}\left(T_{\mathrm{tot}}-\frac{1}{A+\overline{u}}\right)\right]\right.\\
\nonumber &-T_{\mathrm{tot}}\left.M\dot{A}\left(\frac{c^2}{2A}+1\right)\left[\frac{1}{(A+u)^2}+\frac{1}{(A+\overline{u})^2}\right]\right\},
\end{align}
where $u=-\tau-i\omega_0$. Notice that  Eq. \eqref{vorf} is invariant under $u\mapsto \overline{u}$.
We evaluate Eq. \eqref{vorf}  at $\Omega=0$; we have a simplification since  $\left(\frac{c^2}{2A}+1\right)=0$. Setting also $\omega_0=0$ and using $\tau=1/c^{2}$, we obtain:
$$\left.\frac{d\mathbb{E}\left[{Y}\right]}{d\Omega}\right|_{\Omega=\omega_0=0}=\frac{16}{3}\tau(T_{\mathrm{tot}}+\frac{2}{3}\tau)\mathrm{Re}(iM)\approx\frac{16}{3}\tau T_{\mathrm{tot}}\mathrm{Re}(iM)$$
in the limit $T_{\mathrm{tot}}$ large. 
Hence the level of scaling as in case (ii) above is realised. 
One can show that in the case $\Omega\neq0$  the  frequency choices $\omega_0=\pm \Omega$ would equally realise the $\mathcal{O}(\tau T_{\mathrm{tot}})$ scaling.    
%We shall see shortly that the reason for this is that $\mathbf{Y}$ is similar to a frequency band-limited measurement of width $1/\tau$ 
\end{exmp}

\subsubsection{Interpretation of our Measurement in the Frequency Domain}

In general Eq. \eqref{cows1} (and hence the SNR)  is difficult to calculate directly.
Our choice of measurement and the reason why the weighting  factor $K(t,t')$ enhanced the SNR may better understood in the frequency domain. Recall from Sec. \ref{poke11}  that since the input-output  map acts separately on different frequencies and the input 
state is a product,  the QFI is the integral of the QFI for each individual frequency with respect to frequency. 
We saw how the main contributions to this integral came from two  small intervals $|\omega\pm\mathrm{Im}(z_1)|\sim \frac{1}{\tau}$ where $z_1$ is the  eigenvalue of $A$ closest to the imaginary  axis.  Recall that the value of QFI rate in this interval was $\tau^2$, so that the overall QFI  is of the order $T_{\mathrm{tot}}\times\frac{\tau^2}{\tau} =T_{\mathrm{tot}}\tau$. 
We shall now see that $\mathbf{Y}$ is similar to a frequency band-limited measurement of all frequency  quadratures  
$$\mathbf{X}(\omega):=\int^{T_{\mathrm{tot}}}_{0}e^{i\omega t} \mathbf{X}(t)\mathrm{d}t$$
  over a bandwidth of $1/\tau$ centred on the frequency $\omega_0$. Therefore, the fact that there are two frequency bands containing increased information about the parameter  from the calculation in Sec. \ref{poke11} (see Fig. \ref{picf}), explains why two frequency choices $\omega_0=\pm \Omega$ in example \ref{lastex} would have worked equally well.

Consider measuring  $\mathbf{X}(\omega)$, over frequency band $|\omega-\mathrm{Im}(z_1)|\sim\frac{1}{\tau}$. Denote the outcomes by  $X(\omega)$. The information about $\theta$ is contained within the second order moments, hence we consider  the following as an estimator:
\begin{align}\label{lick}
Z:=2\tau \int^{\mathrm{Im}(z_1)+\frac{1}{\tau}}_{\mathrm{Im}(z_1)-\frac{1}{\tau}}X(\omega)X(\omega)^{\dag}\mathrm{d}\omega.
\end{align} Note that the factor $2\tau$ is a normalisation factor. 
Now 
\begin{align*}
Z&=2\tau  \int^{\mathrm{Im}(z_1)+\frac{1}{\tau}}_{\mathrm{Im}(z_1)-\frac{1}{\tau}}\left(\int^{T_{\mathrm{tot}}}_{0}\int^{T_{\mathrm{tot}}}_{0}e^{i\omega (t-t')} X(t)X(t')\mathrm{d}t\mathrm{d}t'\right)\mathrm{d}\omega\\
&\approx 2  \tau    \int^{T_{\mathrm{tot}}}_{0}\int^{T_{\mathrm{tot}}}_{0}  \left(\int^{\infty}_{-\infty}  
 e^{i\omega (t-t')} e^{-|\omega-  \mathrm{Im}(z_1)|\tau}      X(t)X(t')                           \mathrm{d}\omega\right)\mathrm{d}t \mathrm{d}t' \\
 &= \int^{T_{\mathrm{tot}}}_{0}\int^{T_{\mathrm{tot}}}_{0} e^{i \mathrm{Im}(z_1) (t-t')}\frac{4\tau^2}{\tau^2+(t-t')^2} X(t)X(t')  \mathrm{d}t\mathrm{d}t'\\
%& \approx T_{\mathrm{tot}}  \int^{\tau}_{0}     e^{i \mathrm{Im}(z_1) t}  X(t)X(0)  \mathrm{d}t\\
&\approx \left.Y\right|_{\omega_0=\mathrm{Im}(z_1)}\\
\end{align*}
Here we have \eqref{lick} in the  first equality and then switched the order integration in the second.

Therefore, the estimator $Z$ is approximately equal to the estimator $Y$.  Hence, as mentioned above, we can interpret our estimator $Y$ in the frequency domain as frequency band-limited measurement  of width $1/\tau$ centred on the frequency $\omega_0$.
This calculation allows us to calculate the numerator of the SNR \eqref{SNR} by using $Z$ rather than $Y$ and show that it is of the order in case (ii) above. To this end we have 
\begin{align}
\nonumber 
\frac{d\mathbb{E}\left[{Z}\right]}{d\theta}&= 2\tau \int^{\mathrm{Im}(z_1)+\frac{1}{\tau}}_{\mathrm{Im}(z_1)-\frac{1}{\tau}}    
\frac{d\mathbb{E}\left[{X(\omega)X(\omega)^{\dag}}\right]}{d\theta} \mathrm{d}\omega\\
\label{biker}&=\mathcal{O}(\tau T_{\mathrm{tot}}),  
\end{align}
which follows by using the observation \eqref{tease2}\footnote{We originally only showed \eqref{tease2} for the case of passive QLS, but it is also true for general QLSs.}.
Therefore, if 
we choose $\omega_0=  \pm i\mathrm{Im}(\hat{z}_1)$,  where ${\hat{z}_1}$ is an estimator of $z_1$  (the spectral gap of  $A$), we can extract information from either of the most informative frequency intervals  $|\omega-\mathrm{Im}(z_1)|\sim\frac{1}{\tau}$.
Here, our estimator ${\hat{z}_1}$ has been obtained from a preliminary experiment   with MSE $\mathcal{O}(\tau^{-2})$. 
\begin{remark} 
Notice that our estimator, which claims Heisenberg level scaling, seemingly requires Heisenberg level knowledge of the $z_1$; that is the MSE of ${\hat{z}_1}$ is $\mathcal{O}(\tau^{-2})$. 
However, an adaptive procedure like the one in Sec. \ref{toy} for the time-dependent approach may be applied to overcome this problem and achieve Heisenberg scaling (we do not discuss this any further here).   
\end{remark}

\begin{remark}
Notice that the scaling $\mathcal{O}(\tau T_{\mathrm{tot}})$ in Eq. \eqref{biker} is possible with or without the absorber system. The action of the dual system in the frequency domain on the independent modes is a frequency dependent rotation. That is, the absorber doesn't determine or shift the most informative frequency interval, but rather rotates phase space on each particular frequency; the result of this on the scaling in Eq. \eqref{biker} is a  constant factor. We discuss this is more detail in Sec. \ref{optoe1}.
\end{remark}

%must choose complex m...actually gives best input....what does it actually mean...is it momentum squeezed and measuring X?....could revisit later.

%-2c^2\mathrm{Re}\left(M\dot{A}T\left(\frac{c^2}{2A}+1\right)\left(\frac{1}{(A+u)^2}+\frac{1}{(A+\overline{u})^2}\right)\right)        \right)

Now  the quadrature, $\mathbf{X}(t)$, that we chose to measure above may not be the best quadrature to measure. Moreover, since our measurement is very similar to a frequency band-limited measurement above, where each frequency behaves independently, it could be the case that the optimal quadrature may differ across frequencies. If we were to measure  these frequency domain quadratures with frequency dependent phases, we would get the exact optimal measurement. However, this measurement will not translate into one that can be done sequentially in time; for that we would need to have equal phases for all frequencies, which happens in the measurement  \eqref{cof2}.

\subsection{Optimal Estimation Using Adaptive Measurements}\label{optoe1}
As mentioned above, our strategy in the previous subsection is by no means optimal; it fails to be so by a constant factor. We now discuss a method based on the work in \cite{Monras1, Monras2} to reach optimality (asymptotically in the limit of large times). In this subsection we shall showcase the capability of the absorber system for estimation.

\begin{figure}
\centering
\includegraphics[scale=0.18]{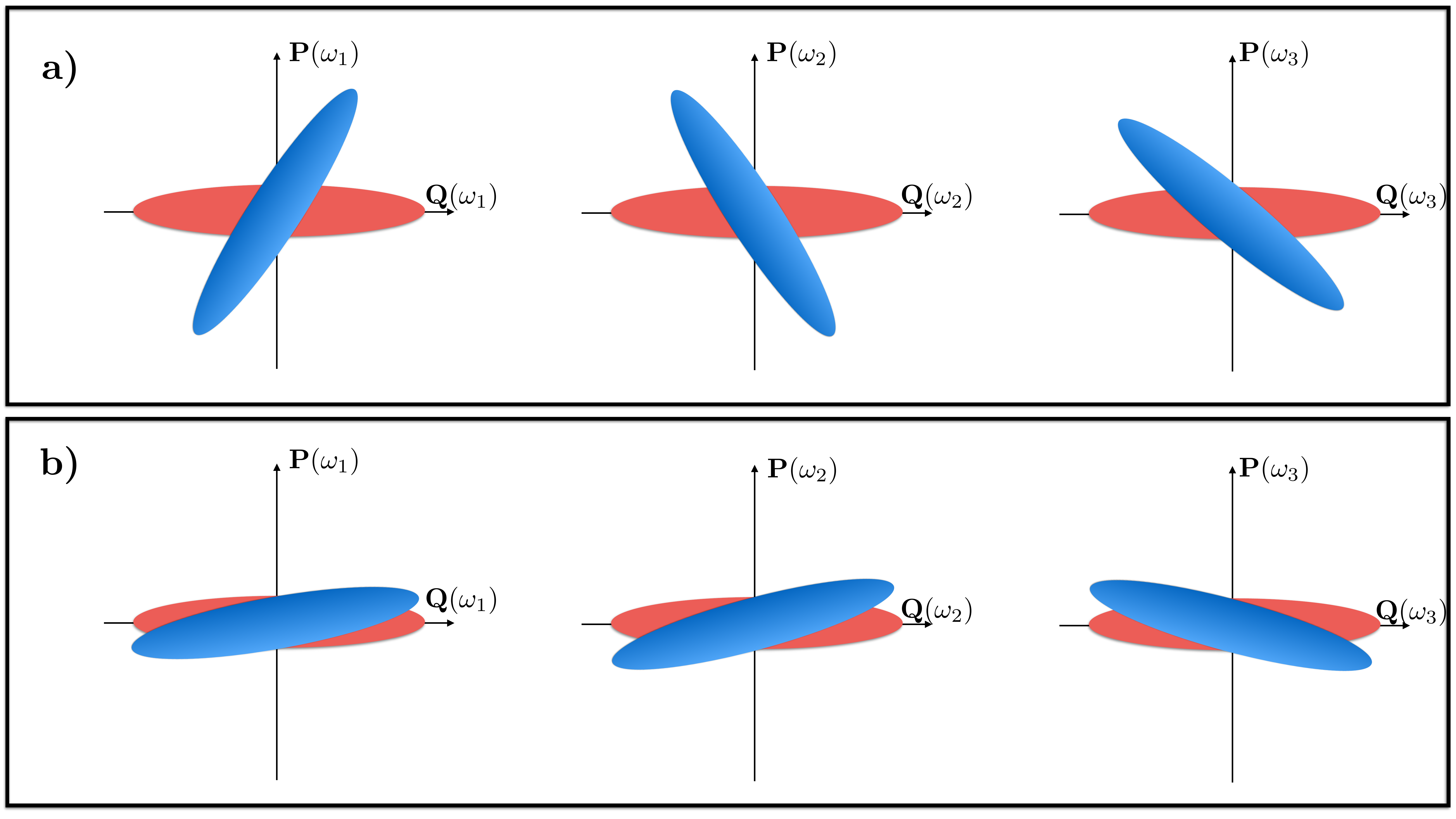}
\caption{This figure shows the action of the system on the input viewed in the  frequency domain. The inputs are shown in red and the output in blue for three choices of frequency. Note that the input is the same for all frequencies. In a) the system rotates by a frequency dependent angle. Part b) shows that  the effect of the dual system is to make all rotations smaller. 
\label{rotat}}
\end{figure}

For simplicity, we consider a SISO PQLS in this subsection with input $V(N,M)$ (rather than vacuum), as in Sec. \ref{dfsw}.  In the frequency domain, all modes are independent and  the input-output map is a rotation of a squeezed state by a frequency dependent parameter (see Fig. \ref{rotat}). The  action of the absorber  system   will be to make the rotation action smaller (on every frequency mode). This is the essence of the trick to reach optimality; we will use some of our (time) resources to obtain a rough estimate for the system and then perform a second experiment using the dual. The effect of this will be to zoom in on the `interval of uncertainty'; that is,  rather than considering the whole of the parameter space the dual  enables us to restrict  our analysis to the smaller unknown region of it.

The above is all a bit hand wavy, so let's be more precise. Let us consider one frequency in the input-output map. 
The QFI per unit time of a particular frequency was given by \eqref{tease} (recall that all frequencies are independent so the QFI is  additive across all frequencies). 
The question of the optimal  measurement on a particular frequency mode is essentially a Gaussian estimation problem and has been solved  
\cite{Monras1, Monras2}.
The optimal measurement is  homodyne (and realises the scaling in the QFI), however the direction to measure depends on the true value \cite{Monras1, Monras2}.
This suggests that we need an adaptive strategy in order to attain optimal precision in the limit of large time. That is, using information from previous steps in order to attain bounds. 
Returning to our problem across all frequencies, this issue  is further complicated by the fact that the optimal quadrature to measure may be different for each  frequency.

The strategy that we take is as follows:
\begin{enumerate}
\item Use the estimator from Sec. \ref{SIDF} to perform a preliminary estimate of $\theta$, given by $\hat{\theta}_0$, using time resource $T_{pre}$ that has  MSE $\mathcal{O}\left(\tau T_{pre}\right)^{-1}$. 
%That is, we are in the regime where the spectral gap is small.
\item Use a dual system to `negate' the original system with parameter $\hat{\theta}_0$. The upshot of this is that the input-output map has the action of a small rotation  (see Fig. \ref{rotat}). The rotation on frequency $\omega$, denoted by $\lambda(\omega)$, will depend on $\theta$ in the following way
$$\lambda(\omega)\sim
 \begin{cases}
\theta \tau    & |\pm\omega-\omega_0|\sim \frac{1}{\tau} \\
\theta &\text{otherwise  }\\
    \end{cases}$$
for some constant $\omega_0$ (see Fig. \ref{picf}).

\item Suppose that  $\phi$ is the direction of squeezing of the input, then perform a measurement 
of the quadrature with phase $\phi+\frac{\pi}{2}+\phi(N)$. The function $\phi(N)$ is just a constant function dependent on the level of squeezing in the input; we do not state it here (see \cite{Monras2}). For example, if the input is momentum squeezed then one should measure the quadrature with phase $\phi(N)$.
\item  Choose the  estimator of $\theta$ as the maximum likelihood estimator.
%
%$\hat{\theta}=\int_{\omega: |\omega-\theta|\leq\tau^{-1}} \hat{\theta(\omega)\mathrm{d}\omega,$$   where $\hat{\theta} is the maximum likelPerform maximum likelihood estimation on each frequency choose the following
%\item Each frequency mode consists of a frequency dependent rotation, which we denote by $\tilde{\theta}(\omega)$. Using step 1, we can find an estimate of $\tilde{\theta}(\omega)$, given by $\tilde{\theta}_0(\omega)$, with  MSE $\mathcal{O}(T)$. Observe that   that $\tilde{\theta}_0(\omega)$ is related to the unknown parameter, $\theta_0$ (i.e the estimate of the true parameter $\theta$) in the following way:
%$$\theta_0(\omega)\sim
% \begin{cases}
%\theta_0 \tau    & |\pm\omega-\theta|\sim \frac{1}{\tau} \\
%\theta_0 &\text{otherwise  }\\
%    \end{cases}.$$
%(see Sec. \ref{poke})
  %  [[[Could comment that is  another adaptive procedure on top of the other earlier]]]
   % [[[Comment that could also use another measurement so long as get this precision]]]
%\item Use a dual system to `negate' the original system with parameter $\theta_0$. The upshot of this is that the input-output map has the action of a small rotation (proportional to $T^{\epsilon}$ on each frequency (see Fig. \ref{rotat})). If $\phi$ is the direction of squeezing of the input, then measure the quadrature with phase $\phi=\frac{\pi}{2}+\phi(N)$. The function $\phi(N)$ is just a constant function dependent on the level of squeezing in the input; we do not state it here (see \cite{Monras2}). For example, if the input is momentum squeezed then one should measure the quadrature with phase $\phi(N)$. 
\end{enumerate}
If the total time of the experiment, $T_{\mathrm{tot}}$, is chosen so that 
 $T_{pre}=T^{\epsilon}_{\mathrm{tot}}$ for $\epsilon>\frac{1}{2}$ then by 
 \cite{Monras1}  we obtain optimal precision.   % [[[Could comment that info obtained as a rate according to earlier formula...i.e T like Monras1 N]]]
The key point is that
  working  time domain,  the optimal scaling may be achieved (in the small gap approximation) by measuring sequentially in time, i.e., performing a  homodyne measurement and using the estimator \eqref{cof2}, where this time $\mathbf{X}(t)$ is the quadrature from (3).

\begin{remark}
    In \cite{Monras1, Monras2} (i.e. the Gaussian estimation problem) they didn't use a  dual system. Instead they negated the known part of the system in the measurement by adding an extra phase to the measurement and measuring $\phi+\hat{\theta}_0$. However, in the frequency domain this would require measuring a different quadrature for each frequency. The power of using the dual is in that allows one to simplify the final  measurement so that the same quadrature can be measured across all frequencies. 
    \end{remark}

\section{Further Applications and Outlook for Quantum Absorbers}\label{appos}

In this chapter we have shown that for a given QLS there always exists a stable coherent quantum absorber. We discussed how quantum absorbers can be advantageous  for quantum estimation for the estimation problem from Ch. \ref{FEDERER}. Let us now discuss other potential applications. 

We have seen in Sec. \ref{purple6} that the dual QLS is unique up to symplectic equivalence classes (see Theorem \ref{symplecticequivalence}). 
 Hence the transfer function belonging to this family of duals  is  unique; we call this the \textit{dual transfer function}.  The  dual system is in some sense capturing the field output.  This poses two interesting open questions.

Firstly,  can the dual solution be used to find an alternative proof to the main identifiability result in Ch. \ref{powers34} (as in Theorem \ref{main})? In fact this was our main reason for studying quantum absorbers. Since the mapping between the transfer function and dual transfer function is one-to-one then it remains to show that the mapping between the dual transfer function and field output is one-to-one. However, this problem turns out to be just as difficult as the original problem problem in Ch. \ref{powers34}, as it still requires solving  a spectral factorisation-type problem (this time with dual systems).  This is because if 
$$\Xi_D(s)\Xi(s)V_{\mathrm{vac}}\Xi(-\overline{s})^{\dag}\Xi_D(-\overline{s})^{\dag}={V}_{\mathrm{vac}}$$
then $$\Psi(s)=\Xi_D(-\overline{s})^{\flat}{V}_{\mathrm{vac}}\left(\Xi_D({s})^{\dag}\right)^{\flat},$$
where $\Xi_{D}(s)$ is the dual transfer function.

   Now the dual system capturing the output of the QLS also means that it  characterises the field output.  That is, rather than the system being correlated with the field it becomes correlated with the dual system.
This is a considerable simplification from working with an infinite number of modes in the field to a finite number in the dual. Perhaps this simplification isn't  so surprising after all following the \textit{Gaussian Schmidt Decomposition} in  \cite{Botero1} (see Theorem \ref{WOLFF}).
It would be interesting to investigate how precise we can make the above; for instance, is it the case that    there is an isomorphism between the modes in the dual and the important modes in the output, i.e. those that are correlated with the internal system at stationarity. 
Further, can notions such as entanglement (or more general geometrical distances) or measurements  be represented in an advantageous (due to their simplicity) way on the dual. If so, this could lead to  potential applications in quantum control and metrology.

Quantum absorbers may also offer interesting applications to  quantum communication,  as  they could enable sharing of entanglement at a distance \cite{Voll1}.

%Another application is to quantum control. Consider switching the positions of the system and the dual. One can invisage a feedback protocol resulting from this cascaded system, where the controller (dual) can affect the system without being affected by the latter .  That is, the dynamics of the system ....[[[talk about quantum observers...sqitch order again?]][QUANTUM OBSERVERS]

Another application is to quantum control.   In control theory a \textit{quantum observer} is a purely quantum  system that is capable of mimicking the behaviour of the system. The three characteristic properties  of quantum observers are the following \cite{Amini1}:
\begin{itemize}
\item The observables of the observer should
 converge to the system variables (in some limit). 
 \item The observer should be a physically realisable
 \item The plant should feedforward to the observer, but not the other way around. That is, there should be  no back-action of the observer on the plant. 
 \end{itemize}
 Using observers in control problems is an example of a coherent feedback scheme. The advantages of coherent feedback schemes are that they make use of non-commutative quantum signals  so that there is no loss of information or coherence, which is not the as  case in measurement-based feedback schemes. Despite recent results indicating that coherent schemes typically perform better than measurement-based schemes for some control objectives \cite{Ham1, Petersen4}, it is still in its incipient stage and is a very active research area.
 Typically the goal of a quantum observer is to map the mean values of the system. In classical control theory, the optimal way to do this is the Kalman filter. In \cite{Amini1} they compare a quantum observer based scheme with  a measurement-based observer using the classical Kalman filter. The  coherent scheme outperforms the classical one in  terms of the estimation error. However, the quantum observer there is by no means optimal. The main difficulty with designing quantum observers is the physical relizability requirement. 
 Now our quantum absorbers in Sec. \ref{purple6} make natural quantum observers for the second order moments. In the asymptotic limit they map the covariance matrix of the system exactly. Our observers/absorbers satisfy all the desirable properties of observers and are particular interesting as they provide a purification of the system's stationary state. Therefore, one can  imagine  these being  useful in certain control applications. However, one potential downside of using absorbers as observers  is that the system will be maximally entangled with the observer and therefore share no entanglement with  the field.

\chapter{Conclusion}

In this thesis we have discussed various aspects of system identification split into two contrasting approaches: (1) Time-dependent input (or transfer function) identifiability and (2) stationary inputs (or power spectrum) identifiability. 

In Ch. \ref{T.F} we considered the time-dependent approach, where equivalent systems were characterised by the property that they have the same transfer function. We answered the first two questions set out in the abstract for these inputs. We have also addressed the same problems in the noisy scenario (recall that noise is modelled by additional channels that we cannot access). Finally, we studied  \textit{noise unobservable subspaces}, where part of the system is shielded from the noise, and found that such systems can be identified as in the noiseless case. Our noise unobservable subspaces work would make a great starting point for future research. 

In Ch. \ref{powers34} we considered the stationary approach. The characteristic quantity is the power spectrum in this case, which itself depends quadratically on the transfer function. Our main result was that under global minimality the power spectrum implies the transfer function uniquely, which is contrary to the analogous classical problem \cite{Youla1, Hayden1, Zhou1, Kalman1}. We also developed  identification methods, that is, how to construct a system realisation of the transfer function. We considered several extensions to our problem, including mixed inputs and the use of  ancillary channels.  One possible direction to extend this work would be to noisy inputs (i.e the analogous problem to the time-dependent input one discussed in Sec. \ref{noisek}).

The next main problem that we considered was  parameter estimation for PQLSs in the time-dependent approach (Ch. \ref{QEEP}). 
The emphasis was on finding a realistic scheme that can be implemented with current technology; our scheme is similar to the one used at LIGO to implement gravitational waves \cite{LIGO1, LIGO2}. We considered both the single and multiple parameter scenario. 
An optimisation problem remains in the multiple parameter passive case; although we have proved that Heisenberg scaling is possible. 
We also extended a multiple quantum metrology result from \cite{Humphreys1} to our domain; that there is a linear advantage with respect to the number of parameters in terms of estimation accuracy by using entangled states between the independent channels. However, we also expressed our concerns in that their result (and indeed ours) is only valid under the assumption of  a fixed total photon number. We discussed in Ch. \ref{QEEP}  the main avenues to extend these results, i.e.,   to active systems or noisy inputs.

In the previous problem energy was the resource constraint. We also considered an estimation problem in Ch. \ref{FEDERER} under a time constraint. We saw that when the system becomes dynamically unstable the scaling becomes quadratic (rather than linear) in the observation time. Ch. \ref{FEDERER} was devoted to investigating and exploiting this phenomenon in both the time and stationary approaches. 

%[[[DIscuss outlook of chapter 9 after completed there]]]

Finally, we developed the notion of a  quantum absorber for QLSs in Ch. \ref{DUAL}. A quantum absorber for a given QLS is another QLS such that the output of the resultant system (obtained by combining the original with the absorber) is equal to the input. They are particularly useful 
because they provide a natural purification for the stationary state of the system.
We showed that for any QLS, such an absorber exists. There are potentially many applications for absorbers, such as   quantum control, quantum communication and quantum estimation \cite{Amini1, Ham1}. We discussed one application of the latter in detail in Sec. \ref{blog}.

%Note that this result holds even for systems with non-trivial squeezing/ scattering matrix.

%----------------------------------------------------------------------------------------
%	THESIS CONTENT - APPENDICES
%----------------------------------------------------------------------------------------

\appendix % Cue to tell LaTeX that the following "chapters" are Appendices

%%%%%%%%%%%%%%%%%%%%%%%%%%%%%%%%
\chapter{Finding a Minimal Classical Realization}\label{APP7}
%%%%%%%%%%%%%%%%%%%%%%%%%%%%%%%%

In this appendix a set of (nonphysical) minimal and doubled-up matrices $(A_0, B_0, C_0)$ are found that realizes the transfer function \eqref{TF1}, which describes a (minimal) physical system $(A, C)$. 

We assume that the matrix  $A$ for   the $n$-mode minimal system, $(A, C)$, possesses  $2n$ distinct eigenvalues each with a nonzero imaginary part. This requirement  can be seen to be generic in the space of all quantum systems \cite{Nurdin3}. Moreover, it can also be shown that if  $\lambda_i$ is a complex eigenvalue of $A$ with right eigenvector $\left(\begin{smallmatrix}R_i\\S_i\end{smallmatrix}\right)$ and left eigenvector  $\left(U_i, V_i\right)$, then $\overline{\lambda}_i$ is also an eigenvalue with right eigenvector $\left(\begin{smallmatrix}{S}^{\#}_i\\{R}^{\#}_i\end{smallmatrix}\right)=\Sigma{\left(\begin{smallmatrix}R^{\#}_i\\S^{\#}_i\end{smallmatrix}\right)}$ and left eigenvector $   \left({V}^{\#}_i, {U}^{\#}_i\right)= {\left(U^{\#}_i, V^{\#}_i\right)}\Sigma_n$, where     $R_i, S_i\in\mathbb{C}^{1\times n}$, $U_i, V_i\in\mathbb{C}^{n\times1}$ and $\Sigma_n:=\left(\begin{smallmatrix} 0_n&1_n\\1_n&0_n\end{smallmatrix}\right)$. That is, for each eigenvalue and eigenvector, there exists a corresponding mirror pair. 
This property follows from the fact that $A$ has the doubled-up form  $A:=\Delta\left(A_{-}, A_{+}\right)$.

We now construct a minimal realization called Gilbert's realization \cite{Zhou1}. 
The only thing that we need to take care of is that the realization we obtain is of the doubled-up form. 

As the transfer function may be written as 
\[\Xi(s)=\frac{N(s)}{\prod_{i=1}^n (s-\lambda_i)(s+\lambda_i)}.\]
we can perform a partial fraction expansion, so that 
\[\Xi(s)={1}+\sum_{i=1}^n \frac{P_i}{(s-\lambda_i)} +\frac{Q_i}{\left(s-{\lambda}^{\#}_i\right)}.\]
As we show below, the matrices $P_i, Q_i$ are rank 1. Therefore there exist matrices $B_i\in\mathbb{C}^{1\times 2}$, $B'_i\in\mathbb{C}^{1\times 2}$, $C_i\in\mathbb{C}^{2\times 1}$, and $C'_i\in\mathbb{C}^{2\times 1}$ such that 
\begin{equation*}
C_iB_i=P_i \,\,\,\mathrm{and}\,\,\, C'_iB'_i=Q_i.
\end{equation*}
The Gilbert realization $A_0, B_0, C_0$ is 
\begin{equation*}
A_0:=\mathrm{Diag}\left(\lambda_1, \hdots, \lambda_n,\overline{\lambda}_1,\hdots, \overline{\lambda}_n\right),
\end{equation*}
\begin{equation*}
B_0:= \left[
  \begin{array}{c}
    %\horzbar & 
    B_1 
    %& \horzbar 
    \\
     \vdots              \\
    %\horzbar & 
    B_n 
    %& \horzbar 
    \\
     %\horzbar & 
     B'_1 
     %& \horzbar 
     \\
             %&
              \vdots   
              % &         
               \\
    %\horzbar & 
    B'_n 
    %& \horzbar
  \end{array}
\right]
\end{equation*}
and
\begin{equation*}
C_0:= 
\left[
  \begin{array}{cccccc}
  %  \vertbar &         & \vertbar &\vertbar&&\vertbar \\
    C_{1}    &      \ldots & C_{n} &C'_1&\ldots&C'_n   \\
  %  \vertbar  &        & \vertbar &\vertbar&&\vertbar
  \end{array}
\right].
\end{equation*}
From the expression of the physical transfer function we have
\begin{equation*}
C\left(s-A\right)^{-1}C^{\flat}=
\sum^n_{i=1} \frac{W_i}{s-\lambda_i}+
\frac{\Sigma{W}^{\#}_i \Sigma}{s-{\lambda}^{\#}_i }
\end{equation*}
where $W_i$ are the rank-one matrices
\begin{equation*}
W_i=\left(\begin{smallmatrix}C_-R_i+C_+S_i\\{C}^{\#}_+R_i+{C}^{\#}_iS_i\end{smallmatrix}\right) \left(\begin{smallmatrix}U_iC_-^{\dag}-V_iC_+^{\dag}&U_iC_+^T+V_iC^T_-\end{smallmatrix}\right).
\end{equation*}
Having fixed $B_i$ and $C_i$ the matrices $B'_i$ and $C'_i$ can then be chosen as 
\begin{equation}
B'_i={B}^{\#}_i\Sigma_2 \,\,\, \mathrm{and}\,\,\,
C'_i=\Sigma_2{C}^{\#}_i
\end{equation}
and so the matrices 
$(A_0, B_0, C_0)$
are of the doubled-up type.

Note that using Gilbert's realization on MIMO systems can also be seen to give a minimal doubled-up realization, but we do not discuss this any further here.

\chapter{Proof of Lemma \ref{hud}}\label{goh}

 \begin{proof}
Firstly define $\{e_1, ..., e_{4n}\}$ as the canonical basis of $\mathbb{C}^{4n}$. 
%
%Write, $\tilde{A}$ and $\tilde{A}'$ as 
%
%$$\tilde{A}=\left(\begin{smallmatrix}\tilde{A}_1&0\\ \tilde{A}_2&\tilde{A}_3\end{smallmatrix}\right)
%\quad\mathrm{and}\quad
%$\tilde{A}'=\left(\begin{smallmatrix}\tilde{A}'_1&0\\ \tilde{A}'_2&\tilde{A}'_3\end{smallmatrix}\right).
%$$
By property \eqref{pil1} of  proper LBT matrices it is clear that $y^{(i)}:= \left(\begin{smallmatrix}0\\y_2^{(i)}\end{smallmatrix}\right)\in\mathrm{Span}\{e_{2n+1}, ..., e_{4n}\}$. Further, as there are $2n$ of them they must form a basis of $\mathrm{Span}\{e_{2n+1}, ..., e_{4n}\}$.
Suppose $y^{(i)}$ has generalised eigenvector rank $m_i$, then as
 as $\tilde{A}'=T\tilde{A}T^{-1}$ we have 
%
%Additionally as $\tilde{A}'=T\tilde{A}T^{-1}$ we have 
%
\begin{align*}
\left(\tilde{A}'-\lambda^{(i)}\right)^{m_i}  Ty^{(i)}&=\left(T\tilde{A}T^{-1}-\lambda^{(i)}\right)^{m_i}  Ty^{(i)}\\
&=T\left(\tilde{A}-\lambda^{(i)}\right)^{m_i}  y^{(i)}\\
&=0.
\end{align*}
Therefore, $Ty^{(i)}$ are generalised eigenvectors of $\tilde{A}'$ associated to $\lambda^{(i)}$. Hence, because $\tilde{A}'$ is also assumed to be proper LBT, it follows that $\mathrm{Span}\{Ty^{(i)}\}\subset\mathrm{Span}\{e_{2n+1}, ..., e_{4n}\}$. Finally, 
\begin{align*}
T\mathrm{Span}\{e_{2n+1}, ..., e_{4n}\}&=T\mathrm{Span}\{y^{(i)}\}\\
&=\mathrm{Span}\{Ty^{(i)}\}\\
&\subset\mathrm{Span}\{e_{2n+1}, ..., e_{4n}\}.
\end{align*}
The invertibility of $T$ has been used in getting from the first to the second line. 
This implies that $T$ is LBT, as required. 
\end{proof}

%%%%%%%%%%%%%%%%%%%%%%%%%%%%%%%%%%%%%%%%%%%%%%%%%%%%%%%%%
\chapter{Showing  Existence of a Minimal Physical System with Transfer Function \eqref{huck2}}\label{APP2}
%%%%%%%%%%%%%%%%%%%%%%%%%%%%%%%%%%%%%%%%%%%%%%%%%%%%%%%%%
Firstly, since we know that the system described by $\Xi(s)$ is physical, then the result of connecting it in series to another physical quantum system will be physical. To this end, consider the system 
$$\tilde{\mathcal{G}}=\mathcal{G}\triangleleft \mathcal{G}_n\triangleleft...\triangleleft\mathcal{G}_1,$$
where $G$ was our original system and $G_i$ is a single mode (unstable) active system with coupling $c_-=0$, $c_+=\sqrt{2\mathrm{Re}\mu_i}$, and Hamiltonian $\Omega_-=\mathrm{Im}\lambda_i$, $\Omega_+=0$. The system $\tilde{\mathcal{G}}$ is physical and is described by the transfer function $\Xi^{(2)}(s)$ (Eq. \eqref{huck2}).
Also $\tilde{\mathcal{G}}$
  must be stable because the transfer function $\Xi^{(2)}(s)$ has poles in the left-complex plane only because 
     $\Xi^{}(s)$ and $\Xi^{(1)}(s)$ do.
However, it is not minimal. 

To find a minimal system employ the quantum Kalman decomposition from \cite{Zhang4}. The result is that this system may be written in the form of eq. (103) and (104) in \cite{Zhang4}. Hence the system is TFE to the \textit{minimal} system  with matrices  (in quadrature form) $\left(\tilde{A}_{co}, B_{co}, C_{co}\right)$ from \cite{Zhang4}. This system gives a minimal realisation of the transfer function $\Xi^{(2)}(s)$.
 It  can also can be verified that it is physical (this either  follows  because its transfer function is doubled-up and symplectic  \cite{Peter1} or alternatively from the results in \cite{Zhang4}) and that the matrices $\left(\tilde{A}_{co}, B_{co}, C_{co}\right)$ are of doubled-up type, as required.
 
 Finally, since  two stable and minimal quantum systems connected in series is always minimal (see a proof of this below), then it is clear that $\Xi^{(2)}(s)$ must necessarily be of size  $n-k$. To see the previous claim,  suppose that we have two minimal systems $(C_1, A_1)$ and $(C_2, A_2)$, where $C_i$ is the coupling matrix of the system and $A_1$ is the usual system matrix. Connecting these systems in series ($(C_1, A_1)$ into $(C_2, A_2)$) we get the resultant coupling and system matrices:
\begin{equation}\label{tenth}
(C, A):=\left(\left(\begin{smallmatrix} C_1&C_2\end{smallmatrix}\right) ,\left(\begin{smallmatrix} A_1&0\\-C_2^{\flat} C_1&A_2\end{smallmatrix}\right)\right).
\end{equation}

Recall that in order to show that the QLS (C,A) is minimal it is enough to show that the pair $(A, -C^\flat)$ is controllable \cite{Gough2}.  Controllability of    $(A, -C^\flat)$    is equivalent to the statement: for all eigenvalues and left-eigenvectors of 
$A$, i.e. $vA=v\lambda $ then $vC^\flat \neq0$. 
Letting $v=\left(y_1, y_2 \right)$ where $y_i$ are both of size $2n$, then there are two cases to consider; either $y_2=0$ or $y_2\neq0$.
\begin{itemize}
\item If $y_2=0$ then Eq. \eqref{tenth} implies that $y_1A_1=\lambda y_1$. Therefore, by controllability of  the first system we must have $vC^{\flat}=y_1C_1^{\flat}\neq 0$ and so the system is controllable.
\item If $y_2\neq0$ then 
$\left(y_1, y_2 \right)A=\left(y_1, y_2 \right)\lambda $ implies that $y_2A_2=y_2\lambda$. Note that by stability $\mathrm{Re}(\lambda)<0$. Hence by controllability of the second system $y_2C_2^\flat\neq0$. Suppose to the contrary that $(A, -C^\flat)$ is not controllable. Then $y_1C^\flat_1+y_2C^\flat_2=0$, which together with $\left(y_1, y_2 \right) A=\left(y_1, y_2 \right)\lambda $ would imply that 
\begin{equation}\label{apeq}
y_1\left(A_1+C^{\flat}_1C_1\right)=y_1\lambda.
\end{equation}
Or equivalently 
\begin{equation}\label{apeq1}
y_1A_1^{\flat}=-y_1\lambda.
\end{equation}
This equation implies that $\mathrm{Re}(\lambda)>0$, which is a contradiction. 
\end{itemize}

\chapter{Proof of Theorem \ref{main}}\label{p1}

As outlined in the proof sketch,  we need to show (1)-(3).% We go through each of these in turn.

\section{Step (1)}\label{s1}
Firstly, the condition $\tilde{B}'=T\tilde{B}$ is equivalent to  
$$\left(\begin{smallmatrix}-C'^{\dag}V_{\mathrm{vac}}\\{C'}^{\dag}\end{smallmatrix}\right)=K\Sigma T\Sigma K \left(\begin{smallmatrix}-{C}^{\dag}V_{\mathrm{vac}}\\{C}^{\dag}\end{smallmatrix}\right),$$
where
$$K=\left(\begin{smallmatrix} J&0\\0&-J\end{smallmatrix}\right) \quad \mathrm{and} \quad \Sigma=\left(\begin{smallmatrix} 0&1\\1&0\end{smallmatrix}\right).$$ 
Hence 
$$\left(\begin{smallmatrix}-V_{\mathrm{vac}}\tilde{C'}&\tilde{C'}\end{smallmatrix}\right)=\left(\begin{smallmatrix}-V_{\mathrm{vac}}\tilde{C}&\tilde{C}\end{smallmatrix}\right)K\Sigma T^{\dag}\Sigma K.$$ 
Therefore, combining this the condition $\tilde{C}'=\tilde{C}T^{-1}$ we have
\begin{equation}\label{in10}
\tilde{C}=\tilde{C}K\Sigma T^{\dag}\Sigma KT.
\end{equation}
Now, 
\begin{align}
\nonumber \tilde{A}'&=-K\Sigma\tilde{A}'^{\dag}\Sigma K\\
\nonumber&=-K\Sigma \left(T^{\dag}\right)^{-1} \tilde{A}^{\dag}T^{\dag}\Sigma K\\
\label{as}&=K\Sigma  \left(T^{\dag}\right)^{-1} \Sigma K \tilde{A}K\Sigma T^{\dag}\Sigma K,
\end{align}
where $\tilde{A}'=T\tilde{A}T^{-1}$ has been used to obtain the second line.

% $\tilde{A}$ (or $\tilde{A}'$) satisfies the relation
%$$\tilde{A}=-K\Sigma\tilde{A}^{\dag}\Sigma K,$$ we can  use \eqref{c1} to obtain
%
%\begin{align}
%\nonumber \tilde{A}'&=-K\Sigma\tilde{A}'^{\dag}\Sigma K\\
%\nonumber&=-K\Sigma \left(T^{\dag}\right)^{-1} \tilde{A}^{\dag}T^{\dag}\Sigma K\\
%\label{as}&=K\Sigma  \left(T^{\dag}\right)^{-1} \Sigma K \tilde{A}K\Sigma T^{\dag}\Sigma K.
%\end{align}
%
And so 
\begin{align*}
\tilde{C} \tilde{A}T^{-1}&=\tilde{C}' \tilde{A}'\\
&=\tilde{C}'K\Sigma  \left(T^{\dag}\right)^{-1} \Sigma K \tilde{A}K\Sigma T^{\dag}\Sigma K\\
&=\left(\tilde{C}T^{-1}K\Sigma  \left(T^{\dag}\right)^{-1} \Sigma K\right) \tilde{A}K\Sigma T^{\dag}\Sigma K\\
&=\tilde{C} \tilde{A}K\Sigma T^{\dag}\Sigma K,
\end{align*}
where  Eq. \eqref{in10} has been used to obtain the fourth line.
Thus
\begin{equation}\label{in2}
\tilde{C} \tilde{A}=\left(\tilde{C} \tilde{A}\right)K\Sigma T^{\dag}\Sigma KT.
\end{equation}

\begin{Claim}
%The following expression holds:
%
\begin{equation}
\tilde{C} \tilde{A}^k=\left(\tilde{C} \tilde{A}^k\right)K\Sigma T^{\dag}\Sigma KT.
\end{equation}
for all $k\geq0$. 
\end{Claim}

\begin{proof}
We  prove this by induction. Note that we already know it to be true for $k=0$ and $k=1$ (see Eq. \eqref{in10} and \eqref{in2}). To this end, suppose that it is true for $k-1$. Therefore,
\begin{align*}
\tilde{C}'\tilde{A}'^k&=\tilde{C}'\left(\tilde{A}'\right)^{k-1}\tilde{A}'\\
&=\tilde{C}'\left(\tilde{A}'\right)^{k-1}K\Sigma  \left(T^{\dag}\right)^{-1} \Sigma K \tilde{A}K\Sigma T^{\dag}\Sigma K\\
&=\left(\tilde{C}\tilde{A}^{k-1}T^{-1}K\Sigma  \left(T^{\dag}\right)^{-1} \Sigma K \right)\tilde{A}K\Sigma T^{\dag}\Sigma K\\
&=\tilde{C}\tilde{A}^{k}K\Sigma T^{\dag}\Sigma K.
\end{align*}
by using Eqs. \eqref{as} and \eqref{c1}. Finally, using the observation $\tilde{C}'\tilde{A}'^k=\tilde{C}\tilde{A}^kT^{-1}$ completes the proof. %
 % and \eqref{c3} on the third and the inductive hypothesis on the fourth. Finally, using \eqref{c1} and \eqref{c3} we arrive at the desired result. 
\end{proof}
Finally, following this claim we  have:
\begin{align*}
\mathcal{O}&=\mathcal{O}K\Sigma T^{\dag}\Sigma KT
\\&=\mathcal{O}\left(\begin{smallmatrix} T_4^{\flat} &0\\ -T_3^{\flat}& T_1^{\flat}\end{smallmatrix}\right)T.
\end{align*}
%where $\mathcal{O}$ is the observability matrix. 

%Or equivalently, writing $T$ in the above  block form 
 %we have
 %
%$$\mathcal{O}=\mathcal{O}\left(\begin{smallmatrix} T_4^{\flat} &0\\ -T_3^{\flat}& T_1^{\flat}\end{smallmatrix}\right)T.$$ 
%

\section{Step (2)}
For this step it is sufficient to prove the following claim. 

\begin{Claim}
$$\tilde{C}\tilde{A}^k\left(\begin{smallmatrix}T_4^{\flat}-T_1^{\flat}\\-T_3^{\flat}\end{smallmatrix}\right)=0$$
for all $k=0,1,2,...$.
\end{Claim}

\begin{proof}
Using the results of Appendix \ref{s1} we know that equivalent systems   are related via
\begin{equation}\label{start}
\tilde{C'}\tilde{A'}^k=\tilde{C}\tilde{A}^k\left(\begin{smallmatrix} T_4^{\flat}&0\\-T^{\flat}_3&T_1^{\flat}\end{smallmatrix}\right).
\end{equation} 
Also     note that the condition $C'A'^k=CA^kT^b_1$ holds.
% (by using the de$\tilde{A}$ and $\tilde{C}$ in \eqref{cask}).

We first see this result for $k=0$. Eq. \eqref{start} for $k=0$ reads 
$$\left(\begin{smallmatrix}-V_{\mathrm{vac}}C', C'\end{smallmatrix}\right)=\left(\begin{smallmatrix}-V_{\mathrm{vac}}CT_4^{\flat}-CT_3^{\flat}, CT_1^{\flat}\end{smallmatrix}\right).$$ 
Therefore, adding the first entry to $V_{\mathrm{vac}}$ times the second entry:
$$0=-V_{\mathrm{vac}}C\left(T_4^{\flat}-T_1^{\flat}\right)+C\left(-T_3^{\flat}\right),$$
which shows the result for $k=0$.

The result for $k\in\mathbb{N}$ goes along the same lines, but is a little more involved. Firstly, observe that $\tilde{A}^{k}$ may be written as
$$\tilde{A}^k=\left(\begin{smallmatrix}  \left(-A^{\flat}\right)^k &0\\ e_k &A^k\end{smallmatrix}\right),$$
where $e_k=A^0C^{\flat}V_{\mathrm{vac}}C\left(-A^{\flat}\right)^{k-1}+...+A^{k-1}C^{\flat}V_{\mathrm{vac}}C\left(-A^{\flat}\right)^{0}$ (and similarly for the primed matrices).
Now, from Eq.  \eqref{start} we have 
\begin{align}
\nonumber&\left(\begin{smallmatrix} -V_{\mathrm{vac}}C'\left(-A'^{\flat}\right)^k+C'A'^{k-1}C'^{\flat}V_{\mathrm{vac}}C'-C'e'_{k-1}A'^{\flat},&C'A'^k\end{smallmatrix}\right)\\
&\label{combo}=\left(-\begin{smallmatrix} V_{\mathrm{vac}}C\left(-A^{\flat}\right)^kT_4^b+CA^{k-1}C^{\flat}V_{\mathrm{vac}}CT_4^{\flat}-Ce_{k-1}A^{\flat}T_4^{\flat}-CA^kT_3^{\flat},&CA^kT_1^{\flat}\end{smallmatrix}\right). 
\end{align}
Again adding the first block to $V_{\mathrm{vac}}$ times the second block gives 
\begin{equation}\label{pod}
H'=\tilde{C}\tilde{A}^k\left(\begin{smallmatrix} T_4^{\flat}-T_1^{\flat}\\
-T^{\flat}_3 \end{smallmatrix}\right)+HT_1^b,
\end{equation}
where
\begin{align*}
H:&=-V_{\mathrm{vac}}C\left(-A^{\flat}\right)^k+CA^{k-1}C^{\flat}V_{\mathrm{vac}}C-Ce_{k-1}A^{\flat}+V_{\mathrm{vac}}CA^k\\
H':&=-V_{\mathrm{vac}}C'\left(-A'^{\flat}\right)^k+C'A'^{k-1}C'^{\flat}V_{\mathrm{vac}}C'-C'e'_{k-1}A'^{\flat}+V_{\mathrm{vac}}C'A'^k
\end{align*}

Now, observe that
\begin{align}
\nonumber H'&=V_{\mathrm{vac}}C'A'^k+\left(V_{\mathrm{vac}}C'\left(-A'^{\flat}\right)^{k-1}-C'e'_{k-1}\right)A'^{\flat}+C'A'^{k-1}C'^{\flat}V_{\mathrm{vac}}C'\\
\nonumber&=V_{\mathrm{vac}}C'A'^k+\left(V_{\mathrm{vac}}C'\left(-A'^{\flat}\right)^{k-1}-C'e'_{k-1}\right)\left(-A'-C'^{\flat}C'\right)\\
\nonumber&+C'A'^{k-1}C'^{\flat}V_{\mathrm{vac}}C'\\
\nonumber&=V_{\mathrm{vac}}C'A'^k+\left(-V_{\mathrm{vac}}C'\left(-A'^{\flat}\right)^{k-1}+C'e'_{k-1}   \right)A'
\\\nonumber& +  \left(-V_{\mathrm{vac}}C'\left(-A'^{\flat}\right)^{k-1}C'^{\flat}+C'e'_{k-1}C'^{\flat}+ C'A'^{k-1}C'^{\flat}V_{\mathrm{vac}}    \right)C'    \\
&\label{roff}=V_{\mathrm{vac}}C'A'^k+G'_{k-1}A'-\tilde{C}'\left(\tilde{A}'\right)^{k-1}\tilde{B}'C',
\end{align}
where $G'_k:=-V_{\mathrm{vac}}C'\left(-A'^{\flat}\right)^{k}+C'e'_{k}$. Here we have used the realisability condition $A+A^{\flat}+C^{\flat}C=0$ on the second line and then rearranged. 

Now, let us obtain a recursive expression for $G_k$. Firstly, using the definition of $e_k$ and the substitution $A'+A'^{\flat}+C'^{\flat}C'=0$:

\begin{align*} 
G'_k&=-V_{\mathrm{vac}}C'\left(-A'^{\flat}\right)^{k}   +    \sum^{k-1}_{j=0} C'A'^{k-1-j}C'^{\flat}V_{\mathrm{vac}}C'\left(-A'^{\flat}\right)^{j} 
%+ \left(-1\right)^{k-2}C'A'^1C'^{\flat}XC'\left(A'^{\flat}\right)^{k-2}+
\\
&=-V_{\mathrm{vac}}C'\left(-A'^{\flat}\right)^{k-1}\left(A'+C'^{\flat}C'\right)
\\&+  \sum^{k-1}_{j=1}  C'A'^{k-1-j}C'^{\flat}V_{\mathrm{vac}}C'\left(-A'^{\flat}\right)^{j-1} \left(A'+C'^{\flat}C'\right)   \\
%&+ \left(-1\right)^{k-2}C'A'^1C'^{\flat}XC'\left(A'^{\flat}\right)^{k-3}\left(-A'-C'^{\flat}C'\right)+
%&+...+\left(-1\right)^{1}C'A'^{k-2}C'^{\flat}XC'\left(A'^{\flat}\right)^{0}\left(-A'-C'^{\flat}C'\right)\\
&            +
C'A'^{k-1}C'^{\flat}V_{\mathrm{vac}}C'\left(-A'^{\flat}\right)^{0}
\end{align*}
Rearranging this and using the definition of $e_k$ again we obtain 
\begin{align*}
G'_k&=\left(   -V_{\mathrm{vac}}C'\left(-A'^{\flat}\right)^{k-1}C'^{\flat}+ C'A'^{k-1}C'^{\flat}V_{\mathrm{vac}}\right.\\
&\left.+C'\left[ \sum^{k-2}_{j=0}   A'^{k-2-j}C'^{\flat}V_{\mathrm{vac}}C'\left(-A'^{\flat}\right)^{j}
%...+\left(-1\right)^0A'^{k-2}C'^{\flat}XC' \left(A'^{\flat}\right)^0 
   \right]C'^{\flat}      \right)C'\\
&+\left(   -V_{\mathrm{vac}}C'\left(-A'^{\flat}\right)^{k-1}  +C'\sum^{k-2}_{j=0}  A'^{k-2-j}C'^{\flat}V_{\mathrm{vac}}C'\left(-A'^{\flat}\right)^{j}\right)A'\\
&
=\left( -V_{\mathrm{vac}}C'\left(-A'^{\flat}\right)^{k-1}C'^{\flat}+ C'A'^{k-1}C'^{\flat}V_{\mathrm{vac}}+C'e'_{k-1}C'^{\flat}\right)C'\\
&+\left( -V_{\mathrm{vac}}C'\left(-A'^{\flat}\right)^{k-1} +C'e'_{k-1}\right)A'\\
&=-\tilde{C}'\tilde{A}'^{k-1}\tilde{B}'C'+G_{k-1}A'.
\end{align*}
%
%
%
%Here we have used the substitution $A'+A'^{\flat}+C'^{\flat}C'=0$ in the second equality and then rearranged to get the third equality. Finally, we have  used the definition of $e_k$ again and $G_k$  in the fourth and fifth equalities, respectively. 
Also note that 
\begin{align*}
G_1&=V_{\mathrm{vac}}C'A'^{\flat}+C'C'^{\flat}V_{\mathrm{vac}}C'\\
&=-V_{\mathrm{vac}}C'A'+\left(-V_{\mathrm{vac}}C'C'^{\flat}+C'C'^{\flat}V_{\mathrm{vac}}\right)C'\\
&=-V_{\mathrm{vac}}C'A'-\tilde{C}'\tilde{B}'C'.
\end{align*}

Using our recursive expression for $G_k$, and continuing on from Eq.  \eqref{roff} we have 
\begin{align}\label{nun}
\nonumber H'=&V_{\mathrm{vac}}C'A'^k-\tilde{C}'\tilde{A}'^{k-1}\tilde{B}'C'+G'_{k-1}A'\\
&=V_{\mathrm{vac}}C'A'^k-\tilde{C}'\tilde{A}'^{k-1}\tilde{B}'C'-\tilde{C}'\tilde{A}'^{k-2}\tilde{B}'C'A'+G_{k-2}\tilde{A}'^2\nonumber\\
&\,\,\vdots \quad \quad\quad \quad\quad   \vdots \quad\quad\quad\quad\quad \vdots\quad \quad\quad \quad\quad   \vdots\nonumber\\
&=V_{\mathrm{vac}}C'A'^k-\sum^{k-1}_{j=1}\tilde{C}'\tilde{A}'^{j}\tilde{B}'C'A'^{k-1-j}+G_1A'^{k-1}\nonumber\\
&=V_{\mathrm{vac}}C'A'^k-  \sum^{k-1}_{j=0}\tilde{C}'\tilde{A}'^{j}\tilde{B}'C'A'^{k-1-j}                -V_{\mathrm{vac}}C'A'^k\nonumber\\
&=-\sum^{k-1}_{j=0}\tilde{C}'\tilde{A}'^{j}\tilde{B}'C'A'^{k-1-j}         .
\end{align}
%
%
%
%That is, 
%
%\begin{align}\label{goy}
%&\left(-1\right)^{k-1}XC'\left(A'^{\flat}\right)^k+C'A'^{k-1}C'^{\flat}XC'-C'e'_{k-1}A'^{\flat}+XC'A'^k\\&=\tilde{C}'\tilde{A}'^{k-1}\tilde{B}'C'A'^0+...+\tilde{C}'\tilde{A}'^{0}\tilde{B}'C'A'^{k-1}.
%\end{align}
%
%It should be noted that \eqref{goy} also holds without the primes. 
Furthermore, as $\tilde{C}'\tilde{A}'^{k-1}\tilde{B}'=\tilde{C}\tilde{A}^{k-1}\tilde{B}$ for all $k$ and $C'A'^k=CA^kT^b_1$,  then we may conclude that 
\begin{align}\label{compare1}
H'=-  \left(   \sum^{k-1}_{j=0}\tilde{C}\tilde{A}^{j}\tilde{B}CA^{k-1-j}                    \right)T_1^{\flat}.
\end{align}

On the other hand, by using an identical argument to above, 
\begin{equation}\label{compare2}
H=-\sum^{k-1}_{j=0}\tilde{C}\tilde{A}^{j}\tilde{B}CA^{k-1-j}.      
\end{equation}  
Therefore, using Eqs \eqref{compare1} and \eqref{compare2} in \eqref{pod} completes the proof. 
\end{proof}

\section{Step (3)} To show that the system is doubled-up we use the observability of the quantum system. Observe that $C_1A_1^k$, $C_2A_2^k$ must be of the of this doubled up form for $k\in\{0,1,2,...\}$. Writing $C_1A_1^k$, $C_2A_2^k$ and $T_1$ as $\left(\begin{smallmatrix} P_{(k)} &Q_{(k)}\\{Q}^{\#}_{(k)}&{P}^{\#}_{(k)}\end{smallmatrix}\right)$, $\left(\begin{smallmatrix} P'_{(k)} &Q'_{(k)}\\{Q}^{'\#}_{(k)}&{P}^{'\#}_{(k)}\end{smallmatrix}\right)$ and $T_1=  \left(\begin{smallmatrix} S_1 &S_2\\S_3&S_4\end{smallmatrix}\right)$, and using the result, $C_1A_1^k= C_2A_2^k T_1^\flat $, 
it follows that 
\[P_{(k)}(S_1^{\dag}-S_4^T)+Q_{(k)}(S_3^T-S_2^{\dag})=0\]
\[{Q}^{\#}_{(k)}(S_1^{\dag}-S_4^T)+{P}^{\#}_{(k)}(S_3^T-S_2^{\dag})=0.\]
Hence 
\[\mathcal{O}\left[\begin{smallmatrix}    S^{\dag}_1-S_4^T\\  S_3^T-S_2^{\dag}\end{smallmatrix}\right]=0\] and by using the fact that $\mathcal{O}$ is full rank implies that 
\[T_1=\left(\begin{smallmatrix} S_1 &S_2\\{S}^{\#}_2&{S}^{\#}_1\end{smallmatrix}\right).\] 
%Hence the result is proven. 

\chapter{Finding a Classical Realisation of the Power Spectrum}\label{nog}

We assume that the matrix  $A$ for   the $n$-mode minimal system, $(A, C)$, possesses  $2n$ distinct eigenvalues each with non-zero imaginary part. This requirement  cis generic in the space of all quantum systems \cite{Nurdin3}. 

Firstly, observe that if  $\lambda_i$ is a complex eigenvalue of $A$ with right eigenvector $\left(\begin{smallmatrix}R_i\\S_i\end{smallmatrix}\right)$ and left eigenvector  $\left(U_i, V_i\right)$, then $\overline{\lambda}_i$ also an eigenvalue with right eigenvector $\left(\begin{smallmatrix}{S}_i\\{R}_i\end{smallmatrix}\right)^{\#}=\Sigma{\left(\begin{smallmatrix}R_i\\S_i\end{smallmatrix}\right)}^{\#}$ and left eigenvector $   \left({V}_i, {U}_i\right)^{\#}= {\left(U_i, V_i\right)}^{\#}\Sigma_n$, where     $R_i, S_i\in\mathbb{C}^{1\times n}$, $U_i, V_i\in\mathbb{C}^{n\times1}$ and $\Sigma_n:=\left(\begin{smallmatrix} 0_n&1_n\\1_n&0_n\end{smallmatrix}\right)$. 
%That is, for each eigenvalue and eigenvector there exists a corresponding mirror pair. 
This property follows from the fact that $A$ has the doubled-up form  $A:=\Delta\left(A_{-}, A_{+}\right)$.
Furthermore, from the  system  \eqref{cask} $\tilde{A}$ may be diagonalised as $\tilde{A}=P\tilde{A}_0P^{-1}$ where 
$$\tilde{A}_0=\left(\begin{smallmatrix}-{A}^{\flat}_0&0\\0&A_0\end{smallmatrix}\right)$$
and $A_0$ is diagonal and doubled-up. Here $P$ and $P^{-1}$ are lower block triangular (Lemma \ref{hud})
written as 
$$P=\left(\begin{smallmatrix} P_1&0\\P_2&P_3\end{smallmatrix}\right) \quad \mathrm{and}\quad
P^{-1}=\left(\begin{smallmatrix} P^{-1}_1&0\\-P_3^{-1}P_2P_1^{-1}&P^{-1}_3\end{smallmatrix}\right),$$
where 
$$P_3=\left(\begin{smallmatrix}R_1&\hdots &R_n&S_1&\hdots&S_n\\ {S}^{\#}_1&\hdots&{S}^{\#}_n&{R}^{\#}_1&\hdots&{R}^{\#}_n\end{smallmatrix}\right) \quad\mathrm{and}\quad
P^{-1}_1=\left(\begin{smallmatrix}U_1&V_1\\\vdots&\vdots\\
U_n&V_n\\
 {V}^{\#}_1&{U}^{\#}_1\\
 \vdots&\vdots\\
 {V}^{\#}_n&{U}^{\#}_n\end{smallmatrix}\right).
$$
Hence, the power spectrum, $\Psi(s)J$, of the system in Eq. \eqref{cask} may be written 
\begin{align}
%\nonumber&\Psi(s)J=
%\nonumber&V_{\mathrm{vac}}-\left(-V_{\mathrm{vac}}C, C\right)\left(\begin{smallmatrix}S+A^{\flat}&0\\-C^{\flat}V_{\mathrm{vac}}C &S-A\end{smallmatrix}\right)\left(\begin{smallmatrix}C^{\flat}\\C^{\flat}V_{\mathrm{vac}}\end{smallmatrix}\right)\\
\label{gut}
V_{\mathrm{vac}}-\left(-V_{\mathrm{vac}}CP_1+CP_2, CP_3\right)\left(\begin{smallmatrix}s+A^{\flat}_0&0\\0 &s-A_0\end{smallmatrix}\right)\left(\begin{smallmatrix}P_1^{-1}C^{\flat}\\-P_3^{-1}P_2P_1^{-1}C^{\flat}+ P_3^{-1}C^{\flat}V_{\mathrm{vac}}\end{smallmatrix}\right).
\end{align}

We can construct a minimal realisation called \textit{Gilbert's realisation} \cite{Zhou1} by expanding as partial fractions:
\begin{equation}\label{gut1}
\Psi(s)J=V_{\mathrm{vac}}+\sum_{i=1}^n \frac{I_i}{(s+\overline{\lambda}_i)} +\frac{K_i}{(s+\lambda_i)}
+\frac{T_i}{(s-\lambda_i)}+\frac{W_i}{(s-\overline{\lambda}_i)},
\end{equation}
with $\mathrm{Re}(\lambda_i)<0$.
The matrices $I_i, K_i, T_i, W_i$ are necessarily rank-one. Therefore there exist matrices $B_{1,i}, B_{2,i}, B'_{1,i}, B'_{2,i}\in\mathbb{C}^{1\times 2m}$ and $C_{1,i}, C_{2,i}, C'_{1,i}, C'_{2,i}\in\mathbb{C}^{2m\times 1}$  such that 
\begin{equation*}
C_{1,i}B_{1,i}=I_i, C'_{1,i}B'_{1,i}=K_i\quad \,\,\,\mathrm{and}\,\,\, C_{2,i}B_{2,i}=T_i, \quad C'_{2,i}B'_{2,i}=W_i
\end{equation*}
and are each uniquely determined from $I_i, K_i, T_i, W_i$ up to a constant\footnote{For example $\frac{1}{\nu}C_{1,i}$ and $\nu B_{1,i}$ are also solutions to $I_i$, where $\nu$ is a constant.}.
The Gilbert realisation $\tilde{A}_0, \tilde{B}_0, \tilde{C}_0$ is 
\begin{equation*}
\tilde{A}_0:=\mathrm{Diag}\left(-\overline{\lambda}_1, ..., -\overline{\lambda}_n, -\lambda_1, ..., -\lambda_n, \lambda_1, ..., \lambda_n,\overline{\lambda}_1,..., \overline{\lambda}_n\right),
\end{equation*}
\begin{equation*}
\tilde{B}_0:=\left[\begin{smallmatrix}B_1\\B_2\end{smallmatrix}\right], \quad 
\tilde{C}_0:= 
\left[C_1, C_2\right]
\end{equation*}
where 
\begin{align*}
& B_1:=\left[\begin{smallmatrix}B_{1,1}\\\vdots\\B_{1,n}\\B'_{1,1}\\\vdots\\B'_{1,n}\end{smallmatrix}\right] 
\quad 
B_2:=\left[\begin{smallmatrix}B_{2,1}\\\vdots\\B_{2,n}\\B'_{1,1}\\\vdots\\B'_{1,n}\end{smallmatrix}\right],   \\
 &C_1:=\left[\begin{smallmatrix}C_{1,1}&\hdots&C_{1,n}&C'_{1,1}&\hdots&C'_{1,n}\end{smallmatrix}\right], \\& C_2:=\left[\begin{smallmatrix}C_{2,1}&\hdots&C_{2,n}&C'_{2,1}&\hdots&C'_{2,n}\end{smallmatrix}\right] .
\end{align*}
At the moment this Gilbert realisation doesn't satisfy the properties required by Sec. \ref{256}, i.e., $B_1$ and $C_2$ are not doubled-up. We can take care of this in the following way. 
Firstly, in this realisation  $I_i$ is equal to the $i^{\mathrm{th}}$  column of $\left(-V_{\mathrm{vac}}CP_1+CP_2\right)$ multiplied by the $i^{\mathrm{th}}$  row of $P_1^{-1}C^{\flat}$ and $K_i$ is equal to the $(n+i)^{\mathrm{th}}$ column of $\left(-V_{\mathrm{vac}}CP_1+CP_2\right)$ multiplied by the $(n+i)^{\mathrm{th}}$ row of $P_1^{-1}C^{\flat}$ (see Eq. \eqref{gut}).  Therefore,  the $i^{\mathrm{th}}$ row of $B_1$ differs from the $i^{\mathrm{th}}$ row of the doubled-up matrix $P_1^{-1}C^{\flat}$ by an (unknown) multiplicative constant.  
Finally, by   multiplying the rows of $B_1$ in our Gilbert realisation by suitable constants (and hence multiplying the corresponding columns of $C_1$ by the inverse of these constants so that the power spectrum remains unchanged) we can obtain a doubled-up $B_1$. A similar technique may be used to obtain a doubled-up $C_2$ by using the fact that $CP_3$ is doubled-up.

\chapter{Proof of Theorem \ref{twomode}}\label{nike}

In this section we prove Theorem \ref{twomode}.

\begin{proof}
Firstly, by Theorem \ref{equivalence}, there exists a TFE basis whereby the system can be split into a series product of a system with a pure stationary state and one with a mixed (with the output) stationary state. Furthermore, we may assume that the pure and mixed component may be decomposed as a series of one-mode PQLSs by the results in Sec. \ref{seriesp}. There are three cases to consider: (i) the system is globally minimal and the stationary state is fully mixed, (ii) pure component is a one mode PQLS or (iii) pure component is a two-mode PQLS and the power spectrum is trivial. 

In this proof we are interested in cases (ii) and (iii). Firstly case( ii) is straightforward as it reduces to Lemma \ref{onemode}.  Now for case (iii), let the first system in the cascade be $(c,\Omega_1)$ and the second be $(d, \Omega_2)$. For the power spectrum to be trivial it is required that 
\[\Xi(-i\omega)N^T=N^T\Xi(-i\omega) \,\, \mathbf{and} \,\, \Xi(-i\omega)M=M\overline{\Xi(+i\omega)}.
\]
Now, considering the poles of $\Xi(-i\omega)M=M\overline{\Xi(+i\omega)}$ we must have either:
\begin{enumerate}
\item $\Omega_1=-\Omega_2$ and $c^{\dag}c=d^{\dag}d$, or
\item $\Omega_1=-\Omega_2=0$, or
\item $\Omega_1\neq-\Omega_2$
\end{enumerate}

\underline{Case 1)}:

In this case the denominators of $\Xi(-i\omega)N^T=N^T\Xi(-i\omega) \,\, \mathbf{and} \,\, \Xi(-i\omega)M=M\overline{\Xi(+i\omega)}$ are equal, therefore we may equate their numerators. Expanding in powers of $\omega$ we obtain the following set of equations:
\begin{align}
\label{once}& \left(cc^{\dag}+dd^{\dag}\right)N^T=N^T\left(cc^{\dag}+dd^{\dag}\right)
\end{align}
\begin{align}
\nonumber \left(i\Omega_1+\frac{1}{2}c^{\dag}c-cc^{\dag}\right)&\left(-i\Omega_1+\frac{1}{2}d^{\dag}d-dd^{\dag}\right)N^T\\
&=N^T\left(i\Omega_1+\frac{1}{2}c^{\dag}c-cc^{\dag}\right)\left(-i\Omega_1+\frac{1}{2}d^{\dag}d-dd^{\dag}\right)
\end{align}
\begin{align}
\label{fors}&\left(cc^{\dag}+dd^{\dag}\right)M=M\overline{\left(cc^{\dag}+dd^{\dag}\right)}
\end{align}
\begin{align}
\nonumber\left(i\Omega_1+\frac{1}{2}c^{\dag}c-cc^{\dag}\right)&\left(-i\Omega_1+\frac{1}{2}d^{\dag}d-dd^{\dag}\right)M\\
&=M\left(-i\Omega_1+\frac{1}{2}c^{\dag}c-\overline{cc^{\dag}}\right)\left(i\Omega_1+\frac{1}{2}d^{\dag}d-\overline{dd^{\dag}}\right).
\end{align}
The second equation here may be further split into real and imaginary parts. Using also the condition $c^{\dag}c=d^{\dag}d$  and the assumption $\Omega\neq0$ we obtain the following pair of equations: 
\begin{align}
&\label{twice}\left(cc^{\dag}-dd^{\dag}\right)N^T=N^T\left(cc^{\dag}-dd^{\dag}\right)\\
&\label{thrice}\left(cc^{\dag}dd^{\dag}-\frac{1}{2}c^{\dag}cdd^{\dag}-\frac{1}{2}d^{\dag}dcc^{\dag}\right)N^T=N^T\left(cc^{\dag}dd^{\dag}-\frac{1}{2}c^{\dag}cdd^{\dag}-\frac{1}{2}d^{\dag}dcc^{\dag}\right)
\end{align}
Similarly, the fourth equation may be split into symmetric and antisymmetric parts, i.e.,
\begin{align}
\label{fives}&\left(cc^{\dag}-dd^{\dag}\right)M=-M\overline{\left(cc^{\dag}-dd^{\dag}\right)}\\
\label{sixes}&\left(cc^{\dag}dd^{\dag}-\frac{1}{2}c^{\dag}cdd^{\dag}-\frac{1}{2}d^{\dag}dcc^{\dag}\right)M=M\overline{\left(cc^{\dag}dd^{\dag}-\frac{1}{2}c^{\dag}cdd^{\dag}-\frac{1}{2}d^{\dag}dcc^{\dag}\right)}
\end{align}
Therefore combing \eqref{once} \eqref{twice} and \eqref{thrice} we obtain $cc^{\dag}N^T=Ncc^{\dag}$, $dd^{\dag}N^T=N^Tdd^{\dag}$ and $cc^{\dag}dd^{\dag}N^T=N^Tcc^{\dag}dd^{\dag}$. Further, combining \eqref{fors} \eqref{fives} and \eqref{sixes} we obtain $cc^{\dag}M=M\overline{dd{\dag}}$,  and $cc^{\dag}dd^{\dag}M=M\overline{cc^{\dag}dd^{\dag}}$. Hence we obtain conditions (3) in the Theorem.

\underline{Case 2)}

The result is equations \eqref{once} \eqref{thrice} \eqref{fors} and \eqref{sixes}. Hence we obtain conditions   (4) in the Theorem.  

\underline{Case 3)}

Because the poles are different in $\Xi(-i\omega)M=M\overline{\Xi(+i\omega)}$ then it must be the case that we have two one-mode cancellations as in Lemma \eqref{onemode}. That is, $\Xi_i(-i\omega)N^T\Xi_i(-i\omega)^{\dag}=N^T \,\, \mathbf{and} \,\, \Xi_i(-i\omega)M\Xi_i(-i\omega)^T=M$ for systems $i=1,2$. This gives either conditions (1) or conditions (2) in the  Theorem. 
\end{proof}

\chapter{Supplementary Proof for Sec. \ref{thermy}} \label{noobs}

We show here that non-observability of the cascaded system \eqref{cask2} implies that the original system is non-globally minimal for the case of distinct eigenvalues. Recall from Sec. \ref{thermy} that non-observability implies that there exists an eigenvalue-eigenvector pair $\lambda$, $y$ of $A^{\dag}$, i.e., $A^{\dag}y=\lambda y$ such that $Cy$ is also an eigenvector of the input $N$. 

Now, list all the left and right eigenvectors of $A$ as $L_i$ and $R_i$ respectively. Therefore we can write $\Upsilon(-i\omega)$  as 
$$\Upsilon(-i\omega)=\left(1-\sum_i \frac{(CR_i)(L_iC^{\dag})}{-i\omega-\lambda_i}\right)N\left(1-\sum_i\frac{ (CL_i^{\dag})(R_i^{\dag}C^{\dag})}{i\omega-\overline{\lambda}_i}\right)$$

Now, suppose that the cascaded system \eqref{cask2} is non-observable as per the assumption above. This means that $CL^{\dag}_i$ is an eigenvector of $N$ for some $i$. Suppose that this is true for $i=k$. It follows fairly straightforwardly by using the unitarity of the transfer function that 
$$\Upsilon(-i\omega)=\left(1-\sum_{i\neq k} \frac{(CR_i)(L_iC^{\dag})}{-i\omega-\lambda_i}\right)N\left(1-\sum_{i\neq k}\frac{ (CL_i^{\dag})(R_i^{\dag}C^{\dag})}{i\omega-\overline{\lambda}_i}\right).$$ 
Moreover, the function $\left(1-\sum_{i\neq k} \frac{(CR_i)(L_iC^{\dag})}{-i\omega-\lambda_i}\right)$ must also be unitary. Hence by \cite{Petersen2} there exist a physical QLS with this transfer function. Crucially, this system has one less mode than the original system, which implies that the original PQLS cannot be globally minimal.

\chapter{Proof of Theorem \ref{Q.F.I}}\label{ONCE1}
Before proving this theorem we discuss a little theory. 

The operators $ \mathbf{B}_{i}(t)$ introduced in Sec. \ref{Bosonic} are the quantum analogue of `classical' Wiener process and can be used to define \textit{quantum stochastic integrals}, such as
$$\mathbf{I}(t)=\int^t_0 d\mathbf{B}^{\dag}(s) \mathbf{M}(s)+\mathbf{N}^{\dag}(s)d\mathbf{B}(s)+\mathbf{P}(s)ds$$
\cite{Bouten1, Bouten2, Guta4}, where $ \mathbf{M}(s),  \mathbf{N}(s),  \mathbf{P}(s)$ are time-adaptive operators.
When multiplying stochastic integrals, $\mathbf{I}_1(t)$ and $\mathbf{I}_2(t)$, the product is a stochastic integral with increment 
$$d\left(\mathbf{I}_1(t)\mathbf{I}_2(t)\right)=d\mathbf{I}_1(t)\cdot \mathbf{I}_2(t)+\mathbf{I}_1(t)\cdot d\mathbf{I}_2(t)+d\mathbf{I}_1(t)\cdot d\mathbf{I}_2(t).$$ Notice the extra \textit{Ito correction} term, which is not present in ordinary calculus. The Ito term can be calculated by using the Ito rules \cite{Parth1}:
\[
\begin{tabular}{c|cc}\centering
$\times$ & $d\mathbf{A}_k$  & $d\mathbf{A}^*_k$  \\
\hline
$d\mathbf{A}_i$ & $(\delta_{ik}+N_{ki})dt$ &$M_{ik}dt$ \\
$d\mathbf{A}^*_j$          & $\overline{M}_{ki}dt$&$N_{ik}dt $ \\
\end{tabular}\]
Notice that the Ito rules depends on the input $V(N, M)$.

\begin{proof}[Proof of Theorem \ref{Q.F.I}]\label{cherry}
The evolution of the system and field is described by the unitary $\mathbf{U}_{\theta}(t)$, which satisfies the QSDE
\begin{equation}\label{QSDE27}
d\mathbf{U}_{\theta}(t)=\left(\sum_{i=1}^m{\mathbf{L}}_id{\mathbf{B}}^*_i(t)-{\mathbf{L}}_{i}^{*}d{\mathbf{B}}_i(t)-({\mathbf{K}}+i{\mathbf{H}})dt
\right){U}_{\theta}(t),
\end{equation}
where $\mathbf{K}=\frac{1}{2}\left(\mathbf{L}^{\#},\mathbf{L}^T\right)JV(N,M)J\left(\mathbf{L}^T, \mathbf{L}^{\dag}\right)=\frac{1}{2}\tilde{\mathbf{L}}^{\dag}\tilde{\mathbf{L}}$ for modified coupling operator given in the theorem. Note that these modified coupling operators in $\mathbf{K}$ are the only difference between Eqs. \eqref{QSDE27} and   \eqref{eq.unitary.cocycle} (the unmodified ones, as in Eq. \eqref{eq.unitary.cocycle} correspond to the case of vacuum input).

Now by considering the generator $\mathbf{G}_{\theta}(t):=\mathbf{U}^*_{\theta}(t)\dot{\mathbf{U}}_{\theta}(t)$, the QFI, $F_{\theta}$, in Eq. \eqref{qfiformula} may be written as ( we drop the subscript $\theta$ here)  \[F_{\theta}=4\mathrm{Re}\left(\braket{\phi\otimes \xi| \mathbf{G}^*(t)\mathbf{G}(t)|\phi\otimes \xi}     -  \braket{\phi\otimes \xi|\mathbf{G}^*(t)|\phi\otimes \xi} \braket{\phi\otimes \xi|\mathbf{G}(t)|\phi\otimes \xi}\right).\]   
The method we take is to show that the generator can be written as a QSDE, from which it may be solved.  From the QSDE \eqref{QSDE27} we have
\begin{align*}
d\mathbf{U}^*(t)&=\mathbf{U}^*(t)\left(\sum_{i=1}^m{\mathbf{L}}^*_id{\mathbf{B}}_i(t)-{\mathbf{L}}_{i}d{\mathbf{B}}^*_i(t)-({\mathbf{K}}-i{\mathbf{H}})dt
\right),\\
d\dot{\mathbf{U}}(t)&=\left(\sum_{i=1}^m{\mathbf{L}}_id{\mathbf{B}}^*_i(t)-{\mathbf{L}}_{i}^{*}d{\mathbf{B}}_i(t)-(\mathbf{K}+i\mathbf{H})dt
\right)\dot{\mathbf{U}}(t)\\
&+
\left(\sum_{i=1}^m\dot{\mathbf{L}}_id{\mathbf{B}}^*_i(t)-\dot{\mathbf{L}}_{i}^{*}d{\mathbf{B}}_i(t)-(\dot{\mathbf{K}}+i\dot{\mathbf{H}})dt
\right){\mathbf{U}}(t).
 \end{align*}
Hence by applying the Ito rules we obtain (dropping the subscript $\theta$)
\begin{align*}
d\mathbf{G}_{}(t)&=\left(d\mathbf{U}^*(t)\right)\dot{\mathbf{U}}(t) + \mathbf{U}^*(t)\left(d\dot{\mathbf{U}}(t)\right)+\left(d\mathbf{U}^*(t)\right) \left(d\dot{\mathbf{U}}(t)\right)  \\
&=\sum_{i=1}^m\left( j_t(\dot{\mathbf{L}}_i)d\mathbf{B}_i^*(t) -    j_t(\dot{\mathbf{L}}^*_i)d\mathbf{B}_i(t)\right)
\\&+\mathbf{U}^*(t)\left(-2\mathbf{K}dt+\left[    \sum_{i=1}^m{\mathbf{L}}^*_id{\mathbf{B}}_i(t)-{\mathbf{L}}_{i}d{\mathbf{B}}^*_i(t)\right]\left[\sum_{i=1}^m{\mathbf{L}}_id{\mathbf{B}}^*_i(t)-{\mathbf{L}}_{i}^{*}d{\mathbf{B}}_i(t)\right]\right)\dot{\mathbf{U}}(t)\\
&+j_t\left(-(\dot{\mathbf{K}}+i\dot{\mathbf{H}})+  \left[\sum_{i=1}^m{\mathbf{L}}^*_id{\mathbf{B}}_i(t)-{\mathbf{L}}_{i}d{\mathbf{B}}^*_i(t)\right]\left[   \sum_{i=1}^m\dot{\mathbf{L}}_id{\mathbf{B}}^*_i(t)-\dot{\mathbf{L}}_{i}^{*}d{\mathbf{B}}_i(t)\right]\right)dt
\end{align*}
As $\mathbf{K}=\frac{1}{2}\tilde{\mathbf{L}}^{\dag}\tilde{\mathbf{L}}$ , the second term above is zero and so we may write  
\begin{align*}
d\mathbf{G}_{}(t)&=
\sum_{i=1}^m\left( j_t(\dot{\mathbf{L}}_i)d\mathbf{B}_i^*(t)-    j_t(\dot{\mathbf{L}}^*_i)d\mathbf{B}_i(t)\right)-ij_t(\mathbf{R})dt,
\end{align*}
where the operator $\mathbf{R}$ is given by
 \begin{equation}\label{sevec}
 \mathbf{R}= \dot{\mathbf{H}}+\mathrm{Im}\sum_{i=1}^m\dot{\tilde{\mathbf{L}_i}}^{\dag}\tilde{\mathbf{L}_i}. 
 \end{equation}
 In the following we work with the centred generator, $\mathbf{G}_0(t)$, given by 
\begin{align*}
d\mathbf{G}_{0}(t)&=
\sum_{i=1}^m\left( j_t(\dot{\mathbf{L}}_i)d\mathbf{B}_i^*(t)-    j_t(\dot{\mathbf{L}}^*_i)d\mathbf{B}_i(t)\right)-ij_t(\mathbf{R}_0)dt,
\end{align*}
where $\mathbf{R}_0:=\mathbf{R}-\left<\mathbf{R}\right>_{\rho_{ss}}\mathds{1},$
so that (by ergodicity) its rescaled mean converges to zero, i.e., 
\[\lim_{t \to \infty}\frac{1}{t}\braket{\phi\otimes \xi|\mathbf{G}_0(t)|\phi\otimes \xi}=0. \]
Such centred operators are examples of a more general class of operators, called output \textit{fluctuation operators} (see \cite{Guta4}).
Now, using this fluctuation operator the QFI scales linearly with $t$ by ergodicity and the leading contribution is given by the \textit{quantum Fisher-information rate}
\[f_{\theta}:=\lim_{t\to\infty}\frac{F_\theta}{t}=\lim_{t\to\infty}\frac{1}{t}4\mathrm{Re} \braket{\phi\otimes \xi| \mathbf{G}_0^*(t)\mathbf{G}_0(t)|\phi\otimes \xi}   .\]
For notational convenience, we absorb a factor of $t^{-1/2}$ into $\mathbf{G}_0$, so that the rate is given by 
\[f_{\theta}=\lim_{t\to\infty}4\mathrm{Re}\int_0^t\braket{\phi\otimes \xi| d\left(\mathbf{G}_0^*(s)\mathbf{G}_0(s)\right)|\phi\otimes \xi} \]
for fluctuation operator
\[\mathbf{G}_0(t)=\frac{1}{\sqrt{t}}\int_0^t\left(\sum_{i=1}^m\left( j_t(\dot{\mathbf{L}}_i)d\mathbf{B}_i^*(s)-    j_t(\dot{\mathbf{L}}^*_i)d\mathbf{B}_i(s)\right)-ij_t(\mathbf{R}_0)ds\right).\]

We will now calculate the rate. Firstly, the differential $d\left(\mathbf{G}_0^*(s)\mathbf{G}_0(s)\right)$ can be written
\begin{equation}\label{split}
d\left(\mathbf{G}_0^*(s)\mathbf{G}_0(s)\right)=\mathbf{G}_0^*(s)\cdot d\mathbf{G}_0(s)+d\mathbf{G}_0^*(s)\cdot \mathbf{G}_0(s)+ d\mathbf{G}_0^*(s)\cdot d\mathbf{G}_0(s).
\end{equation}
Let us  calculate these terms in turn.
For the last term 
\begin{align}
\nonumber
&\int_0^t\braket{\phi\otimes \xi| d\mathbf{G}_0^*(s)\cdot d \mathbf{G}_0(s)|\phi\otimes \xi} \\&=\frac{1}{t}\int_0^t
\braket{\phi\otimes \xi| \left[    \sum_{i=1}^mj_s(\dot{\mathbf{L}}^*_i)d{\mathbf{B}}_i(s)-j_s(\dot{\mathbf{L}}_{i})d{\mathbf{B}}^*_i(s)\right]\left[\sum_{i=1}^mj_s(\dot{\mathbf{L}}_i)d{\mathbf{B}}^*_i(s)-j_s(\dot{\mathbf{L}}_{i}^{*})d{\mathbf{B}}_i(s)\right]     |\phi\otimes \xi}\\
&=\frac{1}{t}\int_0^t \braket{\phi|T_s\left(\dot{\tilde{\mathbf{L}}}^{\dag}\dot{\tilde{\mathbf{L}}}\right)|\phi} \overset{t\to\infty}{\longrightarrow}\left<\dot{\tilde{\mathbf{L}}}^{\dag}\dot{\tilde{\mathbf{L}}}\right>_{\rho_{ss}}. \label{qfi1}
\end{align}
 The remaining two terms in Eq. (\ref{split}) are slightly more involved and require the following Lemma.

\begin{Lemma}
\begin{equation}\label{lemmar} 
{\sqrt{s}}    \braket{\xi|\mathbf{G}^*(s)j_s(-i\mathbf{R}_0)|\xi}=\int^s_0T_r\circ\Upsilon\circ T_{s-r}(-i\mathbf{R}_0)dr,
\end{equation}
where $\Upsilon(\mathbf{X})=i\mathbf{R}_0\mathbf{X}+\sum_{i=1}^m\dot{\tilde{\mathbf{L}}}^{\dag}_i\left[\mathbf{X}, \tilde{\mathbf{L}}_i\right]$.
\end{Lemma}
\begin{proof}
Label the LHS of Eq. (\ref{lemmar}) as $K_s(-i\mathbf{R}_0)_s$ then by using the  Ito rules  we have 
\begin{align*}
&dK_s(-i\mathbf{R}_0)=    \braket{\xi|\mathbf{G}^*(s)dj_s(-i\mathbf{R}_0)|\xi}
+\braket{\xi| j_s(i\mathbf{R}_0)^* j_s(-i\mathbf{R}_0)|\xi}ds\\
&-i\braket{\xi|   \left[    \sum_{i=1}^mj_s(\dot{\mathbf{L}}^*_i)d{\mathbf{B}}_i(s)-j_s(\dot{\mathbf{L}}_{i})d{\mathbf{B}}^*_i(s)\right]\left[\sum_{i=1}^mj_s([   \mathbf{R}_0, \mathbf{L}_i])         {d}\mathbf{B}^*_i(s)-    j_s([   \mathbf{R}_0, \mathbf{L}_i^*])     d\mathbf{B}_i(s)\right] \xi}\\
&\quad\quad\quad\quad\quad= \braket{\xi|\mathbf{G}^*(s)dj_s(-i\mathbf{R}_0)|\xi}   
+\braket{\xi| j_s(i\mathbf{R}_0)^* j_s(-i\mathbf{R}_0)|\xi}ds\\
&-i  \braket{\xi    |\sum_{i=1}^m j_s(   \dot{ \tilde{\mathbf{L}}}_i^*[  \mathbf{R}_0, \tilde{\mathbf{L}}_i]     )\xi}ds\\
&\quad\quad\quad\quad\quad=\left[K_s\circ\tilde{\mathcal{L}}(-i   \mathbf{R}_0        )  +T_s\circ\Upsilon(-i\mathbf{R}_0)\right]ds,
\end{align*}
where $\tilde{\mathcal{L}}(\mathbf{X})$ is the (modified) Lindblad generator.
Finally, this first order differential solution can be verified to have solution as stated in the lemma.
\end{proof}

Following this  result we are now able to calculate the expectations of the first and second terms in Eq. (\ref{split}). Thus 
\begin{align}
\int^t_0\braket{\phi\otimes \xi| \mathbf{G}^*(s)d\mathbf{G}(s)| \phi\otimes \xi}&=\frac{1}{t}\int^t_0ds\braket{\phi| \left(\int^s_0 T_r\circ\Upsilon\circ T_{s-r}( -i   \mathbf{R}_0)dr  \right) |\phi} \nonumber \\
&=\int^t_0ds          \braket{\phi| \frac{1}{t}   \left(\int^{t-s}_0\mathrm{dr} T_r\right)\circ\Upsilon\circ T_{s}( -i   \mathbf{R}_0) |\phi} \nonumber \\
&     \overset{t\to\infty}{\longrightarrow}   -i   \left<       \Upsilon \circ \mathbf{W}\right> _{\rho_{ss}}\label{qfi2}.         
   \end{align}
 A similar result holds for the second term, so that we have 
\begin{align*}
f_{\theta}&=  \left<        \dot{\tilde{\mathbf{L}}}^{\dag}\dot{\tilde{\mathbf{L}}}      -i \sum_i \dot{\tilde{\mathbf{L}}}_i[\tilde{\mathbf{L}}_i, \mathbf{W}] \right. \\
&
\left.
+i\left( \sum_i \dot{\tilde{\mathbf{L}}}_i[\tilde{\mathbf{L}}_i, \mathbf{W}]\right)^{\dag}
       +  \mathbf{W} \mathbf{R}_0  +    \mathbf{R}_0  \mathbf{W}                     \right>     . \end{align*}
The first 3 terms are consistent with the statement \eqref{QFIMAIN} in the Theorem. Hence to complete the proof it remains to show that
$$\mathbf{W} \mathbf{R}_0  +    \mathbf{R}_0  \mathbf{W}   =\left([\tilde{\mathbf{L}}_i, \mathbf{W}]\right)^{\dag}[\tilde{\mathbf{L}}_i, \mathbf{W}],$$
which can be seen in   \cite[pg. 25]{Guta4}.
\end{proof}

%%%%%%%%%%%%%%%%%%%%%%%%%%%%%%%%%%%%%%%%%%%%
\chapter{Using Frequency-Entangled States to Improve Estimation}\label{frequent}
%%%%%%%%%%%%%%%%%%%%%%%%%%%%%%%%%%%%%%%%%%%%

In this section we show that there is only a small advantage to be had  in terms of estimation precision by using frequency-entangled states for one parameter models of a fixed photon input. It is illustrated in the form of an example that this advantage is difficult to exploit in practice.

Let us consider a SISO PQLS to begin with. As we have seen above, the optimal 
QFI for fixed photon input states of one frequency is given by 
\begin{eqnarray}
F(\theta)=  N^2 \cdot \sup_\omega \left\|\frac{d\Xi_{\theta}(-i\omega)}{d\theta} \right\|^2.
\end{eqnarray} 

Now, suppose we are to generalise our class of probe states by allowing for $N$-photon inputs spread over $d$ frequencies. Denote our $d$-frequency entangled state as $\ket{\psi}$. Since we are dealing with linear systems, each frequency evolves independently, so that the action of the PQLS on $\ket{\psi}$ is given by
\[\ket{\psi}\mapsto\Xi_{\theta}(-i\omega_1)\otimes\Xi_{\theta}(-i\omega_2)\otimes...\otimes\Xi_{\theta}(-i\omega_d)\ket{\psi}.\]
Hence the QFI is given by 
\[
\hspace{-2.5cm}
     F(\theta)
       =4\mathrm{Var}\left(
           \frac{d\Xi_{\theta}(-i\omega_1)}{d\theta}
               \mathbf{a}^{\dag}(\omega_1)\mathbf{a}(\omega_1)
                        + \ldots + 
                \frac{d\Xi_{\theta}(-i\omega_d)}{d\theta}
                    \mathbf{a}^{\dag}(\omega_d)\mathbf{a}(\omega_d)\right).
\]
Now, the maximum of this variance corresponds to a $\ket{\psi}$ which is an 
equally weighted superposition of eigenvectors, whose eigenvalues are the 
maximum and minimum of this generator. 
These eigenvalues are given by $N\times\lambda_{max}$ and 
$N\times\lambda_{min}$ where $\lambda_{max}
=\max\big\{\max_{i}\frac{d\Xi(-i\omega_i)}{d\theta}, 0\big\}$ 
and $\lambda_{min}=\min\big\{\min_{i}\frac{d\Xi(-i\omega_i)}{d\theta}, 0\big\}$. There are two cases to consider: 
\begin{enumerate}
\item If $\frac{d\Xi_{\theta}(-i\omega)}{d\theta}$ has the same sign for all $\omega$, then the largest and smallest eigenvalues are given by $0$ and $\max_i\frac{d\Xi_{\theta}(-i\omega_{i})}{d\theta}\times N$. To optimise over all probe states, one must select as one of the probe frequencies   $\omega_{opt}=\arg\sup_\omega \left\|\frac{d\Xi_{\theta}(-i\omega)}{d\theta} \right\|^2$. In this case, a single frequency probe,  and in particular the cat-state, is optimal as it has the required eigenvalues.
\item If $\frac{d\Xi_{\theta}(-i\omega)}{d\theta}$ has differing sign then it is possible to obtain a factor of 4 improvement in the QFI by using frequency-entangled states. To see this, let $\omega_{sup}=\arg\sup_\omega\frac{d\Xi_{\theta}(-i\omega)}{d\theta}$ and    $\omega_{inf}=\arg\inf_\omega\frac{d\Xi_{\theta}(-i\omega)}{d\theta}$ then the optimal probe state is seen to be given by   
\[\ket{\psi}=\frac{1}{\sqrt{2}}\left(\ket{0, N(\omega_{sup})}+ \ket{0, N(\omega_{inf})}\right)\]
and in which case the QFI is \[F(\theta)=N^2\left( \left. \frac{d\Xi_{\theta}(-i\omega)}{d\theta}\right|_{\omega_{sup}}-\left.\frac{d\Xi_{\theta}(-i\omega)}{d\theta}\right|_{\omega_{inf}}\right)^2.\]   Now if $\left. \frac{d\Xi_{\theta}(-i\omega)}{d\theta}\right|_{\omega_{sup}}=-\left.\frac{d\Xi_{\theta}(-i\omega)}{d\theta}\right|_{\omega_{inf}}$ then the QFI here is four times larger than the best single frequency input state. 
\end{enumerate}

In conclusion it is possible to get (up to) a factor of 4 improvement in estimation precision for a fixed photon input to a SISO PQLS by using two-frequency-entangled inputs. However, this improvement is intrinsically dependent on the system and on the unknown parameter. For example, the requirement that $\left. \frac{d\Xi_{\theta}(-i\omega)}{d\theta}\right|_{\omega_{sup}}=-\left.\frac{d\Xi_{\theta}(-i\omega)}{d\theta}\right|_{\omega_{inf}}$ is highly restrictive and so this factor of 4 improvement should be correctly interpreted as an upper bound that may in general not be achievable for a given system and unknown parameter.

\begin{exmp}\label{plmcd}
To see an example where this condition does hold, 
consider a SISO PQLS with one internal mode $(n=1)$ characterised by $\left(C=c, \Omega=a\right)$ and assume that $c$ is unknown and $a$ is known. We may write the phase in the transfer function as $\lambda=2\mathrm{arctan}\left(\frac{-2\omega+2a}{c^2}\right)-\pi$. It follows that $\frac{d\lambda}{dc}=\frac{4(-\omega+a)c}{c^4+4(-\omega+a)^2}$. Maximising this over frequency we find the single frequency N00N state with the largest QFI will have frequency $\omega_{opt}=a\pm c^2/2$. In which case $ \frac{d\lambda}{dc}|_{\omega=\omega_{opt}}=\mp\frac{1}{c}$. Assuming that $|a\pm c^2/2|\geq0$ (so that the following frequency choices are physical), then an input state entangled over these two frequency choices would be four times more informative than the best monochromatic probe. However, note that estimating the other parameter in this system would not satisfy this condition.
\end{exmp}

For MIMO PQLSs the result is similar, except for a slight subtlety due to the possibility of mode-entanglement in addition to frequency-entanglement. We don't discuss this further here.

Practically, since our search is for realistic metrology methods for PQLSs, the cost or complexity involved in the creation of these highly-frequency-entangled states far outweigh the benefit in terms of enhanced precision for one parameter systems (note that this is not the case for multi-parameters). Coupled with the fact that the advantage of using them may only be seen for `special' PQLSs leads us to neglect these types of states from our considerations for one-parameter models. 

\chapter{Adaptive Procedure for Feedback Method 2 (MIMO PQLSs)}\label{HUDE12}

Split the time into  three and obtain a rough estimate for $\theta$ with MSE $\mathcal{O}\left(\frac{1}{{T_{\mathrm{tot}}}}\right)$ in step 1. In step 2 choose  $\hat{c}_{i,n}$ and $\Omega_{jn}$,  so that the direct and indirect couplings in \eqref{keen1} and \eqref{keen2} are of the order $\mathcal{O}\left(\frac{1}{T^{1/4}_{\mathrm{Tot}}}\right)$ (in a similar way to the SISO case). Considering the transfer function of the resultant system given by \eqref{fedr2}, this entails that  
 all matrix elements of $\frac{d\Xi_{\theta}}{d\theta}(-i\hat{\omega})$ are of order $\mathcal{O}\left(\sqrt{T_{\mathrm{tot}}}\right)$, where $\hat{\omega}$ is the estimator of $\Omega_2$ in Eq. \ref{fedr2}, which has MSE $\mathcal{O}\left(\frac{1}{{T_{\mathrm{tot}}}}\right)$ from step 1 (see remark \ref{hod1e}). 
 To see this, firstly write $c_i=\left[c_{1,i}, ..., c_{mi}\right]^T$ and $\delta_1=\left[\delta_{1,1}, ..., \delta_{mi}\right]^T$, then the $(i,j)$-component of the matrix $C\left(-i\omega+i\Omega+\frac{1}{2}C^{\dag}C\right)^{-1}C^{\dag}$ may be written as 
 \begin{equation}\label{food1}
 \frac{\mathrm{Det}\left(-i\omega+i\Omega+\frac{1}{2}C^{\dag}C-\left[\begin{smallmatrix}\overline{c}_{j1} \\\vdots \\ \overline{c}_{j,n-1}\\\overline{\delta}_{j,1}\end{smallmatrix}\right]  \left[\begin{smallmatrix}{c}_{i1} &... & {c}_{i,n-1}& {\delta}_{i,1}\end{smallmatrix}\right]  \right)}{\mathrm{Det}\left(-i\omega+i\Omega+\frac{1}{2}C^{\dag}C\right)}-1
 \end{equation}
 using standard matrix results. 
 %Syvester's determinant theorem
   Now $\mathrm{Det}\left(-i\hat{\omega}+i\Omega+\frac{1}{2}C^{\dag}C\right)=\mathcal{O}\left(\frac{1}{\sqrt{  T_{\mathrm{tot}}  }}\right)$ using \eqref{jesus1}. Similarly the term in the numerator of \eqref{food1} is also of order $\mathcal{O}\left(\frac{1}{\sqrt{  T_{\mathrm{tot}}  }}\right)$, which can be seen by writing it  as 
\begin{align*} 
 \left( -i\omega+i\Omega_2\pm\frac{1}{2}\delta_1^{\dag}\delta_1-\overline{\delta}_{j1}\delta_{i1}\right) 
&\mathrm{Det}\left(-i\omega+i\Omega_1\pm\frac{1}{2}c^{\dag}c -\left[\begin{smallmatrix}\overline{c}_{j1} \\\vdots \\ \overline{c}_{j,n-1}\end{smallmatrix}\right]  \left[\begin{smallmatrix}{c}_{i1} &... & {c}_{i,n-1}\end{smallmatrix}\right]  \right.\\&
\left.  -\frac{\left(i\delta_2\pm\frac{1}{2}c^{\dag}\delta_1   - \left[\begin{smallmatrix}\overline{c}_{j1} \\\vdots \\ \overline{c}_{j,n-1}\end{smallmatrix}\right] \delta_{i1}     \right)\left(i\delta_2^{\dag}\pm\frac{1}{2}\delta^{\dag}_1 c    -    \overline{\delta}_{j1}\left[\begin{smallmatrix}{c}_{i1} &... & {c}_{i,n-1}\end{smallmatrix}\right]       \right)}{ -i \omega+i\Omega_2\pm\frac{1}{2}\delta_1^{\dag}\delta_1   -\overline{\delta}_{j1}\delta_{i1}}   \right).
\end{align*}
 %The reason for this is because 
%each component of $C\left(-i\omega+i\Omega+\frac{1}{2}C^{\dag}C\right)^{-1}C^{\dag}$ for the system \eqref{fedr2} is of the form:
%\begin{equation}\label{food1}
%\frac{K_{ij}\left(-i\omega+i\Omega_2\right)+\mathcal{O}\left(\frac{1}{T^{1/2}_{\mathrm{Tot}}}\right)}{\mathrm{Det}\left(-i\omega+i\Omega+\frac{1}{2}C^{\dag}C\right)}.
%\end{equation}
%Therefore, for input  frequency $\hat{\omega}$ (which has standard error $\mathcal{O}\left(\frac{1}{T^{1/2}_{\mathrm{Tot}}}\right)$)  the numerator and denominator of \eqref{food1} are each of order $\mathcal{O}\left(\frac{1}{T^{1/2}_{\mathrm{Tot}}}\right)$. Similarly, assuming that $\Omega_2$ depends on $\theta$  then the leading order term of the derivative of the denominator and numerator at $\omega=\hat{\omega}$ are $\mathcal{O}\left(1\right)$.
 Hence, by using the quotient rule, each element of $\frac{d\Xi_{\theta}}{d\theta}(-i\Omega_2)$ is of order $\sqrt{T_{\mathrm{tot}}}$. Finally, we have 
\begin{align}\label{MIMOR} 
\left|\left|     \frac{d\Xi_{\theta}}{d\theta}(-i\Omega_2)       \right|\right|^2=    \mathcal{O}\left(T_{\mathrm{tot}}\right)\end{align}
(spectral norm), as in the SISO case.

\begin{remark}
In the MIMO case the optimal coherent input is the one with amplitude given by the largest eigenvector of 
$ \frac{d\Xi_{\theta}}{d\theta}(-i\Omega_2)   $ \cite {Guta2}. However this optimisation is not important as we are primarily interested in the scaling with time (this optimisation will only improve the precision by  a constant factor) and so any vector will suffice.
\end{remark}

\begin{remark}
Notice that we have more control parameters than we actually require. One could of course optimise over these but for simplicity we have set them all to be zero ($d_{jk}=0$ for $2\leq j\leq m, 1\leq k \leq n-1$).
\end{remark}

%\begin{exmp}
%Consider a MIMO PQLS with one internal mode. In this case the coupling matrix and Hamiltonian parameters are independent from each other because they are all identifiable \cite{Guta}. Therefore, if we would like to estimate $\Omega$ then it must be the case that the coupling matrix is known in a one unknown parameter model. We can reduce the problem to a SISO one      by setting $m-1$ of the coupling parameters to zero  (with zero error) using  \eqref{est1}. Hence the feedback adaptive procedure will  progress identically to  the toy example earlier.
%\end{exmp}

\chapter{A QFI Proof for Coherent States}\label{HUDE}
We now prove that for a coherent state, with amplitude $\alpha(\theta)\in\mathbb{R}$, the QFI is given by
$$F(\theta)=4\left|\frac{d\alpha(\theta)}{d\theta}\right|^2.$$

Firstly, write the coherent state as $\left|\alpha(\theta)\right>=\left(e^{-\frac{|\alpha(\theta)|^2}{2}}e^{\alpha(\theta)\mathbf{a}^{\dag}}\right)\left|0\right>$ it follows that 
$$\frac{d}{d\theta}\left|\alpha(\theta)\right>=\left(\frac{d\alpha(\theta)}{d\theta}\mathbf{a}^{\dag}-|\alpha(\theta)|\frac{d|\alpha(\theta)|^2}{d\theta}\right)\left|\alpha(\theta)\right>.$$
Therefore, 
$$\left<\alpha(\theta)|\alpha(\theta)'\right>=  \overline{\alpha(\theta)}  \frac{d\alpha(\theta)}{d\theta}-|\alpha(\theta)|\frac{d|\alpha(\theta)|}{d\theta}        $$
and 
\begin{align*}
\left<\alpha(\theta)'|\alpha(\theta)'\right>&=\left(|\alpha|\frac{d|\alpha(\theta)|}{d\theta}\right)^2+(1+|\alpha(\theta)|^2)\left|\frac{d\alpha(\theta)}{d\theta}\right|^2\\&   -
|\alpha(\theta)|\frac{d|\alpha(\theta)|}{d\theta}\left(\alpha(\theta)\frac{d\overline{\alpha(\theta)}}{d\theta}+\overline{\alpha(\theta)}\frac{d{\alpha(\theta)}}{d\theta}\right).
\end{align*}
Finally, the result follows immediately from Eq. \eqref{qfiformula}.

%----------------------------------------------------------------------------------------
%	BIBLIOGRAPHY
%----------------------------------------------------------------------------------------

\printbibliography[heading=bibintoc]

%----------------------------------------------------------------------------------------

\end{document}